\documentclass[sigconf, nonacm]{acmart}

\usepackage{subcaption}
\usepackage[export]{adjustbox}
\usepackage{multirow}
\usepackage{xspace}
\usepackage{algorithmic}
\algsetup{linenosize=\small}
\usepackage[ruled,lined,linesnumbered,noend]{algorithm2e}
\usepackage{enumitem}
\usepackage{amsmath}
\usepackage{fancyvrb}
\usepackage{balance}
\usepackage{graphicx}
\usepackage[colorinlistoftodos, textsize=small]{todonotes}

\newcommand{\safeincludegraphics}[2][]{%
  \IfFileExists{#2}{%
    \includegraphics[#1]{#2}%
  }{%
    \fbox{\texttt{Missing image}}%
  }%
}

\DeclareMathOperator*{\argmax}{arg\,max}
\def\assign{ \text{:--} }

\newcommand{\db}[1]{\mathsf{db}(#1)}

\newcommand{\mi}[1]{\mathit{#1}}
\newcommand{\ins}[1]{\mathbf{#1}}
\newcommand{\adom}[1]{\mathsf{dom}(#1)}
\newcommand{\ra}{\rightarrow}

\newcommand{\dep}{\Sigma}
\newcommand{\sch}[1]{\mathsf{sch}(#1)}
\newcommand{\esch}[1]{\mathsf{edb}(#1)}
\newcommand{\isch}[1]{\mathsf{idb}(#1)}

\newcommand{\body}[1]{\mathsf{body}(#1)}
\newcommand{\head}[1]{\mathsf{head}(#1)}

\newcommand{\class}[1]{\mathsf{#1}}

\newcommand{\base}[1]{\mathsf{base}(#1)}

\newcommand{\diff}[2]{\Delta(#1,#2)}
\newcommand{\dist}[2]{d(#1,#2)}
\newcommand{\diversity}[1]{\delta(#1)}
\newcommand{\diversitytup}[1]{\delta(#1)}

\newcommand{\am}{\mathcal{AM}}

\newcommand{\support}[1]{\mathsf{support}(#1)}

\newcommand{\indegree}[2]{\mathsf{in}_{#2}(#1)}
\newcommand{\outdegree}[2]{\mathsf{out}_{#2}(#1)}
\newcommand{\inflow}[1]{\mathcal{F}^{\mathsf{in}}_{#1}}
\newcommand{\outflow}[1]{\mathcal{F}^{\mathsf{out}}_{#1}}

\newcommand{\arity}[1]{\mathsf{ar}(#1)}

\newcommand{\why}[3]{\mathsf{why}(#1,#2,#3)}
\newcommand{\whyk}[4]{\mathsf{why}^{#1}(#2,#3,#4)}

\newcommand{\starwhy}[3]{\mathsf{why}_{\star}(#1,#2,#3)}
\newcommand{\starwhyk}[4]{\mathsf{why}_{\star}^{#1}(#2,#3,#4)}
\newcommand{\allwhy}[3]{\mathsf{why}_{\mathsf{ALL}}(#1,#2,#3)}
\newcommand{\unwhy}[3]{\mathsf{why}_{\mathsf{UN}}(#1,#2,#3)}

\newcommand{\OMIT}[1]{}

\newcommand{\multiset}[1]{\{\!\{ #1 \}\!\}}
\newcommand{\card}[2]{|#2(#1)|}
\newcommand{\nodes}[1]{V_{#1}}
\newcommand{\hypergraphof}[1]{\mathcal{H}_{#1}}

\newcommand{\cliques}[2]{\mathsf{cliques}^{#1}(#2)}

\newcommand{\ilphypergraphfunc}{\mathsf{hyper\text{-}ILP}_{\mathsf{ALL}}}
\newcommand{\ilphypergraph}[3]{\ilphypergraphfunc(#1,#2,#3)}
\newcommand{\ilpunhypergraphfunc}{\mathsf{hyper\text{-}ILP}_{\mathsf{UN}}}
\newcommand{\ilpunhypergraph}[3]{\ilpunhypergraphfunc(#1,#2,#3)}

\newcommand{\ilpproofdagfunc}{\mathsf{graph\text{-}ILP}_{\mathsf{ALL}}}
\newcommand{\ilpproofdag}[3]{\ilpproofdagfunc(#1,#2,#3)}
\newcommand{\ilpunproofdagfunc}{\mathsf{graph\text{-}ILP}_{\mathsf{UN}}}
\newcommand{\ilpunproofdag}[3]{\ilpunproofdagfunc(#1,#2,#3)}

\newcommand{\ilpstargammafunc}{\mathsf{\gamma\text{-}ILP}_{\star}}
\newcommand{\ilpstargamma}[3]{\ilpstargammafunc(#1,#2,#3)}

\newcommand{\ComputeDiverse}[1]{\mathsf{CompDiverse}[#1]}
\newcommand{\ConstructILPEncoding}{\mathsf{BuildILPEnc}}
	
\newcommand{\boundall}[3]{\mathsf{bound}_{\mathsf{ALL}}(#1,#2,#3)}
\newcommand{\boundun}[3]{\mathsf{bound}_{\mathsf{UN}}(#1,#2,#3)}

\newcommand{\sol}[1]{\mathsf{sol}(#1)}
\newcommand{\opt}[1]{\mathsf{opt}(#1)}

\newcommand{\copyof}[2]{#1^{(#2)}}
\newcommand{\auxvar}{\mathsf{aux}}
\newcommand{\hammd}[2]{\mathsf{hd}(#1,#2)}
\newcommand{\ansk}[3]{\mathsf{ans}^{#1}(#2,#3)}

\newcommand{\DAT}{\class{Dat}}

\def\eqtree{\approx}
\newcommand{\downof}[2]{#1_{\downarrow#2}}
\newcommand{\gri}[2]{\mathsf{gri}(#1,#2)}
\newcommand{\downc}[3]{\mathsf{down}(#1,#2,#3)}

\def\qed{\hfill{\qedboxempty}      % qed with empty box
  \ifdim\lastskip<\medskipamount \removelastskip\penalty55\medskip\fi}

\def\qedboxempty{\vbox{\hrule\hbox{\vrule\kern3pt
                 \vbox{\kern3pt\kern3pt}\kern3pt\vrule}\hrule}}

\def\qedfull{\hfill{\qedboxfull}   % qed with full box
  \ifdim\lastskip<\medskipamount \removelastskip\penalty55\medskip\fi}

\def\qedboxfull{\vrule height 4pt width 4pt depth 0pt}

\newcommand{\markfull}{\qedboxfull}

\newcommand{\plotlabel}[1]{\texttt{#1}\xspace}

\newcommand{\csda}{\plotlabel{csda}}
\newcommand{\galen}{\plotlabel{galen}}
\newcommand{\doctors}{\plotlabel{doctors}}
\newcommand{\graphanalytics}{\plotlabel{graphanalytics}}

\newcommand{\httpd}{\plotlabel{httpd}}
\newcommand{\postgres}{\plotlabel{postgresql}}
\newcommand{\linux}{\plotlabel{linux}}
\newcommand{\facebook}{\plotlabel{facebook}}
\newcommand{\bitcoin}{\plotlabel{bitcoin}}
\newcommand{\factsk}{\plotlabel{100k}}
\newcommand{\fifteen}{\plotlabel{D15}}
\newcommand{\twentyfive}{\plotlabel{D25}}
\newcommand{\forty}{\plotlabel{D40}}
\newcommand{\fifty}{\plotlabel{D50}}

\begin{document}
\title{Diverse Why-Provenance for Datalog Queries}

%%
%% The "author" command and its associated commands are used to define the authors and their affiliations.
\author{Marco Calautti}
\affiliation{%
	\institution{University of Milan}
	\country{Italy}
}
\email{marco.calautti@unimi.it}

\author{Mostafa Milani}
\affiliation{%
	\institution{University of Western Ontario}
	\country{Canada}
}
\email{mostafa.milani@uwo.ca}

\author{Andreas Pieris}
\affiliation{%
	\institution{University of Cyprus}
	\country{Cyprus}
}
\affiliation{%
	\institution{University of Edinburgh}
	\country{United Kingdom}
}
\email{apieris@inf.ed.ac.uk}

\newtheorem{observation}{Observation}
%%
%% The abstract is a short summary of the work to be presented in the
%% article.
\begin{abstract}
Datalog is a powerful rule-based language that allows us to express complex recursive queries and has found numerous applications over the years. Explaining why a result to a Datalog query is obtained is an essential task towards explainable and transparent data-intensive applications that rely on Datalog. A standard way of explaining a query result is the so-called why-provenance, which provides information about the witnesses to
a query result in the form of subsets of the input database that as a whole can be used to derive that result.
The state-of-the-art approach for efficiently computing the why-provenance explanations of a certain query result aims at enumerating, without repetitions, the relevant explanations by relying on off-the-shelf SAT solvers.
However, whenever the set of explanations is very large, outputting the entire set might be prohibitively expensive. Moreover, it is questionable the practical usefulness of flooding the user with a very large number of explanations, many of them possibly differing only marginally. 
In this case, one can argue that it would be practically more useful to compute a small, sufficiently diverse subset of the set of why-provenance explanations.
The present work initiates the study of the latter problem. In particular, we formalize the problem of interest as an optimization problem, and, after showing that it is highly intractable (it is NP-hard even in data complexity), we investigate whether ILP optimization solvers can be employed towards practical solutions.
Our main finding is that this is achievable as long as we keep the requested set of explanations small, or we employ approximation algorithms (with error guarantees).
\end{abstract}

\maketitle

%%% do not modify the following VLDB block %%
%%% VLDB block start %%%
% \pagestyle{\vldbpagestyle}
% \begingroup\small\noindent\raggedright\textbf{PVLDB Reference Format:}\\
% \vldbauthors. \vldbtitle. PVLDB, \vldbvolume(\vldbissue): \vldbpages, \vldbyear.\\
% \href{https://doi.org/\vldbdoi}{doi:\vldbdoi}
% \endgroup
% \begingroup
% \renewcommand\thefootnote{}\footnote{\noindent
% This work is licensed under the Creative Commons BY-NC-ND 4.0 International License. Visit \url{https://creativecommons.org/licenses/by-nc-nd/4.0/} to view a copy of this license. For any use beyond those covered by this license, obtain permission by emailing \href{mailto:info@vldb.org}{info@vldb.org}. Copyright is held by the owner/author(s). Publication rights licensed to the VLDB Endowment. \\
% \raggedright Proceedings of the VLDB Endowment, Vol. \vldbvolume, No. \vldbissue\ %
% ISSN 2150-8097. \\
% \href{https://doi.org/\vldbdoi}{doi:\vldbdoi} \\
% }\addtocounter{footnote}{-1}\endgroup
%%% VLDB block end %%%

%%% do not modify the following VLDB block %%
%%% VLDB block start %%%
% \ifdefempty{\vldbavailabilityurl}{}{
% \vspace{.3cm}
% \begingroup\small\noindent\raggedright\textbf{PVLDB Artifact Availability:}\\
% \textcolor{red}{The source code, data, and/or other artifacts have been made available at \url{\vldbavailabilityurl}.}
% \endgroup
% }
%%% VLDB block end %%%
\section{Introduction}\label{sec:introduction}

Datalog has emerged in the 1980s as a logic-based query language from Logic Programming and has been extensively
studied since then~\cite{AbHV95}. It essentially extends the language of unions of conjunctive queries, which corresponds to the select-project-join-union fragment of relational algebra, with the important feature of recursion, much needed to express some natural queries.
As for any other query language, explaining why a result to a Datalog query is obtained is crucial
towards explainable and transparent data-intensive applications. A standard way for providing such explanations to query answers is, among others, the so-called why-provenance~\cite{BuKT01}. Its essence
is to collect all the subsets of the input database that as a whole can be used to derive a certain answer. More precisely, in the case of Datalog queries, the why-provenance of an answer tuple $\bar t$ is obtained by considering all the possible proof trees $T$ of the fact $\text{\rm Ans}(\bar t)$, with $\text{\rm Ans}$ being the
answer predicate of the Datalog query in question, and then collecting all the database facts that label the leaves of $T$. Recall that a proof tree of a fact $\alpha$ w.r.t.~a database $D$ and a set $\dep$ of Datalog rules forms a tree-like representation of a derivation process that, starting from $D$ and executing the rules occurring in $\dep$, produces $\alpha$~\cite{AbHV95}.

Despite its wide acceptance in theory with numerous theoretical studies (see, e.g.,~\cite{CLPS24,CalauttiLPS24b,DaAA13,DMRT14,Khamis0PSW22}), why-provenance for Datalog queries suffers from a crucial weaknesses, namely it is computationally very expensive. This is revealed by the fact that until 2020, although why-provenance for Datalog queries has been proposed decades ago, only a couple of works have considered implementing it for recursive queries~\cite{ZhSS20,EsLS14}.
The more recent work~\cite{ElKM22} by Elhalawati et al. tried to make further progress by focusing on the more practical setting of computing the why-provenance of a given query answer (a.k.a. on-demand why-provenance), instead of computing the why-provenance for all the query
answers as done in~\cite{ZhSS20,EsLS14}.
Inspired by~\cite{ElKM22} and the idea of on-demand why-provenance, the authors of~\cite{CLPS24aaai} pushed the state of affairs even further by focusing on the task of enumerating, without repetitions, the members of why-provenance of a given answer by relying on off-the-shelf SAT solvers. In other words, the approach of~\cite{CLPS24aaai} builds the why-provenance of a query answer in an incremental fashion, i.e., one explanation at a time, which is much more useful in practice than the one-shot computation of the whole set as done in~\cite{ElKM22}.
Furthermore, the authors of~\cite{CLPS24aaai} attacked another weakness of why-provenanance for Datalog queries, namely it might provide counterintuitive explanations. Indeed, as it has been observed in the literature, there are proof trees that correspond to unnatural derivation processes, e.g., derivations where an atom is
derived from itself~\cite{BBPT22}, and it is fair to say that an explanation witnessed via such an unnatural proof tree can be classified as a counterintuitive one as it does not correspond to an intuitive derivation process that can be extracted from the proof tree; this is further discussed in Section~\ref{sec:why-provenance}. Towards conceptually meaningful explanations, the authors of~\cite{CLPS24aaai} introduced the class of unambiguous proof trees that, unlike arbitrary proof trees, correspond to conceptually meaningful derivation processes, which in turn ensure the conceptual relevance of the obtained explanations. The essence underlying unambiguous proof trees is that they eliminate any ambiguity on how a certain fact $\alpha$ is derived, namely $\alpha$ is derived in a unique way that can be exctracted from its unambiguous proof tree.

To the best of our knowledge, the approach of~\cite{CLPS24aaai} forms the state-of-the-art concerning the effective compution of on-demand why-provenance, with the experimental evaluation of~\cite{CLPS24aaai} providing very encouraging results. Having said that, it remains questionable the practical usefuleness of flooding the user with an extremely large number of why-provenance explanations of the answer of interest, many of them possibly differing only marginally. One can argue that whenever the set of why-provenance explanations of an answer gets very large, it would be practically more useful to provide an overview by outputting a ``representative'' subset of the set of explanations. A general discussion on this matter about arbitary problems (not only about on-demand why-provenance) can be found in~\cite{MePS25}. Two ways for computing a ``representative'' subset of the set of why-provenance explanations are the following:

\begin{description}
\item[Sampling Approach.] Randomly select explanations, similarly to what was done for the problem of computing answers to conjunctive queries~\cite{ACJR21,CZBCKS22}.

\item[Diversity Approach.] Conpute a small ``diverse'' subset of the set of explanations, similarly to what was done for the problem of computing answers to conjunctive queries~\cite{AMPR25,MePS25}.
\end{description}

The present work initiates the study of the problem of computing a ``representative'' subset of the set of why-provenance explanations by adopting the diversity approach. Let us stress that the diversity approach has been applied to numerous problems related to databases (see, e.g.,~\cite{AMPR25,DeFa14,FHXZ24,MePS25}) and symbolic AI (see, e.g.,~\cite{EEEF13,HHOW05,IBST20,PeTr15}).
Following what was done in the literature in other contexts~\cite{IBST20,MePS25}, to properly define our problem of interest, we first need to fix the notion of {\em distance} between any two why-provenance explanations and the notion of {\em diversity measure} that acts as an aggregator that combines pairwise distances.
Having those two notions in place, our problem is naturally defined as an optimization problem that aims to compute a set $\{D_1,\ldots,D_k\}$ of $k$ why-provenance explanations (if such a set exists), for some $k > 0$ specifed by the user, that maximizes the value obtained by aggregating, using the diversty measure, the pairwise distances of the computed explanations $D_1,\ldots,D_k$.

Concerning the distance between explanations, we adopt the symmetric difference distance, that is, the number of facts on which two explanations $D'$ and $D''$ differ, denoted by $d(D',D'')$. Recall that a why-provenance explanation is a subset of the given database (and thus, a set of facts), which means that we can refer to the symmetric difference of $D'$ and $D''$. In other words, $d(D',D'')$ is the cardinality of the set of facts $(D' \setminus D'') \cup (D'' \setminus D')$.
When it comes to the diversity measure, in general, one can adopt any computable function $f$ of the form $\bigcup_{k > 0} \mathbb{N}^{\frac{k(k-1)}{2}}\ra \mathbb{R}$, that is, a function that, for each $k>0$, assigns to each tuple of integers $\left(d_1,\ldots,d_{\frac{k(k-1)}{2}}\right)$, which are essentially the pairwise distances of $k$ why-provenance explanations, a real number. Of course, towards practical solutions to our problem of interest, the function $f$ should be conceptually relevant and efficiently computable (i.e., computable in polynomial time). Concrete such functions that have been considered in the literature are, for example, the sum and min functions that assign to a tuple of integers  $\left(d_1,\ldots,d_{\frac{k(k-1)}{2}}\right)$ the integer $d_1 + \cdots +d_{\frac{k(k-1)}{2}}$ and $\min \left\{d_1,\ldots,d_{\frac{k(k-1)}{2}}\right\}$, respectively~\cite{MePS25}. In this work, we adopt the sum function as one of the most commonly used diversity measures.
We can now give a semi-formal definition of the optimization problem of interest; the formal definition is given in Section~\ref{sec:diverse-why-provenance}:

\medskip

\begin{center}
	\fbox{\begin{tabular}{ll}
			{\small PROBLEM} : & $\mathsf{DiverseWhyProvenance}$
			\\[2mm]
			{\small INPUT} : & A Datalog query $Q$, a database $D$,\\
			& an answer ${\bar t}$ to $Q$ over $D$, and an integer $k > 0$\\[2mm]
			{\small OUTPUT} : & A set $\{D_1,\ldots,D_k\}$ of why-provenance\\ & explanations of $\bar t$ w.r.t.~$D$ and $Q$ that maximizes\\
			& $\sum_{1 \leq i < j \leq k} d(D_i,D_j)$ if it exists; otherwise \textbf{fail} 
	\end{tabular}}
\end{center}

\medskip

The goal of the present work is to study the above problem with the aim of providing effective solutions that can be used in practice. For the sake of completeness, we are interested in both settings where the explanations are obtained from arbitary or unambiguous proof trees. Recall that unambiguous proof trees have been introduced in~\cite{CLPS24aaai} and correspond to conceptually meaningful derivations processes, which in turn lead to conceptually relevant explanations.
The first research question that comes up concerns the complexity of $\mathsf{DiverseWhyProvenance}$:

\medskip

\noindent \textit{\textbf{Question 1.} Is $\mathsf{DiverseWhyProvenance}$ tractable or intractable?}

\medskip

Our theoretical analysis reveals that we are dealing with a highly intractable optimization problem. In particular, we show that our problem is NP-hard, no matter which class of proof trees we consider (arbitrary or unambiguous), even if the Datalog query is fixed (i.e., in data complexity) and consists of a single Datalog rule.
This suggests that for providing practical solutions to our problem we need to exploit existing optimization methods that are able to deal with NP-hard problems. A natural choice is Integer Linear Programming (ILP), which leads to our second research question:

\medskip

\noindent \textit{\textbf{Question 2.} Can ILP optimization solvers lead to practical solutions for $\mathsf{DiverseWhyProvenance}$?}

\medskip

The key idea towards an answer to the above question is to cast $\mathsf{DiverseWhyProvenance}$ as an ILP optimization problem $\Pi$ in such a way that each diverse set of interest can be exctracted from an optimal solution of $\Pi$.
Since there are mutliple ways of achieveing this, we define a general framework that allows us to devise such ILP optimization problems in a uniform way, and present a generic algorithm that exploits this framework for solving our problem.
The main technical ingredient of this framework is the notion of ILP encoding of the why-provenance explanations of a query answer, which is essentially a system $T$ of ILP constraints such that each why-provenance explanation of the answer of interest can be extracted from a solution of $T$.
Given an ILP encoding $T$ of some query answer and an integer $k > 0$, we then define the ILP optimization problem $\Pi[T,k]$ such that each diverse set of $k$ why-provenance explanations of the answer to be explained can be extracted from an optimal solution of $\Pi[T,k]$.
Thus, any pair $(T,k)$, where $T$ is an ILP encoding of some answer and $k > 0$ some  integer, immediately leads to the ILP optimization problem $\Pi[T,k]$.
Now, once we fix our preferred way of building the ILP encoding, we immediately get an algorithm that solves $\mathsf{DiverseWhyProvenance}$. In particular, given a Datalog query $Q$, a database $D$, an answer $\bar t$ to $Q$ over $D$, and an integer $k$, the algorithm 
(i) constructs the ILP optimization problem $\Pi[T,k]$, where $T$ is build according to the fixed preferred way, (ii) finds an optimal solution $\mu$ of $\Pi[T,k]$ if it exists (otherwise, it fails) by calling an off-the-shelf ILP optimization solver, and (iii) returns the diverse set of $k$ why-provenance explanations of $\bar t$ that can be exctracted from the solution $\mu$ of $\Pi[T,k]$.

After experimenting with different ILP encodings, we have observed that providing practical solutions to $\mathsf{DiverseWhyProvenance}$ is attainable, no matter whether the explanations are obtained from arbitary or unambiguous proof trees, as long as we ask for small diverse sets of explanations (i.e., when the diversity parameter $k$ is at most three). 
On the other hand, all the considered ILP encodings lead to algorithms that exhibit poor scalability w.r.t.~the diversity parameter $k$, and thus, quickly become infeasible for computing larger diverse sets of explanations. This brings us to our last research question whether we can devise alternative approaches that trade optimality of the computed diverse sets for performamce:

\medskip

\noindent \textit{\textbf{Question 3.} Can we devise approximation algorithms (with error guarantees) for $\mathsf{DiverseWhyProvenance}$ that perform much faster than the exact algorithms and can deal with larger diverse sets?}

\medskip

Towards this direction, we target $2$-approximation algorithms for $\mathsf{DiverseWhyProvenance}$ that accept as input a Datalog query $Q$, a database $D$, an answer $\bar t$ to $Q$ over $D$, and an integer $k>0$, and produce a set $\{D'_1,\ldots,D'_k\}$ of why-provenance explanations of $\bar t$ (assuming it exists; otherwise, the algorithm fails) such that
\[
\sum_{1 \leq i < j \leq k} d(D'_i,D'_j)\ \geq\ \frac{1}{2} \cdot \sum_{1 \leq i < j \leq k} d(D_i,D_j),
\]
where $\{D_1,\ldots,D_k\}$ is a set of explanations with maximal diversity, i.e., it maximizes $\sum_{1 \leq i < j \leq k} d(D_i,D_j)$.
Similarly to the exact solutions discussed above, we use the generic framework for constructing ILP optimization problems, which in turn are based on the notion of ILP encoding of the why-provenance explanations of a query answer, to devise a generic $2$-approximation algorithm for $\mathsf{DiverseWhyProvenance}$.
After experimenting with the various ILP encodings that underly our exact algorithms, we have concluded that the approximation algorithms are indeed much faster than the exact ones with the runtime being reduced by several orders of magnitude. Even as the diversity parameter $k$ increases, runtimes remain low, with the vast majority of runs completing well below one second, which is a significant improvement over the exact algorithms.
We have also assessed the quality of the approximation algorithms and observed that the approximation ratios remain quite close to one. This tells us that the approximation ratios achieved in practice are generally much better than the theoretical ones.

\medskip

\noindent 
\textbf{Roadmap.} After introducing some preliminary notions in Section~\ref{sec:preliminaries} and the basics of why-provenace for Datalog in Section~\ref{sec:why-provenance}, we formally define the optimization problem of interest and analyze its data complexity in Section~\ref{sec:diverse-why-provenance}. The general framework for devising ILP optimization problems, as well as the generic algorithm for $\mathsf{DiverseWhyProvenance}$, are presented in Section~\ref{sec:exact-encoding}. The various ILP encodings of why-provenance are given in Section~\ref{sec:characterizations}. The experimental evaluation of the exact algorithms is performed and discussed in Section~\ref{sec:exact-experiments}. The generic $2$-approximation algorithm for $\mathsf{DiverseWhyProvenance}$ is presented in Section~\ref{sec:approximations}. The experimental evaluation of the approximation algorithms is performed and discussed in Section~\ref{sec:approximate-experiments}. We finally conclude in Section~\ref{sec:conclusions}.

\medskip

\noindent \emph{The source code and datasets used in our experimental evaluation can be found at https://github.com/mostafamilani/datalog-div-prov.}
\section{Preliminaries}\label{sec:preliminaries}

We consider the disjoint countably infinite sets $\ins{C}$ and $\ins{V}$ of {\em constants} and {\em variables}, respectively. We may refer to constants and variables as {\em terms}. For an integer $n > 0$, we write $[n]$ for the set of integers $\{1,\ldots,n\}$. In what follows, we adopt the terminology and notation from the recent work~\cite{CLPS24}, where the data complexity of why-provenance for Datalog queries has been studied.

\medskip

\noindent 
\textbf{Relational Databases.} A {\em schema} $\ins{S}$ is a finite set of relation names (or predicates) with associated arity. We write $R/n$ to say that $R$ has arity $n \ge 0$; we may also write $\arity{R}$ for $n$.
A {\em (relational) atom} $\alpha$ over $\ins{S}$ is an expression of the form $R(\bar t)$, where $R/n \in \ins{S}$ and $\bar t$ is an $n$-tuple of terms. By abuse of notation, we may treat tuples as the \emph{set} of their elements. A {\em fact} is an atom with only constants.
A {\em database} over $\ins{S}$ is a finite set of facts over $\ins{S}$. The {\em active domain} of a database $D$, denoted $\adom{D}$, is the set of constants in $D$.

\medskip
\noindent
\textbf{Syntax of Datalog Programs.} A {\em (Datalog) rule} $\sigma$ over a schema $\ins{S}$ is an expression of the form
\[
R_0(\bar x_0)\,\,\, \assign\,\,\, R_1(\bar x_1),\ldots,R_n(\bar x_n)
\]
for $n \geq 1$, where $R_i(\bar x_i)$ is a (constant-free) relational atom over $\ins{S}$ for $i \in \{0,\ldots,n\}$, and each variable in $\bar x_0$ occurs in $\bar x_k$ for some $k \in [n]$.
We refer to $R_0(\bar x_0)$ as the {\em head} of $\sigma$, denoted $\head{\sigma}$, and to the expression $R_1(\bar x_1),\ldots,R_n(\bar x_n)$ as the {\em body} of $\sigma$, denoted $\body{\sigma}$, which we may treat as the set of its atoms.
A {\em Datalog program} over a schema $\ins{S}$ is a finite set $\dep$ of Datalog rules over $\ins{S}$. A predicate $R$ occurring in $\dep$ is called {\em extensional} if there is no rule in $\dep$ having $R$ in its head, and {\em intentional} if there is at least one rule in $\dep$ with $R$ in its head.
The {\em extensional (database) schema} of $\dep$, denoted $\esch{\dep}$, is the set of all extensional predicates in $\dep$, while the {\em intentional
schema} of $\dep$, denoted $\isch{\dep}$, is the set of all intensional predicates in $\dep$. Note that, by definition, $\esch{\dep} \cap \isch{\dep} = \emptyset$. The {\em schema} of $\dep$, denoted $\sch{\dep}$, is the set $\esch{\dep} \cup \isch{\dep}$, which is in general a subset of $\ins{S}$ since some predicates of $\ins{S}$ may not appear in $\dep$.

\medskip
\noindent
\textbf{Semantics of Datalog Programs.} An interesting property of Datalog programs is that their semantics can be equivalently defined in three different ways: model-theoretically, via a fixpoint procedure, and proof-theoretically; see, e.g.,~\cite{AbHV95}. We proceed to recall the proof-theoretic semantics of Datalog programs. To this end, we need the key notion of proof tree of a fact, which will anyway play a crucial role in our work.
For a database $D$ and a Datalog program $\dep$, let $\base{D,\dep}$ be the set that collects 
%= \{R(\bar t) \mid R \in \sch{\dep} \text{ and } \bar t \in \adom{D}^{\arity{R}}\}$ that collects 
all the facts that can be formed using predicates of $\sch{\dep}$ and constants of $\adom{D}$.

\begin{definition}[\textbf{Proof Tree}]\label{def:proof-tree}
	Consider a Datalog program $\dep$, a database $D$ over $\esch{\dep}$, and a fact $\alpha$ over $\sch{\dep}$. 
	A {\em proof tree of $\alpha$ w.r.t.~$D$ and $\dep$} is a finite labeled rooted tree $T=(V,E,\lambda)$, with $\lambda : V \ra \base{D,\dep}$, such that the following hold:
	\begin{enumerate}
		\item If $v \in V$ is the root, then $\lambda(v) = \alpha$.
		
		\item If $v \in V$ is a leaf, then $\lambda(v) \in D$.
		
		\item If $v \in V$ has $n\ge1$ children $u_1,\ldots,u_n$, then there exists a rule $R_0(\bar x_0)\ \assign\ R_1(\bar x_1),\ldots,R_n(\bar x_n) \in \dep$ and a function $h : \bigcup_{i \in [n]} \bar x_i \ra \ins{C}$ such that $\lambda(v) = R_0(h(\bar x_0))$ and, for each $i \in [n]$, $\lambda(u_i) = R_i(h(\bar x_i))$. \hfill\markfull
	\end{enumerate} 
\end{definition}

Essentially, a proof tree of a fact $\alpha$ w.r.t.~$D$ and $\dep$ indicates that we can prove $\alpha$ using $D$ and $\dep$, i.e., we can derive $\alpha$ starting from $D$ end executing the rules of $\dep$. Here is a simple example from~\cite{CLPS24}.

\begin{example}\label{exa:proof-tree}
	Consider the Datalog program $\dep$ consisting of
	\[
	A(x)\ \assign\ S(x) \quad \text{ and } \quad A(x)\ \assign\ A(y), A(z),T(y,z,x)
	\]
	that encodes the {\em path accessibility problem}~\cite{Cook74}. The predicate $S$ represents source nodes, $A$ represents nodes that are accessible from the source nodes, and $T$ represents accessibility conditions, that is, $T(y,z,x)$ means that if both $y$ and $z$ are accessible from the source nodes, then so is $x$.
	Let also
	\[
	D\ =\ \{S(a), T(a,a,b),T(a,a,c), T(a,a,d), T(b,c,a)\}.
	\]
	A proof tree of the fact $A(d)$ w.r.t.~$D$ and $\dep$ follows:
	
	\centerline{\includegraphics[width=.24\textwidth]{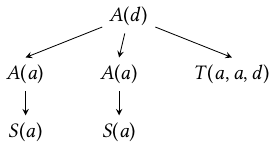}}
	
	\noindent The following is another, slightly more complex, proof tree of the fact $A(d)$ w.r.t.~$D$ and $\dep$:
	
	\centerline{
	\includegraphics[width=.46\textwidth]{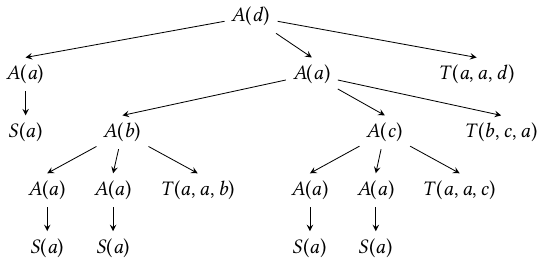}}

	\noindent Note that the above are only two out of the many proof trees of $A(d)$ w.r.t.~$D$ and $\dep$. In fact, there are infinitely many as one can build larger and larger such proof trees: whenever we encounter a node labeled by $A(a)$, we can choose to apply the recursive rule $A(x)\ \assign\ A(y), A(z),T(y,z,x)$. \hfill\markfull
\end{example}

Now, given a Datalog program $\dep$ and a database $D$ over $\sch{\dep}$, the {\em semantics of $\dep$ on $D$} is the set
\[
\dep(D)\ =\ \{\alpha \mid \text{ there is a proof tree of } \alpha \text{ w.r.t. } D \text{ and } \dep\},
\]
that is, the set of facts that can be proven using $D$ and $\dep$.

\medskip
\noindent
\textbf{Datalog Queries.} Having the syntax and the semantics of Datalog programs in place, it is now straightforward to recall the syntax and the semantics of Datalog queries.
A {\em Datalog query} is a pair $Q = (\dep,R)$, where $\dep$ is a Datalog program and $R$ a predicate of $\isch{\dep}$.
Now, for a database $D$ over $\esch{\dep}$, the {\em answer} to $Q$ over $D$ is defined as the set of tuples
\[
Q(D)\ =\ \left\{\bar t \in \adom{D}^{\arity{R}} \mid R(\bar t) \in \dep(D)\right\},
\]
i.e., all those tuples $\bar t$ such that the fact $R(\bar t)$ can be proven using $D$ and $\dep$. The class of all the Datalog queries is denoted $\DAT$. 
\section{Why-Provenance for Datalog}\label{sec:why-provenance}

Why-provenance is a standard way of explaining why a query result is obtained. It essentially collects all the subsets of the database that can be used as a whole to prove (or derive) a query result. We proceed to recall the details following~\cite{CLPS24}.
Given a proof tree $T = (V,E,\lambda)$ (of some fact w.r.t.~some database and Datalog program), the {\em support} of $T$ is the family of sets of atoms
\[
\support{T}\ =\ \left\{\lambda(v) \mid v \in V \text{ is a leaf of } T\right\},
\] 
i.e., the set of facts that label the leaves of the proof tree $T$. Note that $\support{T}$ is a subset of the underlying database since, by definition, the leaves of a proof tree are labeled with database atoms.
The formal definition of why-provenance follows.

\begin{definition}[\textbf{Why-Provenance for Datalog}]\label{def:why-provenance}
	Consider a Datalog query  $Q = (\dep,R)$, a database $D$ over $\esch{\dep}$, and a tuple $\bar t \in \adom{D}^{\arity{R}}$. The {\em why-provenance of $\bar t$ w.r.t.~$D$ and $Q$} is defined as the family of sets of facts 
	\[
	\{\support{T} \mid T \text{ is a proof tree of } R(\bar t) \text{ w.r.t. } D \text{ and } \dep\},
	\]
	which we denote by $\why{\bar t}{D}{Q}$. \hfill\markfull
\end{definition}

Intuitively, a set of facts $D' \subseteq D$ from $\why{\bar t}{D}{Q}$ should be understood as a ``real'' reason why the tuple $\bar t$ is an answer to the query $Q$ over $D$, i.e., $D'$ explains why $\bar t \in Q(D)$. By ``real'' we mean that all the facts of $D'$ are really used in order to derive the tuple $\bar t$ as an answer. Here is a simple example of why-provenance.

\begin{example}\label{exa:why-provenance}
	Let $Q = (\dep,A)$, where $\dep$ is the program that encodes the path accessibility problem as in Example~\ref{exa:proof-tree}, and let $D$ be the database from Example~\ref{exa:proof-tree}. 
	It is easy to verify that the why-provenance of the unary tuple $(d)$ w.r.t.~$D$ and $Q$ consists of $\{S(a), T(a,a,d)\}$ and the database $D$ itself. The former set is actually the support of the first proof tree given in Example~\ref{exa:proof-tree}, while $D$ is the support of the second proof tree. \hfill\markfull
\end{example}

\medskip
\noindent
\textbf{Refined Why-Provenance.} Although the notion of proof tree introduced above is the standard one used to define the notion of why-provennace for Datalog queries, it has been observed in the literature (see, e.g.,~\cite{BBPT22,CLPS24,CLPS24aaai}) that there are counterintuitive proof trees, which in turn lead to counterintuitive explanations based on why-provenance. Such a proof tree is the second one in Example~\ref{exa:proof-tree} that suffers from ``ambiguity''. In particular, there are two nodes labeled with $A(a)$, but their subtrees differ, which means that, according to the proof tree in question, it is ambiguous how $A(a)$ is derived starting from the database $D$. The above observation led to the class of unambiguous proof trees, which essentially forces a fact to be consistently proven via a certain derivation~\cite{CLPS24aaai}; 
see, e.g., the first proof tree in Example~\ref{exa:proof-tree}. To formalize this simple idea we need the notion of isomorphism among rooted trees.

Two rooted trees $T=(V,E,\lambda)$ and $T' = (V',E',\lambda')$ are {\em isomorphic}, denoted $T \eqtree T'$, if there is a bijection $h : V \ra V'$ such that, for each $v \in V$, $\lambda(v) = \lambda'(h(v))$, and for each $u,v \in V$, $(u,v) \in E$ iff $(h(u),h(v)) \in E'$. Let $T[v]$ be the subtree of the proof tree $T$ rooted at node $v$. The definition of unambiguous proof trees follows.

\begin{definition}[\textbf{Unambiguous Proof Tree}]\label{def:unambiguous-proof-tree}
	Consider a Datalog program $\dep$, a database $D$ over $\esch{\dep}$, and a fact $\alpha$. An {\em unambiguous proof tree of $\alpha$ w.r.t.~$D$ and $\dep$} is a proof tree $T = (V,E,\lambda)$ of $\alpha$ w.r.t.~$D$ and $\dep$ such that, for all $v,u \in V$, $\lambda(v) = \lambda(u)$ implies $T[v] \eqtree T[u]$. \hfill\markfull
\end{definition}

Having the notion of unambiguous proof tree in place, we can then refine the notion of why-provenance for Datalog queries by considering only unambiguous proof trees; the formal definition is omitted since it is identical to Definition~\ref{def:why-provenance} with the difference that we only collect the supports of unambiguous proof trees.
Given a Datalog query  $Q = (\dep,R)$, a database $D$ over $\esch{\dep}$, and a tuple $\bar t \in \adom{D}^{\arity{R}}$, we refer to the {\em why-provenance of $\bar t$ w.r.t.~$D$ and $Q$ relative to unambiguous proof trees}, denoted $\unwhy{\bar t}{D}{Q}$.
To be able to treat in a uniform way the notion of why-provenance relative to arbitrary proof trees (as given in Definition~\ref{def:why-provenance}) and the notion of why-provenance relative to unambiguous proof trees, in the rest of the paper we will write $\allwhy{\bar t}{D}{Q}$ instead of $\why{\bar t}{D}{Q}$ for the why-provenance of $\bar t$ w.r.t.~$D$ and $Q$.

At this point, let us clarify that why-provenance relative to unambiguous proof trees is indeed a refinement of why-provenance (relative to arbitrary proof trees). Indeed, given a Datalog query $Q = (\dep,R)$, a database $D$ over $\esch{\dep}$, and a tuple $\bar t \in \adom{D}^{\arity{R}}$, we have $\unwhy{\bar t}{D}{Q} \subseteq \allwhy{\bar t}{D}{Q}$.
Let us stress, however, that this is not the case when it comes to the semantics of Datalog. In other words, whenever there is a proof tree of a fact $\alpha$ w.r.t.~$D$ and $\dep$, then there is an unambiguous one. This is a desirable property since it implies the following: if $\bar t \in Q(D)$, then $\unwhy{\bar t}{D}{Q} \neq \emptyset$, which is what we expect from an explainability notion for Datalog.
\section{Diverse Why-Provenance}
\label{sec:diverse-why-provenance}

As discussed in the Introduction, our goal is, given a Datalog query $Q = (\dep,R)$, a database $D$ over $\esch{\dep}$, and a tuple ${\bar t} \in \adom{D}^{\arity{R}}$, to compute a set of $k > 0$ diverse explanations of why $\bar t \in Q(D)$ according to the notion of why-provenance (relative to arbitrary or unambiguous proof trees) introduced in the previous section. 
To formalize this task, we transfer the general diversity framework for query answers from~\cite{AMPR24} to the setting of why-provenance.

Consider a non-empty subset $\mathcal{S}$ of $\starwhy{\bar t}{D}{Q}$, where $\star \in \{\mathsf{ALL},\mathsf{UN}\}$, if it exists. To determine the diversity of $\mathcal{S}$, we first determine how different the pairs of databases in $\mathcal{S}$ are. This is done by using their symmetric difference. For two databases $D_1,D_2 \in \mathcal{S}$, the \emph{symmetric difference between $D_1$ and $D_2$}, denoted $\diff{D_1}{D_2}$, is the set $(D_1 \setminus D_2) \cup (D_2 \setminus D_1)$. We then define the \emph{distance between $D_1$ and $D_2$}, written $\dist{D_1}{D_2}$, as the integer $|\diff{D_1}{D_2}|$. Finally, the \emph{diversity of $\mathcal{S} = \{D_1,\ldots,D_k\}$}, denoted $\diversity{\mathcal{S}}$, is defined as%
\[
\sum\limits_{i,j \in [k] \text{ and }i<j} \dist{D_i}{D_j}.
\]

We are now ready to introduce the main problems of interest. For a Datalog query $Q = (\dep,R)$, a database $D$ over $\esch{\dep}$, a tuple $\bar t \in \adom{D}^{\arity{R}}$, and an integer $k > 0$, we write $\starwhyk{k}{\bar t}{D}{Q}$, for $\star \in \{\mathsf{ALL},\mathsf{UN}\}$, to denote the collection of all sets $\mathcal{S} \subseteq \starwhy{\bar t}{D}{Q}$ with $|\mathcal{S}| = k$, i.e., each element of $\starwhyk{k}{\bar t}{D}{Q}$ contains exactly $k$ members of $\starwhy{\bar t}{D}{Q}$. 
The main problems of interest are optimization problems that essentially ask to compute a set of $\starwhyk{k}{\bar t}{D}{Q}$ with maximal diversity:

% \medskip

\begin{center}
	\fbox{\begin{tabular}{ll}
			{\small PROBLEM} : & $\star\text{-}\mathsf{DiverseWhyProvenance}$
			\\[2mm]
			{\small INPUT} : & A Datalog query $Q = (\dep,R)$, \\
			& a database $D$ over $\esch{\dep}$,\\
			& a tuple ${\bar t} \in \adom{D}^{\arity{R}}$, and\\
			& an integer $k > 0$\\[2mm]
			{\small OUTPUT} : & A member of $\argmax\limits_{\mathcal{S} \in \starwhyk{k}{\bar t}{D}{Q}} \diversity{\mathcal{S}}$ if it exists;\\
			&  otherwise \textbf{fail} 
	\end{tabular}}
\end{center}

\medskip

Recall that $\argmax$ refers to the arguments of the maxima with
\begin{multline*}
\argmax\limits_{\mathcal{S} \in \starwhyk{k}{\bar t}{D}{Q}} \diversity{\mathcal{S}}\ =\ \left\{\mathcal{S} \in \starwhyk{k}{\bar t}{D}{Q} \mid \delta(\mathcal{S}') \leq \delta(\mathcal{S})\right.\\ 
\left.\text{ for each }
S' \in \starwhyk{k}{\bar t}{D}{Q}\right\}.
\end{multline*}
It is clear that there may be several sets $\mathcal{S} \in \starwhyk{k}{\bar t}{D}{Q}$ with maximal diversity. Our problems simply ask to compute one of those sets.
Observe that in the above definition of our problems, all the relevant components are part of the input, but we are also interested in the scenario where the Datalog query is fixed, i.e., for each Datalog query $Q = (\dep,R)$, we are interested in the problem:

\medskip

\begin{center}
	\fbox{\begin{tabular}{ll}
			{\small PROBLEM} : & $\star\text{-}\mathsf{DiverseWhyProvenance}[Q]$
			\\[2mm]
			{\small INPUT} : & A database $D$ over $\esch{\dep}$,\\
			& a tuple ${\bar t} \in \adom{D}^{\arity{R}}$, and\\
			& an integer $k > 0$\\[2mm]
			{\small OUTPUT} : & A member of $\argmax\limits_{\mathcal{S} \in \starwhyk{k}{\bar t}{D}{Q}} \diversity{\mathcal{S}}$ if it exists;\\
			&  otherwise \textbf{fail} 
	\end{tabular}}
\end{center}

\medskip

We proceed to show that we are dealing with intractable optimization problems, no matter the underlying class of proof trees (arbitrary or unambiguous), even if the Datalog query is fixed and consists of a single Datalog rule.

\begin{theorem}\label{the:diversity-complexity}
For each $\star \in \{\mathsf{ALL},\mathsf{UN}\}$, there exists a Datalog query $Q = (\dep,R)$, where $\dep$ consists of a single Datalog rule, such that $\star\text{-}\mathsf{DiverseWhyProvenance}[Q]$ is NP-hard.
\end{theorem}

Theorem~\ref{the:diversity-complexity} suggests that for providing practical solutions to our problems we need to exploit existing optimization methods that are able to solve NP-hard problems. We proceed in the next section to discuss how this can be achieved via Integer Linear Programming.

\section{Diverse Why-provenance via Integer Linear Programming}\label{sec:exact-encoding}

We aim to encode $\star\text{-}\mathsf{DiverseWhyProvenance}$, for $\star \in \{\mathsf{ALL},\mathsf{UN}\}$, as an Integer Linear Programming (ILP) maximization problem $\Pi$ in such a way that each set $\mathcal{S} \in \starwhyk{k}{\bar t}{D}{Q}$, for some Datalog query $Q = (\dep,R)$, database $D$ over $\esch{\dep}$, tuple $\bar t \in \adom{D}^{\arity{R}}$, and integer $k > 0$, that maximizes $\diversity{\mathcal{S}}$ can be extracted from one of the optimal solutions of $\Pi$.
Of course, as we shall see in the next section, there are multiple ways to construct such an encoding, relying on different characterizations of when a database $D'$ is a member of $\starwhy{\bar t}{D}{Q}$.
Hence, in this section, we first define a general framework for devising such ILP maximization problems in a uniform way, and present a generic algorithm that exploits this framework for solving $\star\text{-}\mathsf{DiverseWhyProvenance}$. We will then show how to instantiate such an algorithm with concrete encodings. We start by recalling the basics of ILP, and then introduce the key properties that an ILP encoding should enjoy to allow us to cast $\star\text{-}\mathsf{DiverseWhyProvenance}$ as an ILP maximization problem.

\begin{comment}
In this section, we discuss how to solve the problem $\mathsf{Diversity}$ by exploiting Integer Linear Programming (ILP). In particular,  our goal is to encode $\mathsf{Diversity}$ as an integer linear programming maximization problem, where each set $\mathcal{S} \in \whyk{k}{\bar t}{D}{Q}$ that maximizes $\diversity{\mathcal{S}}$ can be retrieved from one of the optimal solutions of the optimization problem.

As we are going to see, there can be multiple ways to construct this encoding, depending on different characterizations of when a database $D'$ is a member of the why-provenance $\why{\bar t}{D}{Q}$, for some Datalog query $Q = (\dep,R)$, a database $D$ over $\esch{\dep}$, and tuple $\bar t \in \adom{D}^{\arity{R}}$.

Hence, we first define a general framework for devising such maximization problems, and then we will show how to instantiate this framework with concrete characterizations. We start by defining ILP, and then move to define the key properties that an ILP encoding should have in order to allow us to solve the problem $\mathsf{Diversity}$.
\end{comment}

\subsection{Integer Linear Programming}\label{sec:ILP}

%\medskip
%\noindent
%\textbf{Integer Linear Programming (ILP).} 
Consider the disjoint sets $X$ and $B$ of integer variables and binary variables, respectively. 
%
%For a set $Y \subseteq X \cup B$, we write $\proj{\mu}{Y}$ for the restriction of $\mu$ over $Y$.
%Given a set $S$ of assignments for $X$ and $B$, we write $\proj{S}{Y}$ for the set of assignments $\{\proj{\mu}{Y} \mid \mu \in S\}$.
%for the function $\tau : Y \rightarrow \mathbb{Z}$ with $\tau(y) = \mu(y)$, for all $y \in Y$. 
%We extend the notation $\proj{S}{Y}$ to \emph{sets} $S$ of assignments in the natural way. 
An \emph{ILP constraint over $X$ and $B$} is an expression of the form
\[
a_1 x_1 + \cdots + a_n  x_n\ \diamond\ a_0,
\]
where, for each $i \in [n]$, $a_i \in \mathbb{Q}$, $x_i \in X \cup B$, $\diamond \in \{ \ge, \le, =\}$, and $a_0 \in \mathbb{Q}$. 
An \emph{assignment for $X$ and $B$} is a function $\mu : X \cup B \rightarrow \mathbb{Z}$ such that, for each $x \in B$, $\mu(x) \in \{0,1\}$. 
Such an assignment $\mu$ \emph{satisfies} the constraint $a_1 x_1 + \cdots + a_n  x_n \diamond a_0$ over $X$ and $B$ if it holds that
\[
a_1 \mu(x_1) + \cdots + a_n \mu(x_n) \diamond a_0.
\]
A \emph{system of ILP constraints over $X$ and $B$} is a set $T$ of ILP constraints over $X$ and $B$. A \emph{solution of $T$} is an assignment for $X$ and $B$ that satisfies all the constraints in $T$; we write $\sol{T}$ for the set that collects all the solutions of $T$. If $T$ admits a solution, i.e., $\sol{T} \neq \emptyset$, then we say that $T$ is \emph{feasible}; otherwise, $T$ is \emph{unfeasible}.

An \emph{ILP maximization problem over $X$ and $B$} is a pair $(f,T)$, where $f$, called the \emph{objective function}, is an expression of the form
$c_1 z_1 + \cdots + c_m z_m$,
with $c_1,\ldots,c_n \in \mathbb{Q}$ and $z_1,\ldots,z_n \in X \cup B$, and $T$ is a system of ILP constraints over $X$ and $B$. 
An \emph{optimal solution} of $(f,T)$ is a solution $\mu$ of $T$ that maximizes $f$ among all solutions of $T$, i.e., for every $\mu' \in \sol{T}$, it holds that $\mu(f) = c_1 \mu(z_1) + \cdots + c_m \mu(z_m) \geq \mu'(f) = c_1 \mu'(z_1) + \cdots + c_m \mu'(z_m)$. We write $\opt{f,T}$ for the set of all optimal solutions of $(f,T)$.

\medskip
\noindent
\textbf{XOR and Indicator Constraints.} For the sake of technical clarity, apart from linear constraints as defined above, we are also going to use xor constraints and indicator constraints, which we introduce below. It is well-known that xor and indicator constraints can be easily defined by means of linear constraints. Hence, in our technical development, we will directly use xor and indicator constraints to simplify the presentation, and we will refer to linear, xor, and indicator constraints simply as constraints.

An \emph{xor constraint over $B$} is an expression of the form
\[
x \oplus y\ =\ z,
\]
where $x,y,z$ are binary variables from $B$ and $\oplus$ denotes the xor operator. An assignment $\mu$ for $B$ satisfies the above xor constraint if $\mu(x) \oplus \mu(y) = \mu(z)$. 
Now, an \emph{indicator constraint over $X$ and $B$} is an expression of the form
\[
z\ \equiv\ a_1 x_1 + \cdots a_n x_n \diamond a_0,
\]
where $z \in B$ and $a_1 x_1 + \cdots a_n x_n \diamond a_0$ is a linear constraint over $X$ and $B$. An assignment $\mu$ for $X$ and $B$ satisfies the above indicator constraint when $\mu(z) = 1$ iff $\mu$ satisfies the linear constraint.

\begin{comment}
For integer variables $X$ and binary variables $B$, these include two kind of constraints. One is a \emph{xor constraint} of the form:
$$ x \oplus y = z,$$
where $x,y,z$ are binary variables from $B$, and $\oplus$ denotes the xor operator. An assignment $\mu$ satisfies the xor constraint when $\mu(z)$ coincides with the xor of $\mu(x)$ and $\mu(y)$. The other kind of constraint is an \emph{indicator constraint}, which is of the form:
$$z \equiv a_1 x_1 + \cdots a_n x_n \diamond a_0,$$
where $z$ is a binary variable from $B$, while $a_1 x_1 + \cdots a_n x_n \diamond a_0$ is a linear constraint over $X$ and $B$. The above indicator constraint is satisfied by an assignment $\mu$ when $\mu(z) = 1$ iff $\mu$ satisfies the linear constraint.

It is well-known that xor and indicator constraints can be easily defined by means of linear constraints. Hence, in what follows, we will use these additional kind of constraints directly, in order to simplify the presentation, and we will refer to linear constraints, xor constraints and indicator constraints simply as constraints.
\end{comment}

\subsection{Diverse Why-Provenance as an ILP Maximization Problem}
%\medskip
%\noindent
%\textbf{Diversity as an ILP maximization problem.} 

Having the basics of ILP in place, we can now introduce the general framework for casting $\star\text{-}\mathsf{DiverseWhyProvenance}$, where $\star \in \{\mathsf{ALL},\mathsf{UN}\}$, as an ILP maximization problem based on different ILP encodings in a uniform way. To this end, we first introduce the general notion of ILP encoding of why-provenance.

\begin{definition}[\textbf{ILP Encoding}]\label{def:ilp-encoding-why}
Consider a Datalog query $Q = (\dep,R)$, a database $D = \{\alpha_1,\ldots,\alpha_n\}$ over $\esch{\dep}$, for an integer $n>0$, and a tuple $\bar t \in \adom{D}^{\arity{R}}$. An \emph{ILP encoding of $\starwhy{\bar t}{D}{Q}$}, where $\star \in \{\mathsf{ALL},\mathsf{UN}\}$, is a system $T$ of ILP constraints over a set of integer variables $X$ and a set of binary variables $B$, where $B$ contains $n$ distinguished variables $x_{\alpha_1},\ldots,x_{\alpha_n}$, called the \emph{database variables of $T$}, such that, for all databases $D'$ over $\esch{\dep}$, the following statements are equivalent:
\begin{enumerate}
	\item $D' \in \starwhy{\bar t}{D}{Q}$.
	\item $D' = \{\alpha_i \mid i \in [n] \text{ and } \mu(x_{\alpha_i}) = 1\}$ for some $\mu \in \sol{T}$.
    %, for each $i \in [n]$, $\mu(x_{\alpha_i}) = 1$ if $\alpha_i \in D'$ and $\mu(x_{\alpha_i}) = 0$ if $\alpha_i \not\in D'$. 
    \hfill\markfull
\end{enumerate}
\end{definition}

Therefore, given a Datalog query $Q = (\dep,R)$, a database $D$ over $\esch{\dep}$, and a tuple $\bar t \in \adom{D}^{\arity{R}}$, if we are able to construct an ILP encoding $T$ of $\starwhy{\bar t}{D}{Q}$, then we can compute all the members of $\starwhy{\bar t}{D}{Q}$ by computing the solutions of $T$. 
We now proceed explain how to use an ILP encoding for why-provenance as defined above (Definition~\ref{def:ilp-encoding-why}) to solve $\star\text{-}\mathsf{DiverseWhyProvenance}$.

Consider an ILP encoding $T$ of $\starwhy{\bar t}{D}{Q}$ over $X$ and $B$, with $x_{\alpha_1},\ldots,x_{\alpha_n}$ being the database variables of $T$. For an assignment $\mu$ for $X$ and $B$, we denote by $\db{\mu}$ the \emph{database induced by $\mu$}, that is, the database $\{\alpha_i \mid i \in [n] \text{ and } \mu(x_{\alpha_i}) = 1\}$.
%for an assignment $\mu$ whose domain mentions the database variables $x_{\alpha_1},\ldots,x_{\alpha_n}$ of an ILP encoding of $\why{\bar t}{D}{Q}$, we use $\db{\mu}$ to denote the \emph{database induced by $\mu$}, i.e., the database $\{\alpha_i \mid \mu(x_{\alpha_i}) = 1 \text{ and } i \in [n]\}$.
%
%We can now exploit an ILP encoding of $\why{\bar t}{D}{Q}$ to solve $\mathsf{Diversity}$ as follows.
%Consider an ILP encoding $T$ of $\starwhy{\bar t}{D}{Q}$ with database variables $x_{\alpha_1},\ldots,x_{\alpha_n}$. 
For an integer $i > 0$, the \emph{$i$-th copy of $T$}, denoted $\copyof{T}{i}$, is the system of ILP constraints over $\copyof{X}{i} = \{\copyof{x}{i} \mid x \in X\}$ and $\copyof{B}{i} = \{\copyof{x}{i} \mid x \in B\}$ obtained from $T$ by replacing each integer and binary variable $x$ with $\copyof{x}{i}$.
%, where each integer and binary variable $x$ (including the database variables) is replaced with the fresh variable $\copyof{x}{i}$. 
We are now ready to define the ILP maximization problem that allows us to solve $\star\text{-}\mathsf{DiverseWhyProvenance}$.

\begin{definition}[\textbf{Maximization Problem}]\label{def:ilp-maximization-problem}
Consider a Datalog query $Q = (\dep,R)$, a database $D$ over $\esch{\dep}$, a tuple $\bar t \in \adom{D}^{\arity{R}}$, and an integer $k > 0$. Let $T$ be an ILP encoding of $\starwhy{\bar t}{D}{Q}$, for $\star \in \{\mathsf{ALL},\mathsf{UN}\}$, over $X$ and $B$, with $x_{\alpha_1},\ldots,x_{\alpha_n}$ being the database variables of $T$. %
The \emph{$(T,k)$-problem}
%The \emph{$T$-based encoding of $\mathsf{Diversity}$ w.r.t.\ $D$, $Q$, $\bar t$, and $k$}
is the ILP maximization problem $(f,S)$ over $\bigcup_{i \in [k]}\copyof{X}{i}$ and $\bigcup_{i \in [k]}\copyof{B}{i} \cup 
\{\auxvar^{i,j}_\ell \mid 1 \leq \ell \leq n \text{ and } 1 \leq i < j \leq k\}$, where $f$ is defined as the expression
\[
\sum\limits_{1 \leq i < j \leq k} \left(\sum\limits_{1 \leq \ell \leq n} \auxvar^{i,j}_\ell\right),
\]
and
%where $\copyof{X}{i}$ and $\copyof{B}{i}$ are the integer and binary variables of $\copyof{T}{i}$, for each $i \in [k]$, 
$S$ is the system of constraints
\begin{multline*}
\copyof{T}{1} \cup \cdots \cup \copyof{T}{k}\ \cup \\
 \left\{\copyof{x_{\alpha_\ell}}{i} \oplus \copyof{x_{\alpha_\ell}}{j} = \auxvar^{i,j}_\ell \mid 1 \leq \ell \leq n \text{ and } 1 \leq i < j \leq k\right\}\ \cup \\
 \left\{\sum\limits_{1 \leq \ell \leq n} \auxvar^{i,j}_\ell \ge 1 \mid 1 \leq i < j \leq k\right\}.\tag*{\markfull}
\end{multline*}
%and $f$ is the following function:
%\begin{flalign*}
%&& \sum\limits_{i < j \in [k]} \left(\sum\limits^n_{\ell=1} \auxvar^{i,j}_\ell\right).&& \markfull
%\end{flalign*}
%For each $i \in [k]$, we call the variables in the set $\{\copyof{x_{\alpha_\ell}}{i} \mid \ell \in [n]\}$ the \emph{$i$-th database variables of $(f,S)$}.\hfill\markfull
\end{definition}

Intuitively, the $(T,k)$-problem defined above 
%the $T$-based encoding $(f,S)$ of $\mathsf{Diversity}$ w.r.t.\ $D$, $Q$, $\bar t$, and $k$ 
is a maximization problem which enforces the simultaneous satisfaction of $k$ distinct copies $\copyof{T}{1},\ldots,\copyof{T}{k}$ of $T$. This makes a solution $\mu$ of $S$ to be of the form $\mu_1 \cup \cdots \cup \mu_k \cup \mu'$, where $\mu'$ is over the 
%\auxvar$ 
auxiliary variables, while each $\mu_i$ is a solution of $\copyof{T}{i}$. In addition, $S$ forces $\mu_1,\ldots,\mu_k$ to induce pairwise different databases. The latter is guaranteed by the combination of the xor constraints and the additional linear constraints, apart from those of $\copyof{T}{1},\ldots,\copyof{T}{k}$, added to $S$.
Hence, whenever $S$ is feasible, each solution of $S$ induces $k$ different members $D_1,\ldots,D_k$ of $\starwhy{\bar t}{D}{Q}$. The expression $f$ is then used to maximize the diversity of the set $\{D_1,\ldots,D_k\}$ where, for each $i,j \in [k]$ with $i < j$, the inner sum of $f$ encodes the distance between $D_i$ and $D_j$.
The above discussion is formalized by the following result.
%Formally, we can show the following technical result.

\begin{theorem}\label{thm:generic-encoding}
	Consider a Datalog query $Q = (\dep,R)$, a database $D$ over $\esch{\dep}$, a tuple $\bar t \in \adom{D}^{\arity{R}}$, and an integer $k > 0$. Let $T$ be an ILP encoding of $\starwhy{\bar t}{D}{Q}$, for $\star \in \{\mathsf{ALL},\mathsf{UN}\}$, over $X$ and $B$, with $x_{\alpha_1},\ldots,x_{\alpha_n}$ being the database variables of $T$. 
    %With $(f,S)$ being the $(T,k)$-problem,
%$T$-based ILP encoding $(f,S)$ of $\mathsf{Diversity}$ w.r.t.\ $D$, $Q$, $\bar t$, and $k$ is such that, 
For every set $\{D_1,\ldots,D_k\}$ of databases over $\esch{\dep}$, the following are equivalent:
	\begin{enumerate}
		\item $\{D_1,\ldots,D_k\} \in \argmax\limits_{\mathcal{S} \in \starwhyk{k}{\bar t}{D}{Q}} \diversity{\mathcal{S}}$. 
		\item There is an optimal solution $\mu$ of the $(T,k)$-problem such that, for each $i \in [k]$, $D_i = \{\alpha_\ell \mid \ell \in [n] \text{ and } \mu(\copyof{x_{\alpha_\ell}}{i}) = 1\}$.
        %, for each $\ell \in [n]$ and $i \in [k]$, $\mu(\copyof{x_{\alpha_\ell}}{i}) = 1$ iff $\alpha_\ell \in D_i$.
        %, and $\mu(\copyof{x_{\alpha_\ell}}{i}) = 0$ if $\alpha_\ell \not \in D_i$.
	\end{enumerate}
\end{theorem}

Observe now that Theorem~\ref{thm:generic-encoding} immediately leads to a procedure for solving $\star\text{-}\mathsf{DiverseWhyProvenance}$, for $\star \in \{\mathsf{ALL}, \mathsf{UN}\}$. In particular, assuming that $\ConstructILPEncoding$ is a procedure that takes as input a triple $(\bar t,D,Q)$, where $Q$ is a Datalog query $Q = (\dep,R)$, $D$ is a database $D$ over $\esch{\dep}$, and $\bar t$ a tuple of $\adom{D}^{\arity{R}}$, and constructs an ILP encoding for $\starwhy{\bar t}{D}{Q}$, 
%for some Datalog query $Q = (\dep,R)$, database $D$ over $\esch{\dep}$, and tuple $\bar t \in \adom{D}^{\arity{R}}$, 
we define the procedure $\ComputeDiverse{\ConstructILPEncoding}$, which is depicted in Algorithm~\ref{alg:compute-diverse}, that solves $\star\text{-}\mathsf{DiverseWhyProvenance}$.
The correctness of the algorithm follows immediately from Theorem~\ref{thm:generic-encoding} and the fact that $\ConstructILPEncoding(\bar t,D,Q)$ is an ILP encoding for $\starwhy{\bar t}{D}{Q}$:

\begin{corollary}\label{cor:generic-algorithm}
 Let $\star \in \{\mathsf{ALL}, \mathsf{UN}\}$, and assume that $\ConstructILPEncoding$ is a procedure that takes as input a triple $(\bar t,D,Q)$, where $Q$ is a Datalog query $Q = (\dep,R)$, $D$ is a database $D$ over $\esch{\dep}$, and $\bar t$ a tuple of $\adom{D}^{\arity{R}}$, and constructs an ILP encoding for $\starwhy{\bar t}{D}{Q}$.
 %given as input $(\bar t,D,Q)$, constructs an ILP encoding for $\starwhy{\bar t}{D}{Q}$, for some Datalog query $Q = (\dep,R)$, database $D$ over $\esch{\dep}$, and tuple $\bar t \in \adom{D}^{\arity{R}}$. 
 Then, %The algorithm
$\ComputeDiverse{\ConstructILPEncoding}$ solves $\star\text{-}\mathsf{DiverseWhyProvenance}$.
\end{corollary}
\begin{algorithm}[t]
	\caption{Algorithm $\ComputeDiverse{\ConstructILPEncoding}$}
	\label{alg:compute-diverse}
	
	\BlankLine
	
	\KwIn{A Datalog query $Q = (\dep,R)$, a database $D = \{\alpha_1,\ldots,\alpha_n\}$ over $\esch{\dep}$, a tuple $\bar t \in \adom{D}^{\arity{R}}$, and an integer $k >0 $.}
	\KwOut{A set consisting of $k$ databases, or \textbf{fail}.}
	
	\BlankLine
	
	$T := \ConstructILPEncoding(\bar t,D,Q)$\;
	Construct the $(T,k)$-problem $(f,S)$\;
	\If{an optimal solution $\mu$ of $(f,S)$ exists}{
		Let $x_{\alpha_1},\ldots,x_{\alpha_n}$ be the database variables of $T$\;
		\For{$i = 1$ \textbf{\text{to}} $k$}{
			$D_i := \{\alpha_\ell \mid \ell \in [n] \text{ and } \mu(\copyof{x_{\alpha_\ell}}{i}) = 1\}$\;
		}
		\textbf{return} $\{D_1,\ldots,D_k\}$\;
	}\Else{\textbf{return} \textbf{fail}\;}
\end{algorithm}
%the problem $\mathsf{Diversity}$ by finding an optimal solution of the $T$-based encoding of $\mathsf{Diversity}$, w.r.t.\ the input database $D$, query $Q$, tuple $\bar t$, and integer $k$, where $T$ is any ILP encoding for $\why{\bar t}{D}{Q}$.

\section{ILP Encodings of Why-Provenance}\label{sec:characterizations}

We have seen that we can compute a solution of the problem $\star\text{-}\mathsf{DiverseWhyProvenance}$, where $\star \in  \{\mathsf{ALL},\mathsf{UN}\}$, on input $Q = (\dep,R)$, $D$, $\bar t$, and $k > 0$, by simply finding an optimal solution of the $(T,k)$-problem, where $T$ is any ILP encoding of $\starwhy{\bar t}{D}{Q}$. Moreover, $\ComputeDiverse{\ConstructILPEncoding}$ is a procedure that precisely does that, given that the sub-procedure $\ConstructILPEncoding$ is able to construct the ILP encoding $T$.
In this section, we describe concrete ILP encodings of $\starwhy{\bar t}{D}{Q}$, which will lead to concrete means to implement the sub-procedure $\ConstructILPEncoding$.
To this end, we present the main characterizations from the literature of when a database $D' \subseteq D$ is a member of $\starwhy{\bar t}{D}{Q}$, or equivalently, when there is an arbitrary (resp., unambiguous) proof tree $T$ of the fact $R(\bar t)$ w.r.t.~$D$ and $\dep$ with $\support{T} = D'$.
For each such characterization, we discuss how it can be used to devise ILP encodings of $\starwhy{\bar t}{D}{Q}$.
Before we proceed further, we need to discuss some basic assumptions and recall some auxiliary notions.

%In this section, we present the main characterizations for the existence of a standard (resp., unambiguous) proof tree of a fact $R(\bar t)$ w.r.t.\ a database $D$ and program $\dep$ with a given support that we inherit from the literature, and show how to exploit such characterizations to devise ILP encodings for $\allwhy{\bar t}{D}{\dep}$ (resp., $\unwhy{\bar t}{D}{\dep})$.
%For this we first need to discuss some basic assumptions, and recall some notions. 

\medskip
\noindent
\textbf{Programs.} In order to simplify the technical definitions, in the rest of this section, for a Datalog program $\dep$,
we assume that for each rule $\sigma \in \dep$ of the form $R_0(\bar x_0) \ \assign\ R_1(\bar x_1),\ldots,R_n(\bar x_n)$, the predicates $R_1,\ldots,R_n$ are all distinct, i.e., rule-bodies are self-join-free. The latter assumption can be done w.l.o.g.~since we can always convert $\dep$ into $\dep'$ that satisfies the desired condition without affecting the why-provenance of $R(\bar t)$ w.r.t.~$D$ and $\dep$ relative to both arbitrary and unambiguous proof trees. In particular, for each rule $\sigma$ as above, and for each $i \in [n]$, let $\text{\rm Aux}^\sigma_i \not\in \sch{\dep}$ be an auxiliary predicate of arity $\arity{R_i}$, replace $R_i$ with $\text{\rm Aux}^\sigma_i$ in the body of $\sigma$, and add to $\dep$ the rule $\text{\rm Aux}^\sigma_i(x_1,\ldots,x_{\arity{R_i}})\ \assign\ R_i(x_1,\ldots,x_{\arity{R_i}})$. Clearly, for $\star \in \{\mathsf{ALL},\mathsf{UN}\}$, $\starwhy{\bar t}{D}{Q} = \starwhy{\bar t}{D}{Q'}$, where $Q'=(\dep',R)$.

\medskip
\noindent
\textbf{Multisets.}
A \emph{multiset} $M$ is a pair $(S,\mu)$, where $S$ is a set and $\mu : S \rightarrow \mathbb{N}$ is a function assigning to each element $a \in S$ a strictly positive integer, called the \emph{multiplicity of $a$ in $S$}. We say that $M$ is \emph{empty} if $S$ is empty; otherwise, it is non-empty. The \emph{cardinality} of $M$, denoted $|M|$, is $\sum\limits_{a \in S} \mu(a)$. By abuse of notation, we may write $a \in M$ meaning that $a \in S$. Moreover, we may write $\card{a}{M}$ for $\mu(a)$, if $a \in S$, and $0$ if $a \not \in S$. 
Finally, for a sequence of (not necessarily distinct) objects $o_1,\ldots,o_n$, we write $\multiset{o_1,\ldots,o_n}$ to denote the multiset $(S,\mu)$, where $S = \{o_1,\ldots,o_n\}$ and $\mu$ assigns to each element of $S$ the number of times the element occurs in the sequence $o_1,\ldots,o_n$.

\medskip

\noindent
\textbf{Downward Closure.} We now introduce the so-called {\em downward closure} of a fact w.r.t.~a database and a Datalog program, taken from~\cite{ElKM22}, that will play a key role in all the characterizations we are going to discuss.
Roughly, the downward closure of a fact $\alpha$ w.r.t.~a database $D$ and a Datalog program $\dep$ is a hypergraph that encodes all possible proof trees of $\alpha$ w.r.t.~$D$ and $\dep$. 
For our purposes, a \emph{(directed) hypergraph} $\mathcal{H}$ is a pair $(V,E)$, where $V$ is the set of its \emph{nodes} and $E$ is the set of its \emph{hyperedges}, i.e., pairs of the form $(u,T)$, where $u \in V$ and $T \subseteq V$ is a non-empty set of nodes.
For $u,v \in V$, we say that \emph{$u$ reaches $v$ in $\mathcal{H}$} if either $u=v$, or there exists a sequence of hyperedges $(u_1,T_1),\ldots,(u_n,T_n)$ with $u_1=u$, $v \in T_n$, and $u_i \in T_{i-1}$ for each $i \in \{2,\ldots,n\}$. For a node $u$, we write $\downof{\mathcal{H}}{u}$ for the hypergraph $(V',E')$ obtained from $\mathcal{H}$, where $V'$ contains $u$ and all nodes reachable from $u$, and $E'$ collects all the hyperedges mentioning only nodes of $V'$.

For a Datalog program $\dep$ and a database $D$ over $\esch{\dep}$, the \emph{graph of rule instances (GRI)} of $D$ and $\dep$ is the hypergraph $\gri{D}{\dep}=(V,E)$, with $V \subseteq \base{D,\dep}$, such that
\begin{enumerate}
	\item $D \subseteq V$, and 
	\item if there exists a rule $R_0(\bar x_0)\ \assign\ R_1(\bar x_1),\ldots,R_n(\bar x_n)$ in $\dep$ and a function $h : \bigcup_{i \in [n]} \bar x_i \ra \ins{C}$ such that $\alpha_i = R_i(h(\bar x_i)) \in V$, for $i \in [n]$, then $\alpha_0 = R_0(h(\bar x_0)) \in V$ and $(\alpha_0,\{\alpha_1,\ldots,\alpha_n\})$ is an hyperedge of $E$.
%	$(\alpha_0,\multiset{\alpha_1,\ldots,\alpha_n})$ is an hyperedge of $E$.
\end{enumerate}
Intuitively, $\gri{D}{\dep}$ encodes all the proof trees of some fact $\alpha$ w.r.t.\ $D$ and $\dep$.
Since we are interested only in the proof trees of a specific fact $\alpha$, we do not need to consider $\gri{D}{\dep}$ in its entirety, but only the subhypergraph of $\gri{D}{\dep}$ containing $\alpha$ and all the nodes reachable from it. This leads to the notion of downward closure, which simply keeps from $\gri{D}{\dep}$ the part that is relevant to derive the fact in question. 

\begin{definition}[\textbf{Downward Closure}]\label{def:downward-closure}
	Consider a Datalog program $\dep$, a database $D$ over $\esch{\dep}$, and a fact $\alpha$ occurring in $\gri{D}{\dep}$. The \emph{downward closure of $\alpha$ w.r.t.~$D$ and $\dep$}, denoted $\downc{\alpha}{D}{\dep}$, is the hypergraph $\downof{\gri{D}{\dep}}{\alpha}$. \hfill\markfull
\end{definition}

%The key property of the downward closure is the following
%
%
%\begin{lemma}[\cite{CLPS24aaai}]\label{lem:dag-in-down}
%	Consider a Datalog program $\dep$, a database $D$ over $\esch{\dep}$, a fact $\alpha$ over $\sch{\dep}$, and a proof tree $T=(V,E,\lambda)$ of $\alpha$ w.r.t.\ $D$ and $\dep$. Then, for every node $u \in V$ with outgoing edges $(u,v_1),\ldots,(u,v_n)$ in $T$, we have that $(\lambda(u),\multiset{\lambda(v_1),\ldots,\lambda(v_n)})$ is a hyperedge of $\downc{\alpha}{D}{\dep}$.
%\end{lemma}

In the rest of the section, whenever we refer to an ILP encoding of $\starwhy{\bar t}{D}{Q}$, for a Datalog query $Q = (\dep,R)$, a database $D$ over $\esch{\dep}$, and a tuple $\bar t \in \adom{D}^{\arity{R}}$, we tacitly assume that $R(\bar t)$ occurs in $\gri{D}{\dep}$, or, equivalently, that $\bar t \in Q(D)$, so that the hypergraph $\downc{R(\bar t)}{D}{\dep}$ is well-defined. This is w.l.o.g.~since, if $R(\bar t)$ does not occur in $\gri{D}{\dep}$, then $\starwhy{\bar t}{D}{Q} = \emptyset$, and thus the unfeasible system consisting of the two constraints $v \ge 1$ and $v \le 0$, where $v$ is one of the database variables, is trivially an ILP encoding of $\starwhy{\bar t}{D}{Q}$.

We are now ready to present the ILP encodings of interest. Those encodings are classified into two main categories, namely the hypergaph- and the graph-based ones, depending on whether the underlying characterization of when a database is a member of $\starwhy{\bar t}{D}{Q}$, for $\star \in  \{\mathsf{ALL},\mathsf{UN}\}$, uses a hypergraph or a (directed acyclic) graph as a way for compactly representing proof trees.

\subsection{Hypergraph-based ILP Encodings}\label{sec:hypergraph-based}

We start by recalling the characterizations of interest that rely on hypergraphs, and then use them to devise the hypergraph-based ILP encodings of why-provenance.

\medskip
\noindent
\underline{\textbf{Hypergraph-based Characterizations}}

\smallskip
\noindent Consider a Datalog query $Q = (\dep,R)$, a database $D$ over $\esch{\dep}$, and a tuple $\bar t \in \adom{D}^{\arity{R}}$.
For a multiset $E = (E',\mu)$ of hyperedges from $\downc{R(\bar t)}{D}{\dep}$, we denote by $\nodes{E}$ the \emph{set} of nodes occurring in the hyperedges of $E$ and by $\hypergraphof{E}$ the hypergraph $(\nodes{E},E')$. For a node $\alpha \in \nodes{E}$, we define the \emph{in-degree of $\alpha$ in $E$}, denoted $\indegree{\alpha}{E}$, as 
the number of (possibly repeated) hyperedges of $E$ in which $\alpha$ appears as a target, namely
\[
\indegree{\alpha}{E} = \sum\limits_{(\beta,T) \in E' \text{ with } \alpha \in T} \mu((\beta,T)).
\]
%the number of times $\alpha$ appears as a target in all (possibly repeated) hyperedges of $E$, i.e.,
%$$\indegree{\alpha}{E} = \sum\limits_{(\beta,T) \in E'} \card{\alpha}{T} \cdot \mu((\beta,T)).$$
%
Moreover, the \emph{out-degree of $\alpha$ in $E$}, denoted $\outdegree{\alpha}{E}$, is the number of (possibly repeated) hyperedges in $E$ having $\alpha$ as the source, i.e.,
$$\outdegree{\alpha}{E} = \sum\limits_{(\alpha,T) \in E'} \mu((\alpha,T)).$$
We are now ready to recall the hypergraph-based characterization of when there is a proof tree of $R(\bar t)$ w.r.t.\ $D$ and $\dep$ with a certain support $D' \subseteq D$. This relies on the existence of a multiset of hyperedges with certain properties that one can extract from $\downc{R(\bar t)}{D}{\dep}$. We proceed to introduce those properties.
%define when a multiset of hyperedges enjoys these properties and then present the characterization.

\begin{definition}[\textbf{ (Unambiguous) Proof Tree Description}]\label{def:hypergraph-properties}
	Consider a Datalog query $Q=(\dep,R)$, a database $D$ over $\esch{\dep}$, a tuple $\bar t \in \adom{D}^{\arity{R}}$, and a  database $D' \subseteq D$. We say that a multiset $E$ of hyperedges from $\downc{R(\bar t)}{D}{\dep}$ is a \emph{description of a proof tree of $R(\bar t)$ w.r.t.~$D$ and $\dep$ with support $D'$} if the following hold:
	\begin{enumerate}
		\item $R(\bar t) \in V_E$ and $\outdegree{R(\bar t)}{E} - \indegree{R(\bar t)}{E} = 1$,
		\item $\outdegree{\alpha}{E} - \indegree{\alpha}{E} = 0$, for each $\alpha \in \nodes{E} \setminus (D \cup \{R(\bar t)\})$,
		\item $\sum\limits_{\alpha \in \nodes{E} \cap D} \indegree{\alpha}{E} \ge 1$,
		\item $R(\bar t)$ reaches all nodes of $\hypergraphof{E}$ in $\hypergraphof{E}$, and
		\item $D' = \nodes{E} \cap D$.
	\end{enumerate}
    We say that $E$ is a \emph{description of an unambiguous proof tree of $R(\bar t)$ w.r.t.~$D$ and $\dep$ with support $D'$} if, in addition to the above conditions, for every pair $(\alpha,T),(\alpha,T') \in E$, it holds that $T = T'$. \hfill\markfull
\end{definition}

The desired hypergraph-based characterizations follows:
%, which we call the \emph{hypergraph-based characterization}, follows:

\begin{theorem}[\cite{CalauttiLPS24b}]\label{thm:arbitrary-hypergraph}
	Consider a Datalog query $Q=(\dep,R)$, a database $D$ over $\esch{\dep}$, a tuple $\bar t \in \adom{D}^{\arity{R}}$, and a database $D' \subseteq D$. The following are equivalent:
	\begin{enumerate}
		\item There is a proof tree (resp., an unambiguous proof tree) of $R(\bar t)$ w.r.t.~$D$ and $\dep$ with support $D'$.
		\item There is a description $E$ of a proof tree (resp., an unambiguous proof tree) of $R(\bar t)$ w.r.t.~$D$ and $\dep$ with support $D'$.
	\end{enumerate}
\end{theorem}

Note that the version of Theorem~\ref{thm:arbitrary-hypergraph} for unambiguous proof trees has not been explicitly stated in~\cite{CalauttiLPS24b}, but it is an easy adaptation of its version for arbitrary proof trees.
\begin{comment}
One can easily adapt the above characterization to unambiguous proof trees by additionally requiring that no two hyperedges $(\alpha,T)$ and $(\alpha,T')$ with $T \neq T'$ appear in the multiset $E$.

\begin{definition}[\textbf{Unambiguous Proof Tree Description}]\label{def:un-hypergraph-properties}
	Consider a Datalog query $Q=(\dep,R)$, a database $D$ over $\esch{\dep}$, a tuple $\bar t \in \adom{D}^{\arity{R}}$, and a  database $D' \subseteq D$. We say that a multiset $E$ of hyperedges from $\downc{R(\bar t)}{D}{\dep}$ is a \emph{description of an unambiguous proof tree of $R(\bar t)$ w.r.t.~$D$ and $\dep$ with support $D'$} if:
    \begin{enumerate}
    \item $E$ is a description of a proof tree of $R(\bar t)$ w.r.t.~$D$ and $\dep$ with support $D'$, and 
    \item for every pair $(\alpha,T),(\alpha,T') \in E$, it holds that $T = T'$.\hfill\markfull
    \end{enumerate}
\end{definition}

The unambiguous hypergraph-based characterization follows:

\begin{theorem}[{\normalfont Implicit in }\cite{CalauttiLPS24b}]\label{thm:unambiguous-hypergraph}
	Consider a Datalog query $Q=(\dep,R)$, a database $D$ over $\esch{\dep}$, a tuple $\bar t \in \adom{D}^{\arity{R}}$, and a database $D' \subseteq D$. The following are equivalent:
	\begin{enumerate}
		\item There is an unambiguous proof tree of $R(\bar t)$ w.r.t.~$D$ and $\dep$ with support $D'$.
		\item There is a description $E$ of an unambiguous proof tree of $R(\bar t)$ w.r.t.~$D$ and $\dep$ with support $D'$.
	\end{enumerate}
\end{theorem}
\end{comment}
We now show
%, given a Datalog query $Q = (\dep,R)$, a database $D$ over $\esch{\dep}$, and a tuple $\bar t \in \adom{D}^{\arity{R}}$, 
how to exploit Theorem~\ref{thm:arbitrary-hypergraph} to devise an ILP encoding of why-provenance.
%$\starwhy{\bar t}{D}{Q}$, where $\star \in \{\mathsf{ALL},\mathsf{UN}\}$.
%$\allwhy{\bar t}{D}{Q}$ and $\unwhy{\bar t}{D}{Q}$, respectively, 

\medskip
\noindent
\underline{\textbf{Hypergraph-based ILP Encoding}}

\smallskip
\noindent In what follows, fix a Datalog query $Q=(\dep,R)$, a database $D$ over $\esch{\dep}$, a tuple $\bar t \in \adom{D}^{\arity{R}}$, and let $\downc{R(\bar t)}{D}{\dep} = (V,E)$.
We start with the ILP encoding for $\allwhy{\bar t}{D}{Q}$, called the \emph{hypergraph-based ILP encoding of $\allwhy{\bar t}{D}{Q}$} and denoted as $\ilphypergraph{\bar t}{D}{Q}$.
The high-level idea is to devise a set of constraints whose solutions correspond to descriptions of proof trees of $R(\bar t)$ w.r.t.\ $D$ and $\dep$ with \emph{some} support $D'$. This will imply that from each solution of $\ilphypergraph{\bar t}{D}{Q}$ one can extract the support $D'$ of a proof tree of $R(\bar t)$ w.r.t.\ $D$ and $\dep$, and thus $\ilphypergraph{\bar t}{D}{Q}$ would indeed be an ILP encoding for $\allwhy{\bar t}{D}{Q}$.
Formally, $\ilphypergraph{\bar t}{D}{Q}$ is a system of constraints over integer variables $X$ and binary variables $B$ as follows.

\smallskip
\noindent
\textbf{The Variables of $\ilphypergraph{\bar t}{D}{Q}$.} For each hyperedge $(\alpha,T)$ of $\downc{R(\bar t)}{D}{\dep}$, the set $X$ contains one integer variable $x_{(\alpha,T)}$. Intuitively, the value of $x_{(\alpha,T)}$ in a solution of $\ilphypergraph{\bar t}{D}{Q}$ denotes the multiplicity of $(\alpha,T)$ in the underlying multiset that is a description of a proof tree of $R(\bar t)$ w.r.t.\ $D$ and $\dep$ with some support $D'$. Let us denote this multiset as $E^*$. In particular, when $x_{(\alpha,T)}$ is $0$, it means that the hyperedge is not part of $E^*$. 
Moreover, assuming $\hat{E}$ denotes the set of (standard) edges induced by the hyperedges of $\downc{R(\bar t)}{D}{\dep}$, i.e., 
$\hat{E} = \{(\alpha,\beta) \mid (\alpha,\{\ldots,\beta,\ldots\}) \in E\}$,
%$\hat{E} = \{(\alpha,\beta) \mid (\alpha,\multiset{\ldots,\beta,\ldots}) \in E\}$, 
$X$ contains an integer variable 
$y_{(\alpha,\beta)}$, for each edge $(\alpha,\beta) \in \hat{E}$. Such variables will be used in some constraints of $\ilphypergraph{\bar t}{D}{Q}$ to guarantee that $R(\bar t)$ reaches all nodes of $\hypergraphof{E^*}$, and they encode the amount of "flow" going through these edges; more details on this later. 
Finally, $B$ contains one binary variable $z_\alpha$ for each $\alpha \in V \cup D$.

% each node $\alpha \in V$ of $\downc{R(\bar t)}{D}{\dep}$.

\smallskip
\noindent
\textbf{The Constraints of $\ilphypergraph{\bar t}{D}{Q}$.} We now move to the constraints. The system $\ilphypergraph{\bar t}{D}{Q}$ is a set of constraints 
\[
T_{\mathsf{degree}} \cup T_{\mathsf{reach}} \cup T_{\mathsf{bind}} \cup T_{\mathsf{pos}}.
\]
The constraints of $T_{\mathsf{degree}}$ force the in- and out-degree of the nodes of the multiset $E^*$ to enjoy the properties (1)-(3) of Definition~\ref{def:hypergraph-properties}. 
The constraints $T_{\mathsf{reach}}$ enforce property (4) of Definition~\ref{def:hypergraph-properties}, i.e., that $R(\bar t)$ reaches all nodes of $\hypergraphof{E^*}$; the latter exploits the variables of the form $y_{(\alpha,\beta)}$ discussed above.
Finally, the constraints of $T_\mathsf{bind}$ simply force each binary variable $z_\alpha$ with $\alpha \in V$ to be true iff at least one hyperedge mentioning $\alpha$ belongs to $E^*$, and in case $\alpha \not \in V$, i.e., $\alpha \in D \setminus V$, force $z_\alpha$ to be always false; this is needed to enforce property~(5) of Definition~\ref{def:hypergraph-properties}, guaranteeing that only database facts occurring in $V$ can form a member of $\allwhy{\bar t}{D}{Q}$.\footnote{Note that we do not refer to any specific set $D'$ of facts when implementing property~(5) of Definition~\ref{def:hypergraph-properties}. This is because we want the solutions of our system to induce all possible members of $\allwhy{\bar t}{D}{Q}$, rather than a specific database $D'$.} Finally, the constraints of $T_\mathsf{pos}$ state that the integer variables, which are meant to represent multiplicities and flows, take non-negative values. We proceed to formalize the above informal description.

%\footnote{Note that no constraint enforces property (5) of Definition~\ref{def:hypergraph-properties}, since we want the solutions of our system to induce all possible members of $\allwhy{\bar t}{D}{Q}$, rather than a specific one. Such a property will be crucial, however, for proving that $\ilphypergraph{\bar t}{D}{Q}$ is indeed an ILP encoding of $\allwhy{\bar t}{D}{Q}$.}

For each node $\alpha \in V$ of $\downc{R(\bar t)}{D}{\dep}$, 
let 
%$\mathcal{I}_\alpha$ and $\mathcal{O}_\alpha$ be the following expressions:
\[
\mathcal{I}_\alpha = \sum\limits_{(\beta,T) \in E \text{ with } \alpha \in T} x_{(\beta,T)} \qquad \text{and} \qquad 
\mathcal{O}_\alpha = \sum\limits_{(\alpha,T) \in E} x_{(\alpha,T)}.
\]
%$$
%\mathcal{I}_\alpha = \sum\limits_{(\beta,T) \in E} \card{\alpha}{T} \cdot x_{(\beta,T)} \qquad \qquad 
%\mathcal{O}_\alpha = \sum\limits_{(\alpha,T) \in E} x_{(\alpha,T)}.
%$$
The expressions $\mathcal{I}_\alpha$ and $\mathcal{O}_\alpha$ essentially compute the integers $\indegree{\alpha}{E^*}$ and $\outdegree{\alpha}{E^*}$, respectively. With those expressions in place, $T_{\mathsf{degree}}$ and $T_\mathsf{bind}$ are easily defined. In particular, $T_\mathsf{degree}$ is the set
\begin{eqnarray*}
	&& \left\{ \mathcal{O}_{R(\bar t)} - \mathcal{I}_{R(\bar t)} = 1 \right\}\\
	&\cup& \left\{ \mathcal{O}_\alpha - \mathcal{I}_\alpha = 0\mid \alpha  \in V \setminus (D \cup \{R(\bar t)\})\right\}\\
	&\cup&\left\{ \sum\limits_{\alpha \in V \cap D} \mathcal{I}_\alpha \ge 1 \right\},
\end{eqnarray*}
and $T_\mathsf{bind}$ is the set of constraints
\[
\left\{ z_\alpha \equiv \mathcal{I}_\alpha + \mathcal{O}_\alpha \ge 1 \mid \alpha \in V \right\} \cup \left\{z_\alpha = 0 \mid \alpha \in D \setminus V \right\}.
\]
Finally, $T_\mathsf{pos}$ is the set of constraints
\[
\left\{ x_{(\alpha,T)} \ge 0 \mid (\alpha,T) \in E \right\} \cup \left\{y_{(\alpha,\beta)} \ge 0 \mid (\alpha,\beta) \in \hat{E} \right\}.
\]

It remains to explain how the set $T_\mathsf{reach}$ is defined. For this, we exploit the variables of the form $y_{(\alpha,\beta)}$, for each $(\alpha,\beta) \in \hat{E}$. Intuitively, each variable $y_{(\alpha,\beta)}$ denotes the amount of ``flow'' the edge carries when traversing it. The constraints of $T_\mathsf{reach}$ will exploit these variables to solve the reachability problem on the (standard) directed graph induced by the hyperedges of $E^*$. We do this using a flow-based encoding. Informally, $T_\mathsf{reach}$ will enforce three constraints (i) every edge $(\alpha,\beta) \in \hat{E}$ carries zero flow, unless some hyperedge of $E^*$ with source $\alpha$ mentions $\beta$ in its target, in which case it can carry an amount of flow that is enough to be distributed to all other nodes, (ii) the net outgoing flow of node $R(\bar t)$, i.e., the difference between the total outgoing flow and the total incoming flow of $R(\bar t)$, is enough to be distributed \emph{evenly} to all other nodes of $E^*$, and (iii) for every other node $\alpha$ that appears in $E^*$, the total outgoing flow of $\alpha$ must decrease by 1 w.r.t.\ the total incoming flow of $\alpha$ (this is to guarantee that $\alpha$ consumes precisely its share of the total flow).
We now formalize the above intuitive discussion.

For a node $\alpha \in V$, let 
\[
\inflow{\alpha} = \sum\limits_{(\beta,\alpha) \in \hat{E}} y_{(\beta,\alpha)} \qquad \text{and} \qquad \outflow{\alpha} = \sum\limits_{(\alpha,\beta) \in \hat{E}} y_{(\alpha,\beta)},
\]
referring the total incoming and outgoing flow of node $\alpha$, respectively.
Moreover, for an edge $(\alpha,\beta) \in \hat{E}$, let 
\[
\mathcal{C}_{(\alpha,\beta)} = \sum\limits_{(\alpha,T) \in E \text{ with } \beta \in T} x_{(\alpha,T)},
\]
which counts, with multiplicities, the hyperedges of $E^*$ that have $\alpha$ as the source and mention $\beta$ in their target; hence, $\mathcal{C}_{(\alpha,\beta)}$ is greater than zero precisely when $(\alpha,\beta)$ is an edge of the directed graph induced by the hyperedges of $E^*$.
%Let us stress that using $\mathcal{C}_{(\alpha,\beta)}$, instead of simply $z_\alpha$, is crucial: the flow must be allowed to travel only along the edges induced by the hyperedges that actually occur in $E^*$, since what we need to enforce is reachability in $\hypergraphof{E^*}$, and not merely reachability in $\downc{R(\bar t)}{D}{\dep}$ restricted to the nodes of $\hypergraphof{E^*}$.
%
The set $T_\mathsf{reach}$ is defined as\footnote{Note that, for the sake of clarity, the constraints do not strictly follow the syntax defined in Section~\ref{sec:ILP} for linear constraints, since we allow linear expressions on both sides of the (in)equalities. However, they can be easily converted into constraints of the form $a_1 x_1 + \cdots + a_n x_n \diamond a_0$, in the obvious way.}:
\begin{eqnarray*}
	&& \left\{ y_{(\alpha,\beta)} \le (|V| - 1) \cdot \mathcal{C}_{(\alpha,\beta)} \mid (\alpha,\beta) \in \hat{E} \right\}\\
	&\cup& \left\{ \outflow{R(\bar t)} - \inflow{R(\bar t)} = \sum\limits_{\alpha \in V} z_\alpha - 1 \right\}\\
	&\cup& \left\{ \outflow{\alpha} = \inflow{\alpha} - z_\alpha \mid \alpha \in V \setminus \{R(\bar t)\} \right\}.
\end{eqnarray*}
This completes the definition of $\ilphypergraph{\bar t}{D}{Q}$. We proceed to formally show that it is an ILP encoding of why-provenance:
%We now formally show that is complete, and we can prove our crucial result:

\begin{theorem}\label{thm:hyper-correct}
	Consider a Datalog query $Q = (\dep,R)$, a database $D = \{\alpha_1,\ldots,\alpha_n\}$ over $\esch{\dep}$, and a tuple $\bar t \in \adom{D}^{\arity{R}}$. The system of ILP constraints $\ilphypergraph{\bar t}{D}{Q}$ is an ILP encoding of $\allwhy{\bar t}{D}{Q}$ with database variables $z_{\alpha_1},\ldots,z_{\alpha_n}$.
\end{theorem}

By Theorem~\ref{thm:hyper-correct} and Corollary~\ref{cor:generic-algorithm}, we obtain that:

\begin{corollary}\label{cor:hyper-all-algorithm}
$\ComputeDiverse{\ilphypergraphfunc}$ solves the problem $\mathsf{ALL}\text{-}\mathsf{DiverseWhyProvenance}$.
\end{corollary}

\medskip
\noindent
\underline{\textbf{Unambiguous Hypergraph-based ILP Encoding}}

\smallskip
\noindent In what follows, fix a Datalog query $Q=(\dep,R)$, a database $D$ over $\esch{\dep}$, a tuple $\bar t \in \adom{D}^{\arity{R}}$, and assume $\downc{R(\bar t)}{D}{\dep} = (V,E)$.
We proceed to show how to construct the ILP encoding of $\unwhy{\bar t}{D}{Q}$. Recall that the unambiguous hypergraph-based characterization only requires one additional property compared to the hypergraph-based characterization. Hence, we can easily define the ILP encoding of $\unwhy{\bar t}{D}{Q}$, dubbed $\ilpunhypergraph{\bar t}{D}{Q}$, by simply adding the required constraints to $\ilphypergraph{\bar t}{D}{Q}$.
Let $X$ and $B$ be the sets of integer and binary variables, respectively, of $\ilphypergraph{\bar t}{D}{Q}$. 
Now, the variables of 
%We obtain 
$\ilpunhypergraph{\bar t}{D}{Q}$ are $X \cup B \cup B'$, where $B'$ is the set of binary variables
%by extending $\ilphypergraph{\bar t}{D}{Q}$ with an additional set $B'$ 
%of binary variables of the form
\[
\{\hat{x}_{(\alpha,T)} \mid (\alpha,T) \in E\}.
\]
Intuitively, in a solution of $\ilpunhypergraph{\bar t}{D}{Q}$, which induces a multiset $E^*$ of hyperedges of $\downc{R(\bar t)}{D}{\dep}$ that is a description of an unambiguous proof tree of $R(\bar t)$ w.r.t.\ $D$ and $\dep$ with some support $D'$, each variable $\hat{x}_{(\alpha,T)}$ is true iff the hyperedge $(\alpha,T)$ occurs in $E^*$. 
%Hence, $\ilpunhypergraph{\bar t}{D}{Q}$ is over the integer variables $X$ and the binary variables $B \cup B'$. 
Finally, $\ilpunhypergraph{\bar t}{D}{Q}$ is defined as
%the set of constraints of the form
\[
\ilphypergraph{\bar t}{D}{Q} \cup T_{\mathsf{UN}},
\]
where $T_{\mathsf{UN}}$ is the set of constraints enforcing the additional required condition, that is, no two hyperedges $(\alpha,T),(\alpha,T')$ appear in $E^*$ with $T \neq T'$. This is achieved by defining $T_{\mathsf{UN}}$ as follows:
\[
	\left\{ \hat{x}_{(\alpha,T)} \equiv x_{(\alpha,T)} \ge 1 \mid (\alpha,T) \in E \right\} \cup 
	\left\{ \sum\limits_{(\alpha,T) \in E} \hat{x}_{(\alpha,T)} \le 1 \mid \alpha \in V \right\}.
\]
%\begin{multline*}
%	\left\{ \hat{x}_{(\alpha,T)} \equiv x_{(\alpha,T)} \ge 1 \mid (\alpha,T) \in E \right\} \cup \\ \left\{ \sum\limits_{(\alpha,T) \in E} \hat{x}_{(\alpha,T)} \le 1 \mid \alpha \in V \right\}.
%\end{multline*}

We can now formally show that $\ilpunhypergraph{\bar t}{D}{Q}$ is indeed an ILP encoding of why-provenance:
%The crucial result follows:
\begin{theorem}\label{thm:un-hyper-correct}
	Consider a Datalog query $Q = (\dep,R)$, a database $D = \{\alpha_1,\ldots,\alpha_n\}$ over $\esch{\dep}$, and a tuple $\bar t \in \adom{D}^{\arity{R}}$. The system of ILP constraints $\ilpunhypergraph{\bar t}{D}{Q}$ is an ILP encoding of $\unwhy{\bar t}{D}{Q}$ with database variables $z_{\alpha_1},\ldots,z_{\alpha_n}$.
\end{theorem}

By Theorem~\ref{thm:un-hyper-correct} and Corollary~\ref{cor:generic-algorithm}, we obtain that:

\begin{corollary}\label{cor:hyper-un-algorithm}
$\ComputeDiverse{\ilpunhypergraphfunc}$ solves the problem $\mathsf{UN}\text{-}\mathsf{DiverseWhyProvenance}$.
\end{corollary}

\subsection{Graph-based ILP Encodings}\label{sec:proofdag-based}

We now recall the characterizations of when there exists an (unambiguous) proof tree of some fact $R(\bar t)$ w.r.t.~some database $D$ and Datalog program $\dep$ with some support $D'$ that rely on the existence of a directed acyclic graph (DAG) with certain properties, called a proof DAG of $R(\bar t)$ w.r.t.~$D$ and $\dep$, and then use them to devise the graph-based ILP encodings of why-provenance.

\medskip
\noindent
\underline{\textbf{Graph-based Characterizations}}

\smallskip
\noindent 
%In this last section we recall the other characterizations for the existence of an (unambiguous) proof tree of some fact $R(\bar t)$ w.r.t.\ some database $D$ and Datalog program $\dep$ with some support $D'$. These characterizations rely on the existence of a directed acyclic graph (DAG) with certain properties, called a proof DAG of $R(\bar t)$ w.r.t.\ $D$ and $\dep$. Roughly, a proof DAG is a compact way of encoding a proof tree. 
%
The characterizations of interest state that every (unambiguous) proof tree admits a corresponding (unambiguous) proof DAG $G$, where for each fact $\alpha$, only a ``small'' number of nodes of $G$ are labeled with $\alpha$. In fact, in the case of arbitrary proof DAGs, this number is bounded by a certain polynomial that depends on some parameters of the downward closure of $\alpha$ w.r.t.~$D$ and $\dep$, whereas in the case of unambiguous proof DAGs, this number is $1$.
%the number of nodes of the proof DAG labeled with $\alpha$ is bounded by a certain polynomial, and this polynomial depends on the family of proof trees we consider and on some parameters of the downward closure of $\alpha$ w.r.t.~$D$ and $\dep$.
%
We proceed to recall (unambiguous) proof DAGs and then present the actual characterizations. We recall that a DAG $G$ is \emph{rooted} if it has exactly one node, the \emph{root}, with no incoming edges. A node of $G$ is a \emph{leaf} if it has no outgoing edges.

\begin{definition}[\textbf{(Unambiguous) Proof DAG}]\label{def:proof-dag}
	Consider a Datalog program $\dep$, a database $D$ over $\esch{\dep}$, and a fact $\alpha$ over $\sch{\dep}$. A \emph{proof DAG} of $\alpha$ w.r.t.\ $D$ and $\dep$ is a finite labeled rooted DAG $G = (V,E,\lambda)$, with $\lambda : V \rightarrow \base{D,\dep}$, such that:
	\begin{enumerate}
		\item If $v \in V$ is the root, then $\lambda(v) = \alpha$.
		
		\item If $v \in V$ is a leaf, then $\lambda(v) \in D$.
		
		\item If $v \in V$ has $n\ge1$ outgoing edges $(v,u_1),\ldots,(v,u_n)$, then there exists a rule $R_0(\bar x_0)\ \assign\ R_1(\bar x_1),\ldots,R_n(\bar x_n) \in \dep$ and a function $h : \bigcup_{i \in [n]} \bar x_i \ra \ins{C}$ such that $\lambda(v) = R_0(h(\bar x_0))$, and $\lambda(u_i) = R_i(h(\bar x_i))$ for each $i \in [n]$. 
	\end{enumerate}
    We say that $G$ is an \emph{unambiguous proof DAG of $\alpha$ w.r.t.~$D$ and $\dep$} if, in addition to the above conditions, for every two nodes $u,v \in V$, it holds that $\lambda(u) = \lambda(v)$ implies $u=v$. \hfill\markfull
\end{definition}

The difference between a proof tree and a proof DAG is that a proof DAG might reuse nodes to compactly represent a proof tree. In fact, every proof tree is a proof DAG where each node has at most one parent. As for proof trees, given a proof DAG $G$ (of some fact w.r.t. some database and Datalog program), its support, denoted $\support{G}$, is the set of facts that label the leaves of $G$.

We can now recall the desired graph-based characterizations. In what follows, with $\downc{R(\bar t)}{D}{\dep} = (V,E)$, for some  Datalog program $\dep$, database $D$ over $\esch{\dep}$, and fact $R(\bar t)$, let 
\begin{eqnarray*}
\boundall{R(\bar t)}{D}{\dep} &=& ( |V| \cdot d + 1) \cdot d \cdot b,\\
\boundun{R(\bar t)}{D}{\dep} &=& 1,
\end{eqnarray*}
where $d = |V \cap D|$, i.e., $d$ is the number of nodes of $\downc{R(\bar t)}{D}{\dep}$ that correspond to database facts, and $b = \max\{|T| \mid (\alpha,T) \in E\}$ is the maximum ``branching'' of a hyperedge of $\downc{R(\bar t)}{D}{\dep}$.

%the downward closure.
%Finally, we let $\boundun{R(\bar t)}{D}{\dep} = 1$.

\begin{theorem}[\cite{CLPS24,CLPS24aaai}]\label{thm:arbitrary-proof-dag}
	Consider a Datalog query $Q=(\dep,R)$, a database $D$ over $\esch{\dep}$, a tuple $\bar t \in \adom{D}^{\arity{R}}$, and a database $D' \subseteq D$. The following are equivalent:
	\begin{enumerate}
		\item There is a proof tree (resp., an unambiguous proof tree) $T$ of $R(\bar t)$ w.r.t.\ $D$ and $\dep$ with support $D'$.
		\item There is a proof DAG (resp., an unambiguous proof DAG) $G = (V,E,\lambda)$ of $R(\bar t)$ w.r.t.\ $D$ and $\dep$ with support $D'$ such that the following hold:
		\begin{enumerate}
			\item For each $v \in V$, $\lambda(v)$ is a node of $\downc{R(\bar t)}{D}{\dep}$.
			\item For each $v \in V$ with outgoing edges $(v,u_1),\ldots,(v,u_n)$,
			$(\lambda(v),\{\lambda(u_1),\ldots,\lambda(u_n)\})$ belongs to $\downc{R(\bar t)}{D}{\dep}$.
%			$(\lambda(v),\multiset{\lambda(u_1),\ldots,\lambda(u_n)})$ belongs to $\downc{R(\bar t)}{D}{\dep}$.
			\item For each fact $\alpha$, it holds that $|\{v \in V \mid \lambda(v) = \alpha\}| \le \boundall{R(\bar t)}{D}{\dep}$ (resp., $\boundun{R(\bar t)}{D}{\dep}$).
		\end{enumerate} 
	\end{enumerate}
\end{theorem}

\begin{comment}
We call the above result the \emph{proof DAG-based characterization}.
The characterization for unambiguous proof trees, which we dub \emph{unambiguous proof DAG-based characterization}, is similar. The only difference is that the number of nodes in the proof DAG labeled with a certain fact $\alpha$ must be one. That is:

\begin{theorem}[\cite{CLPS24,CLPS24aaai}]\label{thm:unambiguous-proof-dag}
	Consider a Datalog query $Q=(\dep,R)$, a database $D$ over $\esch{\dep}$, a tuple $\bar t \in \adom{D}^{\arity{R}}$, and a database $D' \subseteq D$. The following are equivalent:
	\begin{enumerate}
		\item There is an unambiguous proof tree $T$ of $R(\bar t)$ w.r.t.\ $D$ and $\dep$ with support $D'$.
		\item There is a proof DAG $G = (V,E,\lambda)$ of $R(\bar t)$ w.r.t.\ $D$ and $\dep$ with support $D'$ such that the following hold:
		\begin{enumerate}
			\item For each $v \in V$, $\lambda(v)$ is a node of $\downc{R(\bar t)}{D}{\dep}$.
			\item For each $v \in V$ with outgoing edges $(v,u_1),\ldots,(v,u_n)$,
			$(\lambda(v),\{\lambda(u_1),\ldots,\lambda(u_n)\})$ belongs to $\downc{R(\bar t)}{D}{\dep}$.
%			$(\lambda(v),\multiset{\lambda(u_1),\ldots,\lambda(u_n)})$ belongs to $\downc{R(\bar t)}{D}{\dep}$.
			\item For each fact $\alpha$, it holds that $|\{v \in V \mid \lambda(v) = \alpha\}| \le 1$.\hfill\markfull
		\end{enumerate} 
	\end{enumerate}
\end{theorem}
\end{comment}

Let us now discuss how Theorem~\ref{thm:arbitrary-proof-dag} can be used in order to devise an ILP encoding of why-provenance. 
%, given a Datalog query $Q = (\dep,R)$, a database $D$ over $\esch{\dep}$, and a tuple $\bar t \in \adom{D}^{\arity{R}}$, 
%how to use  Theorem~\ref{thm:arbitrary-proof-dag} to devise an ILP encoding of why-provenance.
%$\starwhy{\bar t}{D}{Q}$, where $\star \in \{\mathsf{ALL},\mathsf{UN}\}$.

%As we did for the hypergraph-based characterizations, we now show how to use the above proof DAG-based characterizations to devise ILP encodings for $\allwhy{\bar t}{D}{Q}$ and $\unwhy{\bar t}{D}{Q}$, for a given Datalog query $Q = (\dep,R)$, database $D$ over $\esch{\dep}$, and tuple $\bar t \in \adom{D}^{\arity{R}}$.

\medskip
\noindent
\underline{\textbf{Graph-based ILP Encoding}}

\smallskip
\noindent Fix a Datalog query $Q=(\dep,R)$, a database $D$ over $\esch{\dep}$, a tuple $\bar t \in \adom{D}^{\arity{R}}$, and let $\downc{R(\bar t)}{D}{\dep} = (V,E)$.
We start with the ILP encoding for $\allwhy{\bar t}{D}{Q}$, which we dub the \emph{graph-based ILP encoding of $\allwhy{\bar t}{D}{Q}$}, and denote it by $\ilpproofdag{\bar t}{D}{Q}$. The encoding is a direct implementation of Theorem~\ref{thm:arbitrary-proof-dag}, namely the set of constraints $\ilpproofdag{\bar t}{D}{Q}$ searches for a proof DAG of $R(\bar t)$ w.r.t.\ $D$ and $\dep$ with \emph{some} support, where no more than $\boundall{R(\bar t)}{D}{\dep}$ nodes are labeled with the same fact.
Formally, we construct $\ilpproofdag{\bar t}{D}{Q}$ as a system of constraints over integer variables $X$ and binary variables $B$ as follows.

\medskip
\noindent
\textbf{The Variables of $\ilpproofdag{\bar t}{D}{Q}$.} Let $\hat{E}$ be the set of (standard) edges induced by the hyperedges of $\downc{R(\bar t)}{D}{\dep}$, i.e., 
$\hat{E} = \{(\alpha,\beta) \mid (\alpha,\{\ldots,\beta,\ldots\}) \in E\}$. For a fact $\alpha$,
%node $\alpha \in V$ of $\downc{R(\bar t)}{D}{\dep}$, 
%
an \emph{amplified copy of $\alpha$} is an expression of the form $\alpha^i$, where $i \in [\boundall{R(\bar t)}{D}{\dep}]$ if $\alpha \not\in D$, and $i = 1$ if $\alpha \in D$. In other words, a database fact admits a single amplified copy.
% i.e., it can be represented by at most one node of the proof DAG that we are looking for. This is w.l.o.g.~since the nodes of a proof DAG labeled with a database fact are always leaves, and thus all the nodes labeled with the same database fact can be merged into a single one (see Claim~1 in the proof of Theorem~\ref{thm:dag-correct}); it is also what allows us to use the variables of the form $x_{\alpha^1}$, with $\alpha \in D$, as the database variables of the encoding. 
Moreover, for a hyperedge $(\alpha,T) \in E$, with $T = \{\beta_1,\ldots,\beta_n\}$, of $\downc{R(\bar t)}{D}{\dep}$, an \emph{amplified copy of $(\alpha,T)$} is an expression of the form $(\alpha^i,T')$, with $T' = \{\beta_1^{k_1},\ldots,\beta_n^{k_n}\}$, where $\alpha^i$ is an amplified copy of $\alpha$, and $\beta_j^{k_j}$ is an amplified copy of $\beta_j$, for $j \in [n]$. We now introduce the variables of $\ilpproofdag{\bar t}{D}{Q}$.

The set $X$ is empty, whereas $B$ is the union of three disjoint sets $B_N$, $B_H$, and $B_E$ defined as follows:
\begin{itemize}
\item[-] The set $B_N$ contains a binary variable $x_{\alpha^i}$ for each amplified copy $\alpha^i$ of a fact $\alpha \in V \cup D$. 

\item[-] The set $B_H$ contains a binary variable $y_{(\alpha^i,T')}$ for each amplified copy $(\alpha^i,T')$ of a hyperedge $(\alpha,T) \in E$.

\item[-] The set $B_E$ contains a binary variable $z_{(\alpha^i,\beta^j)}$ for each pair of amplified copies $\alpha^i,\beta^j$ of facts $\alpha,\beta$ such that $(\alpha,\beta) \in \hat{E}$.
\end{itemize}

\noindent 
Intuitively, when $x_{\alpha^i}$ is true, for a fact $\alpha \in V \cup D$, it means that a node of the underlying proof DAG $G$ is labeled with $\alpha$, and among all nodes of $G$ labeled with $\alpha$, this node has $i$ as its identifier. In other words, $\alpha^i$ is the unique name of such a node labeled with $\alpha$ in $G$.
Moreover, the variables of $B_N$ and $B_E$ that become true via a solution of $\ilpproofdag{R(\bar t)}{D}{\dep}$ induce the nodes and the edges of a proof DAG $G$ for $R(\bar t)$ w.r.t.~$D$ and $\dep$, where each true variable of the form $x_{\alpha^i}$ means that among all possible nodes of a proof DAG that can be labeled with $\alpha$, the node with identifier $i$ is part of $G$, and a true variable of the form $z_{(\alpha^i,\beta^j)}$ means that the node with identifier $i$ labeled with $\alpha$ and the node with identifier $j$ labeled with $\beta$ are connected by an edge in $G$.

\medskip
\noindent
\textbf{The Constraints of $\ilpproofdag{\bar t}{D}{Q}$.} We now move to the constraints. In particular, the system $\ilpproofdag{\bar t}{D}{Q}$ is a set of constraints of the form
\[
T_{\mathsf{graph}} \cup T_{\mathsf{root}} \cup T_{\mathsf{proof}} \cup T_{\mathsf{acyclic}}.
\]

The set of constraints $T_{\mathsf{graph}}$ is in charge of guaranteeing consistency between the truth value of the variables in $B_N$ and the variables in $B_E$, i.e., if an edge between two nodes is part of $G$, then the two nodes must belong to $G$ as well. Formally, $T_{\mathsf{graph}}$ is the set
\[
\left\{ z_{(\alpha^i,\beta^j)} \le x_{\alpha^i} \mid z_{(\alpha^i,\beta^j)} \in B_E \right\}\ \cup\
 \left\{ z_{(\alpha^i,\beta^j)} \le x_{\beta^j} \mid z_{(\alpha^i,\beta^j)} \in B_E \right\}.
\]

The set of constraints $T_{\mathsf{root}}$ guarantees that $G$ contains a distinguished node labeled with the fact $R(\bar t)$, which is the root of $G$, and no other node of $G$ can be the root (i.e., every other node must have at least one incoming edge). For this, we pick $R(\bar t)^1$, which we denote as $\rho$, as the representative node labeled with $R(\bar t)$ to play the role of the root. Formally, $T_{\mathsf{root}}$ is the set
\begin{multline*}
	\{x_{\rho} = 1\} \cup
	\left\{ z_{(\alpha^i,\rho)} = 0 \mid z_{(\alpha^i,\rho)} \in B_E \right\} \cup \\
	\left\{ x_{\alpha^i} \le \sum\limits_{ z_{(\beta^j,\alpha^i)} \in B_E} z_{(\beta^j,\alpha^i)} \mid x_{\alpha^i} \in B_N \text{ with } \alpha \in V \text{ and } \alpha^i \neq \rho \right\}.
\end{multline*}

We now discuss $T_{\mathsf{proof}}$. Roughly, this set of constraints is in charge of guaranteeing that the nodes and edges induced by the true variables of $B_N$ and $B_E$ form a valid proof DAG. For this, the constraints ensure, for each fact $\alpha \in D \setminus V$, that each variable of the form $x_{\alpha^i}$ is always false, and, for all other facts $\alpha \in V$, ensure that (i) when a node $\alpha^i$ belongs to $G$, i.e., it is a node labeled with $\alpha$, and (ii) $\alpha$ is an intensional atom $\alpha$, then $\alpha^i$ must have the correct children in $G$. The latter means that the labels of the children of $\alpha^i$ form the target of some hyperedge of $\downc{R(\bar t)}{D}{Q}$ having $\alpha$ as the source. Moreover, no other nodes can be the children of $\alpha^i$.
This is achieved with two families of constraints in $T_{\mathsf{proof}}$. The first set of constraints is in charge of choosing which nodes should the children of $\alpha^i$ in $G$, assuming $\alpha^i$ occurs in $G$. This is done by forcing one of the variables of the form $y_{(\alpha^i,T)} \in B_H$ to be true, whenever $x_{\alpha^i}$ is true. The second set of constraints guarantees that whenever the variable $y_{(\alpha^i,T)} \in B_H$ is true, \emph{all and only} the nodes in $T$ are children of $\alpha^i$ in $G$. Formally, $T_{\mathsf{proof}}$ is the set constraints
\begin{multline*}
	\left\{x_{\alpha^i} = 0 \mid \alpha \in D \setminus V \right\} \cup \\ 
	\left\{ x_{\alpha^i} \le \sum\limits_{y_{(\alpha^i,T)} \in B_H} y_{(\alpha^i,T)} \mid x_{\alpha^i} \in B_N \text{ with }  \alpha \in V \text{ intensional }\right\} \cup  \\
	\left\{ y_{e} \le \ell_{e,\beta^j} \mid y_{e} \in B_H \text{ with } e = (\alpha^i,T) \text{ and } z_{(\alpha^i,\beta^j)} \in B_E \right\},
\end{multline*}
where, for each $e=(\alpha^i,T)$, $\ell_{e,\beta^j}$ denotes $z_{(\alpha^i,\beta^j)}$ if $\beta^j \in T$, and $1- z_{(\alpha^i,\beta^j)}$ otherwise.

We finally discuss $T_{\mathsf{acyclic}}$, which is in charge of checking that $G$, namely the graph whose edges correspond to the true variables in $B_E$, is acyclic.
Checking the acyclicity of a graph via a set of constraints over binary variables defining the edges of the graph is a well-studied problem in different fields with a plethora of encodings. For our construction, we employ the state-of-the-art encoding based on vertex elimination~\cite{RankoohR22}, which is much more efficient than other encodings that compute the full transitive closure of the graph.

This completes the definition of $\ilpproofdag{\bar t}{D}{Q}$ and we can formally show that is an ILP encoding of why-provenance:

%The construction of $\ilpproofdag{\bar t}{D}{Q}$ is now complete, and we can prove the key result of this section:

\begin{theorem}\label{thm:dag-correct}
	Consider a Datalog query $Q = (\dep,R)$, a database $D = \{\alpha_1,\ldots,\alpha_n\}$ over $\esch{\dep}$, and a tuple $\bar t \in \adom{D}^{\arity{R}}$. The system of ILP constraints $\ilpproofdag{\bar t}{D}{Q}$ is an ILP encoding of $\allwhy{\bar t}{D}{Q}$ with database variables $x_{\alpha_1^1},\ldots,x_{\alpha^1_n}$.
\end{theorem}

By Theorem~\ref{thm:dag-correct} and Corollary~\ref{cor:generic-algorithm}, we obtain that:

\begin{corollary}\label{cor:dag-all-algorithm}
$\ComputeDiverse{\ilpproofdagfunc}$ solves the problem $\mathsf{ALL}\text{-}\mathsf{DiverseWhyProvenance}$.
\end{corollary}

%\medskip
\noindent
\underline{\textbf{Unambiguous Graph-based ILP Encoding}}

\smallskip
\noindent Consider a Datalog query $Q=(\dep,R)$, a database $D$ over $\esch{\dep}$, and a tuple $\bar t \in \adom{D}^{\arity{R}}$.
%, and assume $\downc{R(\bar t)}{D}{\dep} = (V,E)$. 
We proceed to show how to construct the ILP encoding of $\unwhy{\bar t}{D}{Q}$.
Recall that the only difference between the graph-based and the unambiguous graph-based characterization, both presented in Theorem~\ref{thm:arbitrary-proof-dag}, concerns only the number of nodes in the proof DAG that can be labeled with the same fact; for unambiguous proof DAGs this mumber is $1$.
%We observe that the only difference between the graph-based characterization of Theorem~\ref{thm:arbitrary-proof-dag}, and the unambiguous graph-based characterization of Theorem~\ref{thm:unambiguous-proof-dag} regards only the number of nodes in the proof DAG that can be labeled with the same fact; for unambiguous proof DAGs, this bound is one. 
Hence, we can define the ILP encoding of $\unwhy{\bar t}{D}{Q}$, called $\ilpunproofdag{\bar t}{D}{Q}$, by adopting the same construction of $\ilpproofdag{\bar t}{D}{Q}$ with the difference that each fact $\alpha$ admits exactly one amplified copy, i.e., the amplified copy $\alpha^1$. We can the show the following:
%Our last key result follows:

\begin{theorem}\label{thm:unambiguous-dag-correct}
	Consider a Datalog query $Q = (\dep,R)$, a database $D = \{\alpha_1,\ldots,\alpha_n\}$ over $\esch{\dep}$, and a tuple $\bar t \in \adom{D}^{\arity{R}}$. The system of ILP constraints $\ilpunproofdag{\bar t}{D}{Q}$ is an ILP encoding of $\unwhy{\bar t}{D}{Q}$ with database variables $x_{\alpha_1^1},\ldots,x_{\alpha^1_n}$.
\end{theorem}

By Theorem~\ref{thm:unambiguous-dag-correct} and Corollary~\ref{cor:generic-algorithm}, we obtain that:

\begin{corollary}\label{cor:dag-un-algorithm}
$\ComputeDiverse{\ilpunproofdagfunc}$ solves the problem $\mathsf{UN}\text{-}\mathsf{DiverseWhyProvenance}$.
\end{corollary}
\section{Exact Computation of Diverse Why-Provenance}
\label{sec:exact-experiments}
In the previous section, for each $\star \in \{\mathsf{ALL}, \mathsf{UN}\}$, we introduced two different strategies for constructing ILP encodings for $\starwhy{\bar t}{D}{Q}$, given a Datalog query $Q = (\dep,R)$, a database $D$ over $\esch{\dep}$, and a tuple $\bar t \in \adom{D}^{\arity{R}}$. This in turn led to four concrete instantiations of the generic algorithm $\ComputeDiverse{\ConstructILPEncoding}$ of Algorithm~\ref{alg:compute-diverse} discussed in Section~\ref{sec:exact-encoding}, by choosing $\ConstructILPEncoding \in \{\ilphypergraphfunc, \ilpproofdagfunc, \ilpunhypergraphfunc, \ilpunproofdagfunc\}$. This provides two algorithms for solving $\mathsf{ALL}\text{-}\mathsf{DiverseWhyProvenance}$, and two algorithms for solving $\mathsf{UN}\text{-}\mathsf{DiverseWhyProvenance}$.
The goal of this section is to experimentally evaluate those algorithms. 

In the rest of the section, we will first discuss how we implemented the algorithms in question, and then present an experimental evaluation of such implementations. The latter will allow us to assess the running time of our implementations on well-established benchmarks, and, in particular, how the running time scales w.r.t.~the requested number $k$ of members of $\starwhy{\bar t}{D}{Q}$.

%We evaluate the runtime of the \emph{exact} algorithm for computing a set of \(k\) diverse why-provenance explanations using the ILP encoding in Section~\ref{sec:characterizations}. Our goal is to understand how this algorithm scales with \(k\) across databases and configurations, and to identify the regimes where exact optimization remains practical.
%
%We consider three configurations for the algorithm combining proof-tree types and ILP characterizations: (\arbitrary, \hypergraph), (\unambiguous, \hypergraph), and (\unambiguous, \dagEnc). We exclude (\arbitrary, \dagEnc) due to its poor results. In preliminary experiments, we constructed the corresponding ILP instances and evaluated them under the same experimental setting used throughout the paper, including a timeout limit of 10 minutes per instance. However, none of the tested instances in this configuration completed within the timeout. The main reason is the rapid growth of the ILP formulation caused by the need for multiple copies of intensional nodes, which leads to extremely large optimization problems in both time and memory consumption.

\subsection{Implementation Details}
\label{sec:implementation-details}

In what follows, fix a family of proof trees $\star \in \{\mathsf{ALL}, \mathsf{UN}\}$, and a desired charaterization $\gamma \in \{\mathsf{hyper}, \mathsf{graph}\}$. We recall that for a given Datalog query $Q = (\dep,R)$, database $D$, and tuple $\bar t$, constructing the ILP encoding $\ilpstargamma{\bar t}{D}{Q}$ of $\starwhy{\bar t}{D}{Q}$ requires first the construction of the downward closure $\downc{R(\bar t)}{D}{\dep}$. Hence, the main tasks that the algorithm $\ComputeDiverse{\ilpstargammafunc}$ needs to perform can be summarized as follows:
\begin{enumerate}
\item compute $\downc{R(\bar t)}{D}{\dep}$,
\item compute the encoding $T = \ilpstargamma{\bar t}{D}{Q}$, together with the $(T,k)$-problem $(f,S)$, by using $\downc{R(\bar t)}{D}{\dep}$, and 
\item solve $(f,S)$.
\end{enumerate}
We briefly describe how we implemented the above steps. 

\medskip
\noindent
\textbf{Constructing the Downward Closure.}
Regarding the downward closure, this is a data structure that is orthogonal to the specific ILP encoding employed, and it has already been used in previous works on why-provenance for Datalog~\cite{ElKM22,CLPS24aaai}. In particular, we employ an existing implementation from~\cite{CLPS24aaai} for computing the downward closure which in turn exploits the DLV Datalog engine~\cite{adrian2018dlv}.

\medskip
\noindent
\textbf{Constructing the ILP Encoding.}
The construction of the ILP encoding $T = \ilpstargamma{\bar t}{D}{Q}$ and the corresponding $(T,k)$-problem mostly follows their definition. The only differences are related to the database variables of $T$ and how we encode indicator constraints. Regarding the database variables, only the ones corresponding to facts appearing in the downward closure are actually generated in $T$ since, by definition of $T = \ilpstargamma{\bar t}{D}{Q}$, all other database variables are always set to $0$, in any solution of $T$, and thus the corresponding database facts will never appear in the output of $\ComputeDiverse{\ilpstargammafunc}$. Regarding the indicator constraints, these only occur in $T$ when $\gamma = \mathsf{hyper}$, and they are all of the form
$$ z\ \equiv\ (\ a_1 x_1 + \cdots a_n x_n\ ) \ge 1,$$
where the expression $a_1 x_1 + \cdots a_n x_n$ can never be negative and it is always an integer. Hence, the indicator constraint essentially asks that $z$ is $1$ when $a_1 x_1 + \cdots a_n x_n$ is non-zero, and zero when $a_1 x_1 + \cdots a_n x_n$ is zero. For
efficiency reasons, we encode the constraint by means of the following set of two constraints:
\[
\{z \cdot (\ a_1 x_1 + \cdots a_n x_n\ ) = a_1 x_1 + \cdots a_n x_n,\,\,\, z \le a_1 x_1 + \cdots a_n x_n\}.
 \]
%\begin{multline*}
%\{z \cdot (\ a_1 x_1 + \cdots a_n x_n\ ) = a_1 x_1 + \cdots a_n x_n,\\ 
%z \le a_1 x_1 + \cdots a_n x_n\}.
%\end{multline*}
%
We note that the first constraint is non-linear, due to the product on the left of the equality, and it might seem that this would result in a harder encoding w.r.t.\ a standard encoding via linear constraints. However, we recall that the standard encoding for the indicator constraint via linear constraints requires finding a constant $M$ that is guaranteed to be larger than $a_1 x_1 + \cdots a_n x_n$, for \emph{all} values of $x_1,\ldots,x_n$. A preliminary analysis has shown that such a constant $M$ would be extremely large, and caused numerical instabilities while solving the ILP encodings, and thus cannot be employed for our implementation. 
Nonetheless, as we are going to see next, we will exploit a state-of-the-art solver that natively supports quadratic constraints, and thus, we can employ our alternative encoding.

\medskip
\noindent
\textbf{Finding Optimal Solutions.}
For solving the $(T,k)$-problem $(f,S)$ built at the previous step, we use Gurobi~\cite{gurobi}, which is a state-of-the-art Mixed-Integer (Quadratic) Programming solver.

\begin{table}[t]
	\centering
	\resizebox{\columnwidth}{!}{%
	\begin{tabular}{lccc}
		\toprule
		Name & Databases & \#Rules & Rec. \\
		\midrule
		
		\multirow{2}{*}{%
		\begin{tabular}{@{}l@{}}
		\graphanalytics \\
		(transitive closure)
		\end{tabular}}
		& \bitcoin~\cite{kumar2016edge} \text{(235K)} & \multirow{2}{*}{2} & \multirow{2}{*}{Yes} \\
		& \facebook~\cite{snap} \text{(88.2K)} & & \\
		
		\midrule
		
		\doctors 1-7 (7 scenarios)~\cite{ElKM22}
		& \factsk \text{(100K)} & 6 & No \\
		
		\midrule
		
		\multirow{2}{*}{\galen~\cite{ElKM22}}
		& \fifteen \text{(26.5K)}, \twentyfive \text{(30.5K)} & \multirow{2}{*}{14} & \multirow{2}{*}{Yes} \\
		& \forty \text{(67K)}, \fifty \text{(82K)} & & \\
		
		\midrule
		
		\multirow{2}{*}{\csda~\cite{fan2019scaling}}
		& \httpd \text{(10M)}, \linux \text{(44M)} & \multirow{2}{*}{2} & \multirow{2}{*}{Yes} \\
		& \postgres \text{(34.8M)} & & \\
		
		\bottomrule
	\end{tabular}%
	}
	\medskip
	\caption{Experimental scenarios. For each database, we give in parenthesis the number of tuples. ``Rec.'' indicates whether the rule set is recursive.}
	\label{tab:scenarios}
\end{table}

\subsection{Experimental Setup} 

\noindent \textbf{System Details.}
All experiments were conducted on a server running Fedora Linux 39 (Cloud Edition), provisioned with 20 vCPUs (Intel Xeon Processor, Cascadelake) and 117GB of RAM. The algorithms $\ComputeDiverse{\ilpstargammafunc}$, with $\star \in \{\mathsf{ALL}, \mathsf{UN}\}$ and $\gamma \in \{\mathsf{hyper}, \mathsf{graph}\}$, are implemented in C++, apart from the step solving the maximization problem $(f,S)$ that is delegated to Gurobi 11.0.0~\cite{gurobi}. We only use Java OpenJDK 21.0.11 as a thin wrapper around the C++ implementation to ease the logging of the results.

\medskip
\noindent \textbf{Experimental Scenarios.} 
We consider a variety of scenarios from the Datalog literature, where each scenario consists of a query $Q = (\dep,R)$ and a family of databases $\mathcal{D}$ over $\esch{\dep}$. All the considered scenarios are listed in Table~\ref{tab:scenarios}; the number of rules in $\dep$, whether the program $\dep$ is recursive or not, and further references about each scenario are given in the table as well. Let us provide further details about the adopted scenarios:
%Further details about the scenarios can be found in the appendix.

\begin{description}
	\item[$\mathsf{TransClosure}$.] This scenario computes the transitive closure of a graph and asks for connected nodes. The database $\mathsf{bitcoin}$ stores a portion of the Bicoin network~\cite{Weber19}, and
	%with nodes representing transactions and directed edges representing flow of money from one transaction to the other~\cite{Weber19}. 
	$\mathsf{facebook}$ stores different ``social circles'' from Facebook~\cite{McAuley12}.
	%, where nodes represent users and an edge from $u$ to $v$ denotes that $u$ has chosen $v$ to be part of her social circle~\cite{McAuley12}.
	%
	%The scenarios $\mathsf{DOCTORS\text{-}1},\ldots,\mathsf{DOCTORS\text{-}7}$, and $\mathsf{GALEN}$ from~\cite{ElhalawatiKM22}, which have been used to experimentally evaluate the computation of ondemand why-provenance for standard proof trees, and the scenarios $\mathsf{ANDERSEN}$ and $\mathsf{CSDA}$ from~\cite{FanMK22}.

	\item[$\mathsf{Doctors}$.] The scenarios $\mathsf{doctors\text{-}}i$, for $i \in [7]$, represent queries obtained from a well-known data-exchange benchmark involving existential rules (the existential variables have been replaced with fresh constants). All such scenarios share the same database with 100K facts.

	\item[$\mathsf{Galen}$.] This scenario implements the ELK calculus~\cite{KazakovKS14} and asks for all pairs of concepts that are related with the $\mathsf{subClassOf}$ relation. The various databases contain different portions of the Galen ontology~\cite{Galen}.
	
	% \item[$\mathsf{Andersen}$.] This scenario used in~\cite{FanMK22} implements the classical Andersen ``points-to'' algorithm for determining the flow of data in procedural programs and asks for all the pairs of a pointer $p$ and a variable $v$ such that $p$ points to $v$. The databases are encodings of program statements of different length.

	\item[$\mathsf{CSDA}$.] This scenario (Context-Sensitive Dataflow Analysis) implements the Andersen ``points-to'' algorithm for determining the flow of data in procedural programs and asks for null references in a program. The databases $\mathsf{httpd}$, ${\mathsf{postgresql}}$, and ${\mathsf{linux}}$ store the statements of the httpd web server, the PostgreSQL DBMS, and the Linux kernel, respectively. 
\end{description}

\medskip
\noindent \textbf{Execution.} For each scenario $s$, consisting of a query $Q = (\dep,R)$ and a family of databases $\mathcal{D}$, we first computed $Q(D)$, for each database $D \in \mathcal{D}$, using an off-the-shelf Datalog engine~\cite{adrian2018dlv}, and selected $100$ tuples $\bar t^1_{s,D},\ldots,\bar t^{100}_{s,D}$ from $Q(D)$ uniformly at random.
Then, for the experimental evaluation, for each of the four combinations of proof tree family $\star \in \{\mathsf{ALL}, \mathsf{UN}\}$ and characterization $\gamma \in \{\mathsf{hyper},\mathsf{graph}\}$, each scenario $s$, each database $D$ of the scenario, each tuple $t^i_{s,D}$, with $i \in [100]$, and each $k \in [5]$, we executed our implementation of the algorithm $\ComputeDiverse{\ilpstargammafunc}$ with input $Q$, $D$, $\bar t^i_{s,D}$, and $k$. We collected the total running time, and the diversity of the resulting set, providing that the execution did not fail. Each execution had a timeout of 10 minutes.

%and $Q$, we compute $Q(D)$ using DLV~\cite{adrian2018dlv} and uniformly select 100 answer tuples at random from the query result. For each selected tuple, we construct its downward closure, i.e., the sub-hypergraph capturing all possible derivations relevant to the tuple. Following~\cite{CLPS24aaai}, the downward closure is computed by evaluating a modified Datalog query over a modified database, whose answers coincide with the hyperedges of the downward closure. This allows us to use the same Datalog engine both for query answering and for explanation construction. We then generate the corresponding ILP formulation and solve the resulting optimization problem using Gurobi to compute a set of $k$ explanations maximizing diversity. Each run is subject to a timeout of 10 minutes per tuple.
%
%\medskip
%\noindent \textbf{System Details.} All experiments were conducted on a server running Fedora Linux 39 (Cloud Edition), provisioned with 20 vCPUs (Intel Xeon Processor, Cascadelake) and 117GB of RAM. The ILP models are constructed in C++ and solved using Gurobi11.0.0~\cite{gurobi}. Python 3.12.7 is used for preprocessing.

\subsection{Experimental Results} \label{sec:exp-results-exact}
Before we discuss the results of our experimental evaluation, we point out that for the ILP encoding $\ilpproofdagfunc$, all executions of the algorithm $\ComputeDiverse{\ilpproofdagfunc}$ did not terminate withing the given timeout. This is due to the extremely high number of variables that are required by the ILP encoding $\ilpproofdagfunc$. The latter is due to the very high number of amplified copies of facts and hyperedges. For this reason, although we will discuss the results related to all our algorithms, we will only present plots for the remaining 3 algorithms $\ComputeDiverse{\ilphypergraphfunc}$, $\ComputeDiverse{\ilpunhypergraphfunc}$, and $\ComputeDiverse{\ilpunproofdagfunc}$.

\medskip
\noindent
\textbf{Visualizing Our Results.}
For each scenario $s=(Q,\mathcal{D})$, database $D \in \mathcal{D}$, $k \in [5]$, and algorithm $\ComputeDiverse{\ilpstargammafunc}$, with $(\gamma,\star) \in \{(\mathsf{hyper},\mathsf{ALL})$, $(\mathsf{hyper},\mathsf{UN})$, $(\mathsf{graph},\mathsf{UN})\}$, we report the runtimes of the executions of the algorithm $\ComputeDiverse{\ilpstargammafunc}$ with input the query $Q$, the database $D$, one of the tuples $\bar t^i_{s,D}$ for $i \in [100]$, and the parameter $k$, in a dedicated box-plot. In particular, each box-plot reports only the runtimes for the tuples $\bar t^i_{s,D}$ that have at least 5 members in $\starwhy{\bar t^i_{s,D}}{D}{Q}$. The reason for this is that answer tuples with, e.g., no more than $2$ explanations would be reported with an identical running time, for all $2 < k \le 5$, thus biasing the comparison across different values of $k$. Nonetheless, out of the 100 randomly selected answer tuples, this filtering retained 6 tuples for \bitcoin, 68 for \facebook, 77 for \httpd, 57 for \linux, 55 for \postgres, and all 100 tuples for each of the seven \doctors scenarios. We treat \galen differently because most tuples $\bar t^i_{s,D}$ had only one member in $\starwhy{\bar t^i_{s,D}}{D}{Q}$. Thus, for \galen, we report execution times only for $k=1$, retaining all 100 sampled tuples for each of the databases $D_{15}$, $D_{25}$, $D_{40}$, and $D_{50}$.

% Nonetheless, for all scenarios, out of the 100 randomly selected answer tuples, we could keep at least \textcolor{red}{XX} tuples, except for \galen, in which most of the tuples $\bar t^i_{s,D}$ had only one member in $\starwhy{\bar t^i_{s,D}}{D}{Q}$. Thus, for this scenario, we only report execution times for $k = 1$.

%
The box-plots are organized in figures and subfigures as follows. Each figure, from Figure~\ref{fig:graph-analytics-exact} to Figure~\ref{fig:csda-exact}, collects the box-plots relative to a specific scenario $s = (Q,\mathcal{D})$. Each figure has a subfigure for each database $D \in \mathcal{D}$ of the scenario $s$, which shows the box-plots relative to the execution of each of the algorithms $\ComputeDiverse{\ilpstargammafunc}$, with $(\gamma,\star) \in \{(\mathsf{hyper},\mathsf{ALL})$, $(\mathsf{hyper},\mathsf{UN})$, $(\mathsf{graph},\mathsf{UN})\}$, over the tuples $\bar t^i_{s,D}$ having at least 5 members of $\starwhy{\bar t^i_{s,D}}{D}{Q}$, for the different values of $k \in [5]$. Since for \galen we focus only on the case of $k = 1$, we group in each subfigure relative to a database $D \in \mathcal{D}$, all the box-plots relative to the three algorithms, executed with $k = 1$. In each box-plot, the bottom and the top borders represent the first and third quartile, i.e., the runtime under which 25\% and 75\% of all the runtimes occur, respectively, and the middle line represents the median runtime. Moreover, the bottom and the top whiskers represent the minimum and maximum runtime, respectively. Each box-plot also shows each individual runtime as a dot, where red dots denote timeouts, and the percentage of tuples for which a timeout was reached is shown in red on the box-plot as well. Finally, the overal distribution of the runtimes is shown around the box-plot by means of a violin plot; wider areas denote runtimes that are more frequent across the considered tuples. Runtimes are in milliseconds over a log-scale.

% \medskip
\noindent
\textbf{Algorithms Comparison.}
We first provide a relative comparison of the algorithms $\ComputeDiverse{\ilpstargammafunc}$, for $\star \in \{\mathsf{ALL},\mathsf{UN}\}$, and $\gamma \in \{\mathsf{hyper},\mathsf{graph}\}$. When focusing on $\star = \mathsf{ALL}$, and thus on the computation of diverse sets of $k$ members of the why-provenance relative to arbitrary proof trees, only the algorithm $\ComputeDiverse{\ilphypergraphfunc}$ is able to complete its execution in most of the scenarios, whereas $\ComputeDiverse{\ilpproofdagfunc}$ timeouts in all scenarios, for all answer tuples, as already discussed at the beginning of the section. Hence, $\ComputeDiverse{\ilphypergraphfunc}$ is the way to go for solving $\mathsf{ALL}\text{-}\mathsf{DiverseWhyProvenance}$.

The picture is different for $\star = \mathsf{UN}$, i.e., when focusing on the family of unambiguous proof trees. In this case, both algorithms $\ComputeDiverse{\ilpunhypergraphfunc}$ and $\ComputeDiverse{\ilpunproofdagfunc}$ perform similarly across all scenarios, with the exception of the \doctors\ scenarios, where the algorithm $\ComputeDiverse{\ilpunhypergraphfunc}$ is generally faster, while the algorithm $\ComputeDiverse{\ilpunproofdagfunc}$ exhibits higher runtimes and a higher number of timeouts. This difference is due to the fact that the ILP encoding $\ilpunproofdagfunc$ employs the set of constraints $T_{\mathsf{acyclic}}$ whose size is highly sensitive on how connected the downward closure is. This is a consequence of the vertex elimination technique employed in $T_{\mathsf{acyclic}}$. This is confirmed by Figure~\ref{fig:timeouts}, which reports the number of variables occurring in $T_\mathsf{acyclic}$ for each of the answer tuples considered in the three scenarios \doctors~1, \doctors~5, and \doctors~7, which are the only \doctors\ scenarios where $\ComputeDiverse{\ilpunproofdagfunc}$ starts exhibiting timeouts; the timeouts occur for $k \in \{2,3,4,5\}$, $k \in \{4,5\}$, and $k \in \{2,3,4,5\}$, respectively. In the plots, red dots denote answer tuples for which a timeout is reached. The box-plots clearly show that timeouts consistently occur for tuples in which the number of variables employed in $T_\mathsf{acyclic}$ is higher.

From the above discussion, we conclude that in the case of  unambiguous proof trees, the algorithms $\ComputeDiverse{\ilpunhypergraphfunc}$ and $\ComputeDiverse{\ilpunproofdagfunc}$ are generally comparable, but the first is to be preferred for a more predictable running time.

\begin{figure}[h]
	\centering
	\small
	
	%================= facebook =================%
	\begin{minipage}[t]{\linewidth}
		\centering
		\begin{subfigure}[t]{0.36\linewidth}
			\centering
			\safeincludegraphics[height=2.57cm]{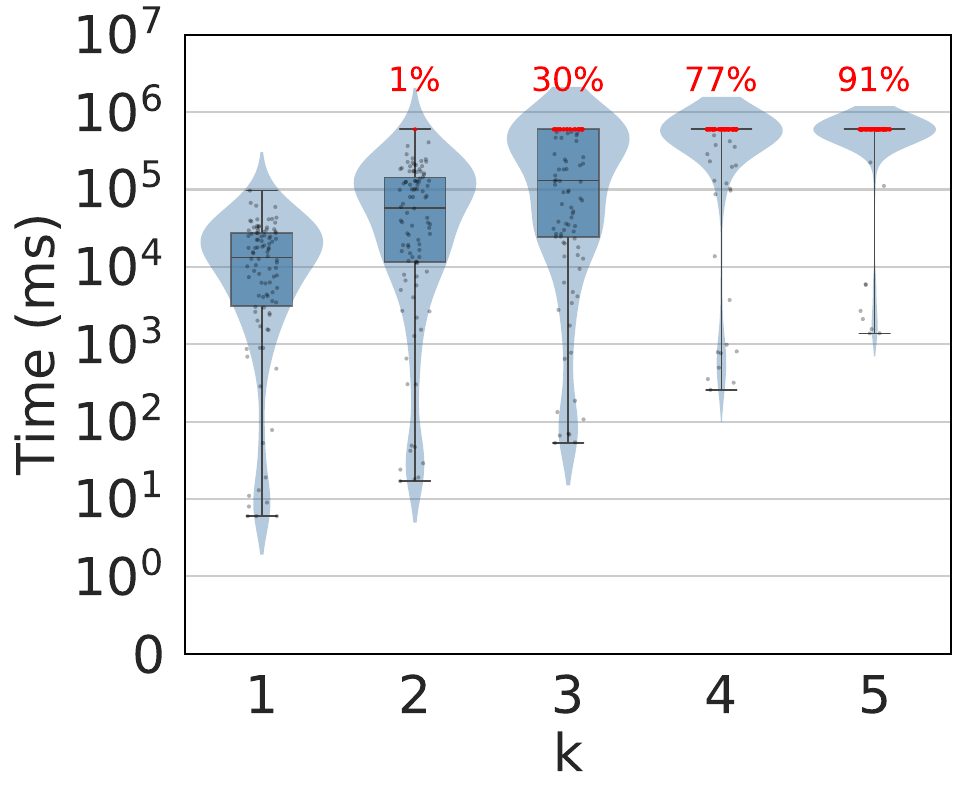}
			\caption{$\ilphypergraphfunc$}
		\end{subfigure}\hfill
		\begin{subfigure}[t]{0.30\linewidth}
			\centering
			\safeincludegraphics[height=2.5cm]{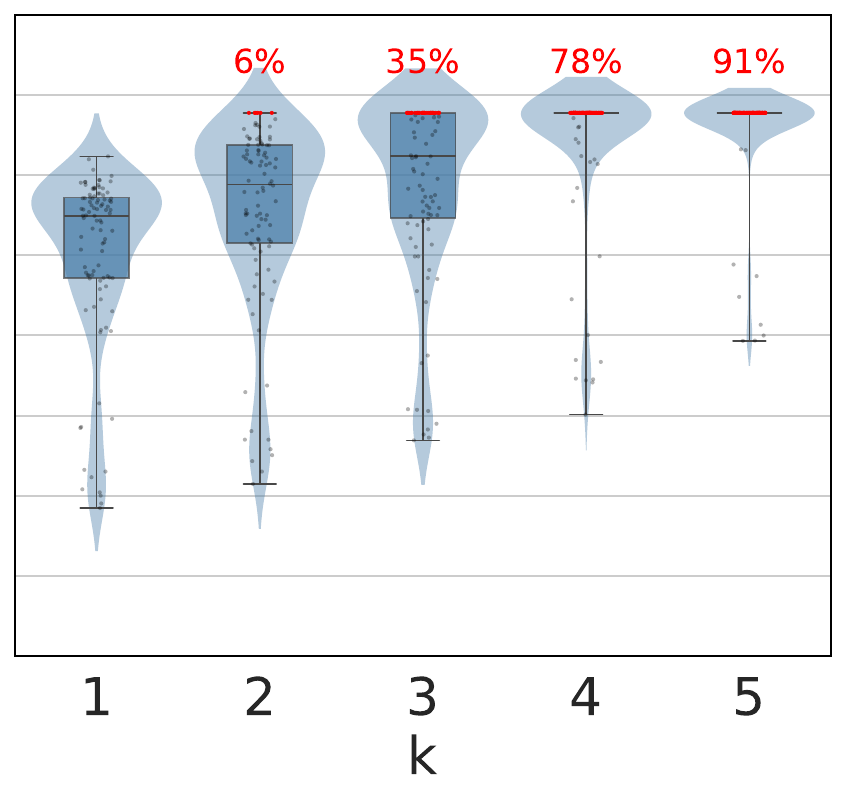}
			\caption{$\ilpunhypergraphfunc$}
		\end{subfigure}\hfill
		\begin{subfigure}[t]{0.30\linewidth}
			\centering
			\safeincludegraphics[height=2.5cm]{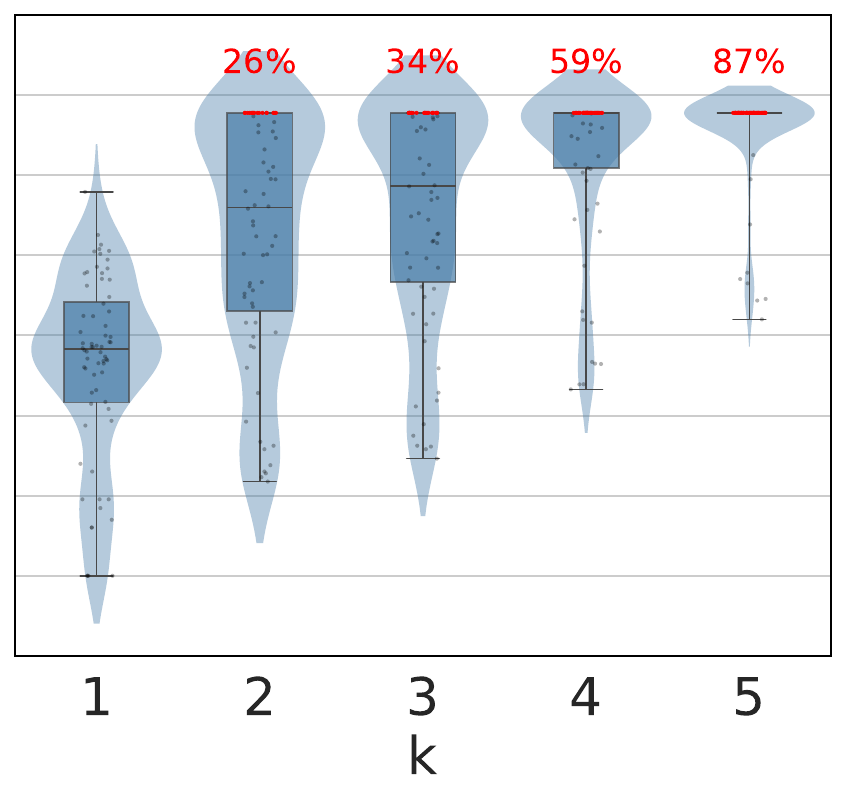}
			\caption{$\ilpunproofdagfunc$}
		\end{subfigure}
		\par\vspace{0.5em}
		\textbf{\graphanalytics\ (\facebook)}
	\end{minipage}
	
	\par\vspace{1em}
	
	%================= bitcoin =================%
	\begin{minipage}[t]{\linewidth}
		\centering
		\begin{subfigure}[t]{0.36\linewidth}
			\centering
			\safeincludegraphics[height=2.57cm]{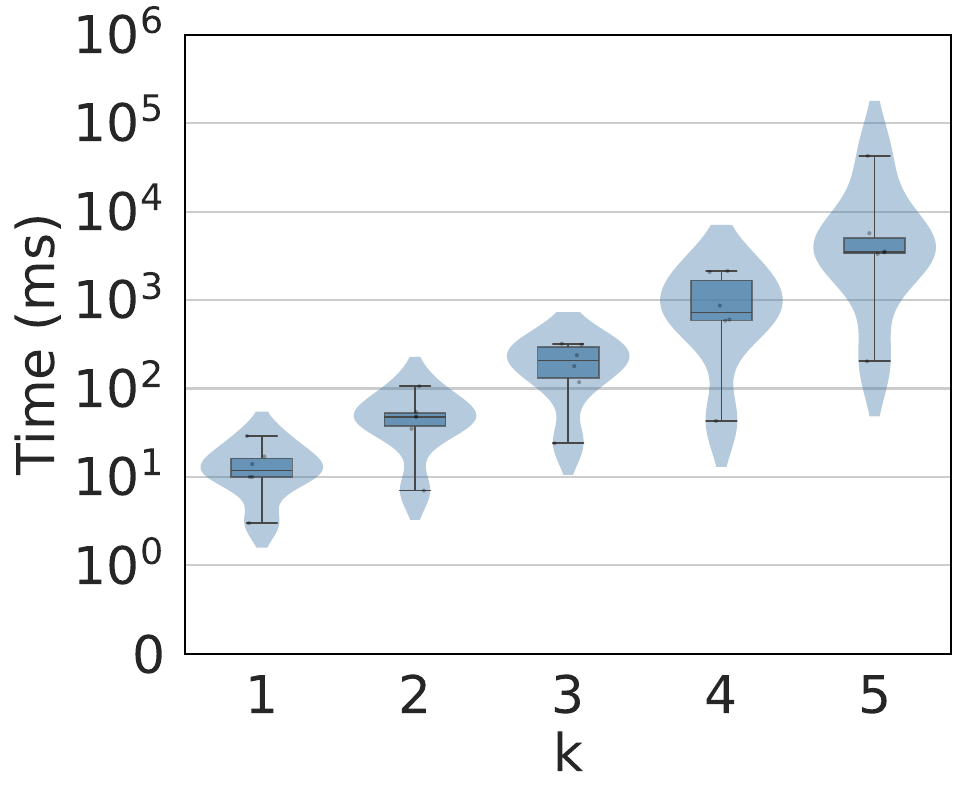}
			\caption{$\ilphypergraphfunc$}
		\end{subfigure}\hfill
		\begin{subfigure}[t]{0.30\linewidth}
			\centering
			\safeincludegraphics[height=2.5cm]{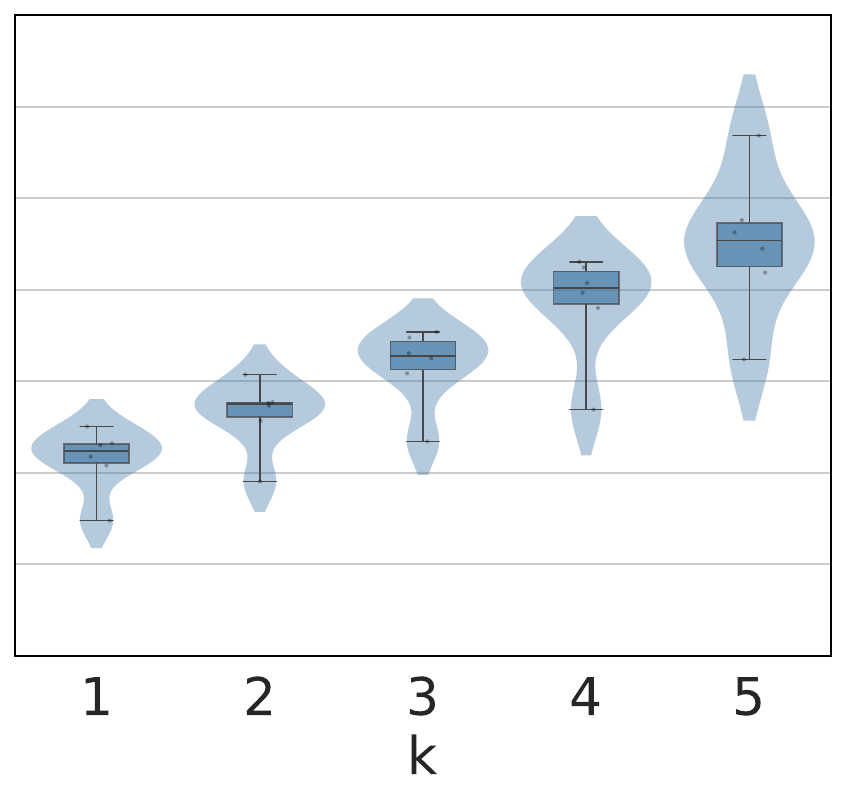}
			\caption{$\ilpunhypergraphfunc$}
		\end{subfigure}\hfill
		\begin{subfigure}[t]{0.30\linewidth}
			\centering
			\safeincludegraphics[height=2.5cm]{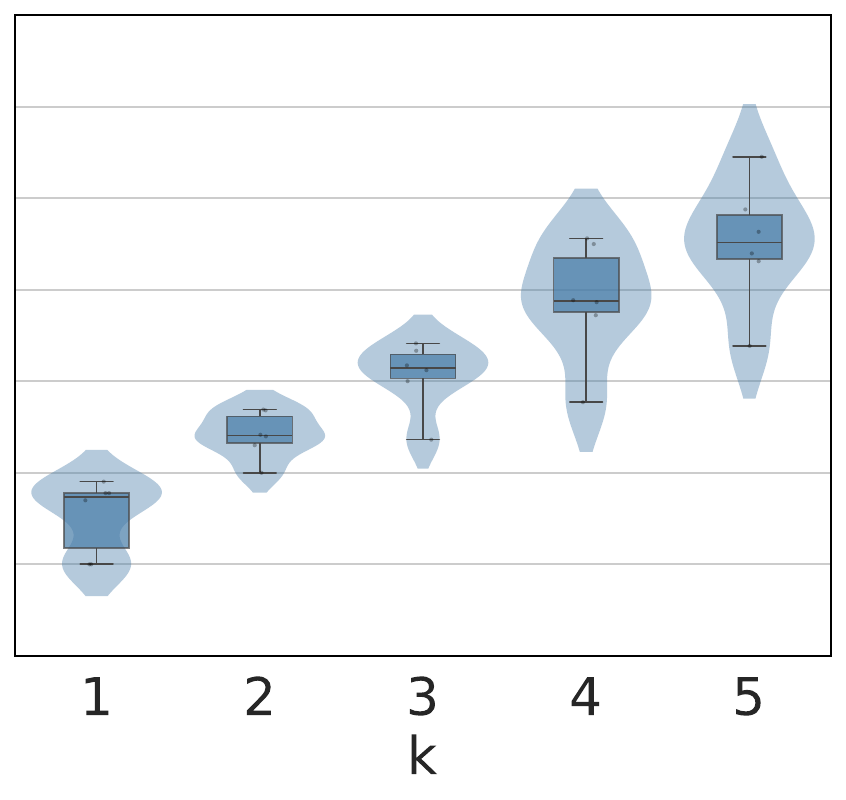}
			\caption{$\ilpunproofdagfunc$}
		\end{subfigure}
		\par\vspace{0.5em}
		\textbf{\graphanalytics\ (\bitcoin)}
	\end{minipage}
	
	\caption{Runtimes for the \graphanalytics\ scenario.}
	\label{fig:graph-analytics-exact}
\end{figure}

\begin{figure}[h]
	\centering
	\small
	
	%==================== doctors 1 ====================%
	\begin{minipage}[t]{\linewidth}
		\centering
		
		\begin{subfigure}[t]{0.36\linewidth}
			\centering
			\safeincludegraphics[height=2.57cm]{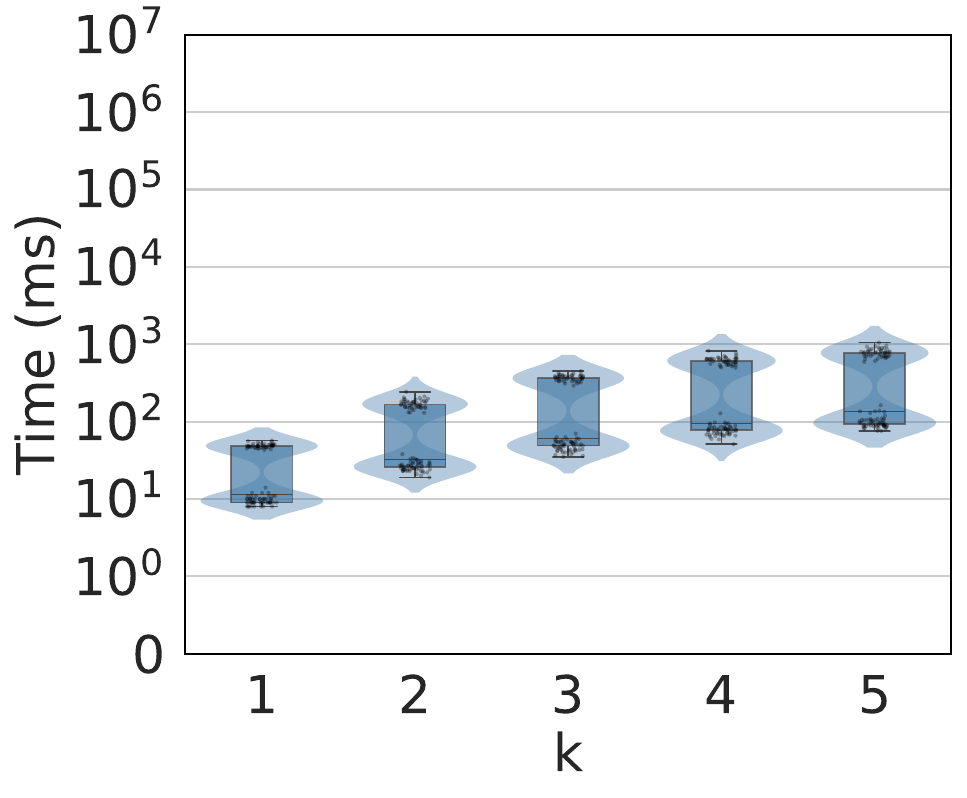}
			\caption{$\ilphypergraphfunc$}
		\end{subfigure}\hfill
		\begin{subfigure}[t]{0.30\linewidth}
			\centering
			\safeincludegraphics[height=2.5cm]{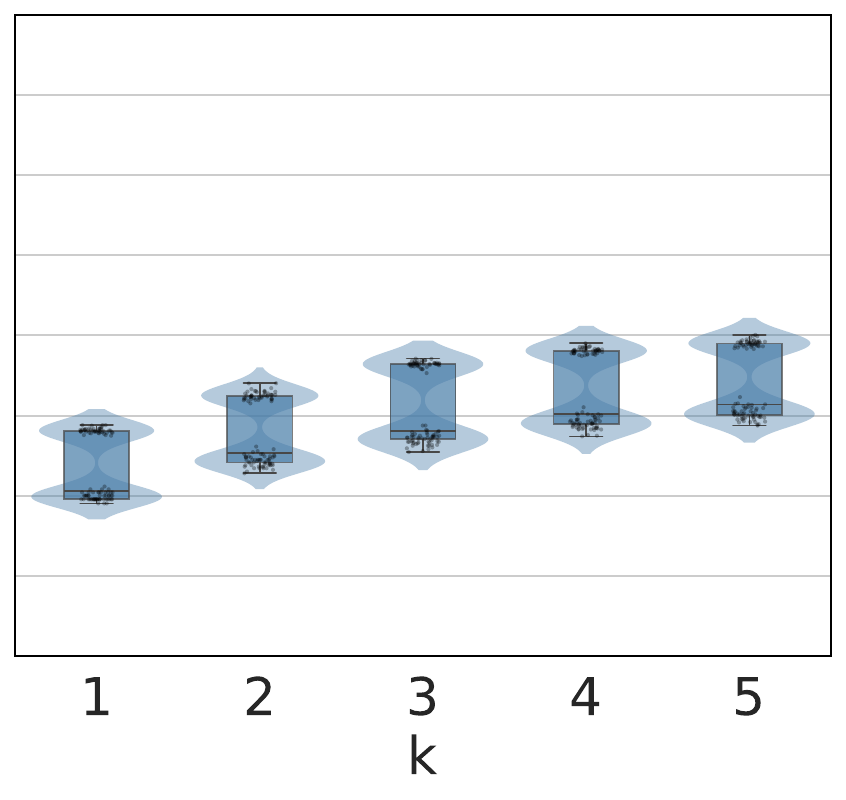}
			\caption{$\ilpunhypergraphfunc$}
		\end{subfigure}\hfill
		\begin{subfigure}[t]{0.30\linewidth}
			\centering
			\safeincludegraphics[height=2.5cm]{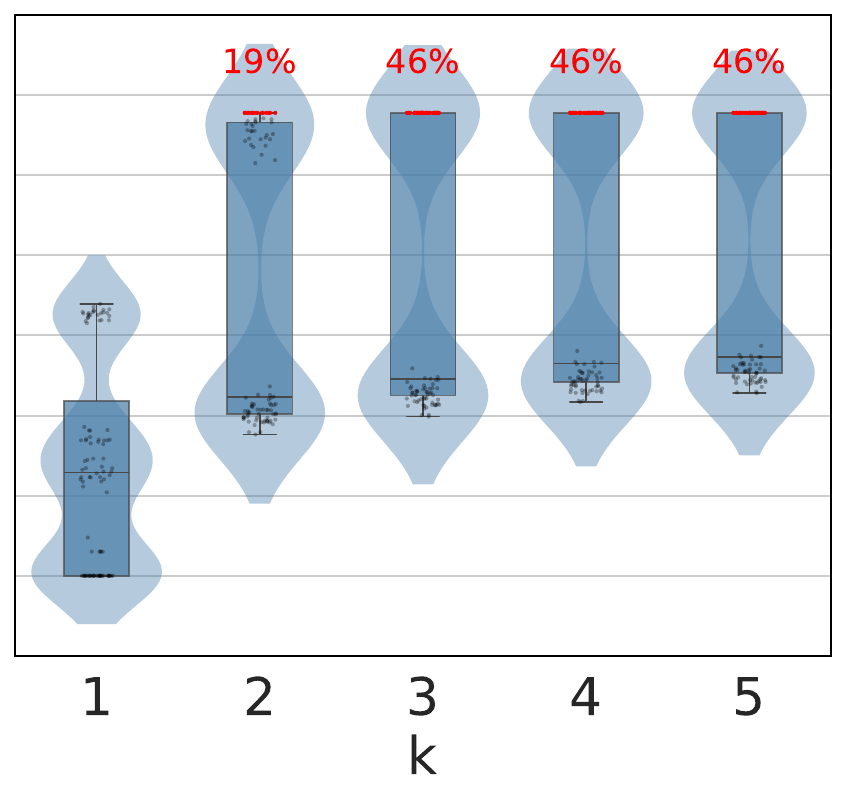}
			\caption{$\ilpunproofdagfunc$}
		\end{subfigure}
		
		\par\vspace{1em}
		\textbf{\doctors 1 (100k)}
	\end{minipage}
	
	\par\vspace{1em}
	
	%================= doctors 2 =================%
	\begin{minipage}[t]{\linewidth}
		\centering
		
		\begin{subfigure}[t]{0.36\linewidth}
			\centering
			\safeincludegraphics[height=2.57cm]{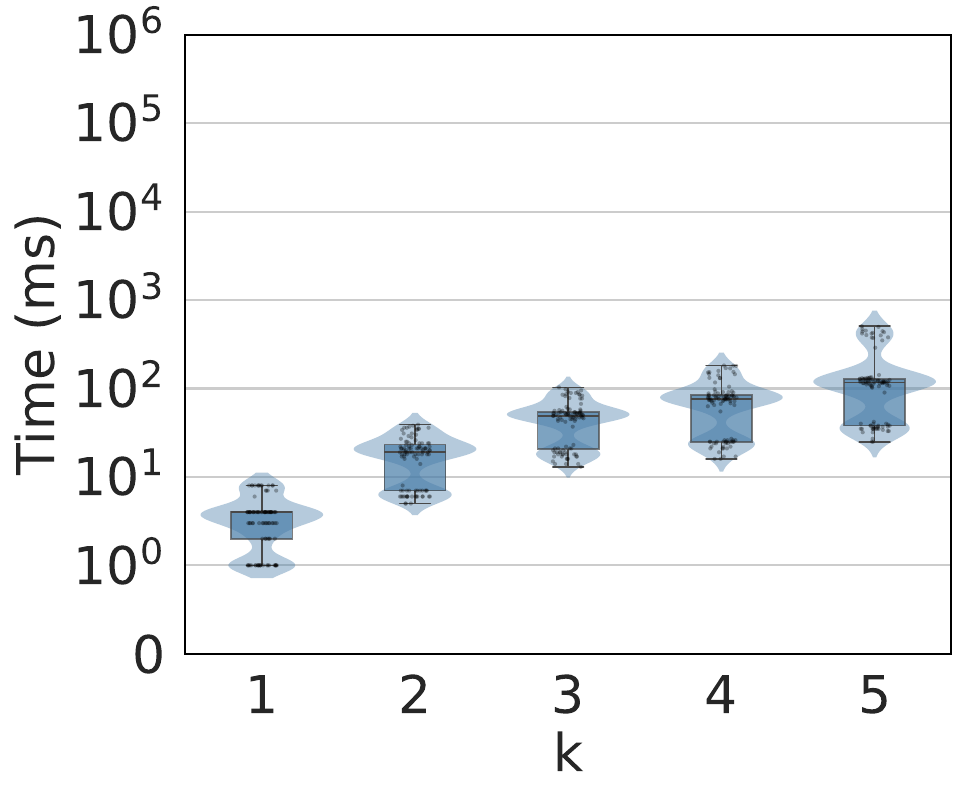}
			\caption{$\ilphypergraphfunc$}
		\end{subfigure}\hfill
		\begin{subfigure}[t]{0.30\linewidth}
			\centering
			\safeincludegraphics[height=2.5cm]{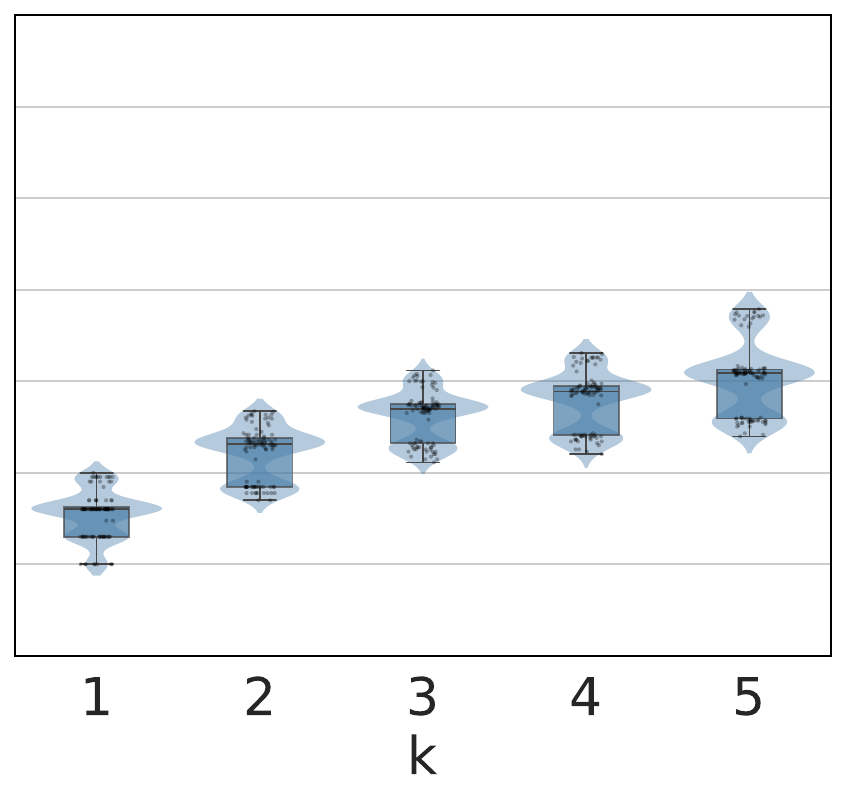}
			\caption{$\ilpunhypergraphfunc$}
		\end{subfigure}\hfill
		\begin{subfigure}[t]{0.30\linewidth}
			\centering
			\safeincludegraphics[height=2.5cm]{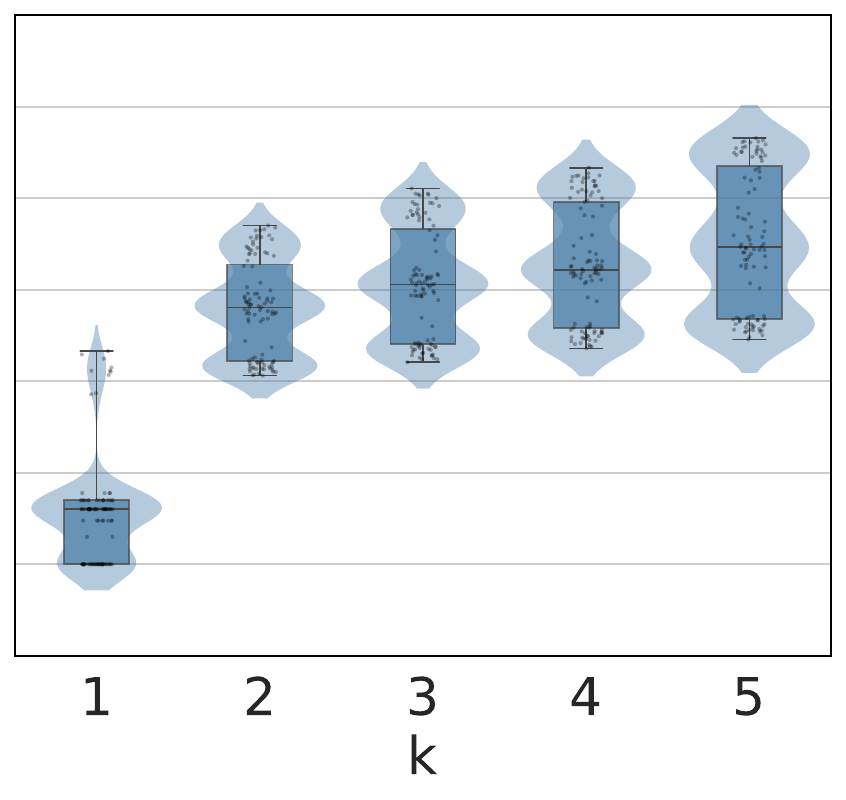}
			\caption{$\ilpunproofdagfunc$}
		\end{subfigure}
		
		\par\vspace{1em}
		\textbf{\doctors 2 (100k)}
	\end{minipage}
	
	\par\vspace{1em}
	
	%================= doctors 3 =================%
	\begin{minipage}[t]{\linewidth}
		\centering
		\begin{subfigure}[t]{0.36\linewidth}
			\centering
			\safeincludegraphics[height=2.57cm]{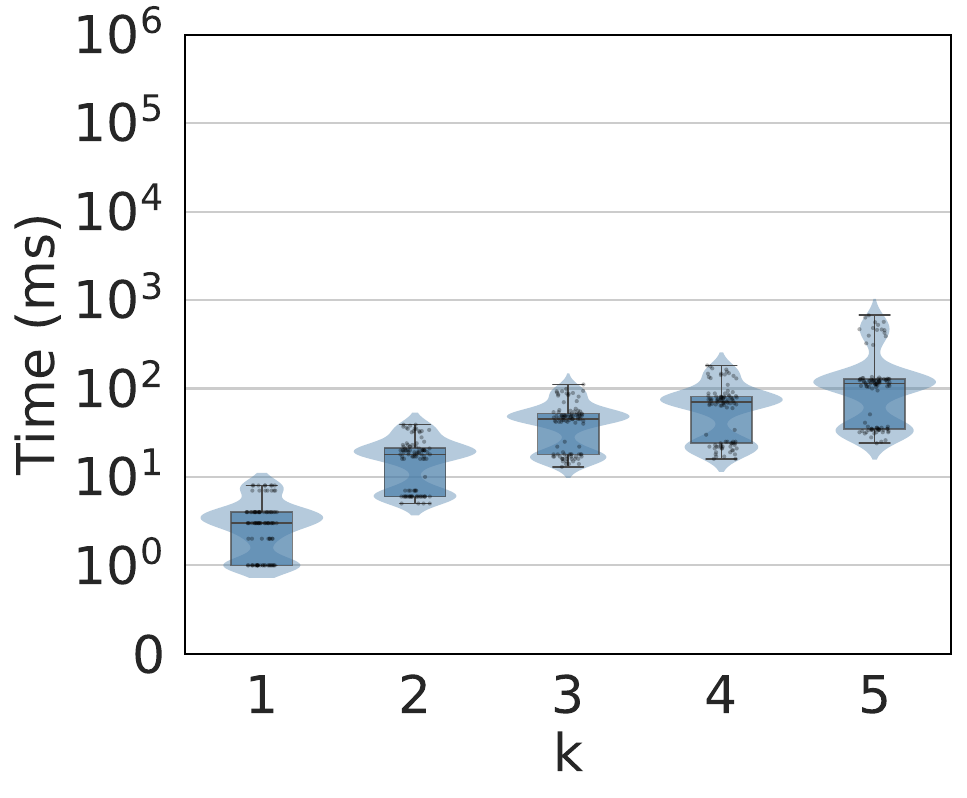}
			\caption{$\ilphypergraphfunc$}
		\end{subfigure}\hfill
		\begin{subfigure}[t]{0.30\linewidth}
			\centering
			\safeincludegraphics[height=2.5cm]{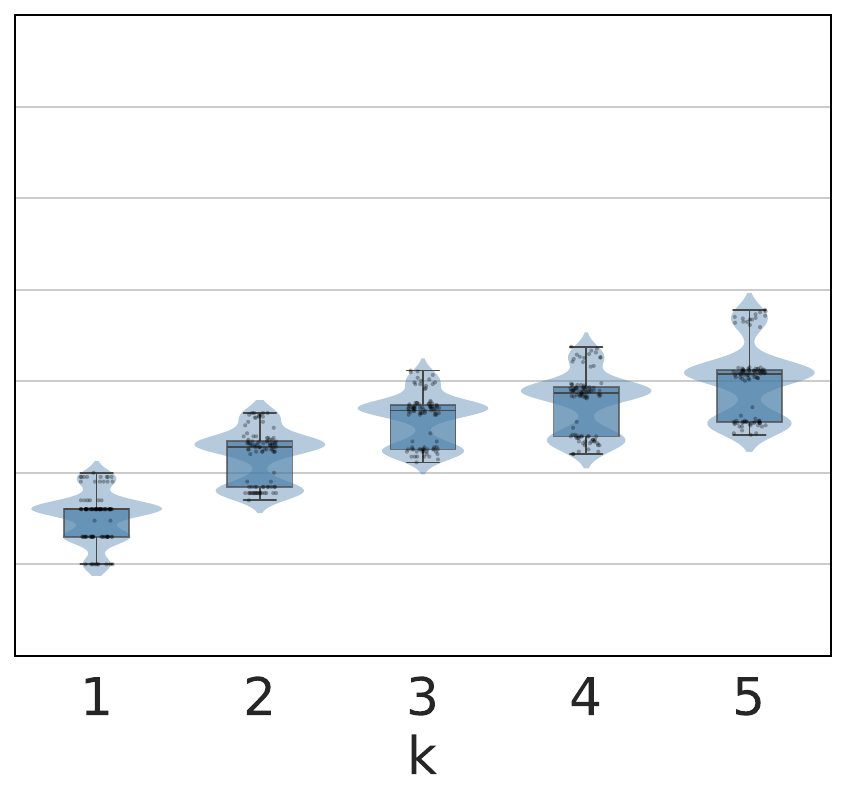}
			\caption{$\ilpunhypergraphfunc$}
		\end{subfigure}\hfill
		\begin{subfigure}[t]{0.30\linewidth}
			\centering
			\safeincludegraphics[height=2.5cm]{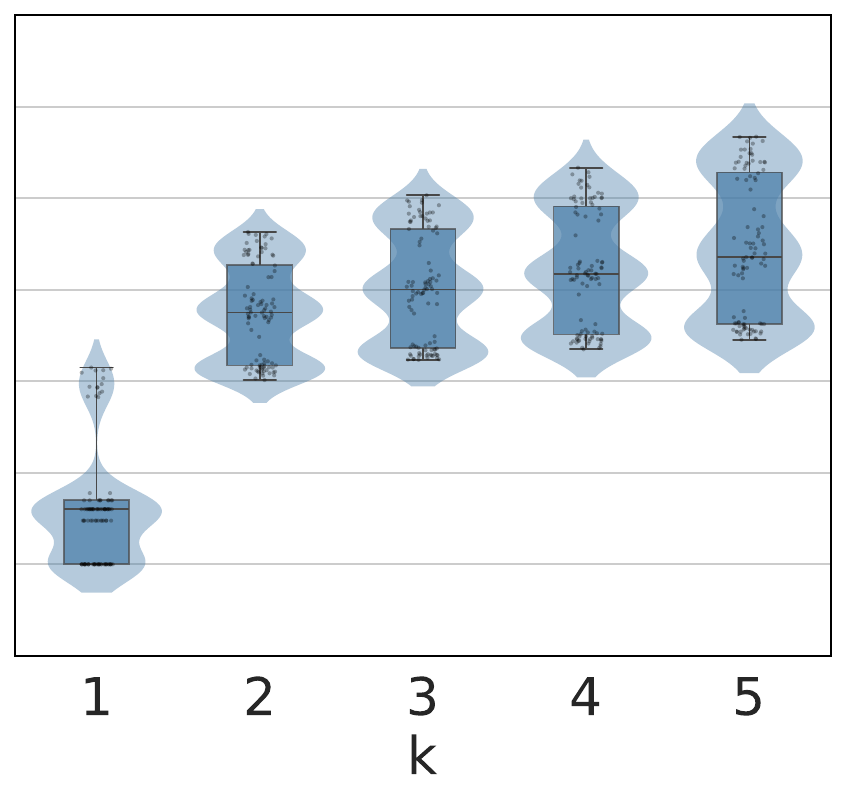}
			\caption{$\ilpunproofdagfunc$}
		\end{subfigure}
		
		\par\vspace{1em}
		\textbf{\doctors 3 (100k)}
	\end{minipage}
	
	\par\vspace{1em}
	
	%================= doctors 4 =================%
	\begin{minipage}[t]{\linewidth}
		\centering
		\begin{subfigure}[t]{0.36\linewidth}
			\centering
			\safeincludegraphics[height=2.57cm]{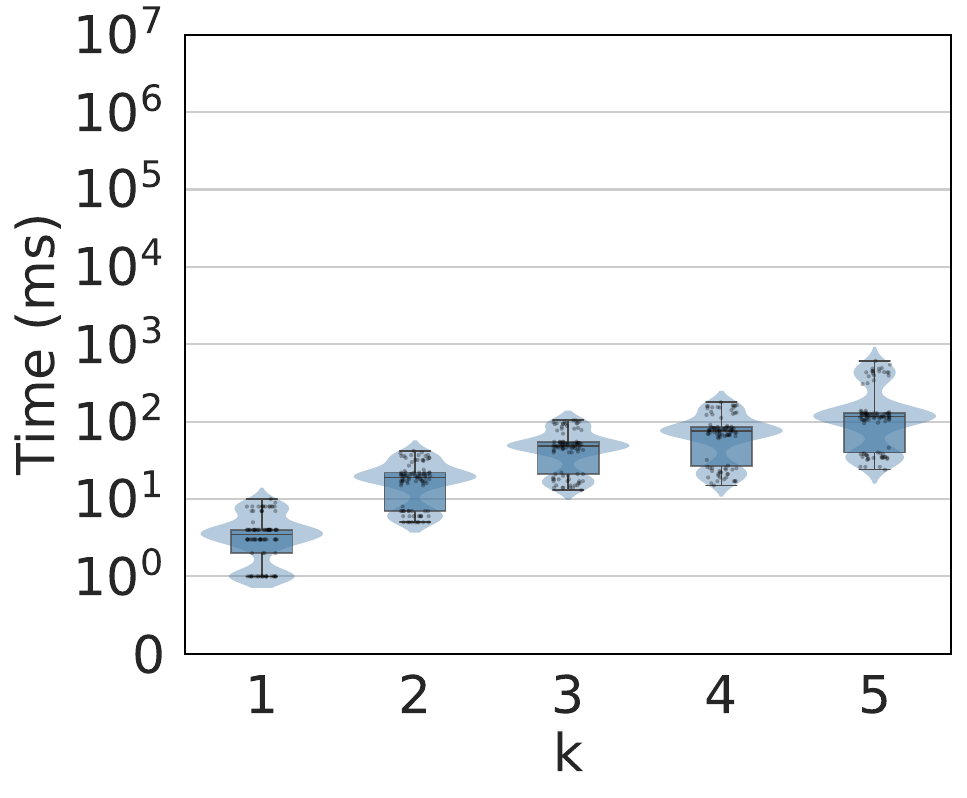}
			\caption{$\ilphypergraphfunc$}
		\end{subfigure}\hfill
		\begin{subfigure}[t]{0.30\linewidth}
			\centering
			\safeincludegraphics[height=2.5cm]{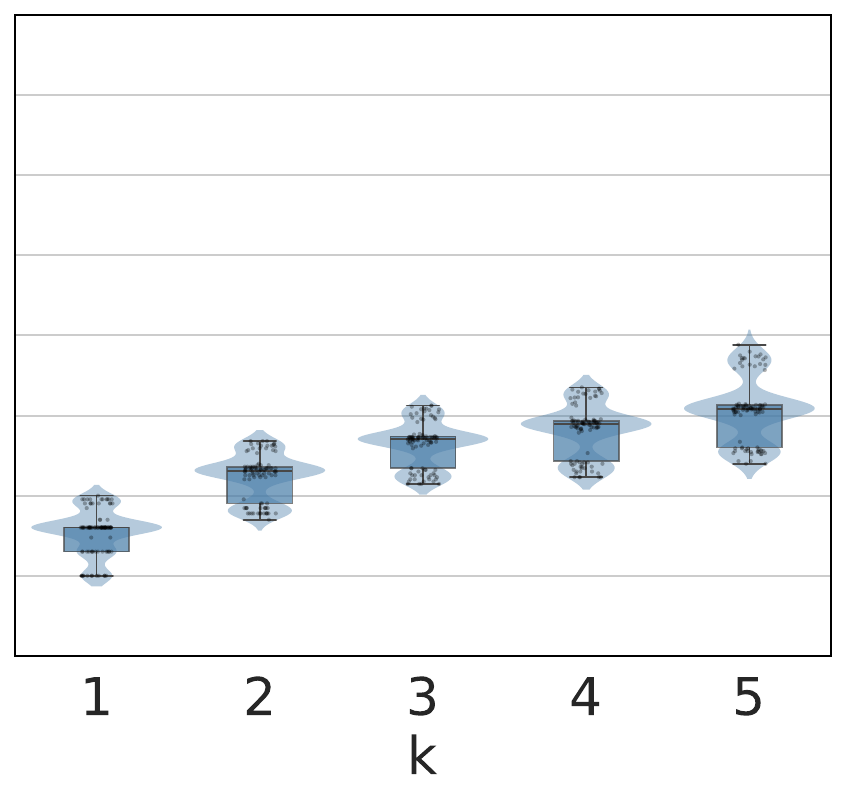}
			\caption{$\ilpunhypergraphfunc$}
		\end{subfigure}\hfill
		\begin{subfigure}[t]{0.30\linewidth}
			\centering
			\safeincludegraphics[height=2.5cm]{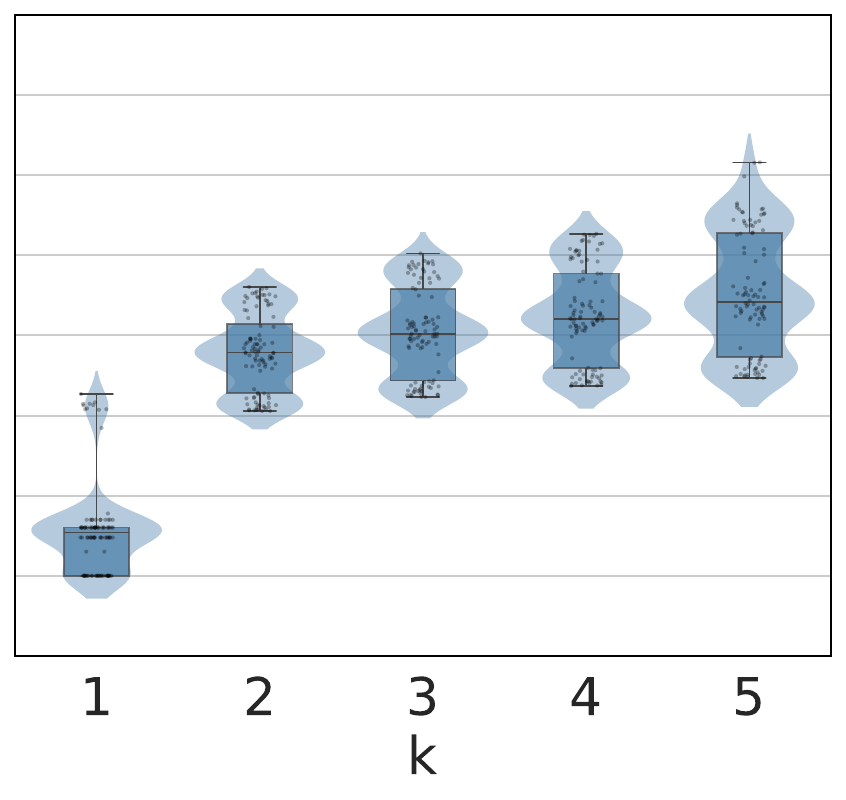}
			\caption{$\ilpunproofdagfunc$}
		\end{subfigure}
		
		\par\vspace{1em}
		\textbf{\doctors 4 (100k)}
	\end{minipage}

\par\vspace{1em}

    %================= doctors 5 =================%
	\begin{minipage}[t]{\linewidth}
		\centering
		
		\begin{subfigure}[t]{0.36\linewidth}
			\centering
			\safeincludegraphics[height=2.57cm]{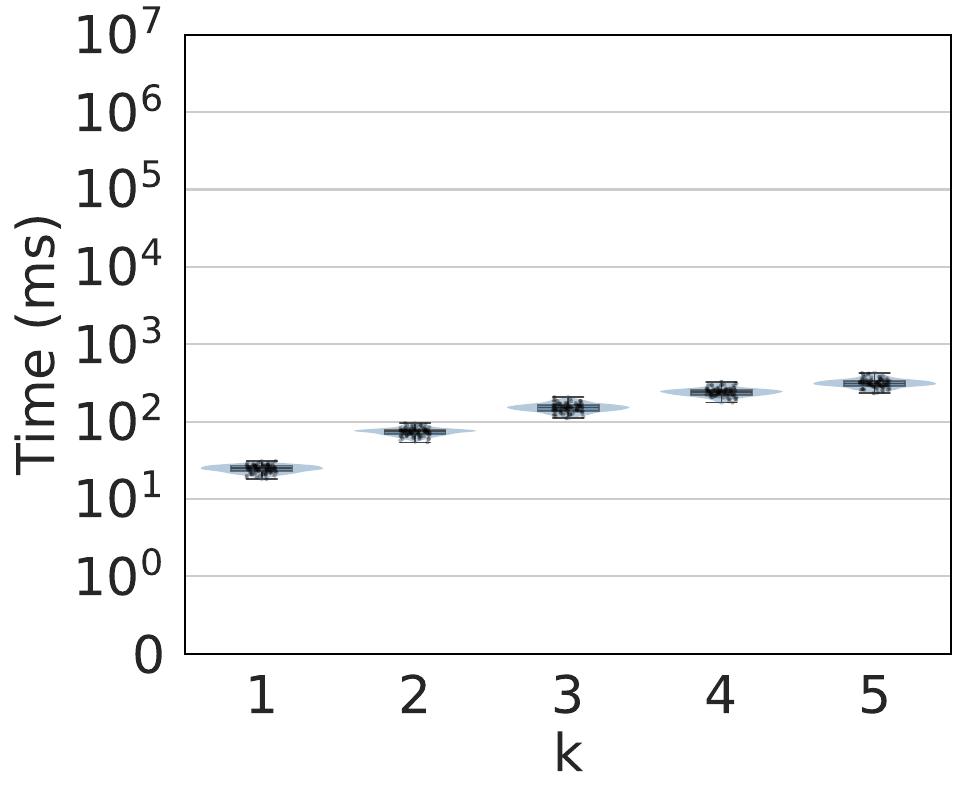}
			\caption{$\ilphypergraphfunc$}
		\end{subfigure}\hfill
		\begin{subfigure}[t]{0.30\linewidth}
			\centering
			\safeincludegraphics[height=2.5cm]{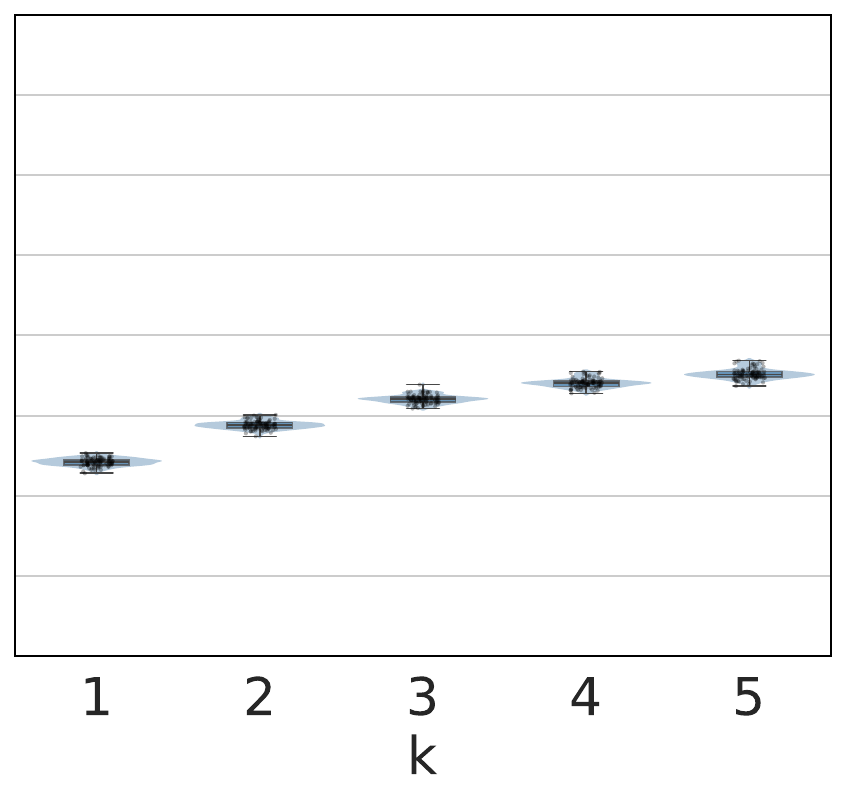}
			\caption{$\ilpunhypergraphfunc$}
		\end{subfigure}\hfill
		\begin{subfigure}[t]{0.30\linewidth}
			\centering
			\safeincludegraphics[height=2.5cm]{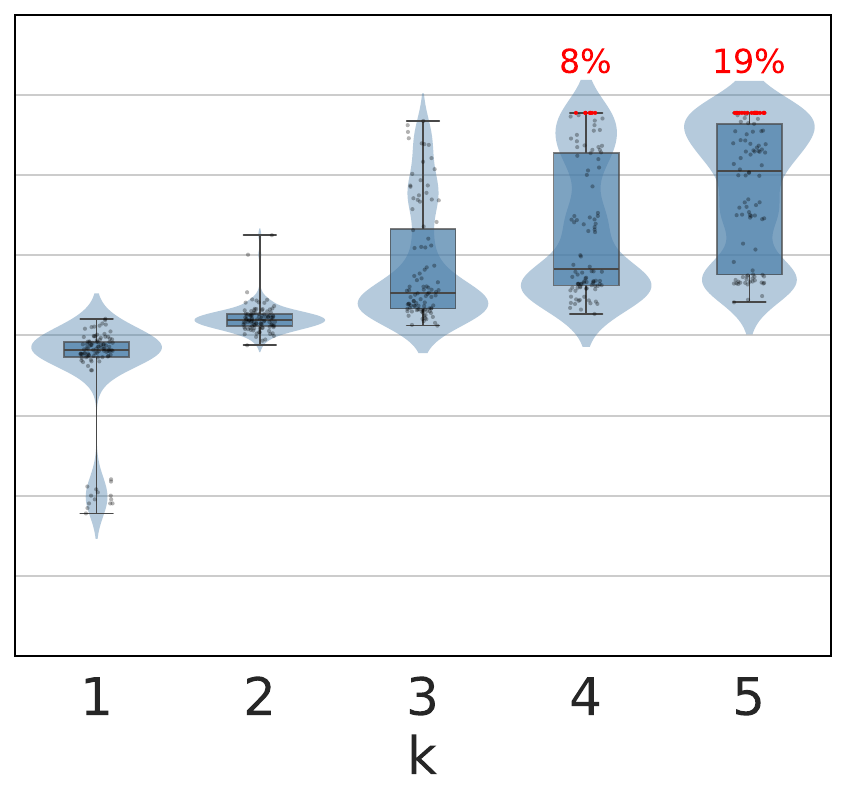}
			\caption{$\ilpunproofdagfunc$}
		\end{subfigure}
		
		\par\vspace{1em}
		\textbf{\doctors\ 5 (100k)}
	\end{minipage}

	\caption{Runtimes for the \doctors (1--5) scenarios.}
	\label{fig:doctors-exact}
\end{figure}

% Part 2: doctors 5-7 (same figure number; subfigure letters continue from (m))
\begin{figure}[!tp]
	\ContinuedFloat
	\centering
	\small

	%================= doctors 6 =================%
	\begin{minipage}[t]{\linewidth}
		\centering
		
		\begin{subfigure}[t]{0.36\linewidth}
			\centering
			\safeincludegraphics[height=2.57cm]{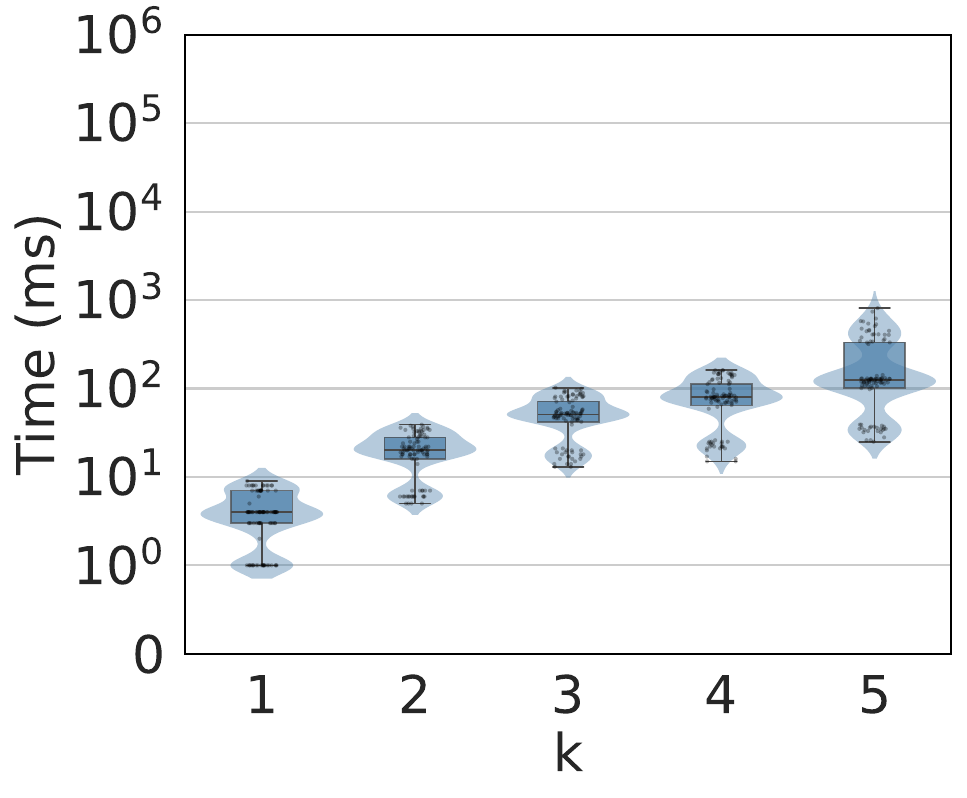}
			\caption{$\ilphypergraphfunc$}
		\end{subfigure}\hfill
		\begin{subfigure}[t]{0.30\linewidth}
			\centering
			\safeincludegraphics[height=2.5cm]{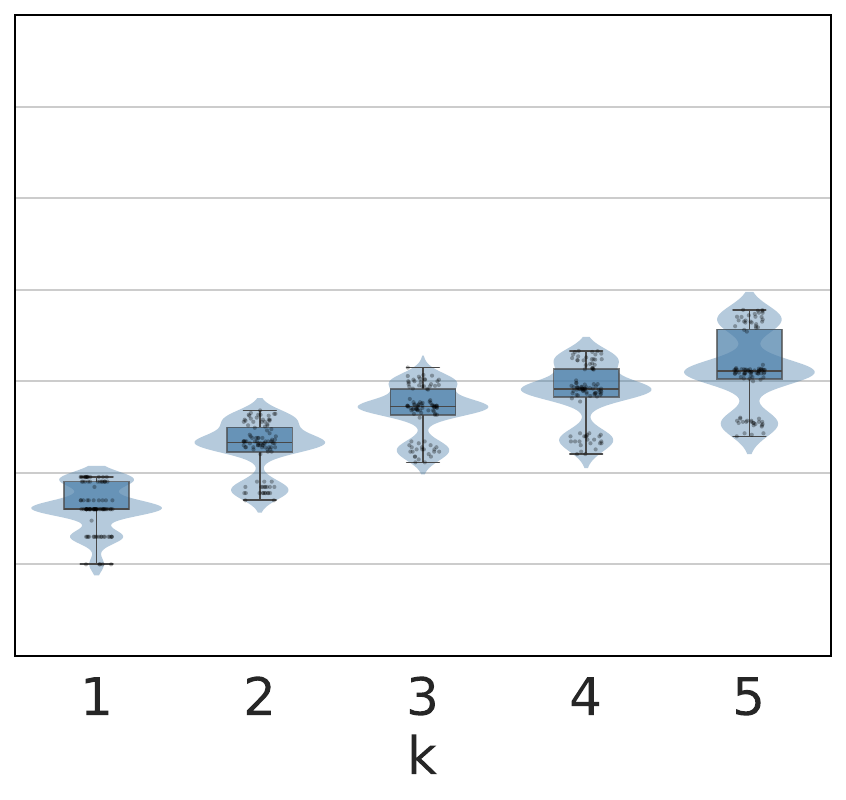}
			\caption{$\ilpunhypergraphfunc$}
		\end{subfigure}\hfill
		\begin{subfigure}[t]{0.30\linewidth}
			\centering
			\safeincludegraphics[height=2.5cm]{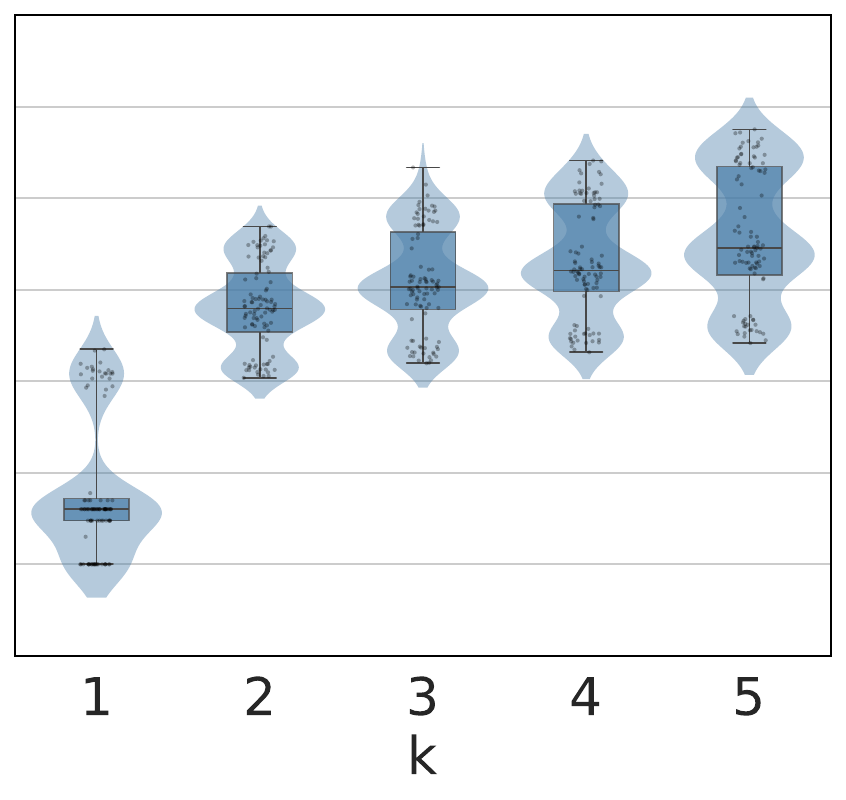}
			\caption{$\ilpunproofdagfunc$}
		\end{subfigure}
		
		\par\vspace{1em}
		\textbf{\doctors\ 6 (100k)}
	\end{minipage}
	
	\par\vspace{1em}
	
	%================= doctors 7 =================%
	\begin{minipage}[t]{\linewidth}
		\centering
		
		\begin{subfigure}[t]{0.36\linewidth}
			\centering
			\safeincludegraphics[height=2.57cm]{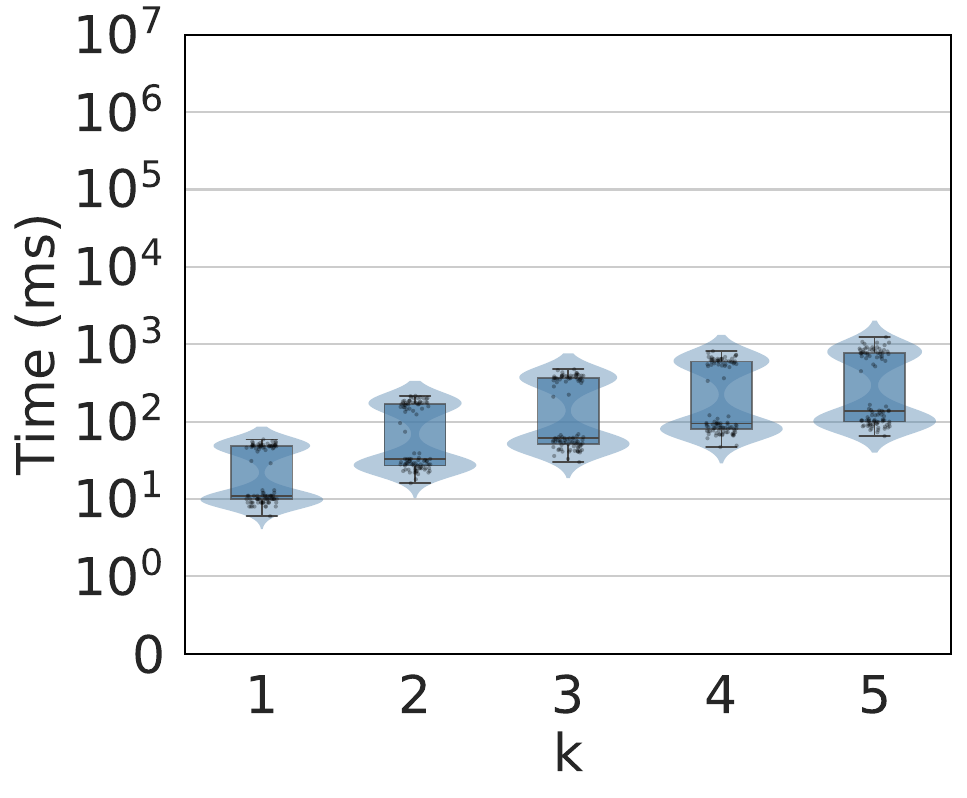}
			\caption{$\ilphypergraphfunc$}
		\end{subfigure}\hfill
		\begin{subfigure}[t]{0.30\linewidth}
			\centering
			\safeincludegraphics[height=2.5cm]{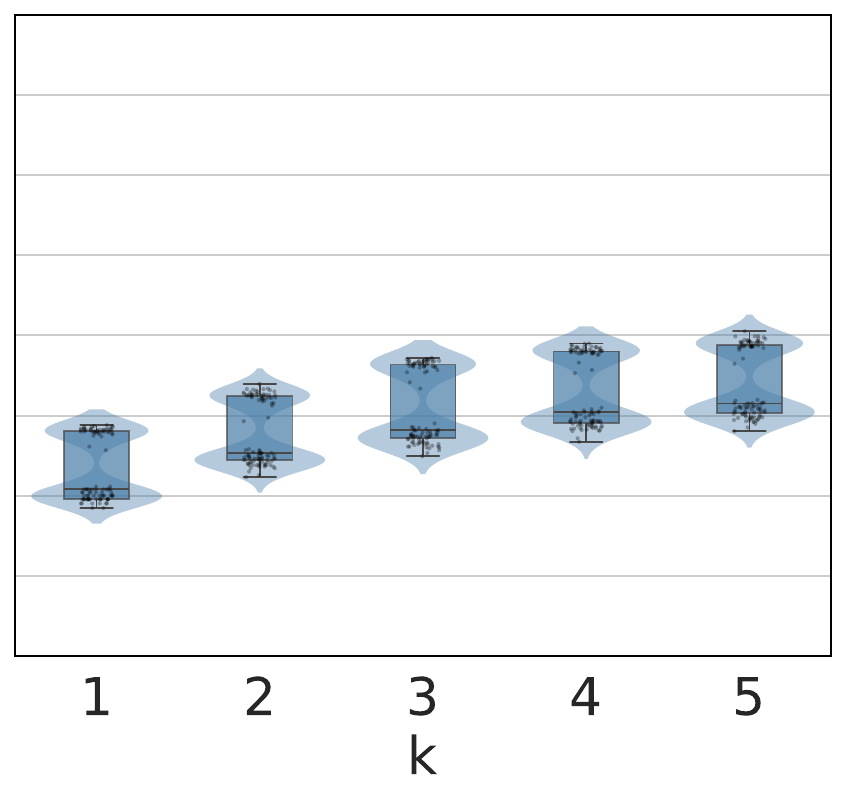}
			\caption{$\ilpunhypergraphfunc$}
		\end{subfigure}\hfill
		\begin{subfigure}[t]{0.30\linewidth}
			\centering
			\safeincludegraphics[height=2.5cm]{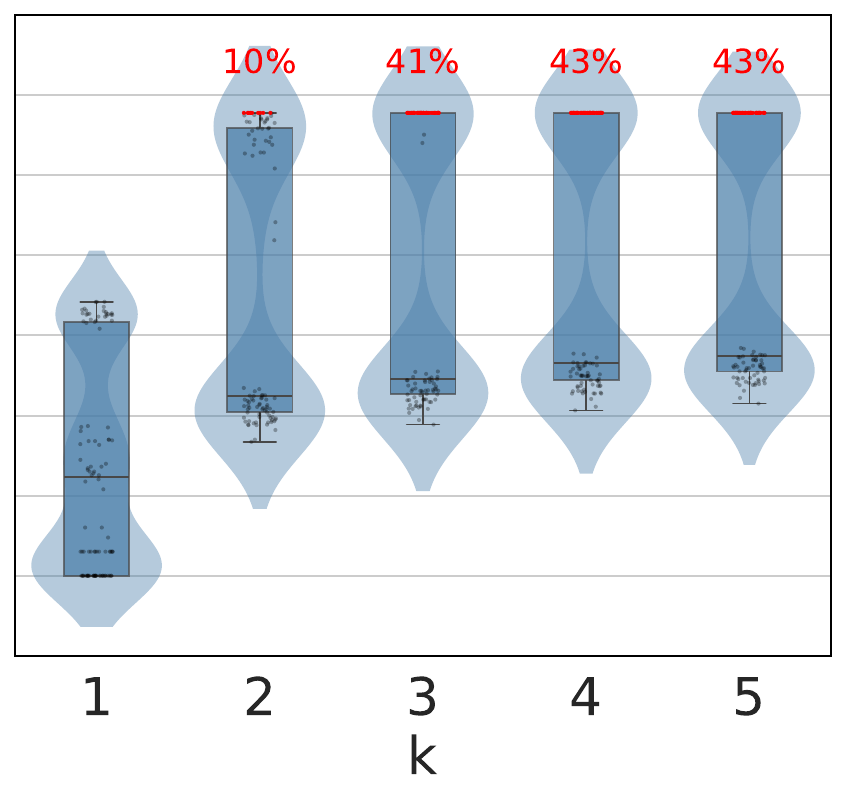}
			\caption{$\ilpunproofdagfunc$}
		\end{subfigure}
		
		\par\vspace{1em}
		\textbf{\doctors\ 7 (100k)}
	\end{minipage}
	
	\caption[]{Runtimes for the \doctors (6 and 7) scenarios.}
\end{figure}

\begin{figure}[!tp]
	\centering
	\small
	
	%==================== all galen ====================%
	\begin{minipage}[t]{\linewidth}
		\centering
		
		% ---- row 1 ----
		\begin{subfigure}[t]{0.49\linewidth}
			\centering
			\safeincludegraphics[height=3cm]{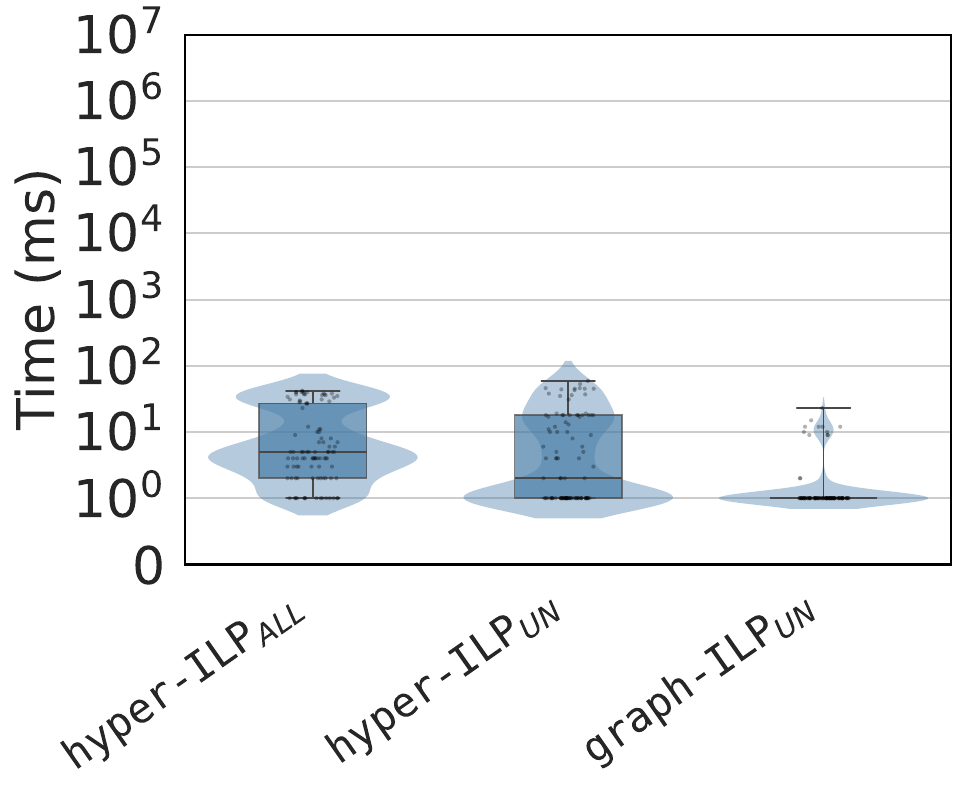}
			\caption{\galen (\fifteen)}
		\end{subfigure}\hfill
		\begin{subfigure}[t]{0.49\linewidth}
			\centering
			\safeincludegraphics[height=3cm]{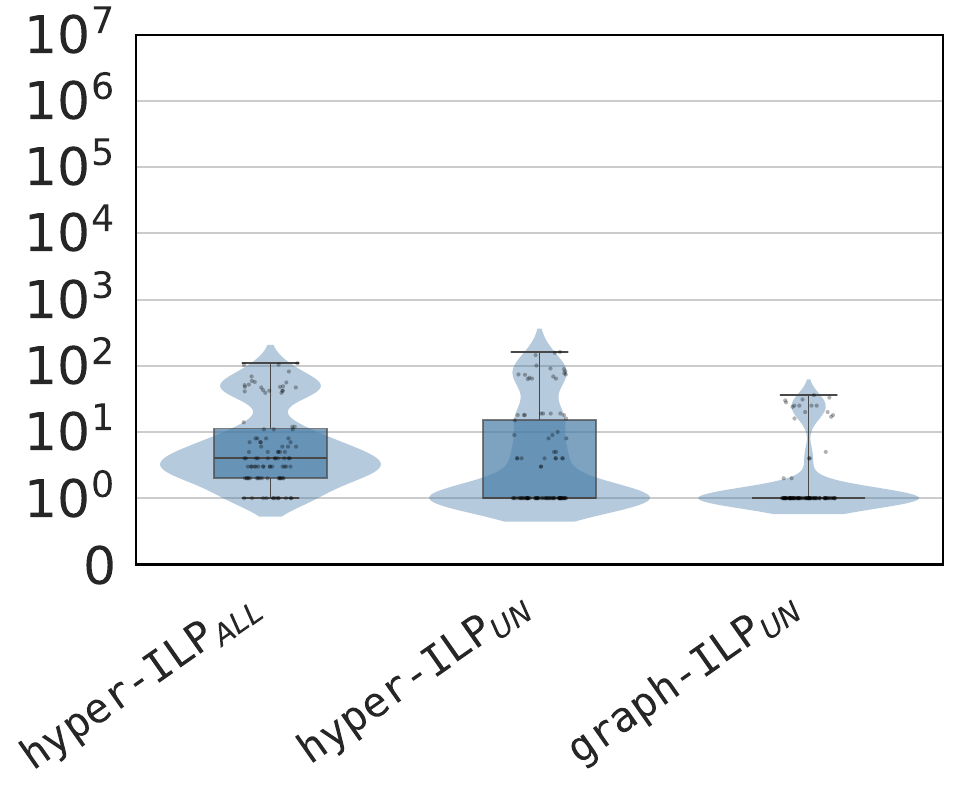}
			\caption{\galen (\twentyfive)}
		\end{subfigure}
		
		\par\vspace{1em}
		
		% ---- row 2 ----
		\begin{subfigure}[t]{0.49\linewidth}
			\centering
			\safeincludegraphics[height=3cm]{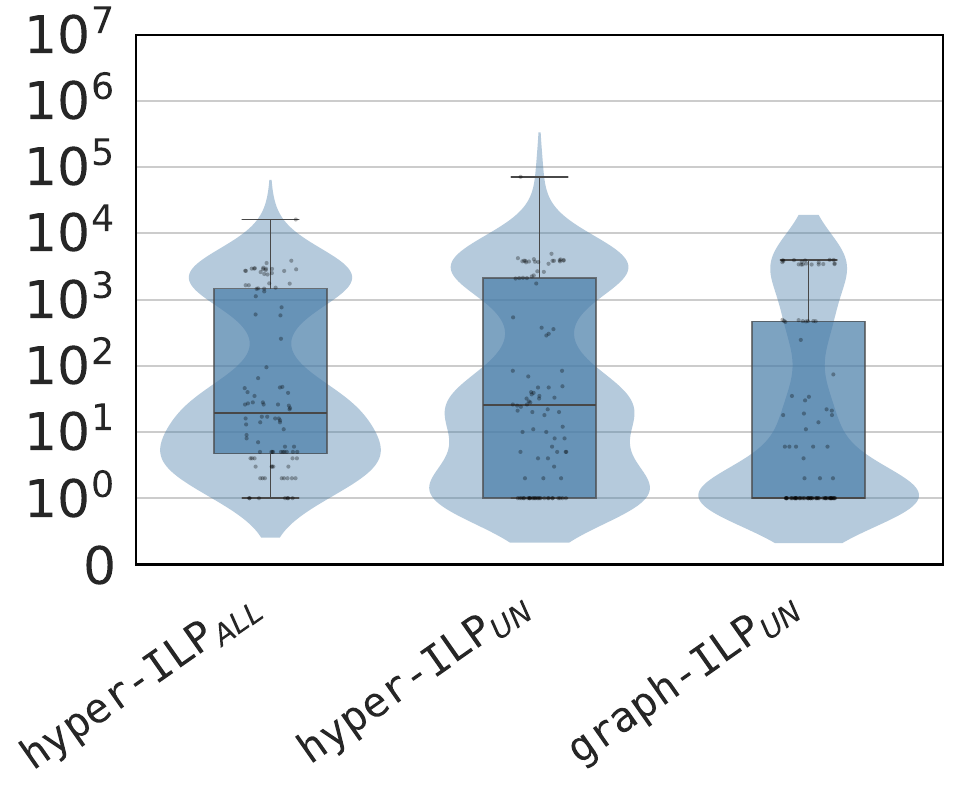}
			\caption{\galen (\forty)}
		\end{subfigure}\hfill
		\begin{subfigure}[t]{0.49\linewidth}
			\centering
			\safeincludegraphics[height=3cm]{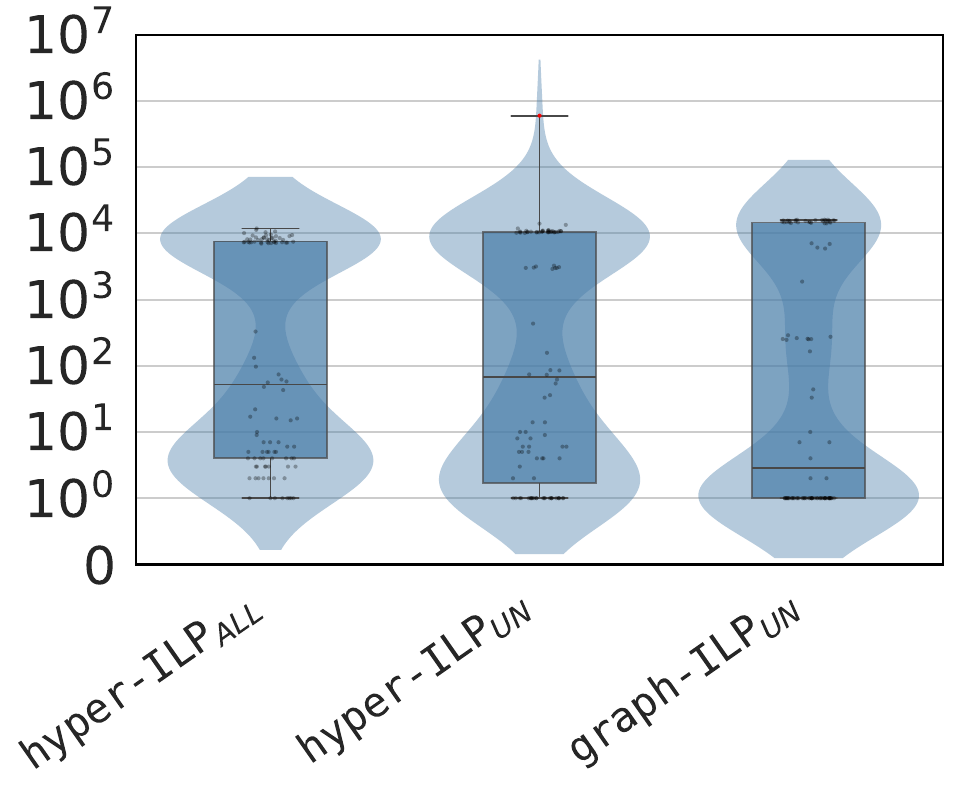}
			\caption{\galen (\fifty)}
		\end{subfigure}
	\end{minipage}
	\caption{Runtimes for the \galen\ scenario.}
	\label{fig:galen-exact}
\end{figure}

\begin{figure}[!tp]
	\centering
	\small
	
	%==================== HTTPD ====================%
	\begin{minipage}[t]{\linewidth}
		\centering
		
		\begin{subfigure}[t]{0.36\linewidth}
			\centering
			\safeincludegraphics[height=2.57cm]{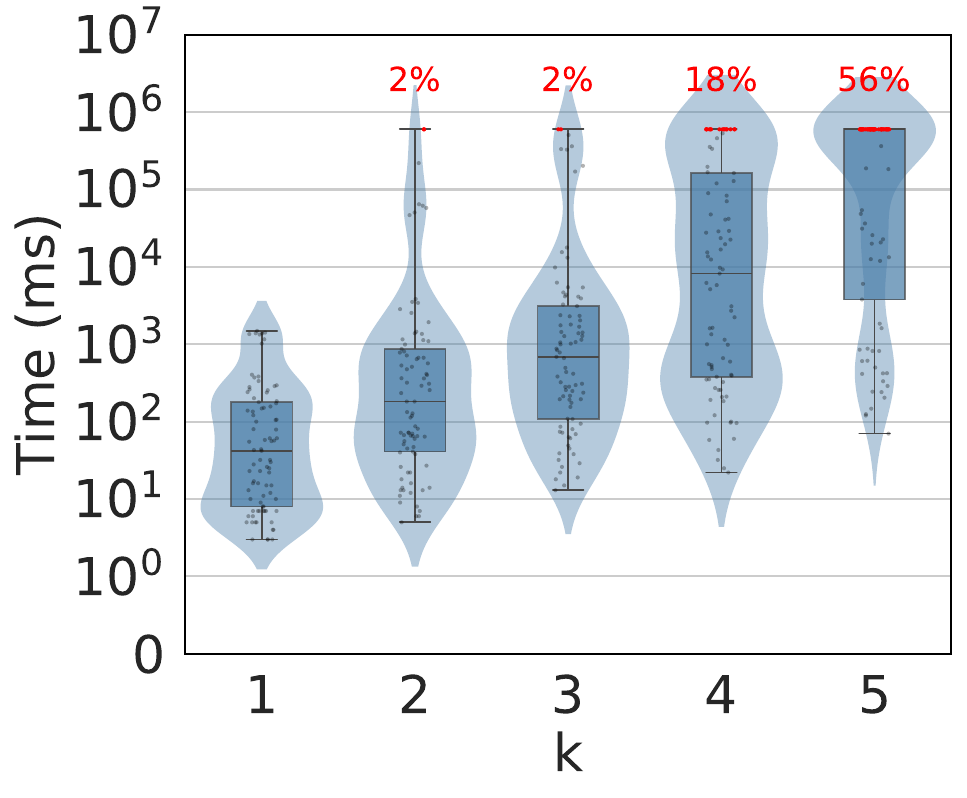}
			\caption{$\ilphypergraphfunc$}
		\end{subfigure}\hfill
		\begin{subfigure}[t]{0.30\linewidth}
			\centering
			\safeincludegraphics[height=2.5cm]{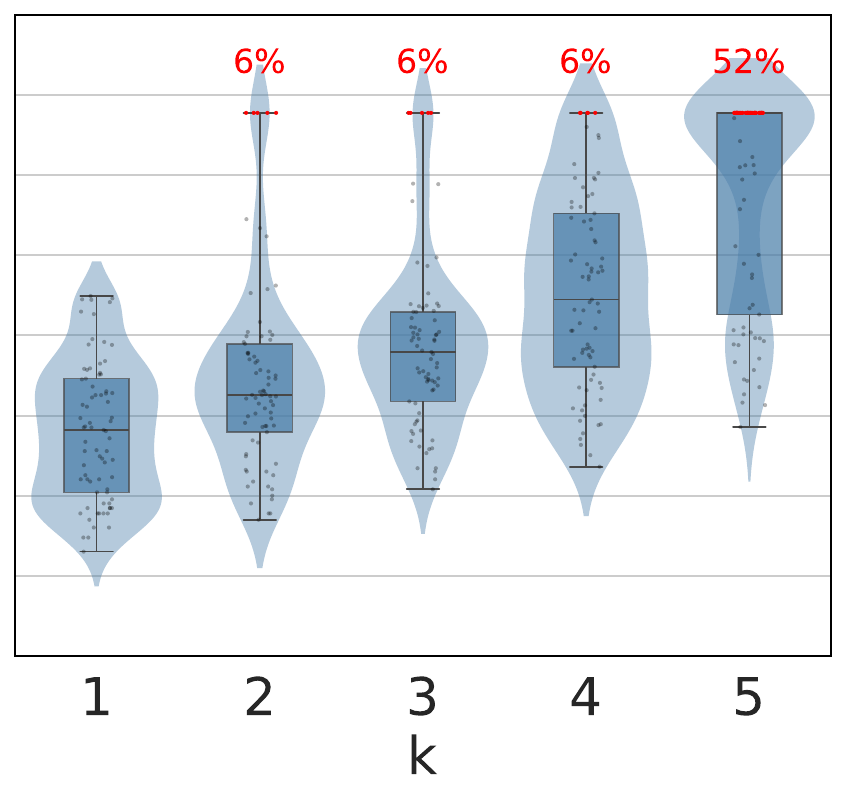}
			\caption{$\ilpunhypergraphfunc$}
		\end{subfigure}\hfill
		\begin{subfigure}[t]{0.30\linewidth}
			\centering
			\safeincludegraphics[height=2.5cm]{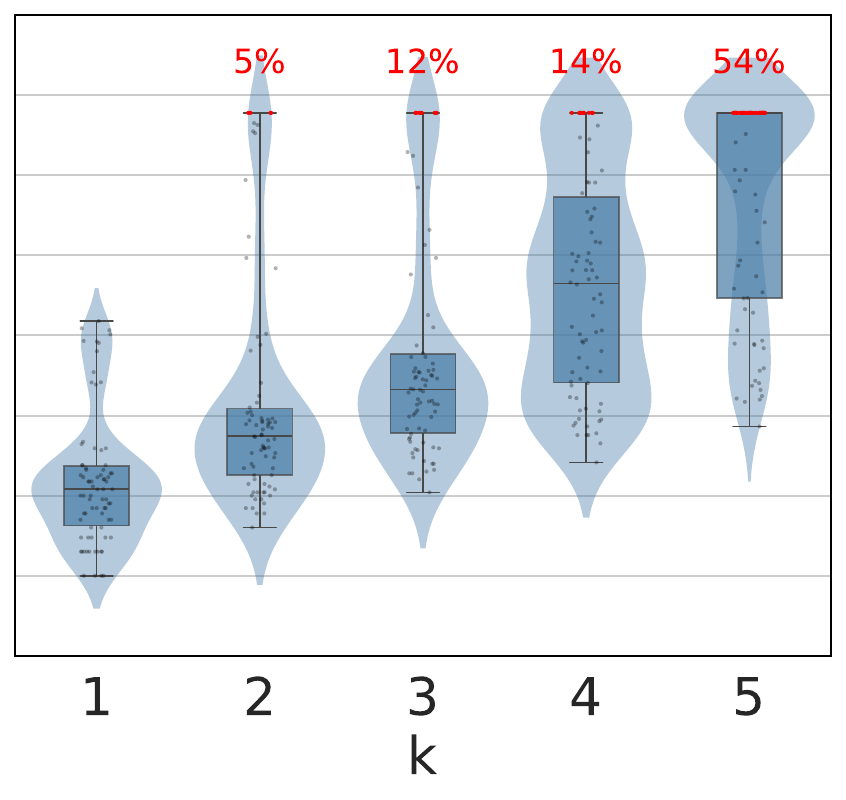}
			\caption{$\ilpunproofdagfunc$}
		\end{subfigure}
		
		\par\vspace{1em}
		\textbf{\csda (\httpd)}
	\end{minipage}
	
	\par\vspace{1em}
	
	%================= PostgreSQL =================%
	\begin{minipage}[t]{\linewidth}
		\centering
		
		\begin{subfigure}[t]{0.36\linewidth}
			\centering
			\safeincludegraphics[height=2.57cm]{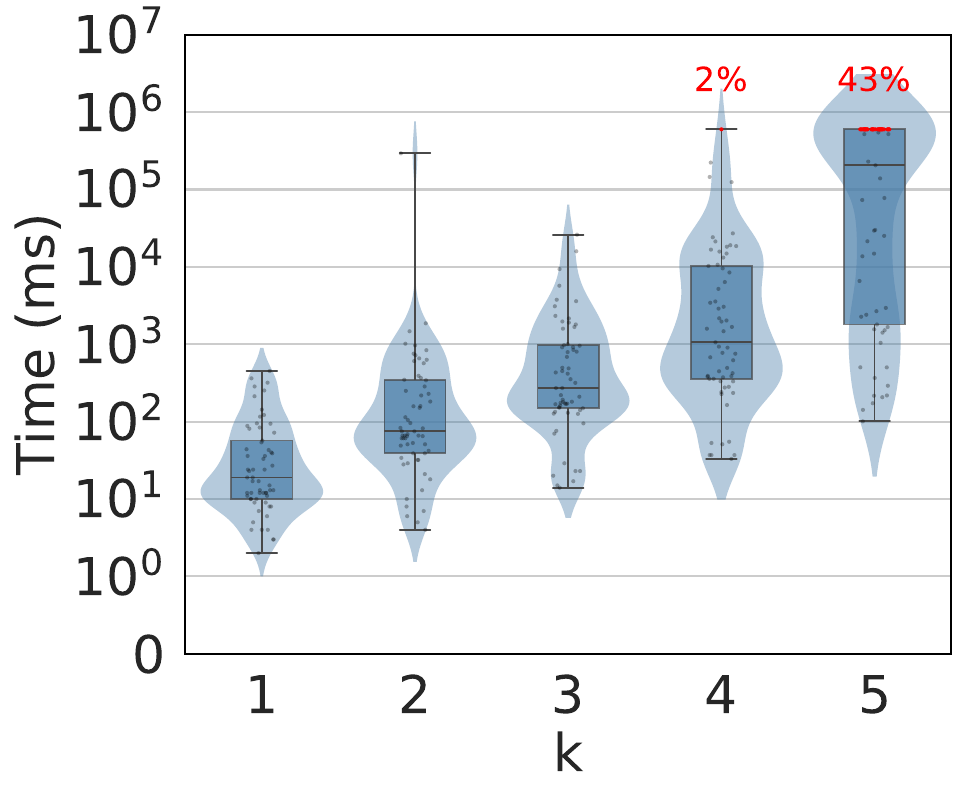}
			\caption{$\ilphypergraphfunc$}
		\end{subfigure}\hfill
		\begin{subfigure}[t]{0.30\linewidth}
			\centering
			\safeincludegraphics[height=2.5cm]{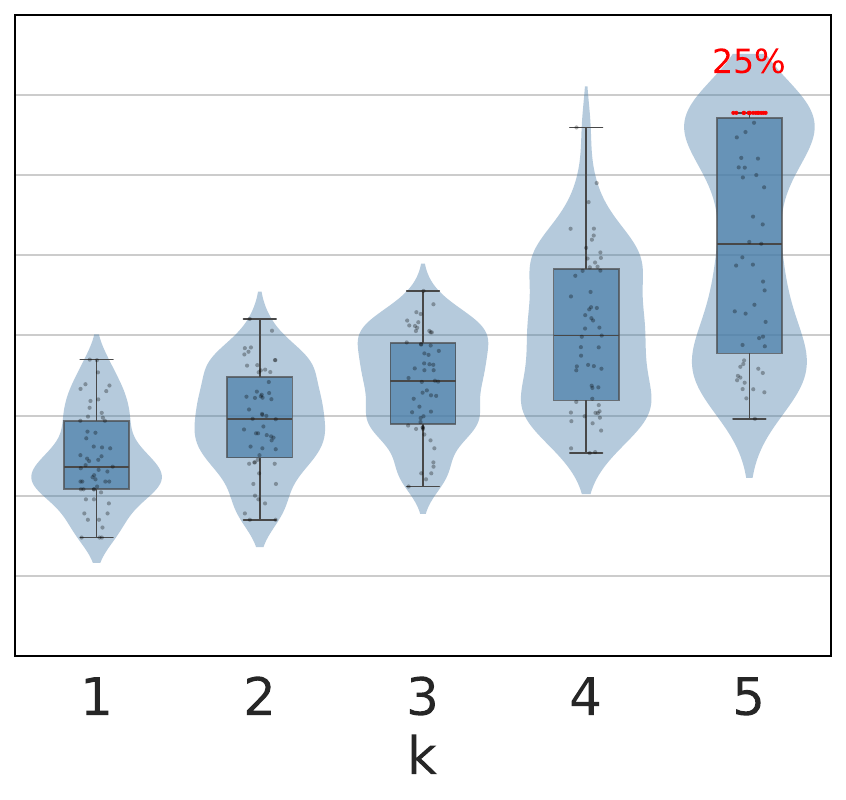}
			\caption{$\ilpunhypergraphfunc$}
		\end{subfigure}\hfill
		\begin{subfigure}[t]{0.30\linewidth}
			\centering
			\safeincludegraphics[height=2.5cm]{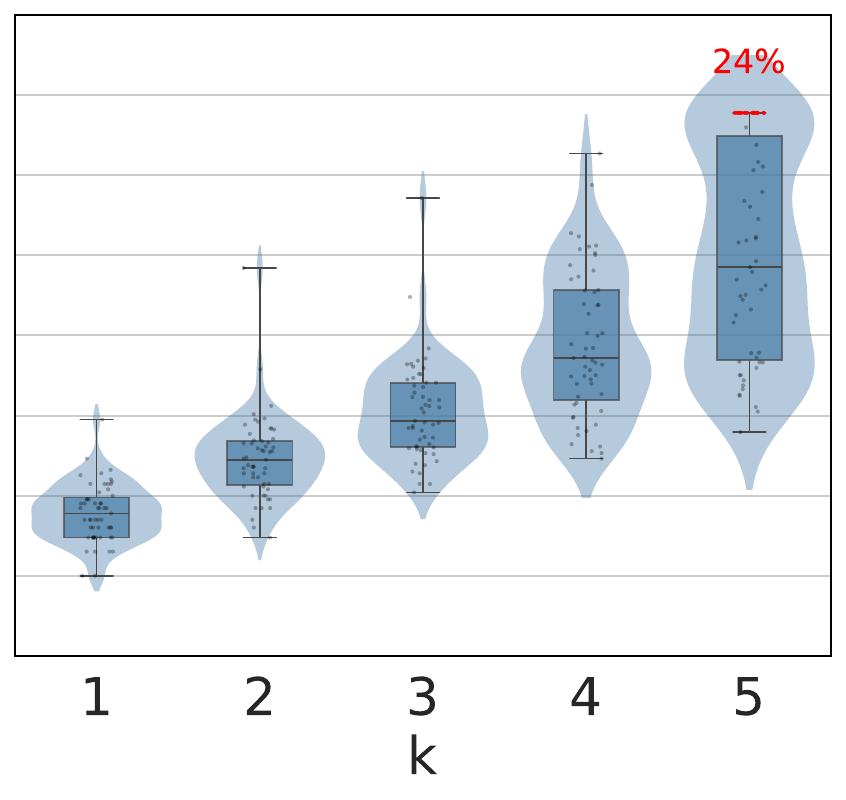}
			\caption{$\ilpunproofdagfunc$}
		\end{subfigure}
		
		\par\vspace{1em}
		\textbf{\csda (\postgres)}
	\end{minipage}
	
	\par\vspace{1em}
	
	%================= Linux =================%
	\begin{minipage}[t]{\linewidth}
		\centering
		\begin{subfigure}[t]{0.36\textwidth}
			\centering
			\safeincludegraphics[height=2.57cm]{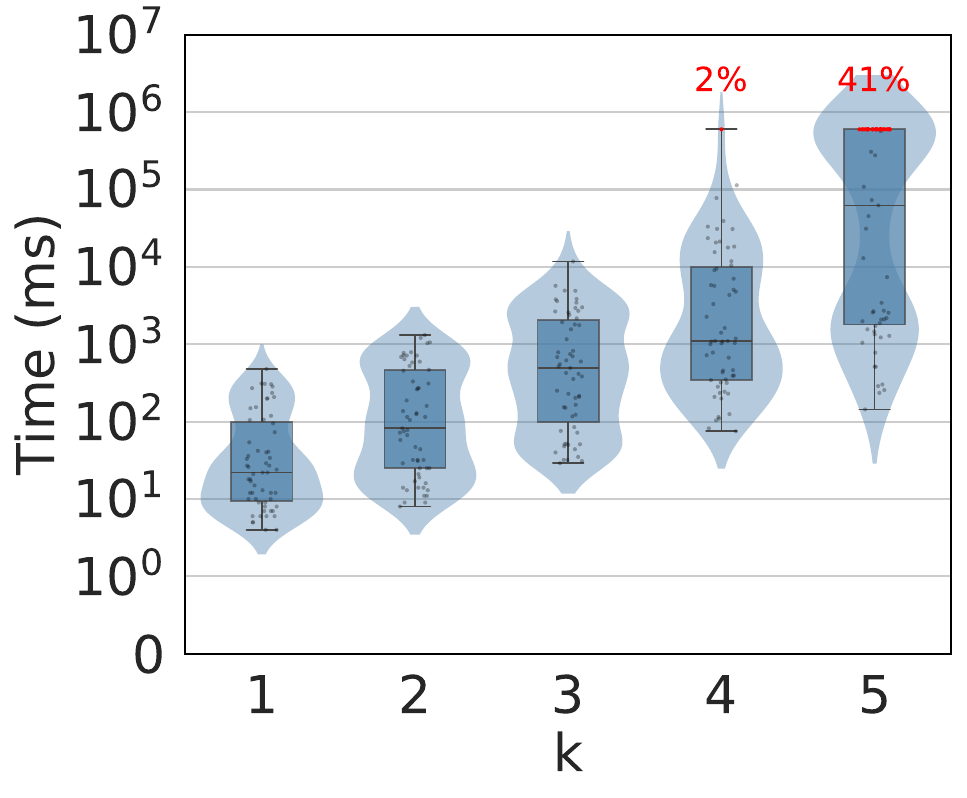}
			\caption{$\ilphypergraphfunc$}
		\end{subfigure}\hfill
		\begin{subfigure}[t]{0.30\textwidth}
			\centering
			\safeincludegraphics[height=2.5cm]{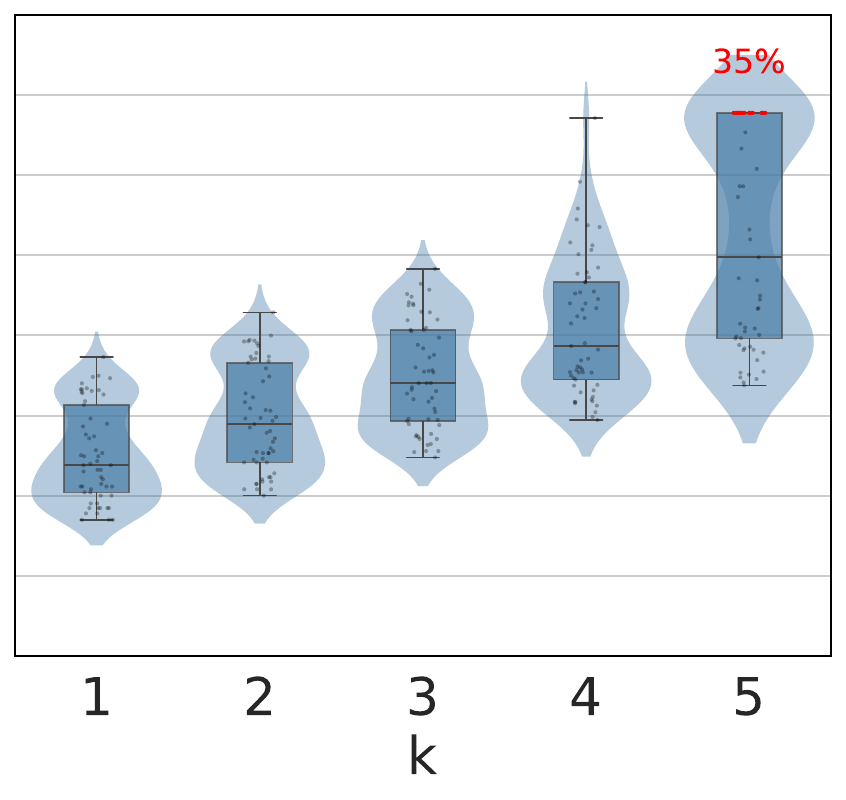}
			\caption{$\ilpunhypergraphfunc$}
		\end{subfigure}\hfill
		\begin{subfigure}[t]{0.30\textwidth}
			\centering
			\safeincludegraphics[height=2.5cm]{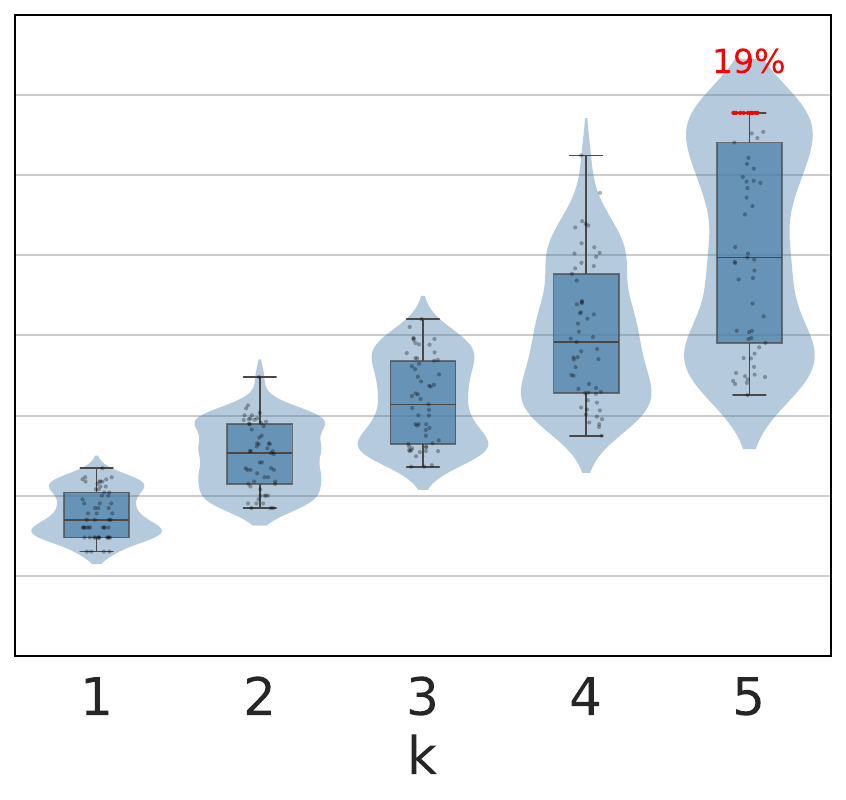}
			\caption{$\ilpunproofdagfunc$}
		\end{subfigure}
		
		\par\vspace{1em}
		\textbf{\csda (\linux)}
	\end{minipage}
	\caption{Runtimes for the \csda\ scenario.}
	\label{fig:csda-exact}
\end{figure}

\begin{figure}[h]
    \centering
    \includegraphics[width=.8\linewidth]{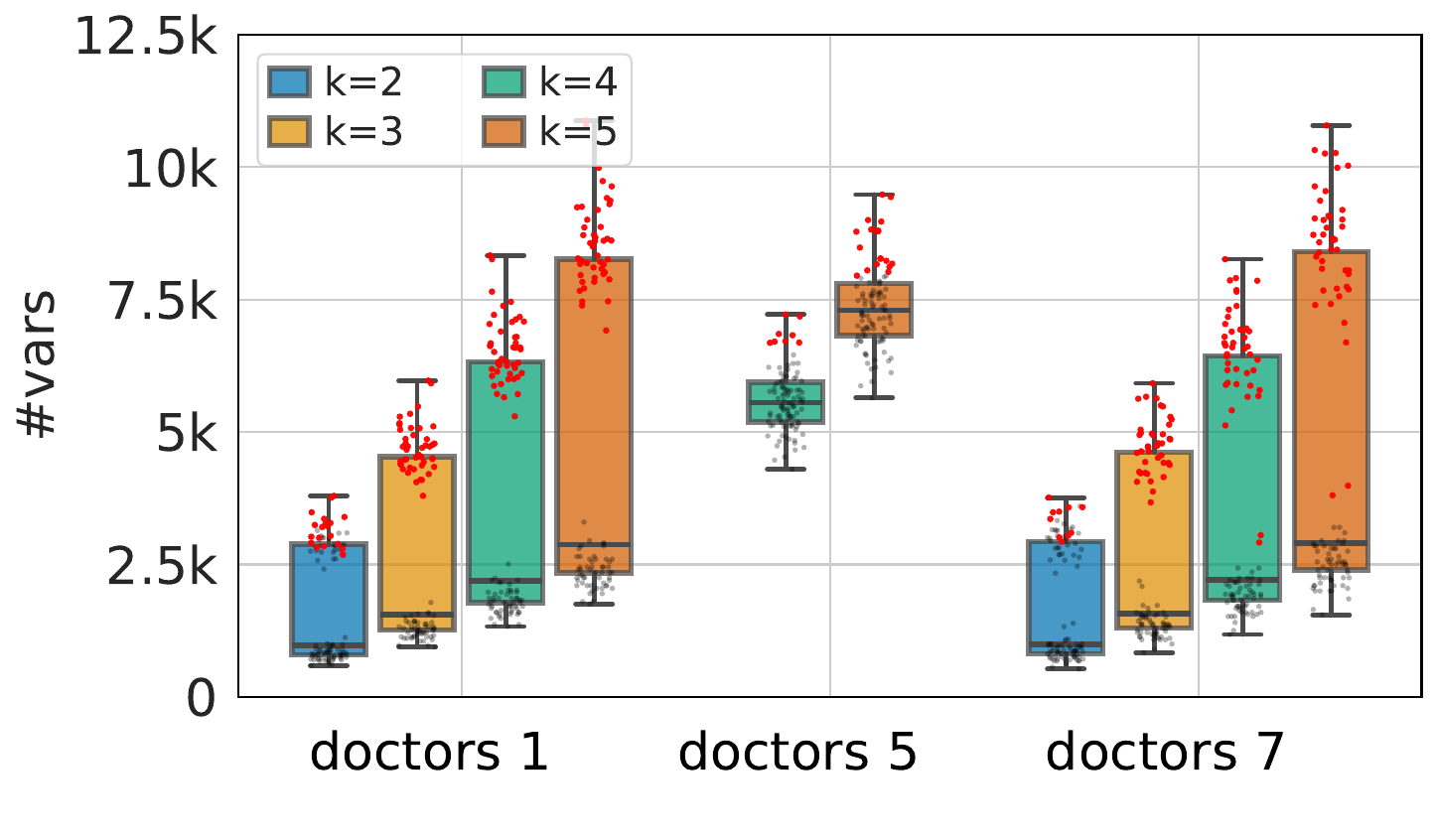}
    \vspace{-3mm}
    \caption{Number of variables occurring in the set $T_\mathsf{acyclic}$ of the ILP encoding $\ilpunproofdagfunc$. Red dots denote timeouts.}
    \label{fig:timeouts}
\end{figure}

\noindent
\textbf{Runtime Evaluation and Scalability with the Diversity Parameter.}
We now analyse the running time of our algorithms in absolute terms, and their scalability w.r.t.\ the input diversity parameter $k$.
Our results show that, apart from $\ComputeDiverse{\ilpproofdagfunc}$, which timeouts in all scenarios, $\ComputeDiverse{\ilpstargammafunc}$, with $(\gamma,\star) \in \{(\mathsf{hyper},\mathsf{ALL})$, $(\mathsf{hyper},\mathsf{UN})$, $(\mathsf{graph},\mathsf{UN})\}$, perform well across all scenarios when computing sets of no more than $k = 3$ maximally diverse explanations. In this setting, such sets can be computed in the order of seconds in most cases, with some scenarios requiring some minutes, for most of the answer tuples. For larger values of $k$, runtimes stay below the 10-minutes timeout for the less demanding scenarios (Figure~\ref{fig:graph-analytics-exact}/\bitcoin, Figure~\ref{fig:doctors-exact}, and Figure~\ref{fig:galen-exact}), while an increasibly large number of timeouts can be observed for the more demanding scenarios (Figure~\ref{fig:graph-analytics-exact}/\facebook, and Figure~\ref{fig:csda-exact}), even when focusing on the best performing algorithm for each family of proof trees. Hence, all the algorithms we consider do not scale well with $k$, in general. This is due to the exponential growth of the search space of the underlying maximization problem solved by each of the algorithms we consider. Indeed, there are $\binom{N}{k}$ possible sets of $k$ explanations to consider, where $N$ is the total number of explanations in $\starwhy{\bar t}{D}{Q}$, which itself can already be exponentially large w.r.t. $|D|$, for a considered query $Q$, database $D$, and tuple $\bar t$.

From our discussion, we conclude that although the above three algorithms are effective at finding diverse sets of small size, they exhibit poor scalability w.r.t.~the diversity parameter $k$, and thus, they quickly become infeasible for computing larger sets. This motivates the need for alternative approaches that trade optimality of the computed sets for performance, namely for approximation techniques. This will be the subjected of the next section.
\section{Approximation Algorithm for Diverse Why-Provenance}\label{sec:approximations}

We proceed to present an approximation algorithm for the problem $\star\text{-}\mathsf{DiverseWhyProvenance}$, where $\star \in \{\mathsf{ALL},\mathsf{UN}\}$. Our goal is to devise an algorithm that, given as input a Datalog query $Q =(\dep,R)$, a database $D$ over $\esch{\dep}$, a tuple $\bar t \in \adom{D}^{\arity{R}}$, and $k > 0$, computes a set of $\starwhyk{k}{\bar t}{D}{Q}$ whose diversity is an approximation of the optimal diversity and it performs much faster, in practice, than the approach discussed in Section~\ref{sec:characterizations} that aims for an exact computation.
We first formalize the kind of approximation algorithm that we target and then explain how we obtain such an algorithm.

\begin{definition}
For any $\beta \ge 1$, a \emph{$\beta$-approximation algorithm for $\star\text{-}\mathsf{DiverseWhyProvenance}$}, where $\star \in \{\mathsf{ALL},\mathsf{UN}\}$, is an algorithm that, given as input a Datalog query $Q = (\dep,R)$, a database $D$ over $\esch{\dep}$, a tuple $\bar t \in \adom{D}^{\arity{R}}$, and an integer $k > 0$, outputs a set $\hat{\mathcal{S}} \in \starwhyk{k}{\bar t}{D}{Q}$ of $k$ databases such that
$$\diversity{\hat{\mathcal{S}}} \ge \dfrac{1}{\beta} \cdot\diversity{\mathcal{S}_{\mathsf{opt}}},$$
if there is $\mathcal{S}_{\mathsf{opt}} \in \argmax\limits_{\mathcal{S} \in \starwhyk{k}{\bar t}{D}{Q}} \diversity{\mathcal{S}}$; otherwise, it fails.
We call $\beta$ the \emph{approximation ratio} of the algorithm. \hfill\markfull
\end{definition}

\medskip
\noindent
\textbf{Diverse Why-Provenance as a Remote-Clique Problem.}
To devise the desired approximation algorithm, we first observe that $\star\text{-}\mathsf{DiverseWhyProvenance}$ can be seen as a (succinct) encoding of the remote-clique problem, a well-studied optimization problem in graph theory. Consider an edge-weighted, undirected graph $G = (V,E,w)$, where $w : E \rightarrow \mathbb{N}$ assigns a non-negative weight to each edge of $G$. By abuse of notation, for an edge $\{u,v\} \in E$, we may write $w(u,v)$ or $w(v,u)$ instead of $w(\{u,v\})$. We further assume that $w(v,v) = 0$, for each node $v \in V$.
We say that $G$ is \emph{complete} if, for every pair of distinct nodes $u,v \in V$, there is an edge $\{u,v\} \in E$. For an edge-weighted, complete, undirected graph $G = (V,E,w)$, we say that $w$ \emph{satisfies the triangle inequality} if, for every three nodes $v_1,v_2,v_3 \in V$, $w(v_1,v_3) \le w(v_1,v_2) + w(v_2,v_3)$.
For an integer $k > 0$, we denote by $\cliques{k}{G}$ the family of all sets $V' \subseteq V$ of $k$ nodes of $G$. Finally, for a set $V'= \{v_1,\ldots,v_k\} \in \cliques{k}{G}$, the \emph{diversity} of $V'$ is defined as the integer
\[
\diversity{V'} = \sum\limits_{i,j \in [|V'|] \text{ and }i<j} w(v_i,v_j).
\]
The remote-clique problem is defined as follows:

\medskip

\begin{center}
	\fbox{\begin{tabular}{ll}
			{\small PROBLEM} : & $\mathsf{Remote\text{-}Clique}$
			\\[2mm]
			{\small INPUT} : & A complete, edge-weighted, undirected graph \\
			& $G = (V,E,w)$, where $w$ satisfies the triangle \\
			& inequality, and an integer $k > 0$\\[2mm]
			{\small OUTPUT} : & A member of $\argmax\limits_{V' \in \cliques{k}{G}} \diversity{V'}$ if it exists;\\
			&  otherwise \textbf{fail} 
	\end{tabular}}
\end{center}

\medskip

It is easy to see that the problem $\star\text{-}\mathsf{DiverseWhyProvenance}$, where $\star \in \{\mathsf{ALL},\mathsf{UN}\}$, can be reduced to $\mathsf{Remote\text{-}Clique}$. In particular, given a Datalog query $Q = (\dep, R)$, a database $D \in \esch{\dep}$, a tuple $\bar t \in \adom{D}^{\arity{R}}$, and an integer $k > 0$, we construct the edge-weighted, complete, undirected graph $G = (V,E,w)$, where $V = \starwhy{\bar t}{D}{Q}$ and $w(D_1,D_2) = |\diff{D_1}{D_2}|$, for each $D_1,D_2 \in V$.\footnote{The weight function $w$ defined as $w(D_1,D_2) = |\diff{D_1}{D_2}|$ satisfies the triangle inequality since, for any three sets $A,B,C$, $\diff{A}{C} \subseteq \diff{A}{B} \cup \diff{B}{C}$.}
It is clear that the output of $\mathsf{Remote\text{-}Clique}$ with input the graph $G$ as defined above and the integer $k$ coincides with the output of $\star\text{-}\mathsf{DiverseWhyProvenance}$ with input $Q$, $D$, $\bar t$, and $k$.

\medskip
\noindent
\textbf{Approximation Algorithm for $\mathsf{Remote\text{-}Clique}$.}
Due to the connection between our problem and the  $\mathsf{Remote\text{-}Clique}$ problem discussed above, we can exploit exiting approximation algorithms for  $\mathsf{Remote\text{-}Clique}$ to provide an approximation algorithm for our problem.
The state of the art approximation algorithm for $\mathsf{Remote\text{-}Clique}$ is a greedy algorithm defined as follows. Given an input $G,k$ of the problem $\mathsf{Remote\text{-}Clique}$, the algorithm performs $k$ iterations, where at each iteration $i \in [k]$, the algorithm picks a node $v_i \in V$ such that $v_i$ is different from all previously selected nodes $v_{j}$, with $j < i$, and $v_{i}$ is such that the quantity $\diversity{\{v_1,\ldots,v_{i}\}} - \diversity{\{v_1,\ldots,v_{i-1}\}}$ is maximized, namely $v_i$ maximizes the \emph{marginal diversity}, i.e., the increase in diversity w.r.t.\ the set of nodes selected so far. For each $i \in [k]$, if a node $v_{i}$ which is different from all nodes $v_j$, with $j < i$ cannot be found, the algorithm halts failing. If instead all $k$ iterations complete without failing, the algorithm outputs the set $\{v_1,\ldots,v_{k}\}$ and halts.
The above algorithm is a $2$-approximation algorithm for $\mathsf{Remote\text{-}Clique}$, where for any $\beta \ge 1$, a $\beta$-approximation algorithm for $\mathsf{Remote\text{-}Clique}$ is an algorithm that, given $G$ and $k$, outputs a set $\hat{V}$ of $k$ nodes such that $\diversity{\hat{V}} \ge \frac{1}{\beta} \cdot \diversity{V_\mathsf{opt}}$, if there is $V_\mathsf{opt} \in \argmax\limits_{V' \in \cliques{k}{G}} \diversity{V'}$; otherwise, it fails.
Moreover, the factor $2$ is \emph{tight} for the approximation algorithm, i.e., there exist families of inputs $G,k$ for which the algorithm returns a set $\hat{V}$ of $k$ nodes with $\diversity{\hat{V}} = \frac{1}{2} \cdot \diversity{V_{\mathsf{opt}}}$.

\medskip
\noindent
\textbf{Approximation Algorithm for Diverse Why-Provenance.} A naive adaptation of the greedy $2$-approximation algorithm discussed above to the problem $\star\text{-}\mathsf{DiverseWhyProvenance}$ would be as follows: given a Datalog query $Q = (\dep,R)$, a database $D$ over $\esch{\dep}$, a tuple $\bar t \in \adom{D}^{\arity{R}}$, and an integer $k > 0$, 
\begin{enumerate}
\item first construct the set $\starwhy{\bar t}{D}{Q}$,
\item then construct $G = (V,E,w)$ with $V = \starwhy{\bar t}{D}{Q}$ and $w(D_1,D_2) = |\diff{D_1}{D_2}|$, for each $D_1,D_2 \in V$, and
\item execute the $2$-approximation algorithm for the problem $\mathsf{Remote\text{-}Clique}$ with input $G,k$.
\end{enumerate}

It should not come as a surprise that this naive approach would never scale in practice since it requires the explicit construction of the whole graph $G$, which in turn requires the construction of the whole set $\starwhy{\bar t}{D}{Q}$ upfront; the latter can be very large. 
Hence, we need a more sophisticated approach that does not require the explicit construction of the graph $G$. 
%In particular, we need strategy for selecting a node $v_i$ at each iteration $i \in [k]$ that maximizes the integer $\diversity{\{v_1,\ldots,v_{i}\}} - \diversity{\{v_1,\ldots,v_{i-1}\}}$ without explicitly having available the set $\starwhy{\bar t}{D}{Q}$.}
%
Towards such a sophisticated approach, we need a clever strategy for selecting a node $v_i$ at each iteration $i \in [k]$ that maximizes the integer $\diversity{\{v_1,\ldots,v_{i}\}} - \diversity{\{v_1,\ldots,v_{i-1}\}}$.
%we observe that for the the approximation algorithm to work, the only requirement is to have a strategy for selecting a node $v_i$ at each iteration $i \in [k]$ that maximizes the integer $\diversity{\{v_1,\ldots,v_{i}\}} - \diversity{\{v_1,\ldots,v_{i-1}\}}$. 
In terms of our problem, this means that we need a way, given an already constructed set of $i-1$ databases $D_1,\ldots,D_{i-1}$ from $\starwhy{\bar t}{D}{Q}$, to construct a new database $D_i \in \starwhy{\bar t}{D}{Q}$ that differs from all $D_1,\ldots,D_{i-1}$, and the marginal diversity $\diversity{\{D_1,\ldots,D_i\}} - \diversity{\{D_1,\ldots,D_{i-1}\}}$ is maximized. We discuss how this can be achieved without, of course, explicitly constructing the set $\starwhy{\bar t}{D}{Q}$, by exploiting once again our ILP-based framework introduced above.

Let $\star \in \{\mathsf{ALL}, \mathsf{UN}\}$, and assume that $\ConstructILPEncoding$ is a procedure that takes as input a triple $(\bar t,D,Q)$, where $Q$ is a Datalog query $Q = (\dep,R)$, $D$ is a database $D$ over $\esch{\dep}$, and $\bar t$ a tuple of $\adom{D}^{\arity{R}}$, and constructs an ILP encoding for $\starwhy{\bar t}{D}{Q}$. The approximation algorithm $\mathsf{Approx}[\ConstructILPEncoding]$, which is parameterized with the procedure $\ConstructILPEncoding$, is depicted in Algorithm~\ref{alg:generic-approximation}.
Note that in the description of the algorithm, we write $g_{\mu,i}$, for each $i \in [n]$, to denote the ``opposite'' of $x_{\alpha_i}$ w.r.t.~$\mu$, i.e., the expression $1 - x_{\alpha_i}$ if $\mu(x_{\alpha_i}) = 1$, and the expression $x_{\alpha_i}$ if $\mu(x_{\alpha_i}) = 0$.
We can show the following:

\begin{theorem}\label{thm:approximation-quality}
Let $\star \in \{\mathsf{ALL}, \mathsf{UN}\}$, and assume that $\ConstructILPEncoding$ is a procedure that takes as input a triple $(\bar t,D,Q)$, where $Q$ is a Datalog query $Q = (\dep,R)$, $D$ is a database $D$ over $\esch{\dep}$, and $\bar t$ a tuple of $\adom{D}^{\arity{R}}$, and constructs an ILP encoding for $\starwhy{\bar t}{D}{Q}$. Then, $\mathsf{Approx}[\ConstructILPEncoding]$ is a $2$-approximation algorithm for the problem $\star\text{-}\mathsf{DiverseWhyProvenance}$.
\end{theorem}
\begin{algorithm}[t]
	\caption{Algorithm $\mathsf{Approx}[\ConstructILPEncoding]$}
	\label{alg:generic-approximation}
	
	\BlankLine
	
	\KwIn{A Datalog query $Q = (\dep,R)$, a database $D = \{\alpha_1,\ldots,\alpha_n\}$ over $\esch{\dep}$, a tuple $\bar t \in \adom{D}^{\arity{R}}$, and an integer $k >0 $.}
	\KwOut{A set of $k$ databases, or \textbf{fail}}
	
	\BlankLine
	
	$T := \ConstructILPEncoding(\bar t,D,Q)$\;
	Let $x_{\alpha_1},\ldots,x_{\alpha_n}$ be the database variables of $T$\;
	$f := 0$\;
	$D_1 := \emptyset$; \ldots ; $D_k := \emptyset$\;
	\For{$i = 1$ \textbf{\text{to}} $k$}{
		$\mu :=$ an optimal solution of $(f,T)$ \textbf{if} exists, \textbf{else} \textbf{fail}\;
		$D_i := \db{\mu}$\;\
		$T := T \cup \left\{\sum\limits^n_{\ell = 1} g_{\mu,\ell} \ge 1\right\}$;\hfill\tcp{\footnotesize Force $D_{i+1}$ to be new}
		$f := f + \sum\limits^n_{\ell = 1} g_{\mu,\ell}$;\hfill\tcp{\footnotesize Maximize $\diversity{\{D_1,\ldots,D_{i+1}\}} - \diversity{\{D_1,\ldots,D_{i}\}}$}
	}
	\textbf{return} $\{D_1,\ldots,D_{k}\}$;
\end{algorithm}

% Moreover, we can prove that $\mathsf{Approx}$ is optimal:
% \begin{theorem}\label{thm:approximation-optimal}
% 	There is no $\beta$-approximation algorithm for $\mathsf{Diversity}$ with $\beta > \frac{1}{2}$.
% \end{theorem}
% \begin{proof}
% 	\textcolor{red}{Missing...}
% \end{proof}
\section{Approximate Computation of Diverse Why-Provenance}
\label{sec:approximate-experiments}

The aim of this section is to experimentally evaluate the approximation algorithms $\mathsf{Approx}[\ilpstargammafunc]$, where $\star \in \{\mathsf{ALL}, \mathsf{UN}\}$ and $\gamma \in \{\mathsf{hyper},\mathsf{graph}\}$. Our goal is to assess both their efficiency and the quality of the sets of explanations that they produce. To this end, we employ the same scenarios as the ones used in Section~\ref{sec:exact-experiments}.
In particular, in Section~\ref{sec:approximate-runtime}, we evaluate the runtime of the algorithm $\mathsf{Approx}[\ilpstargammafunc]$ over all answer tuples which guarantee at least five members of the why-provenance, as in Section~\ref{sec:exact-experiments}, and compare it with the ones of the exact algorithm $\ComputeDiverse{\ilpstargammafunc}$. In Section~\ref{sec:quality}, we evaluate the quality of the \emph{diversity} of the sets of explanations constructed by $\mathsf{Approx}[\ilpstargammafunc]$ by comparing it against the optimal value computed by $\ComputeDiverse{\ilpstargammafunc}$ in Section~\ref{sec:exact-experiments}; we only report the comparison for the answer tuples for which the exact algorithm could finish its execution within the given 10-minutes timeout. The comparison is performed by computing the ratio $\frac{d_\mi{opt}}{d_\mi{apx}}$ between the optimal diversity $d_\mi{opt}$ provided by the algorithm $\ComputeDiverse{\ilpstargammafunc}$ and the approximate diversity  $d_\mi{apx}$ provided by the algorithm $\mathsf{Approx}[\ilpstargammafunc]$. 
Finally, in Section~\ref{sec:large-k}, we study the scalability of the approximation algorithms for larger values of $k$ (beyond those considered in Section~\ref{sec:exact-experiments}).

% Requires subcaption (for \ContinuedFloat in the doctors figure).

% =========================
% Graph analytics
% =========================

\begin{figure}[!tp]
	\centering
	\small
	
	%==================== Facebook ====================%
	\begin{minipage}[t]{\linewidth}
		\centering
		
		\begin{subfigure}[t]{0.36\linewidth}
			\centering
			\safeincludegraphics[height=2.57cm]{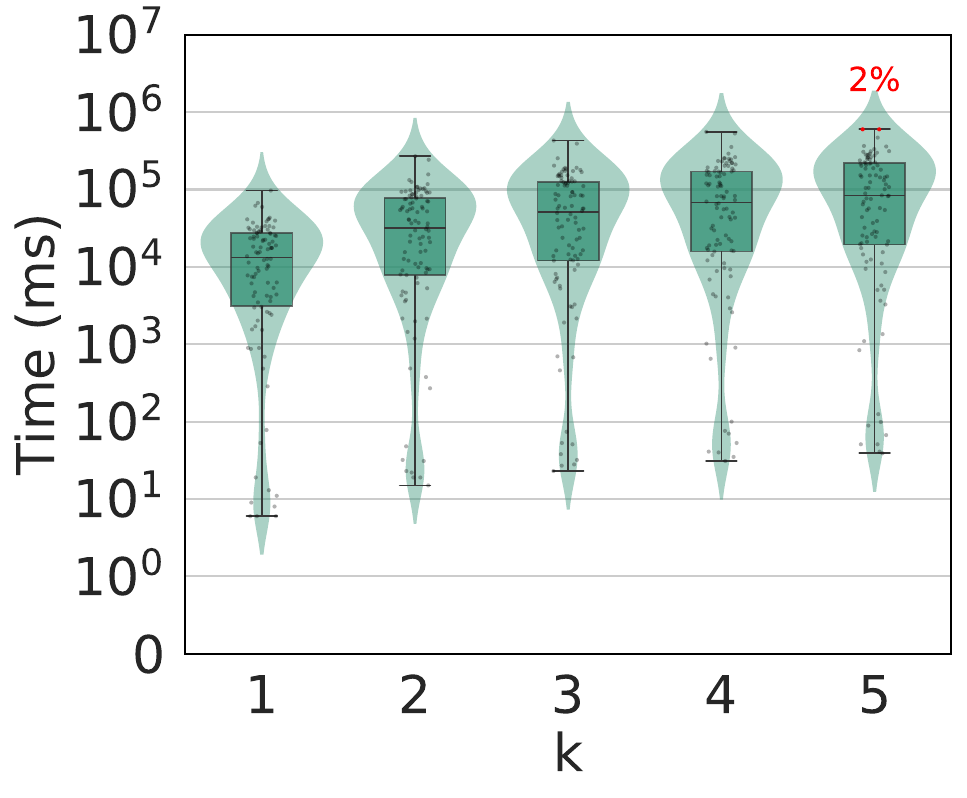}
			\caption{$\ilphypergraphfunc$}
		\end{subfigure}\hfill
		\begin{subfigure}[t]{0.30\linewidth}
			\centering
			\safeincludegraphics[height=2.5cm]{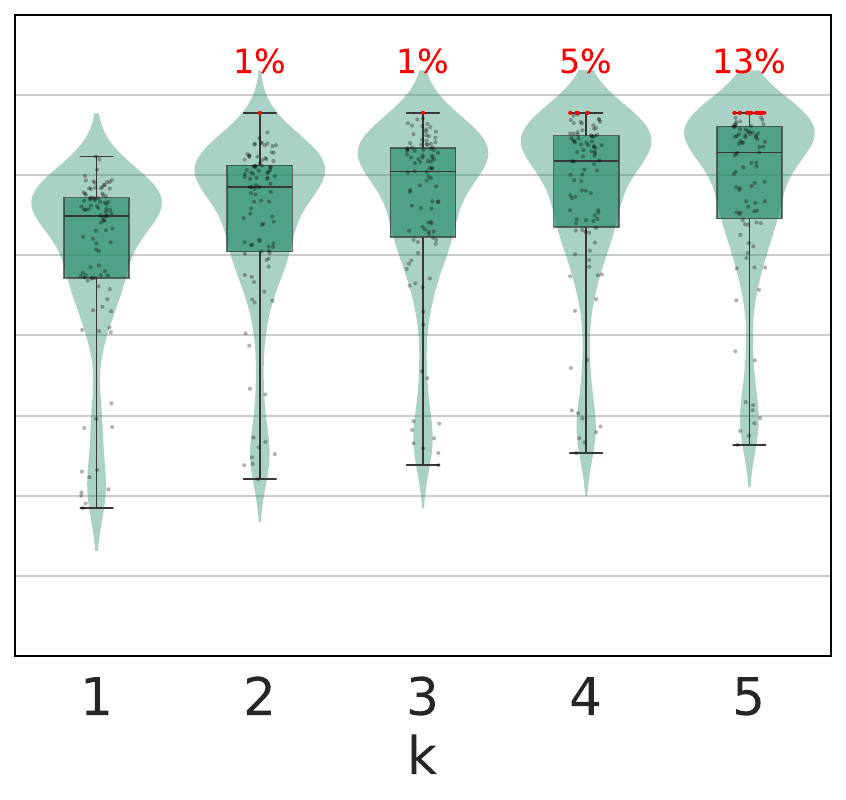}
			\caption{$\ilpunhypergraphfunc$}
		\end{subfigure}\hfill
		\begin{subfigure}[t]{0.30\linewidth}
			\centering
			\safeincludegraphics[height=2.5cm]{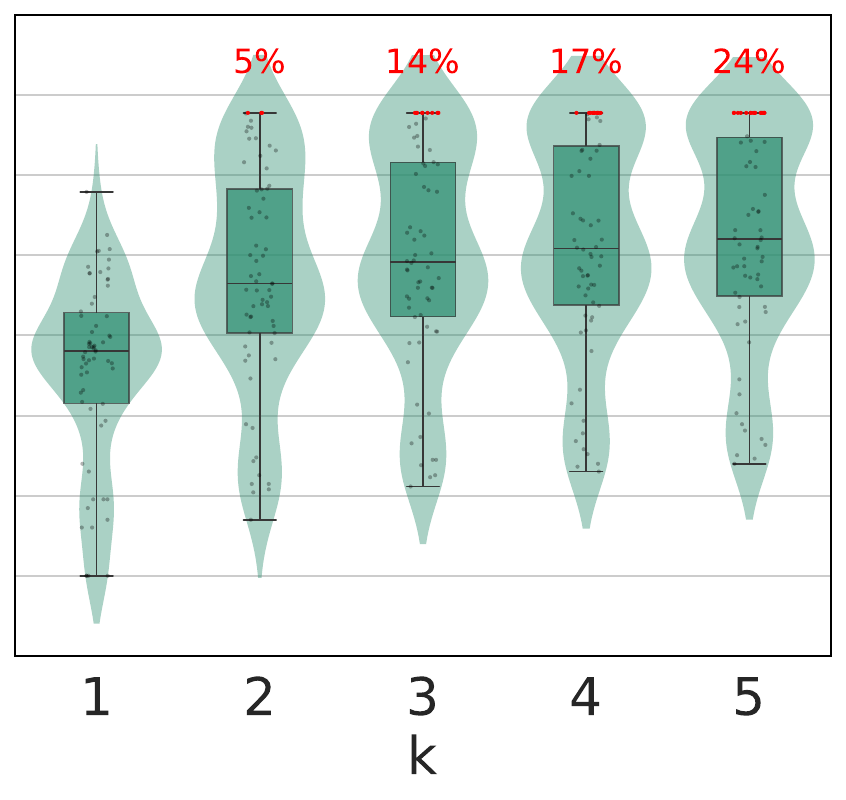}
			\caption{$\ilpunproofdagfunc$}
		\end{subfigure}
		
		\par\vspace{1em}
		\textbf{\graphanalytics (\facebook)}
	\end{minipage}

	\par\vspace{1em}

	%==================== Bitcoin ====================%
	\begin{minipage}[t]{\linewidth}
		\centering
		
		\begin{subfigure}[t]{0.36\linewidth}
			\centering
			\safeincludegraphics[height=2.57cm]{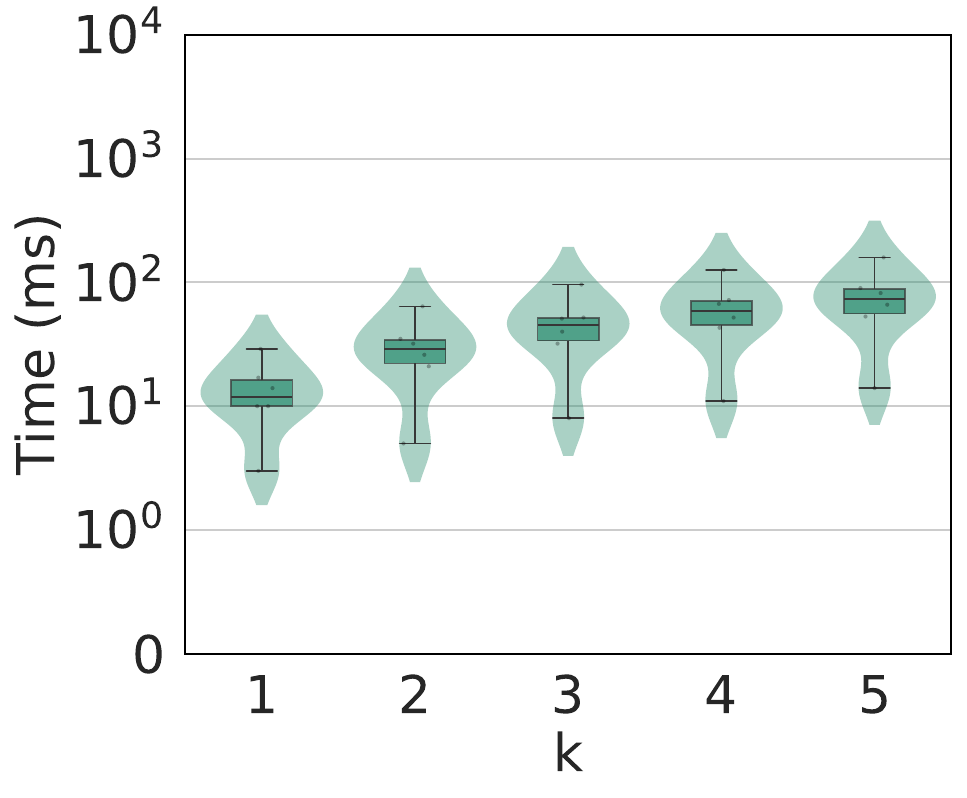}
			\caption{$\ilphypergraphfunc$}
		\end{subfigure}\hfill
		\begin{subfigure}[t]{0.30\linewidth}
			\centering
			\safeincludegraphics[height=2.5cm]{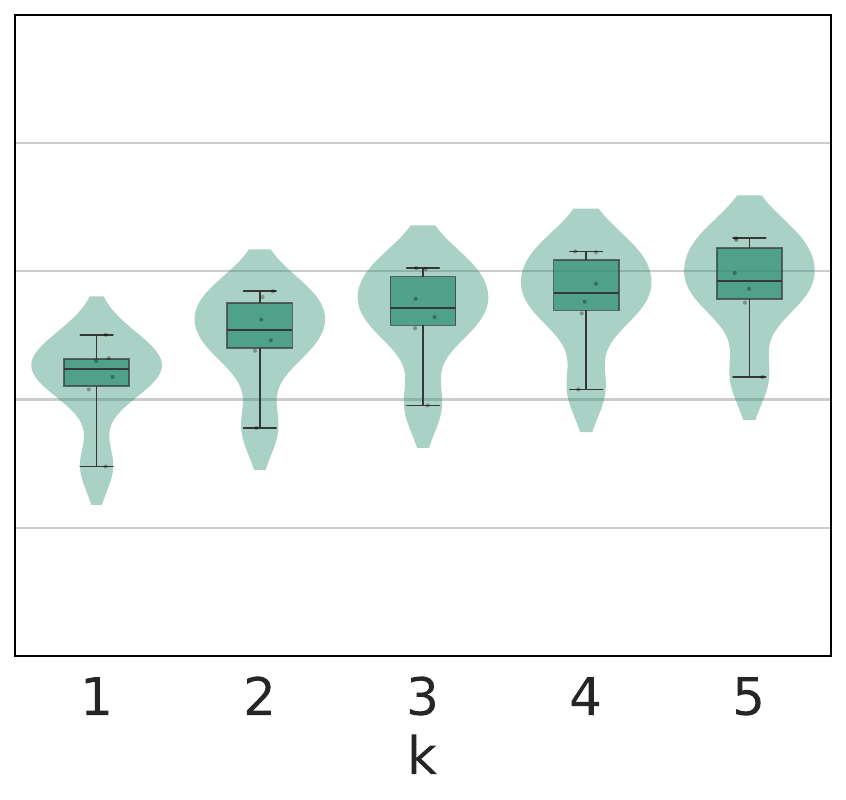}
			\caption{$\ilpunhypergraphfunc$}
		\end{subfigure}\hfill
		\begin{subfigure}[t]{0.30\linewidth}
			\centering
			\safeincludegraphics[height=2.5cm]{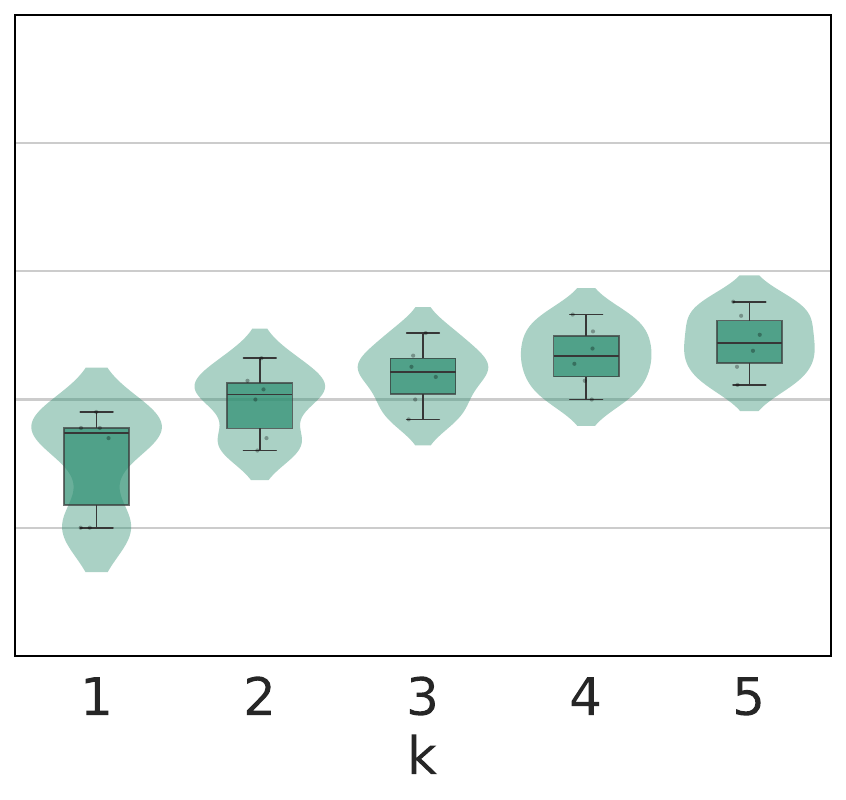}
			\caption{$\ilpunproofdagfunc$}
		\end{subfigure}
		
		\par\vspace{1em}
		\textbf{\graphanalytics (\bitcoin)}
	\end{minipage}
	
	\caption{Runtimes for the \plotlabel{\graphanalytics} scenario using the approximation algorithms.}
	\label{fig:graph-analytics-approximate}
\end{figure}

% =========================
% Doctors (split in two floats; same figure number)
% =========================

% Part 1: doctors 1-4
\begin{figure}[!tp]
	\centering
	\small
	
	%==================== doctors 1 ====================%
	\begin{minipage}[t]{\linewidth}
		\centering
		
		\begin{subfigure}[t]{0.36\linewidth}
			\centering
			\safeincludegraphics[height=2.57cm]{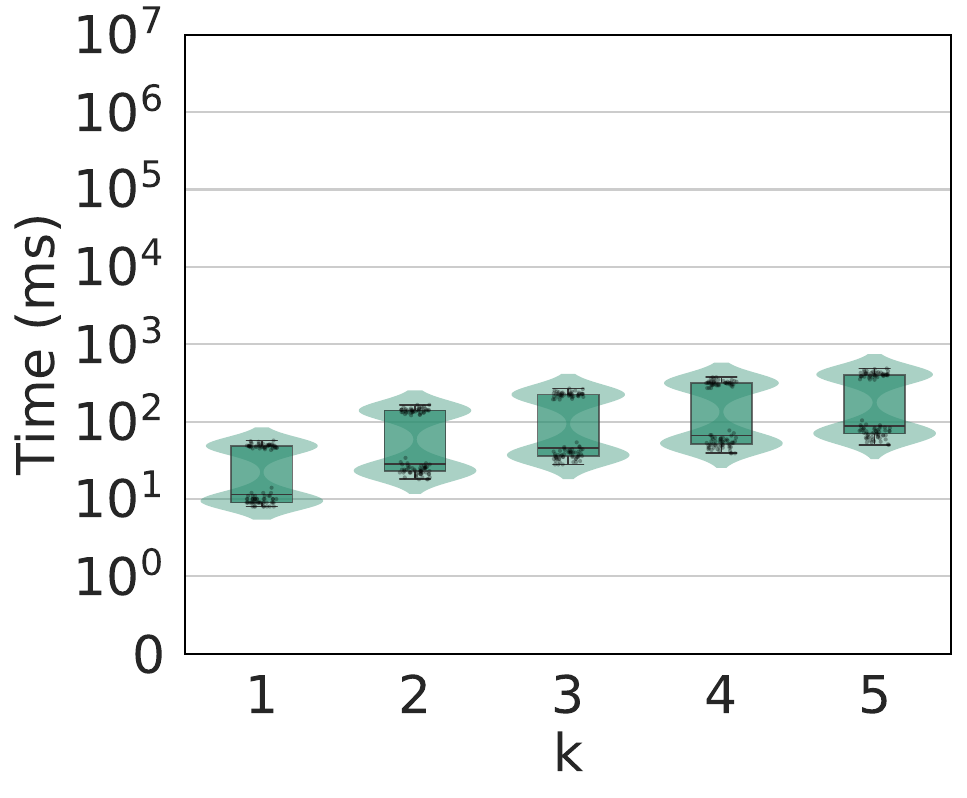}
			\caption{$\ilphypergraphfunc$}
		\end{subfigure}\hfill
		\begin{subfigure}[t]{0.30\linewidth}
			\centering
			\safeincludegraphics[height=2.5cm]{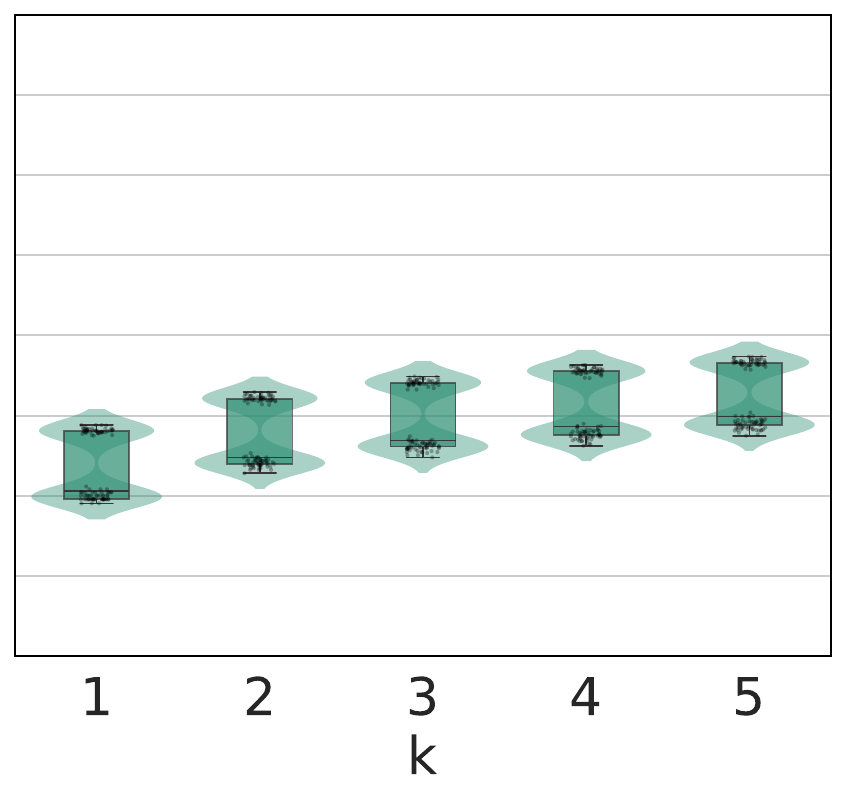}
			\caption{$\ilpunhypergraphfunc$}
		\end{subfigure}\hfill
		\begin{subfigure}[t]{0.30\linewidth}
			\centering
			\safeincludegraphics[height=2.5cm]{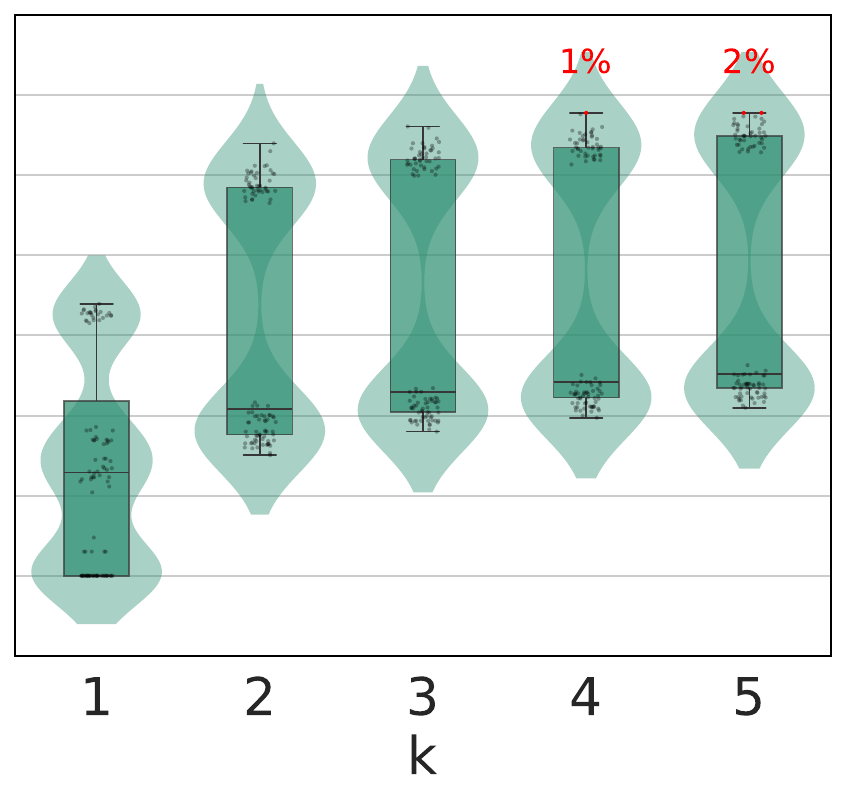}
			\caption{$\ilpunproofdagfunc$}
		\end{subfigure}
		
		\par\vspace{1em}
		\textbf{\doctors\ 1 (100k)}
	\end{minipage}

	\par\vspace{1em}

	%==================== doctors 2 ====================%
	\begin{minipage}[t]{\linewidth}
		\centering
		
		\begin{subfigure}[t]{0.36\linewidth}
			\centering
			\safeincludegraphics[height=2.57cm]{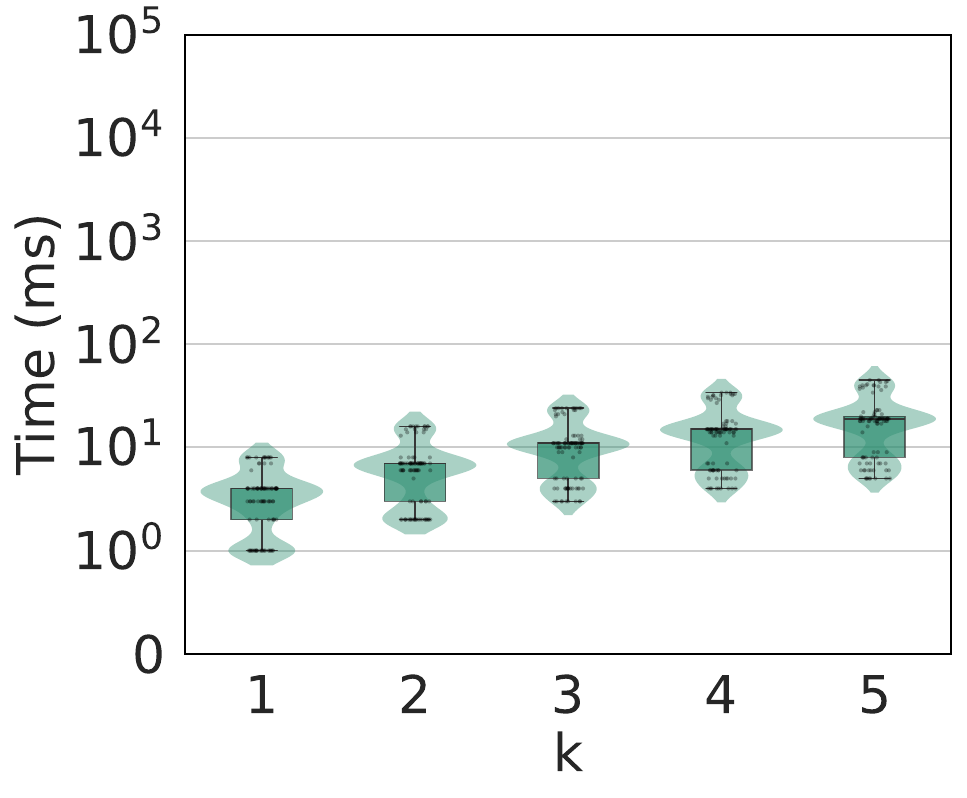}
			\caption{$\ilphypergraphfunc$}
		\end{subfigure}\hfill
		\begin{subfigure}[t]{0.30\linewidth}
			\centering
			\safeincludegraphics[height=2.5cm]{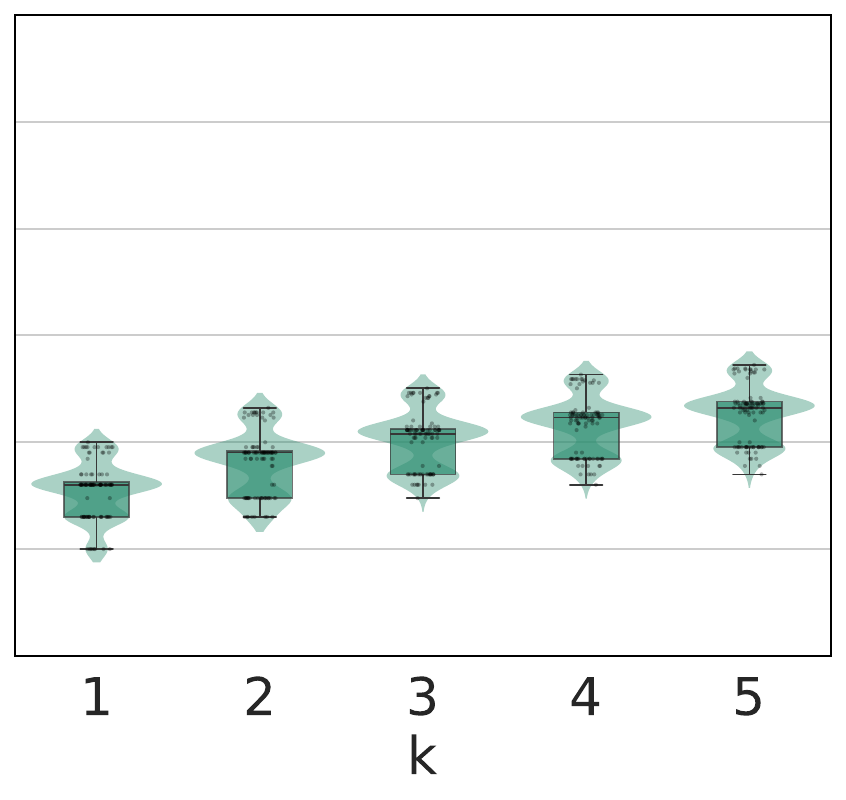}
			\caption{$\ilpunhypergraphfunc$}
		\end{subfigure}\hfill
		\begin{subfigure}[t]{0.30\linewidth}
			\centering
			\safeincludegraphics[height=2.5cm]{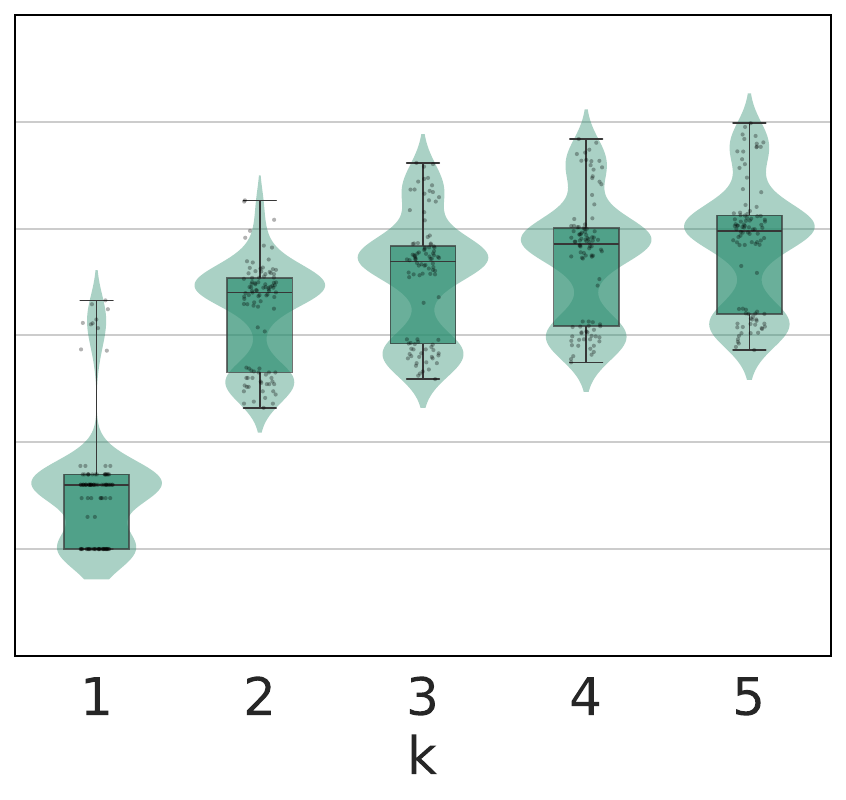}
			\caption{$\ilpunproofdagfunc$}
		\end{subfigure}
		
		\par\vspace{1em}
		\textbf{\doctors\ 2 (100k)}
	\end{minipage}

	\par\vspace{1em}

	%==================== doctors 3 ====================%
	\begin{minipage}[t]{\linewidth}
		\centering
		
		\begin{subfigure}[t]{0.36\linewidth}
			\centering
			\safeincludegraphics[height=2.57cm]{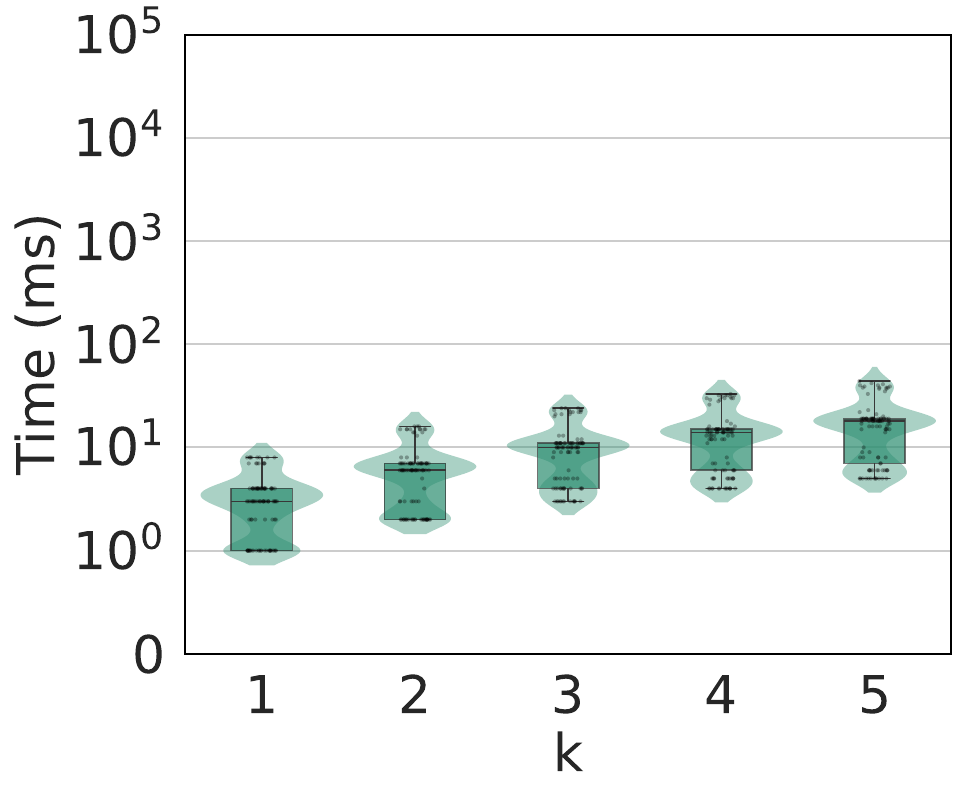}
			\caption{$\ilphypergraphfunc$}
		\end{subfigure}\hfill
		\begin{subfigure}[t]{0.30\linewidth}
			\centering
			\safeincludegraphics[height=2.5cm]{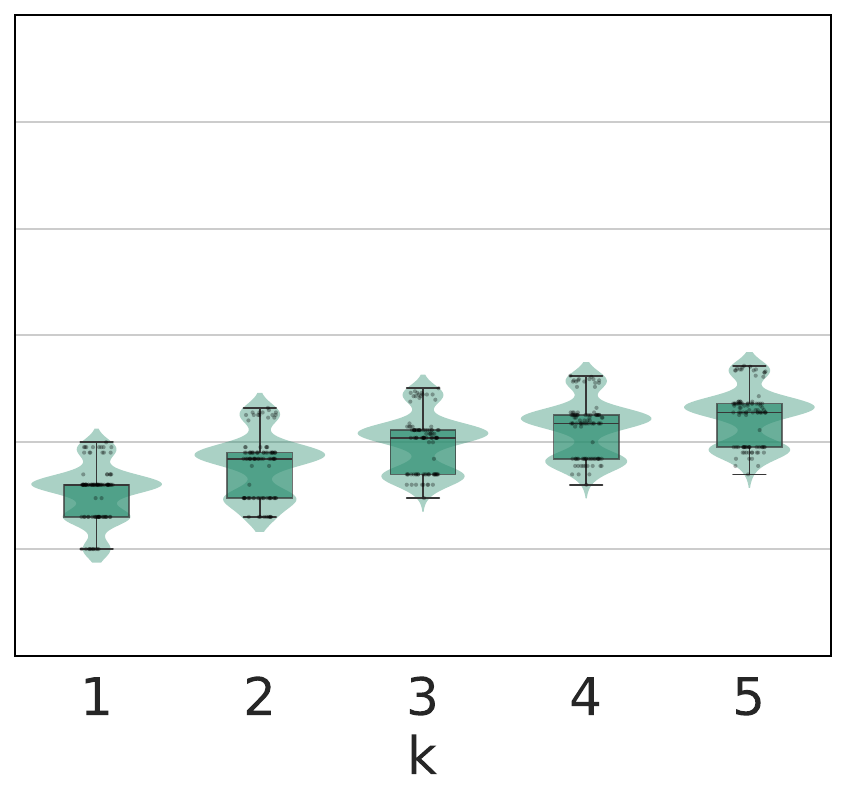}
			\caption{$\ilpunhypergraphfunc$}
		\end{subfigure}\hfill
		\begin{subfigure}[t]{0.30\linewidth}
			\centering
			\safeincludegraphics[height=2.5cm]{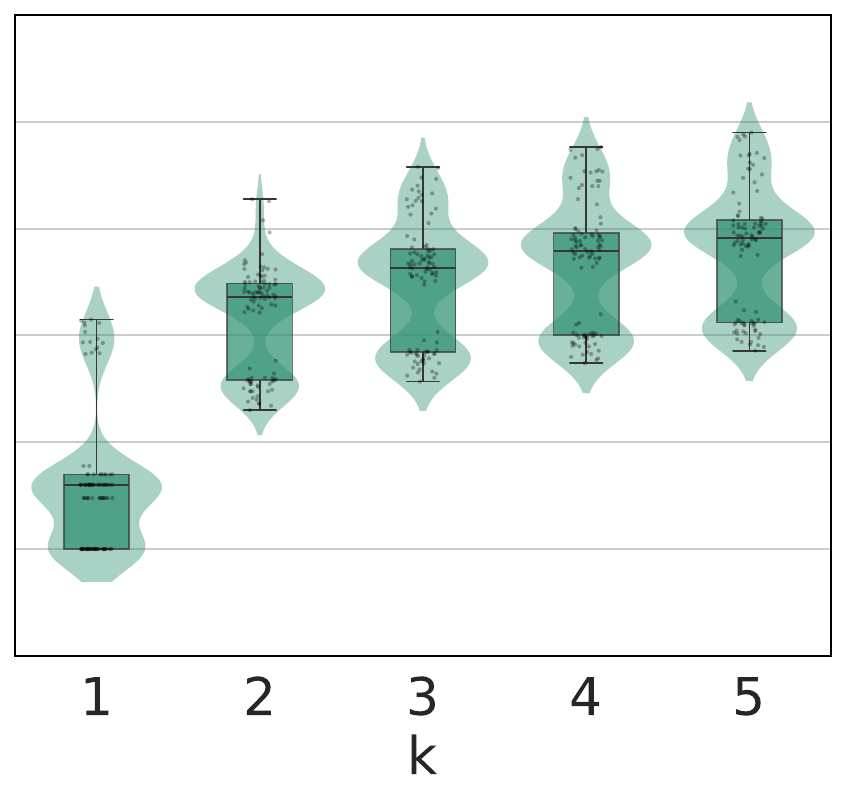}
			\caption{$\ilpunproofdagfunc$}
		\end{subfigure}
		
		\par\vspace{1em}
		\textbf{\doctors\ 3 (100k)}
	\end{minipage}

	\par\vspace{1em}

	%==================== doctors 4 ====================%
	\begin{minipage}[t]{\linewidth}
		\centering
		
		\begin{subfigure}[t]{0.36\linewidth}
			\centering
			\safeincludegraphics[height=2.57cm]{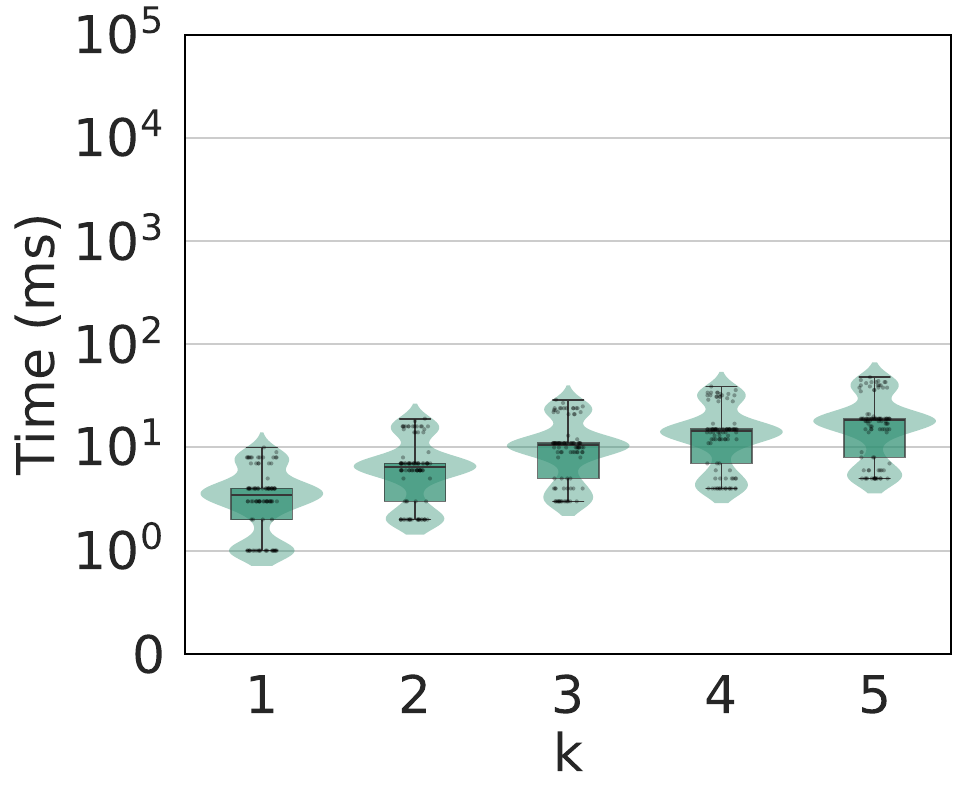}
			\caption{$\ilphypergraphfunc$}
		\end{subfigure}\hfill
		\begin{subfigure}[t]{0.30\linewidth}
			\centering
			\safeincludegraphics[height=2.5cm]{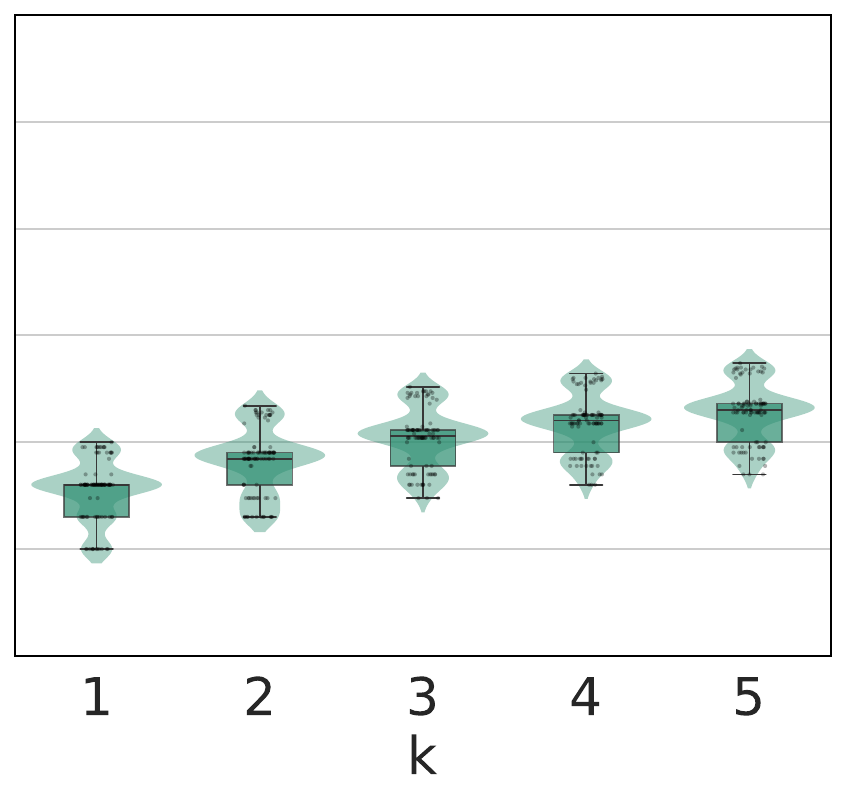}
			\caption{$\ilpunhypergraphfunc$}
		\end{subfigure}\hfill
		\begin{subfigure}[t]{0.30\linewidth}
			\centering
			\safeincludegraphics[height=2.5cm]{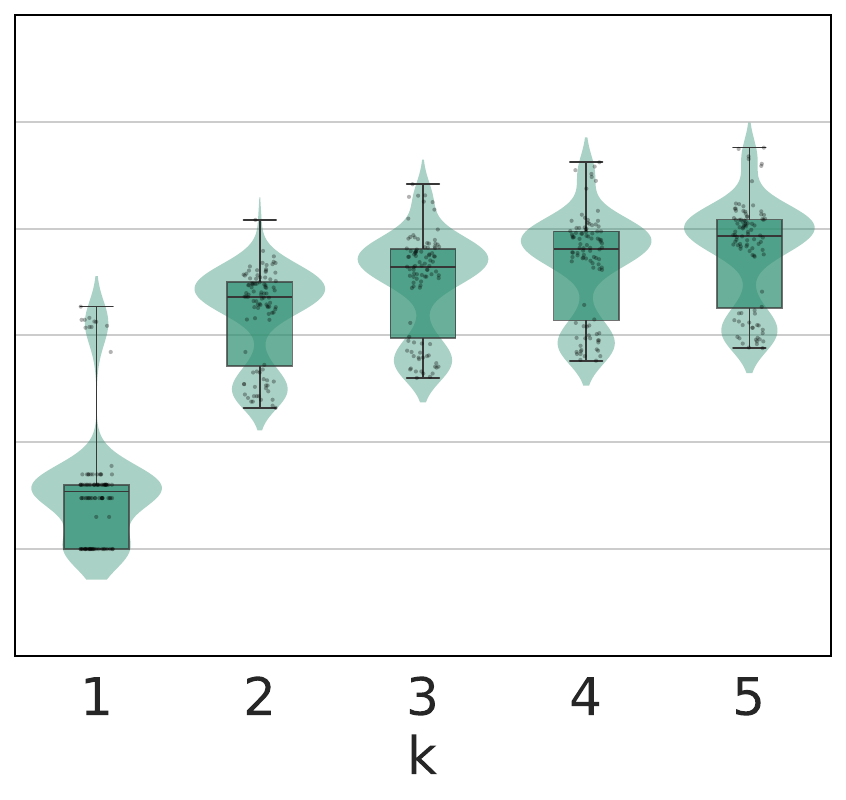}
			\caption{$\ilpunproofdagfunc$}
		\end{subfigure}
		
		\par\vspace{1em}
		\textbf{\doctors\ 4 (100k)}
	\end{minipage}
	
	\caption{Runtimes for the \doctors\ (1--4) scenarios using the approximation algorithms.}
	\label{fig:doctors-approximate}
\end{figure}

% Part 2: doctors 5-7 (subfigure letters continue)
\begin{figure}[!tp]
	\ContinuedFloat
	\centering
	\small
	
	%==================== doctors 5 ====================%
	\begin{minipage}[t]{\linewidth}
		\centering
		
		\begin{subfigure}[t]{0.36\linewidth}
			\centering
			\safeincludegraphics[height=2.57cm]{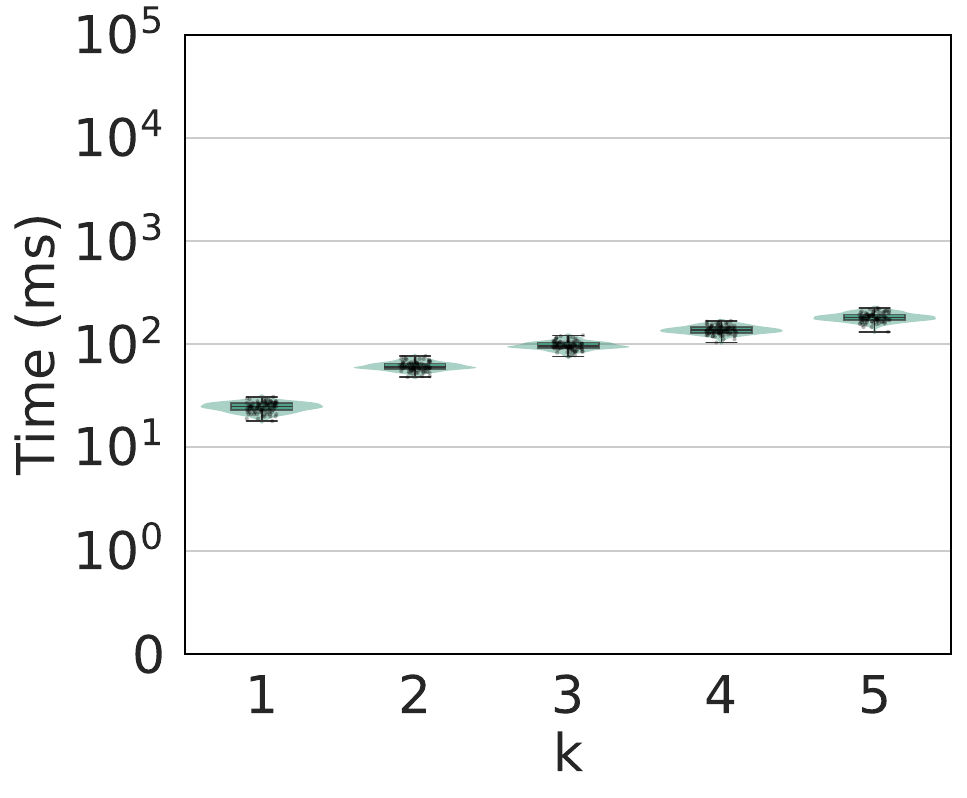}
			\caption{$\ilphypergraphfunc$}
		\end{subfigure}\hfill
		\begin{subfigure}[t]{0.30\linewidth}
			\centering
			\safeincludegraphics[height=2.5cm]{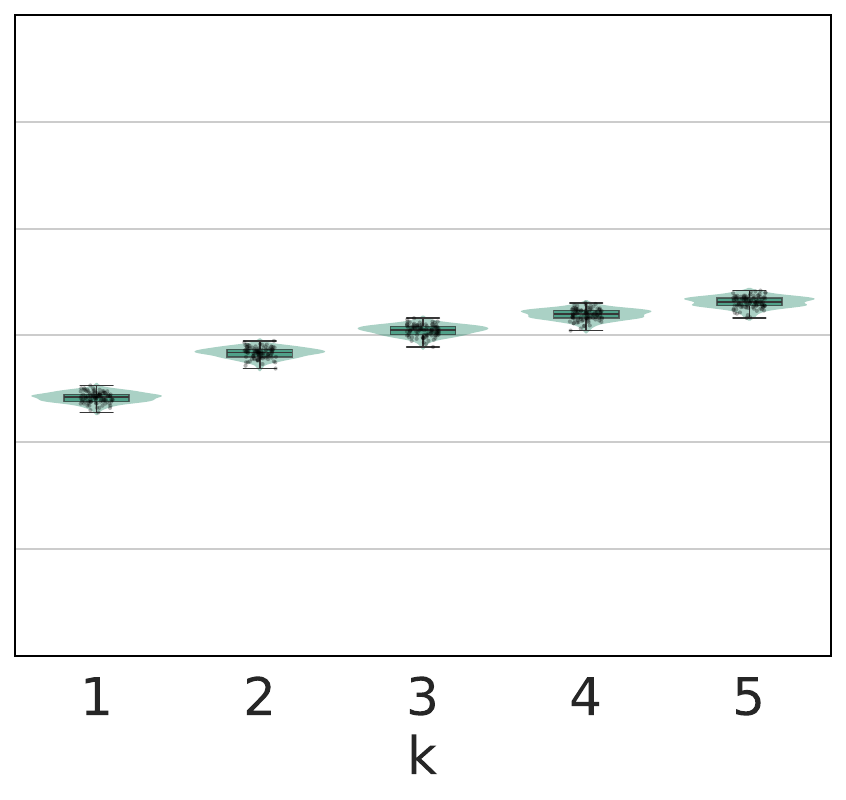}
			\caption{$\ilpunhypergraphfunc$}
		\end{subfigure}\hfill
		\begin{subfigure}[t]{0.30\linewidth}
			\centering
			\safeincludegraphics[height=2.5cm]{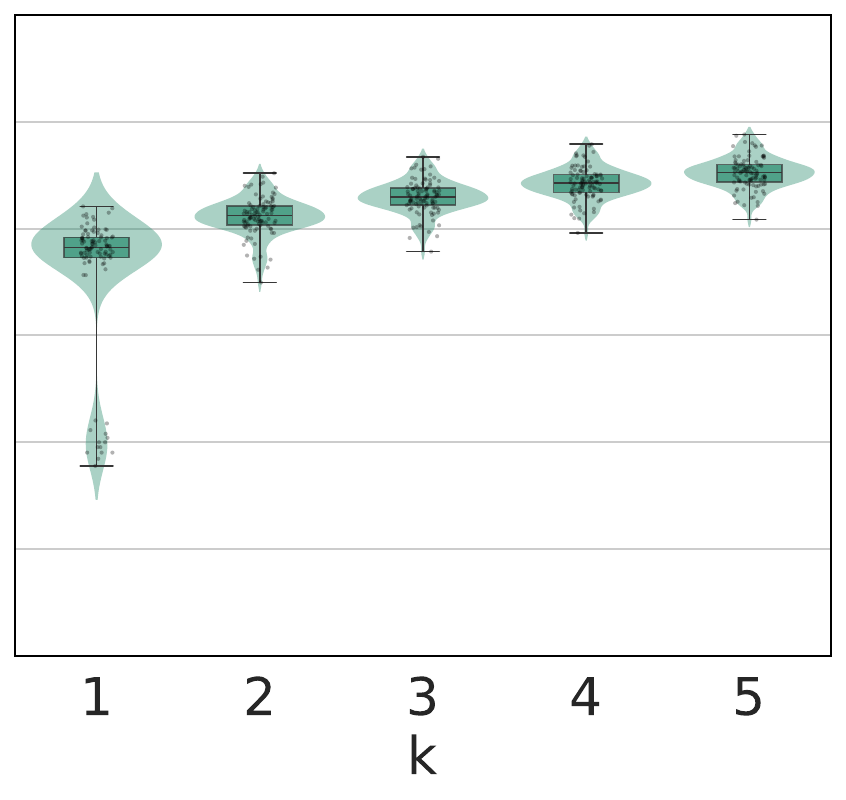}
			\caption{$\ilpunproofdagfunc$}
		\end{subfigure}
		
		\par\vspace{1em}
		\textbf{\doctors\ 5 (100k)}
	\end{minipage}

	\par\vspace{1em}

	%==================== doctors 6 ====================%
	\begin{minipage}[t]{\linewidth}
		\centering
		
		\begin{subfigure}[t]{0.36\linewidth}
			\centering
			\safeincludegraphics[height=2.57cm]{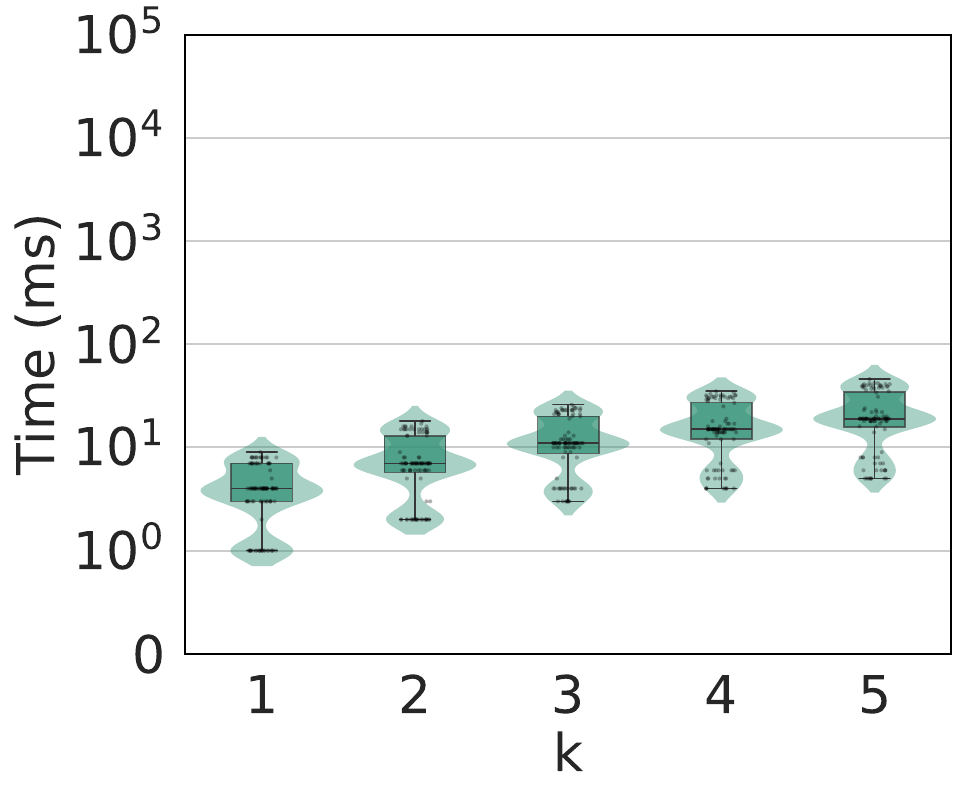}
			\caption{$\ilphypergraphfunc$}
		\end{subfigure}\hfill
		\begin{subfigure}[t]{0.30\linewidth}
			\centering
			\safeincludegraphics[height=2.5cm]{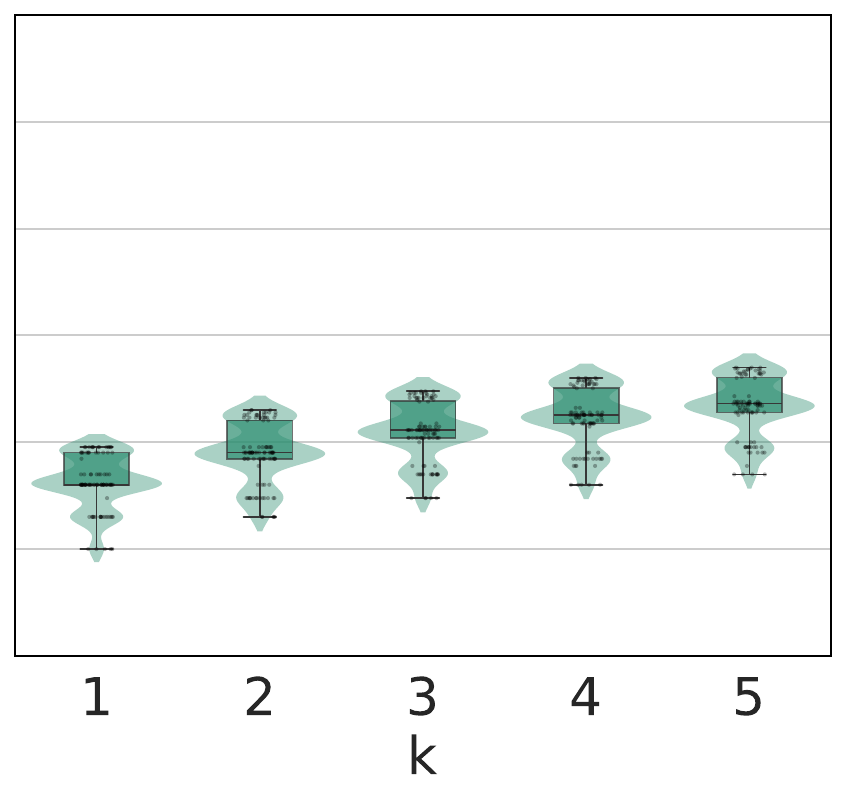}
			\caption{$\ilpunhypergraphfunc$}
		\end{subfigure}\hfill
		\begin{subfigure}[t]{0.30\linewidth}
			\centering
			\safeincludegraphics[height=2.5cm]{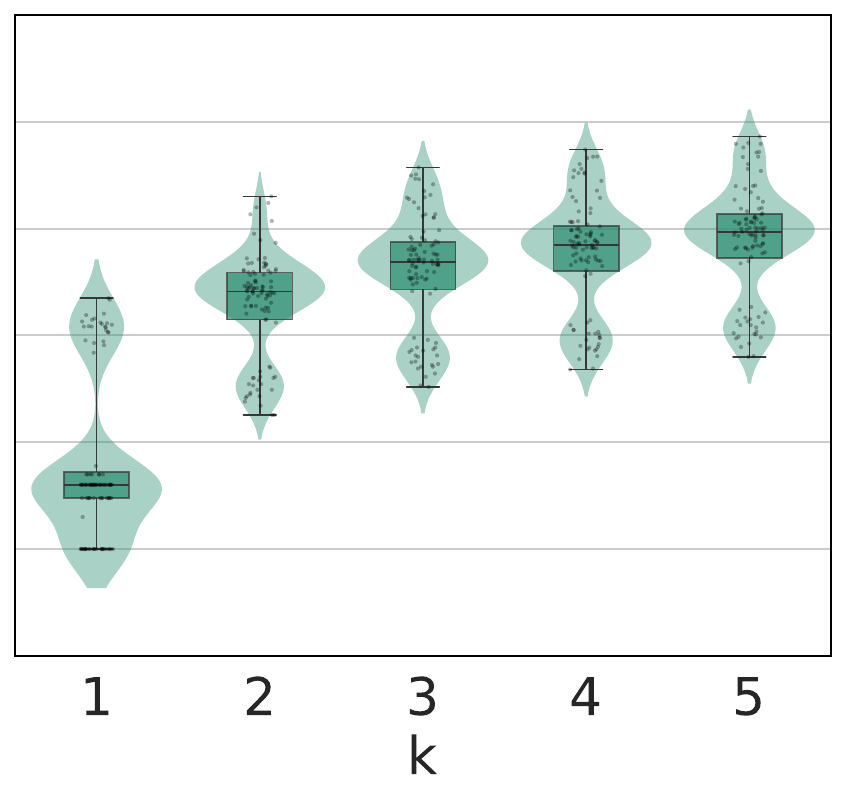}
			\caption{$\ilpunproofdagfunc$}
		\end{subfigure}
		
		\par\vspace{1em}
		\textbf{\doctors\ 6 (100k)}
	\end{minipage}

	\par\vspace{1em}

	%==================== doctors 7 ====================%
	\begin{minipage}[t]{\linewidth}
		\centering
		
		\begin{subfigure}[t]{0.36\linewidth}
			\centering
			\safeincludegraphics[height=2.57cm]{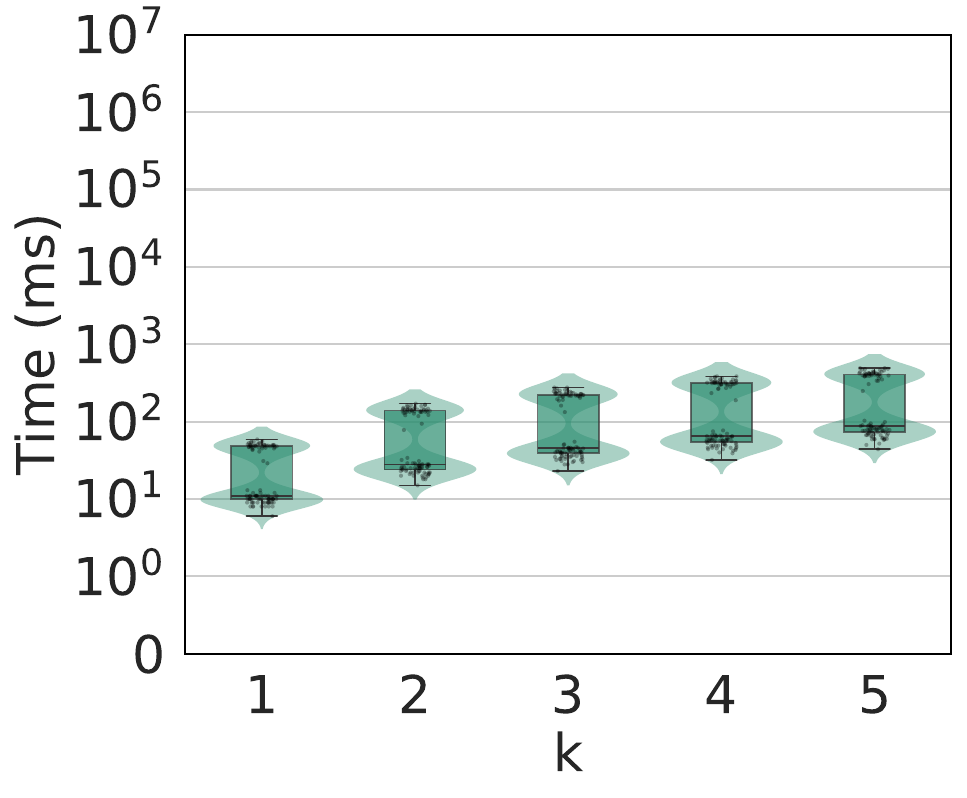}
			\caption{$\ilphypergraphfunc$}
		\end{subfigure}\hfill
		\begin{subfigure}[t]{0.30\linewidth}
			\centering
			\safeincludegraphics[height=2.5cm]{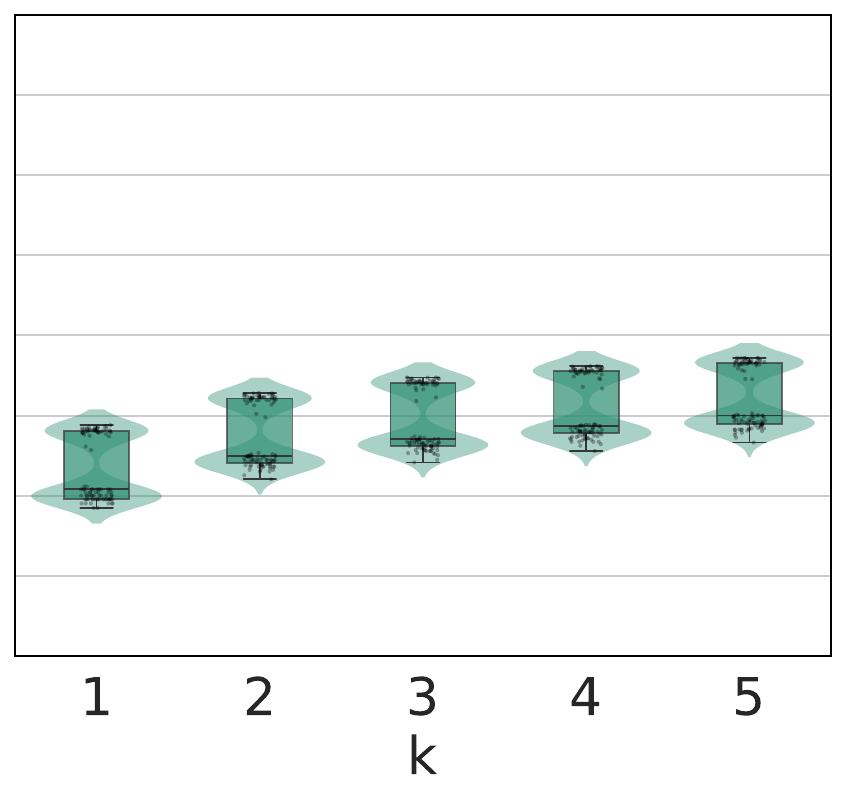}
			\caption{$\ilpunhypergraphfunc$}
		\end{subfigure}\hfill
		\begin{subfigure}[t]{0.30\linewidth}
			\centering
			\safeincludegraphics[height=2.5cm]{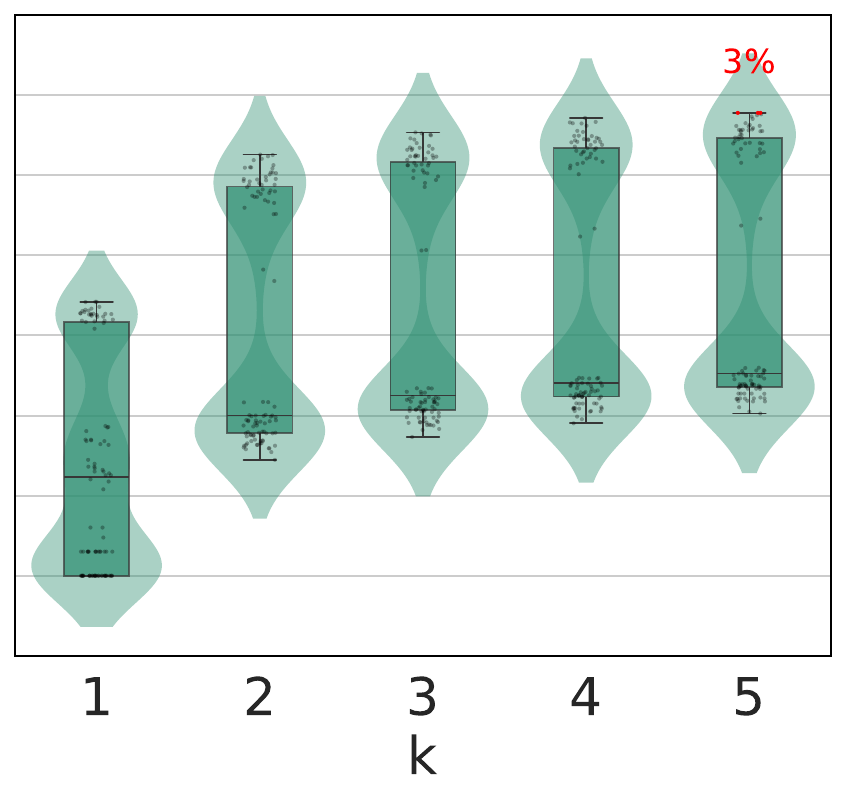}
			\caption{$\ilpunproofdagfunc$}
		\end{subfigure}
		
		\par\vspace{1em}
		\textbf{\doctors\ 7 (100k)}
	\end{minipage}
	
	\caption[]{Runtimes for the \doctors\ (5--7) scenarios using the approximation algorithms.}
\end{figure}

% =========================
% CSDA
% =========================

\begin{figure}[!tp]
	\centering
	\small
	
	%==================== HTTPD ====================%
	\begin{minipage}[t]{\linewidth}
		\centering
		
		\begin{subfigure}[t]{0.36\linewidth}
			\centering
			\safeincludegraphics[height=2.57cm]{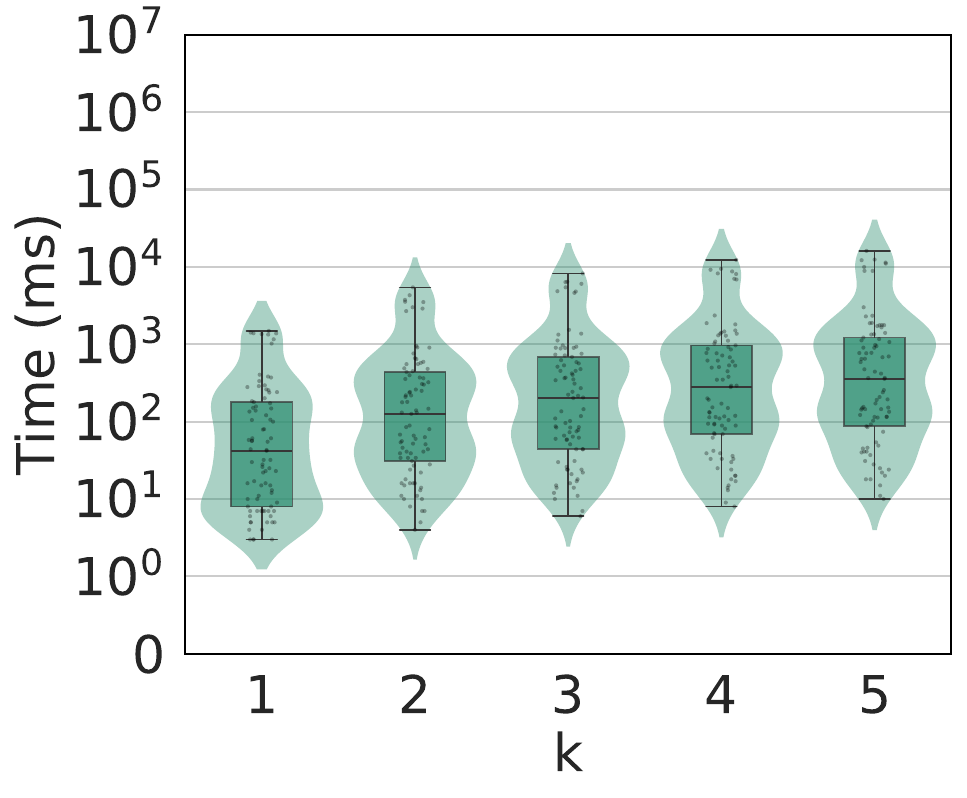}
			\caption{$\ilphypergraphfunc$}
		\end{subfigure}\hfill
		\begin{subfigure}[t]{0.30\linewidth}
			\centering
			\safeincludegraphics[height=2.5cm]{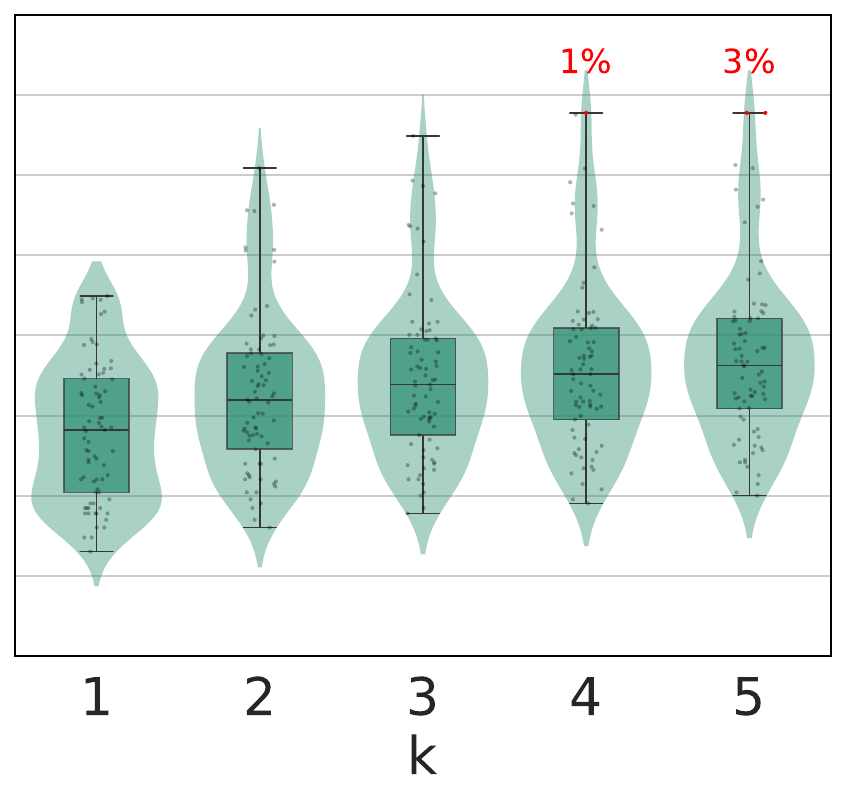}
			\caption{$\ilpunhypergraphfunc$}
		\end{subfigure}\hfill
		\begin{subfigure}[t]{0.30\linewidth}
			\centering
			\safeincludegraphics[height=2.5cm]{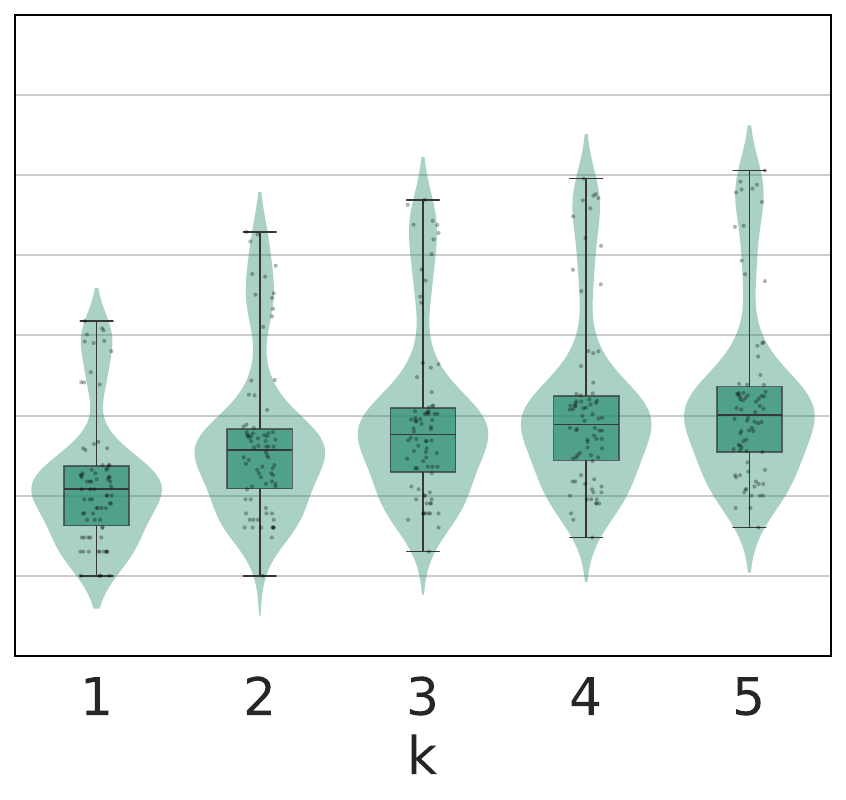}
			\caption{$\ilpunproofdagfunc$}
		\end{subfigure}
		
		\par\vspace{1em}
		\textbf{\csda (\httpd)}
	\end{minipage}

	\par\vspace{1em}

	%==================== PostgreSQL ====================%
	\begin{minipage}[t]{\linewidth}
		\centering
		
		\begin{subfigure}[t]{0.36\linewidth}
			\centering
			\safeincludegraphics[height=2.57cm]{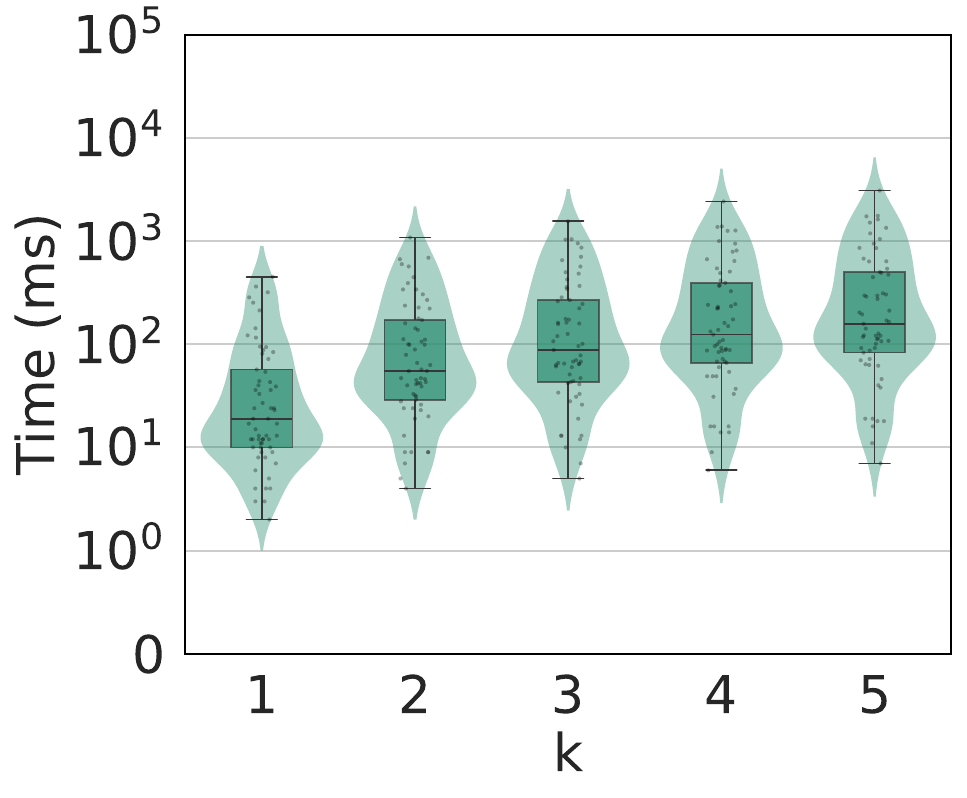}
			\caption{$\ilphypergraphfunc$}
		\end{subfigure}\hfill
		\begin{subfigure}[t]{0.30\linewidth}
			\centering
			\safeincludegraphics[height=2.5cm]{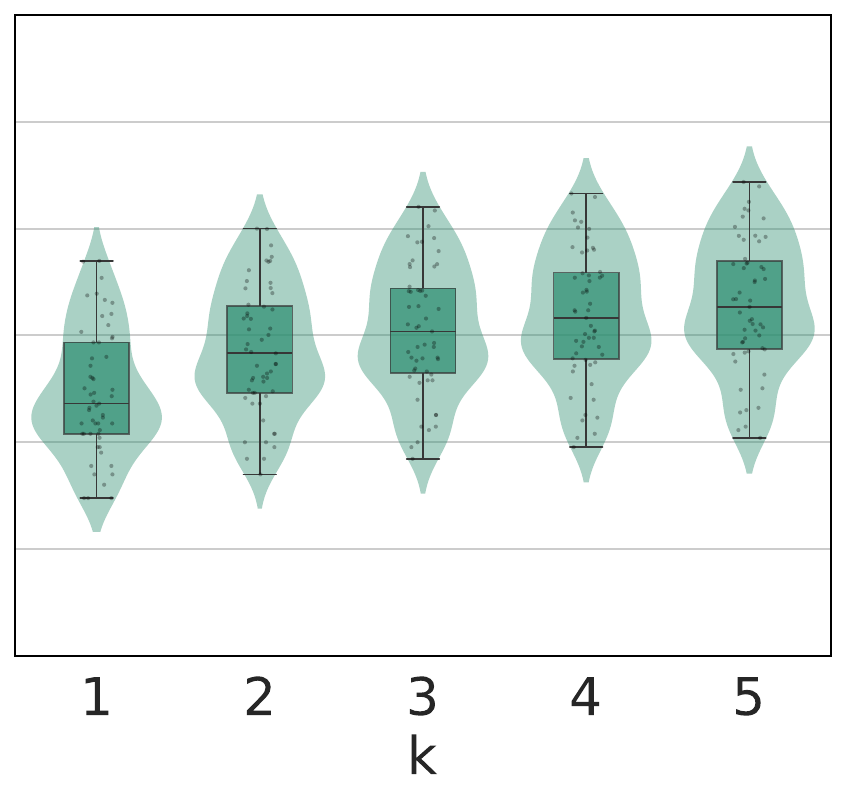}
			\caption{$\ilpunhypergraphfunc$}
		\end{subfigure}\hfill
		\begin{subfigure}[t]{0.30\linewidth}
			\centering
			\safeincludegraphics[height=2.5cm]{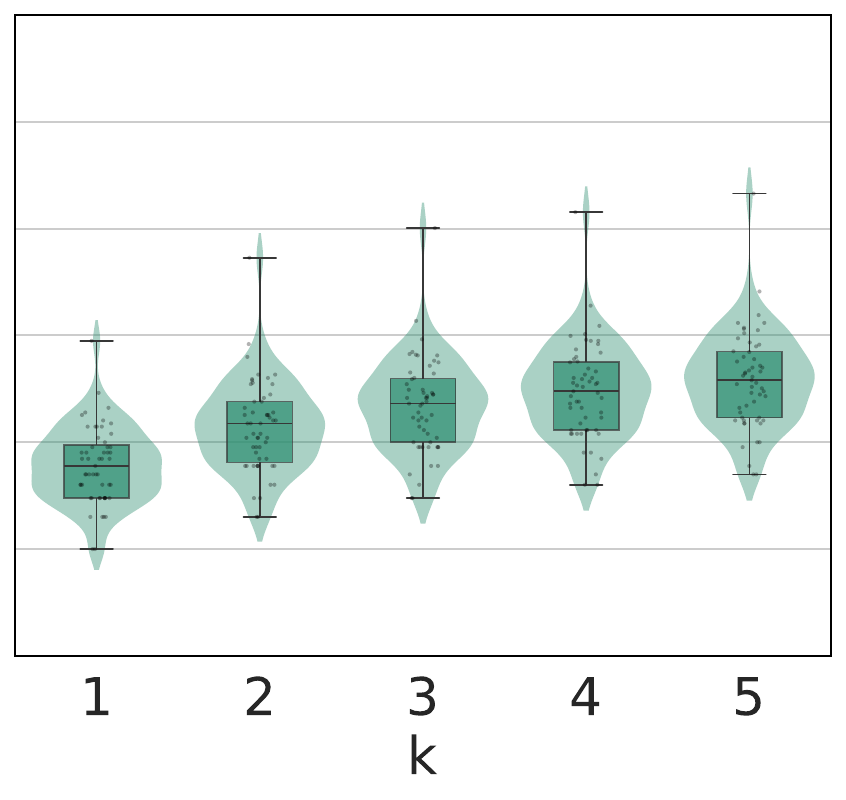}
			\caption{$\ilpunproofdagfunc$}
		\end{subfigure}
		
		\par\vspace{1em}
		\textbf{\csda (\postgres)}
	\end{minipage}

	\par\vspace{1em}

	%==================== Linux ====================%
	\begin{minipage}[t]{\linewidth}
		\centering
		
		\begin{subfigure}[t]{0.36\linewidth}
			\centering
			\safeincludegraphics[height=2.57cm]{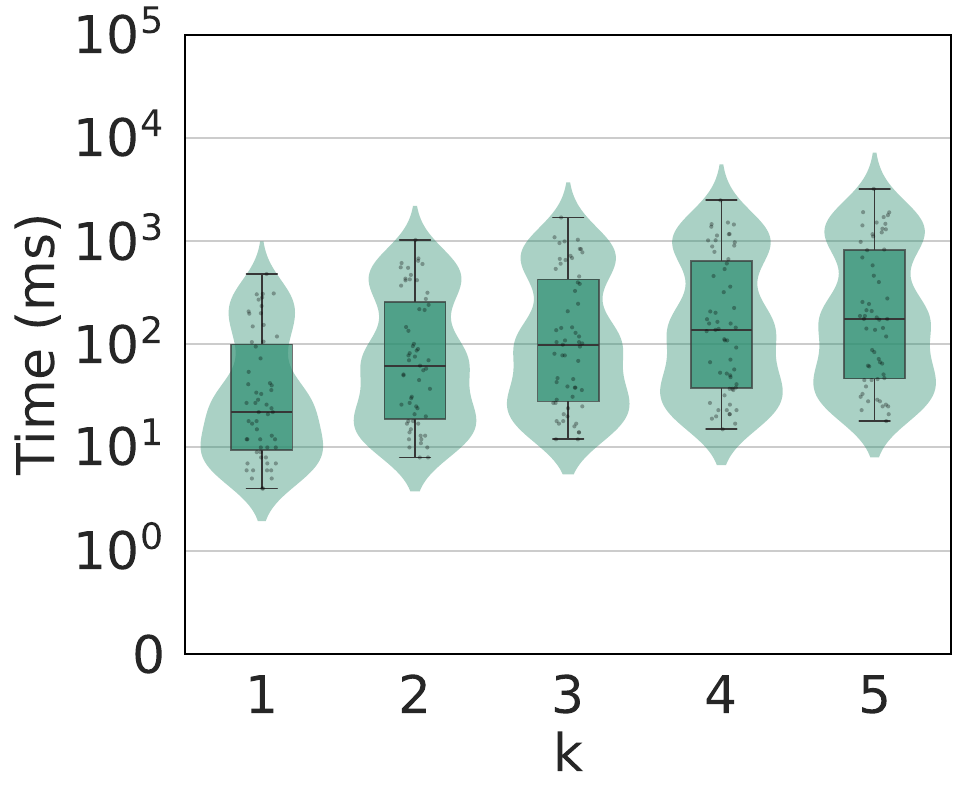}
			\caption{$\ilphypergraphfunc$}
		\end{subfigure}\hfill
		\begin{subfigure}[t]{0.30\linewidth}
			\centering
			\safeincludegraphics[height=2.5cm]{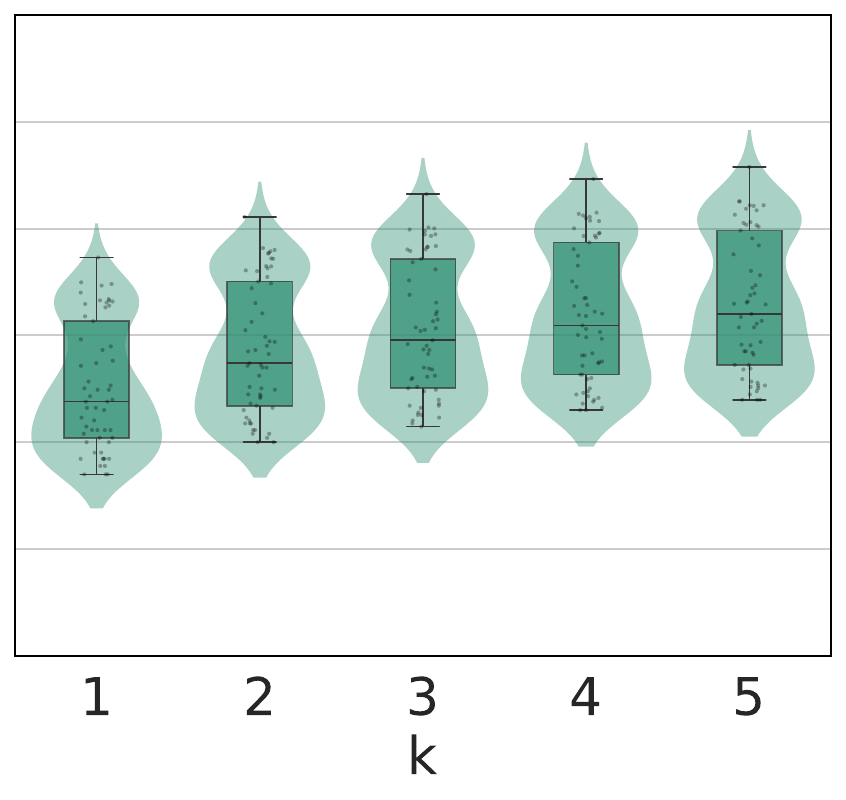}
			\caption{$\ilpunhypergraphfunc$}
		\end{subfigure}\hfill
		\begin{subfigure}[t]{0.30\linewidth}
			\centering
			\safeincludegraphics[height=2.5cm]{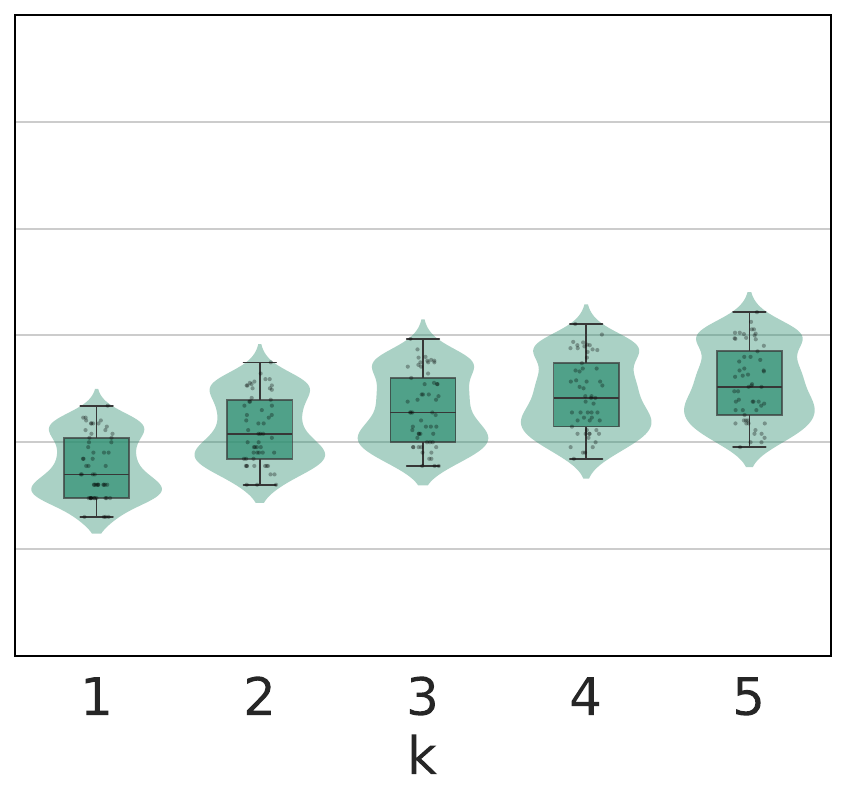}
			\caption{$\ilpunproofdagfunc$}
		\end{subfigure}
		
		\par\vspace{1em}
		\textbf{\csda (\linux)}
	\end{minipage}
	
	\caption{Runtimes for the \csda\ scenario using the approximation algorithms.}
	\label{fig:csda-approximate}
\end{figure}

\begin{figure}[!tp]
	\centering
	\small
	
	%==================== Facebook ====================%
	\begin{minipage}[t]{\linewidth}
		\centering
		
		\begin{subfigure}[t]{0.36\linewidth}
			\centering
			\safeincludegraphics[height=2.57cm]{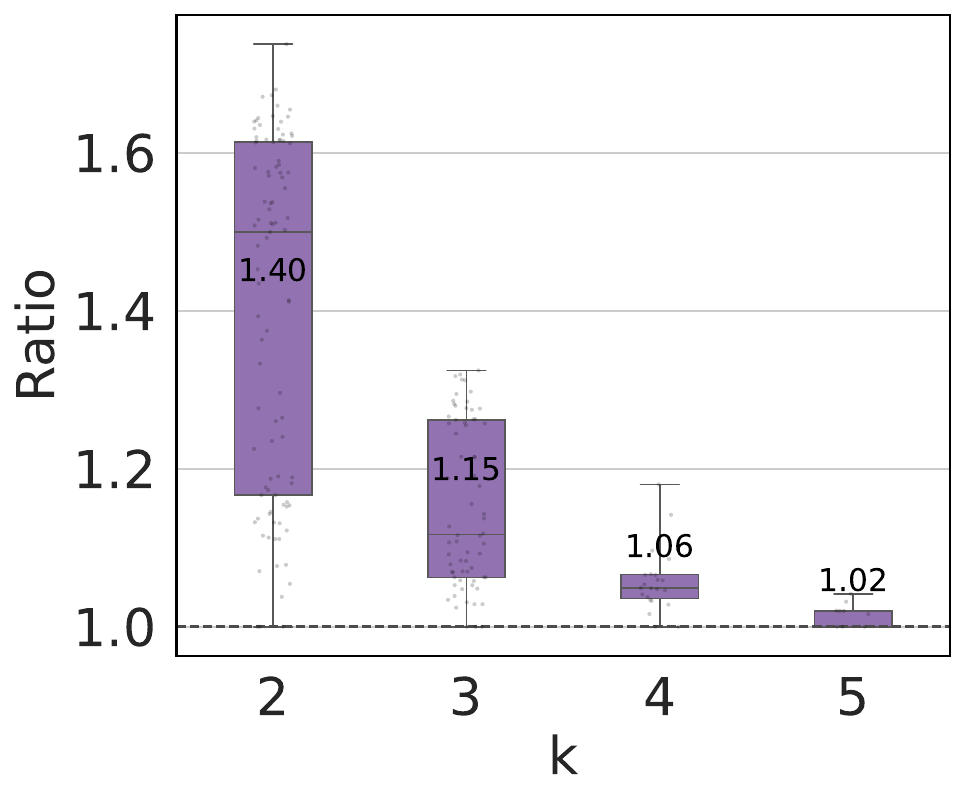}
			\caption{$\ilphypergraphfunc$}
		\end{subfigure}\hfill
		\begin{subfigure}[t]{0.30\linewidth}
			\centering
			\safeincludegraphics[height=2.5cm]{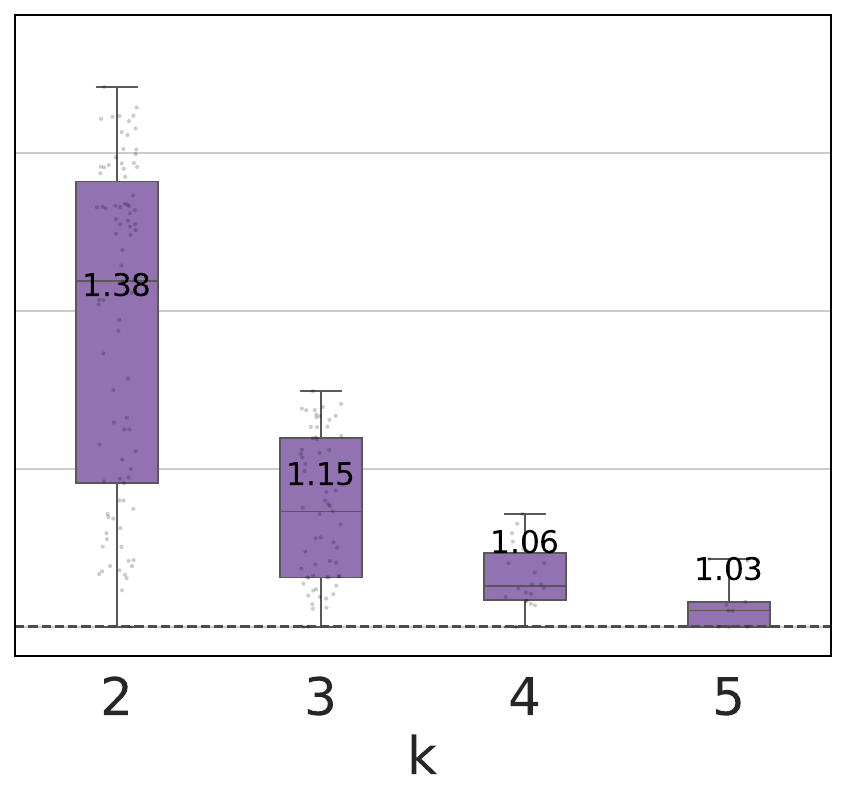}
			\caption{$\ilpunhypergraphfunc$}
		\end{subfigure}\hfill
		\begin{subfigure}[t]{0.30\linewidth}
			\centering
			\safeincludegraphics[height=2.5cm]{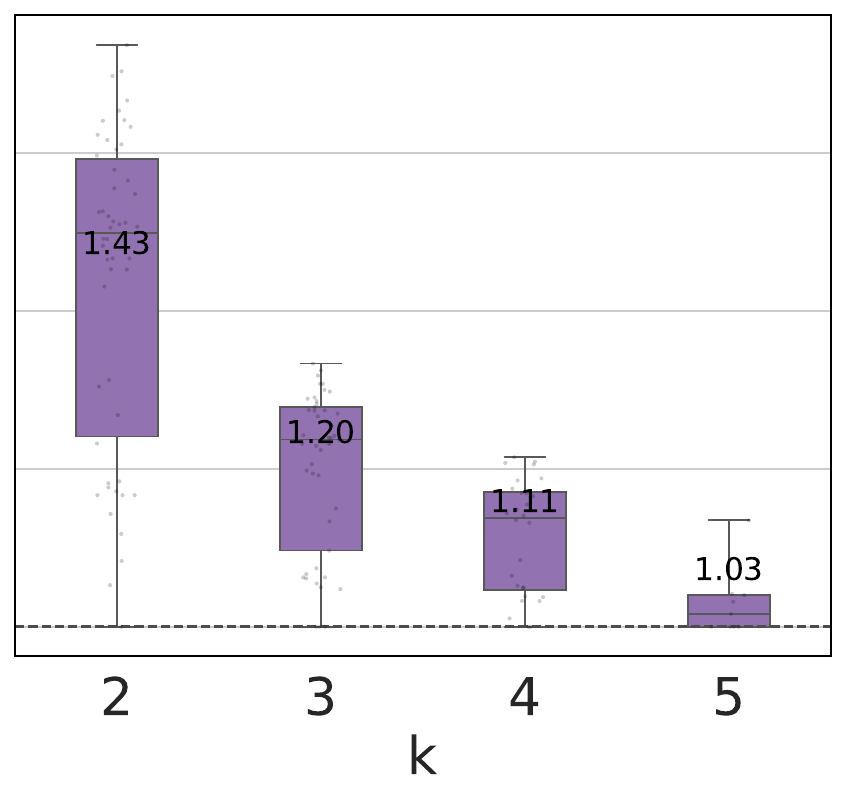}
			\caption{$\ilpunproofdagfunc$}
		\end{subfigure}
		
		\par\vspace{1em}
		\textbf{\graphanalytics (\facebook)}
	\end{minipage}
	
	\par\vspace{1em}
	
	%==================== Bitcoin ====================%
	\begin{minipage}[t]{\linewidth}
		\centering
		
		\begin{subfigure}[t]{0.36\linewidth}
			\centering
			\safeincludegraphics[height=2.57cm]{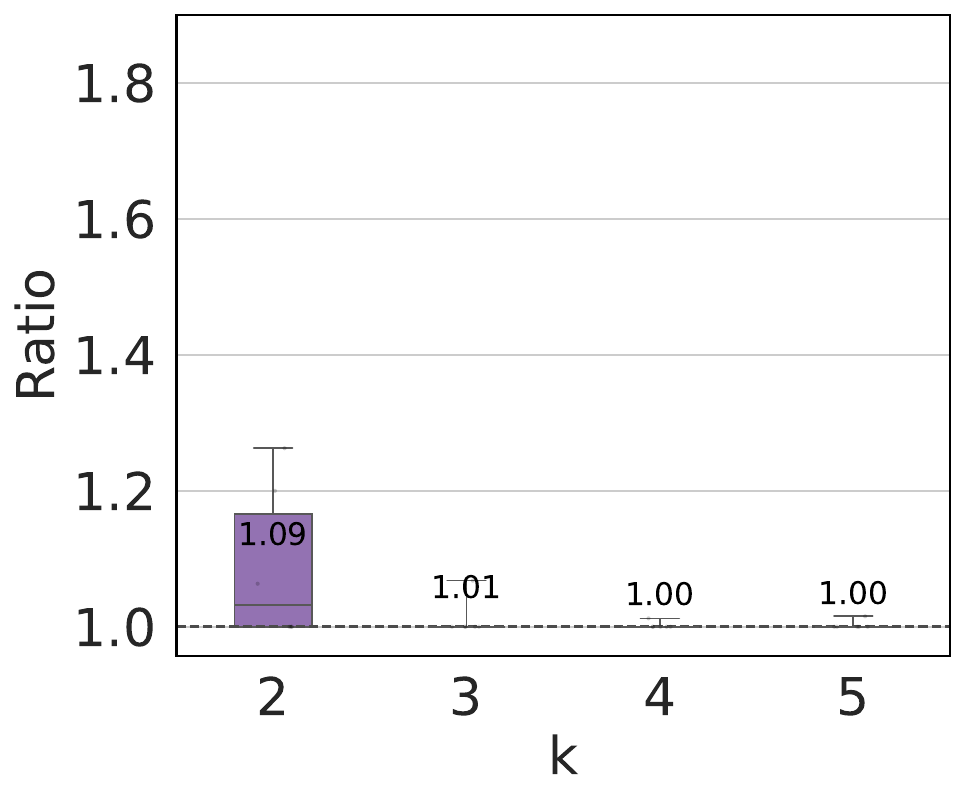}
			\caption{$\ilphypergraphfunc$}
		\end{subfigure}\hfill
		\begin{subfigure}[t]{0.30\linewidth}
			\centering
			\safeincludegraphics[height=2.5cm]{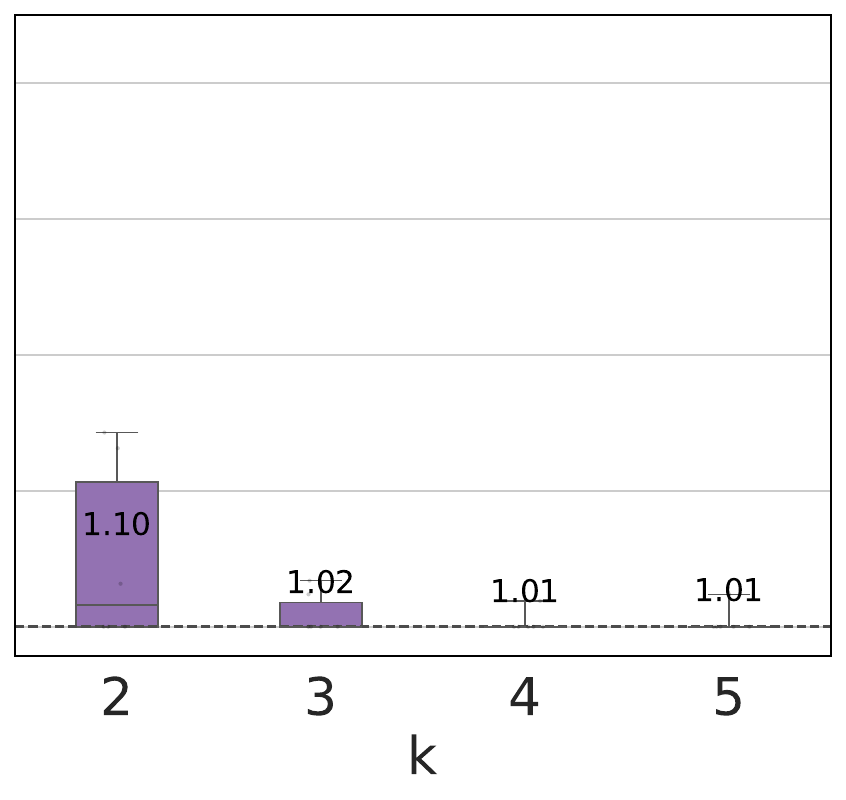}
			\caption{$\ilpunhypergraphfunc$}
		\end{subfigure}\hfill
		\begin{subfigure}[t]{0.30\linewidth}
			\centering
			\safeincludegraphics[height=2.5cm]{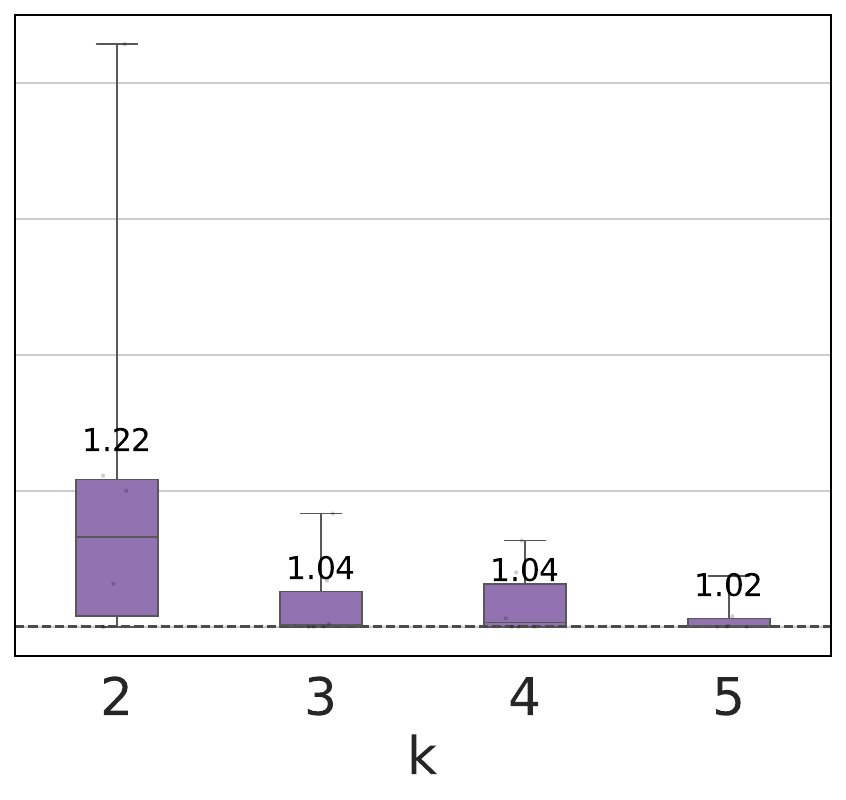}
			\caption{$\ilpunproofdagfunc$}
		\end{subfigure}
		
		\par\vspace{1em}
		\textbf{\graphanalytics (\bitcoin)}
	\end{minipage}
	
	\caption{Ratio (optimal / greedy) distribution across $k$ for the \graphanalytics\ scenario.}
	\label{fig:graph-analytics-rate}
\end{figure}

\begin{figure}[!tp]
	\centering
	\small
	
	%==================== HTTPD ====================%
	\begin{minipage}[t]{\linewidth}
		\centering
		
		\begin{subfigure}[t]{0.36\linewidth}
			\centering
			\safeincludegraphics[height=2.57cm]{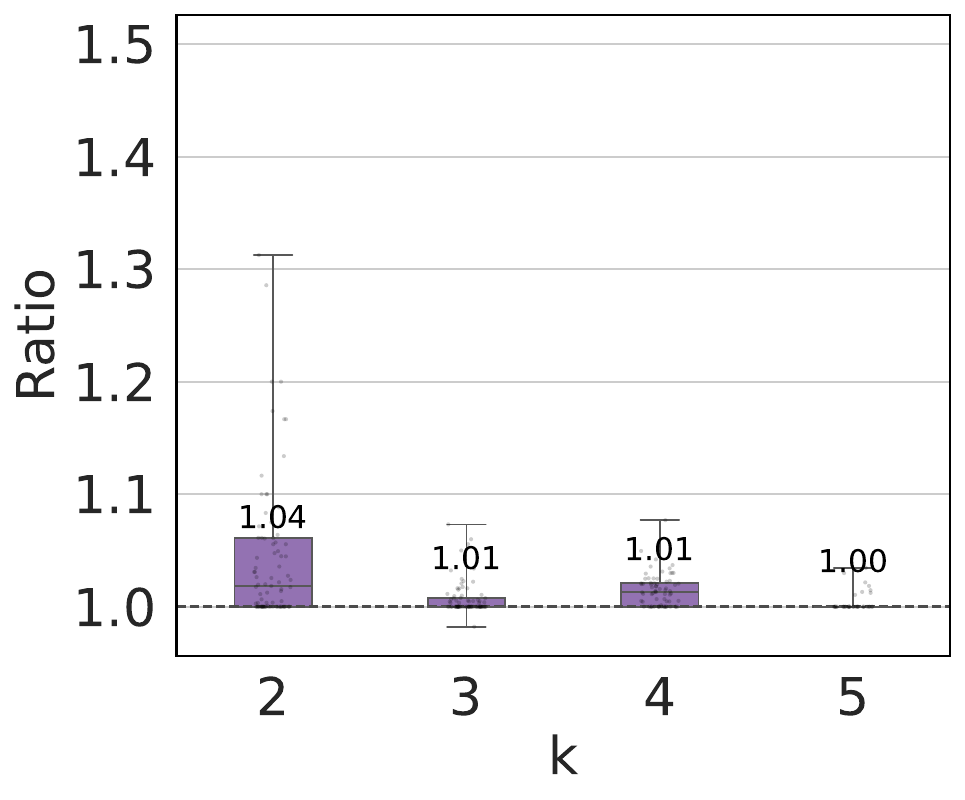}
			\caption{$\ilphypergraphfunc$}
		\end{subfigure}\hfill
		\begin{subfigure}[t]{0.30\linewidth}
			\centering
			\safeincludegraphics[height=2.5cm]{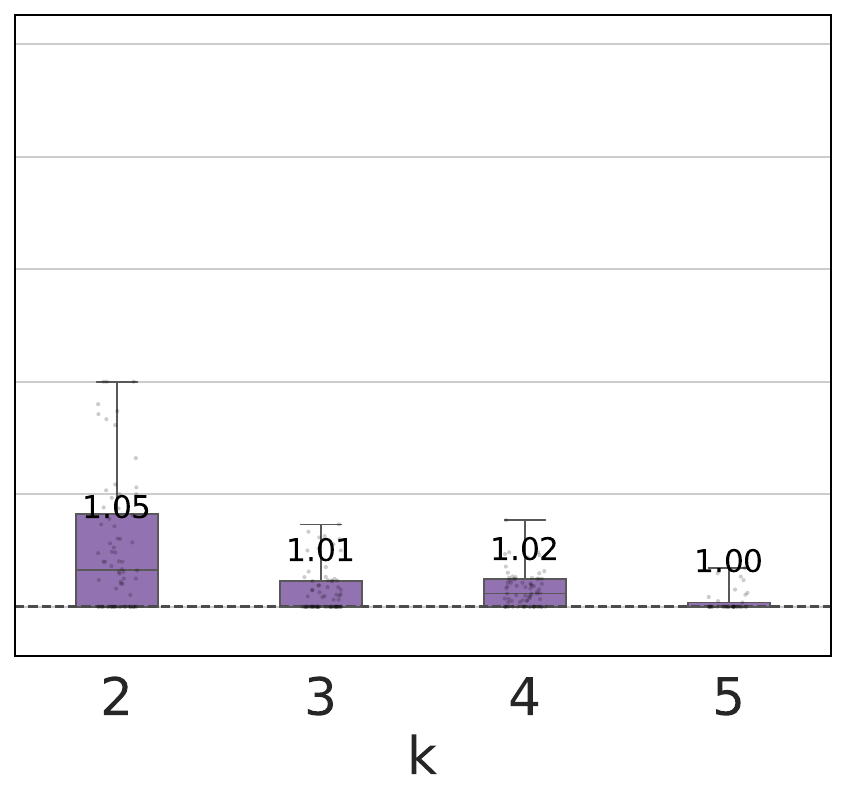}
			\caption{$\ilpunhypergraphfunc$}
		\end{subfigure}\hfill
		\begin{subfigure}[t]{0.30\linewidth}
			\centering
			\safeincludegraphics[height=2.5cm]{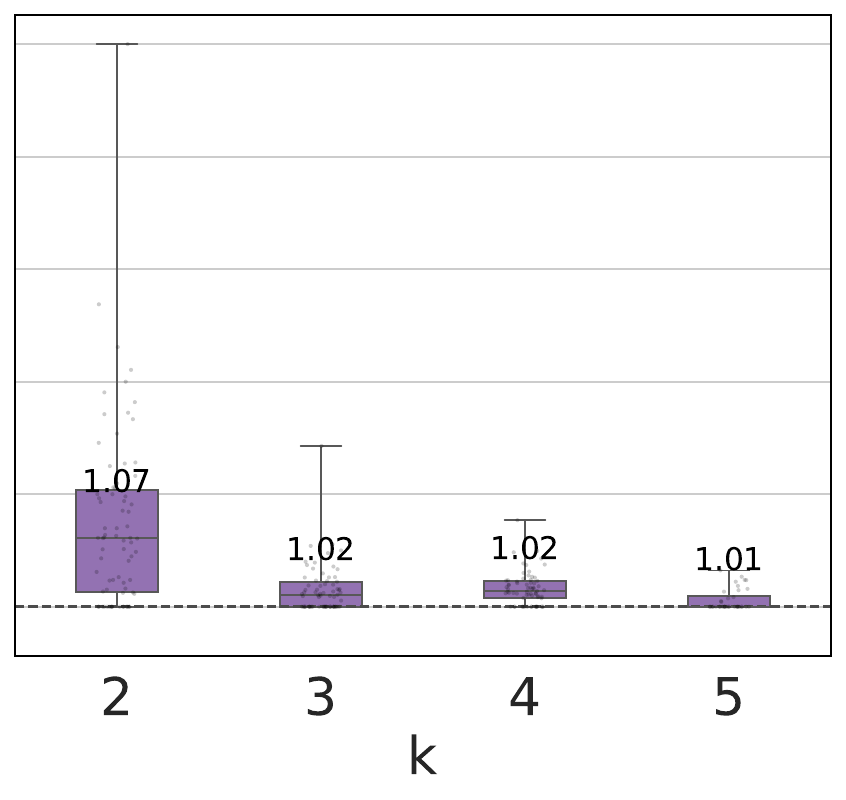}
			\caption{$\ilpunproofdagfunc$}
		\end{subfigure}
		
		\par\vspace{1em}
		\textbf{\csda (\httpd)}
	\end{minipage}
	
	\par\vspace{1em}
	
	%==================== PostgreSQL ====================%
	\begin{minipage}[t]{\linewidth}
		\centering
		
		\begin{subfigure}[t]{0.36\linewidth}
			\centering
			\safeincludegraphics[height=2.57cm]{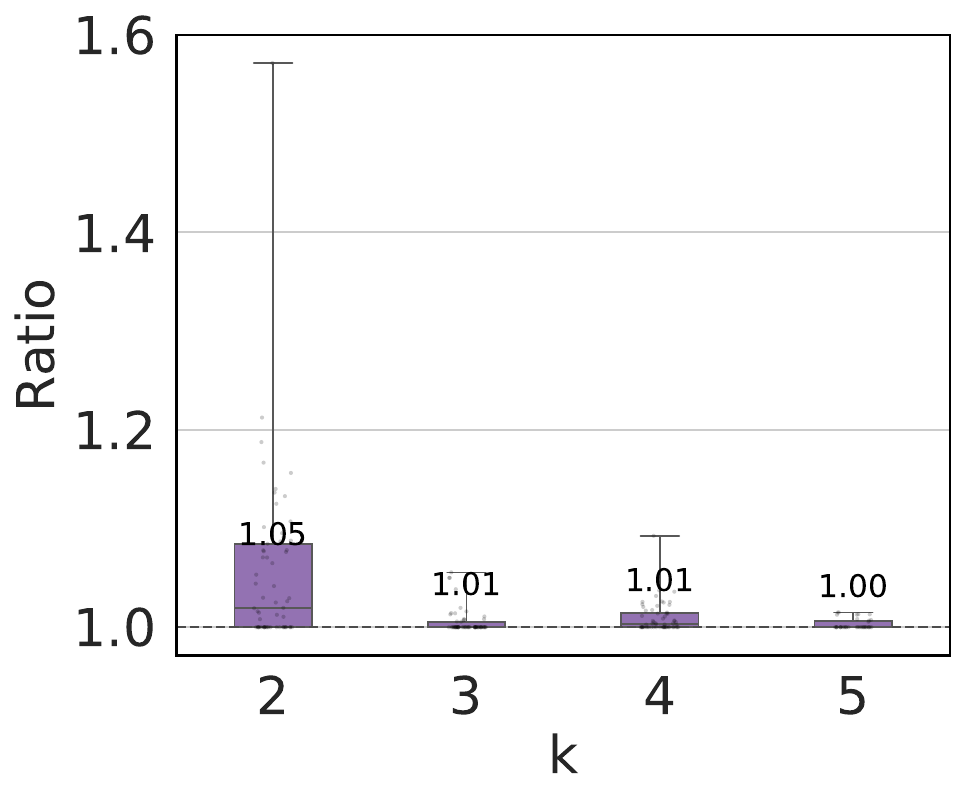}
			\caption{$\ilphypergraphfunc$}
		\end{subfigure}\hfill
		\begin{subfigure}[t]{0.30\linewidth}
			\centering
			\safeincludegraphics[height=2.5cm]{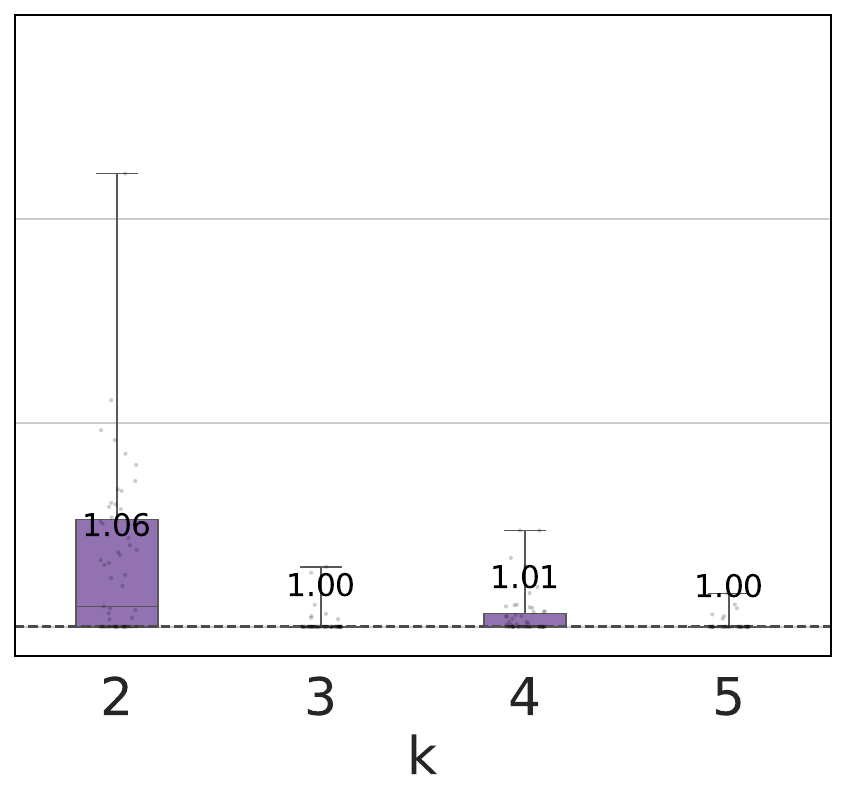}
			\caption{$\ilpunhypergraphfunc$}
		\end{subfigure}\hfill
		\begin{subfigure}[t]{0.30\linewidth}
			\centering
			\safeincludegraphics[height=2.5cm]{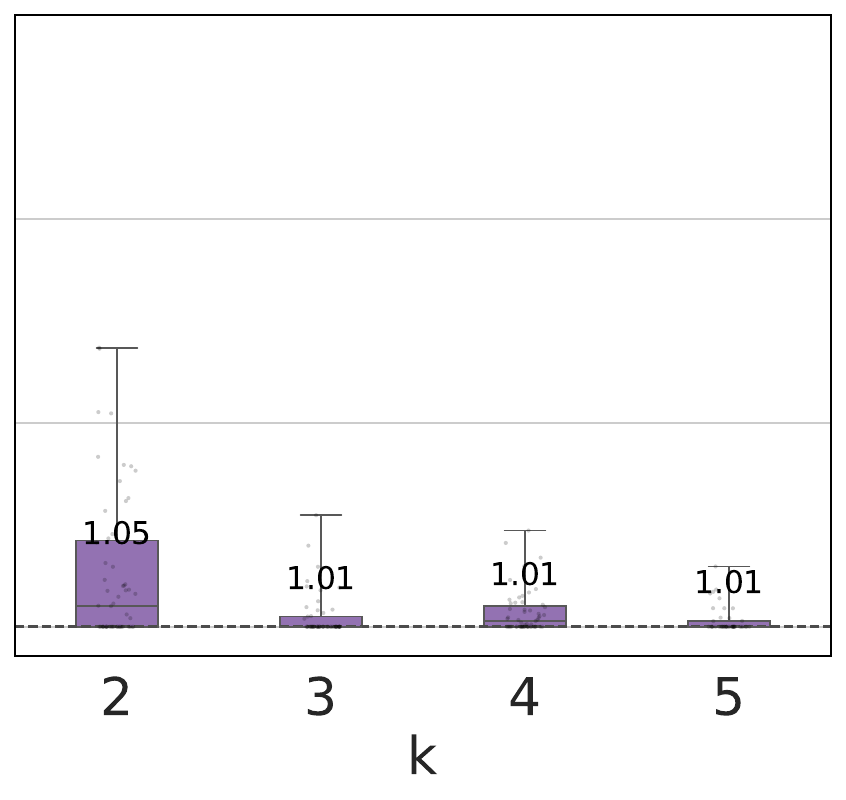}
			\caption{$\ilpunproofdagfunc$}
		\end{subfigure}
		
		\par\vspace{1em}
		\textbf{\csda (\postgres)}
	\end{minipage}
	
	\par\vspace{1em}
	
	%==================== Linux ====================%
	\begin{minipage}[t]{\linewidth}
		\centering
		
		\begin{subfigure}[t]{0.36\linewidth}
			\centering
			\safeincludegraphics[height=2.57cm]{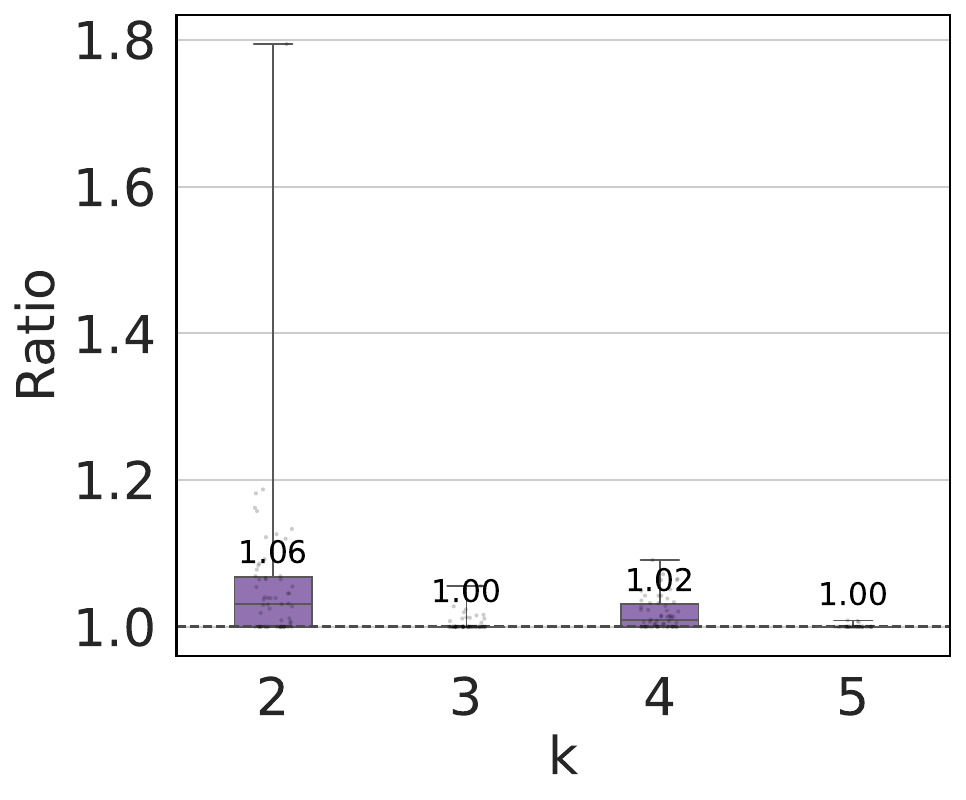}
			\caption{$\ilphypergraphfunc$}
		\end{subfigure}\hfill
		\begin{subfigure}[t]{0.30\linewidth}
			\centering
			\safeincludegraphics[height=2.5cm]{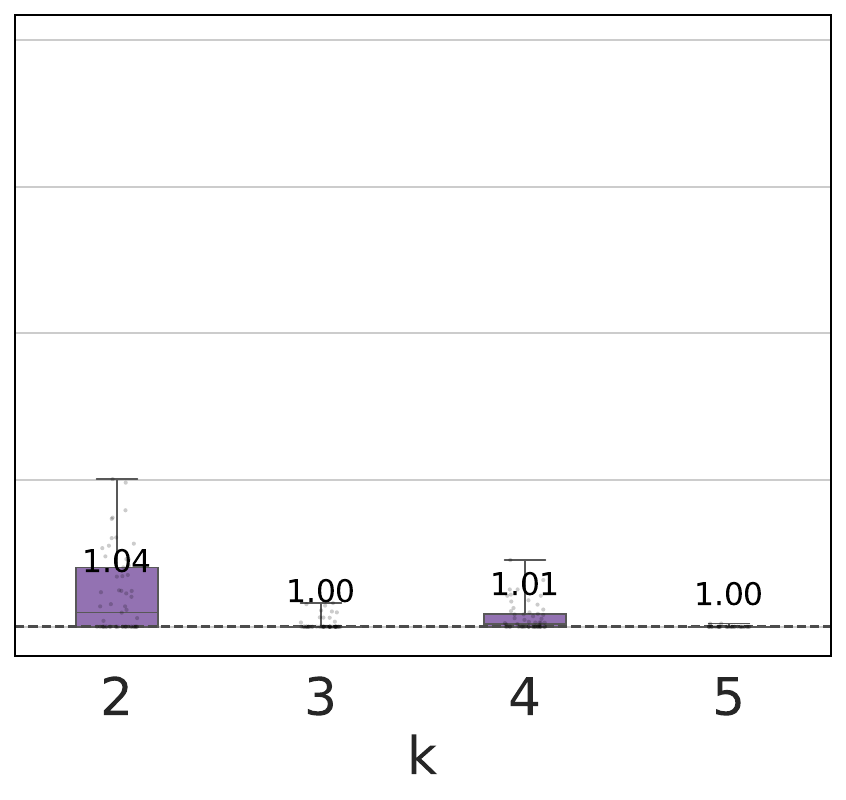}
			\caption{$\ilpunhypergraphfunc$}
		\end{subfigure}\hfill
		\begin{subfigure}[t]{0.30\linewidth}
			\centering
			\safeincludegraphics[height=2.5cm]{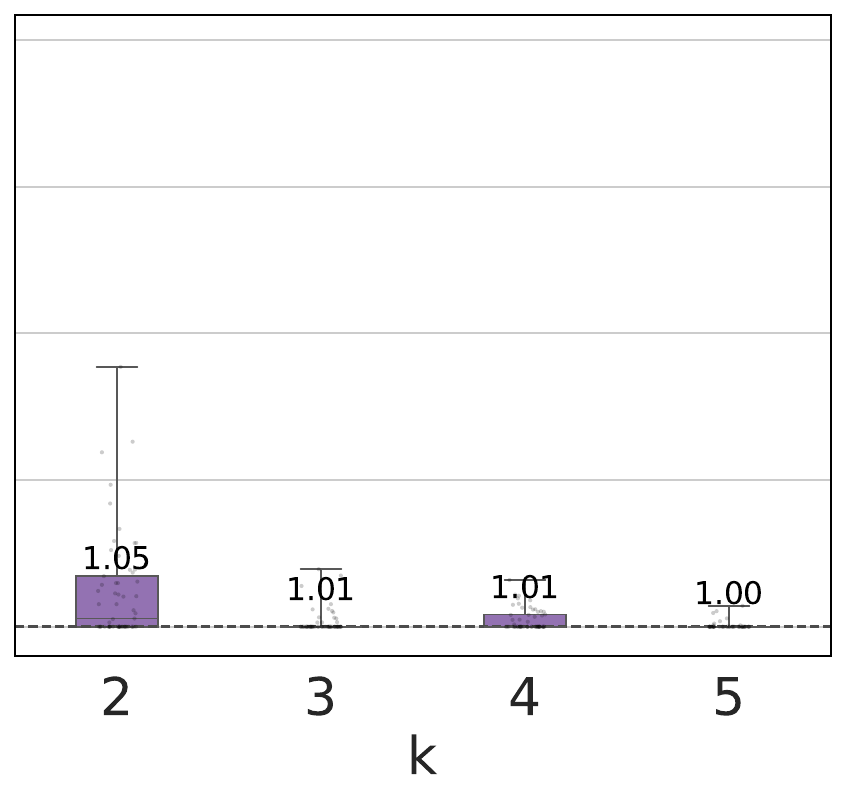}
			\caption{$\ilpunproofdagfunc$}
		\end{subfigure}
		
		\par\vspace{1em}
		\textbf{\csda (\linux)}
	\end{minipage}
	
	\caption{Ratio (optimal / greedy) distribution across $k$ for the \csda\ scenario.}
	\label{fig:csda-rate}
\end{figure}

\begin{figure}[!tp]
	\centering
	\small
	
	%==================== doctors ====================%
	\begin{minipage}[t]{\linewidth}
		\centering
		
		\begin{subfigure}[t]{0.36\linewidth}
			\centering
			\safeincludegraphics[height=2.57cm]{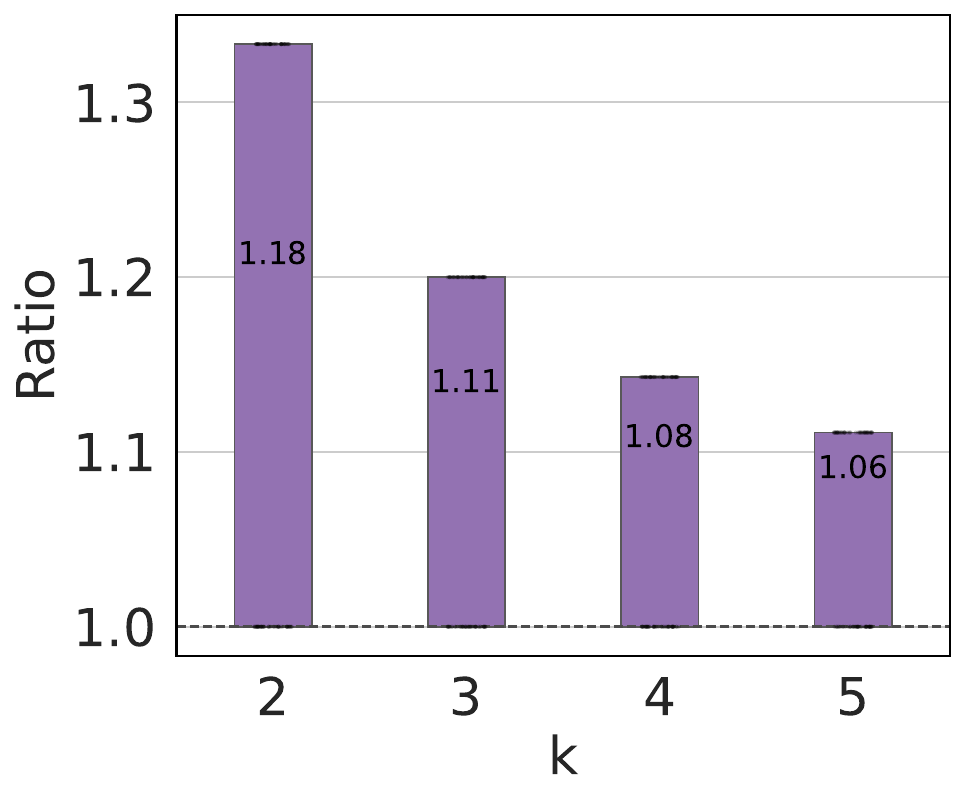}
			\caption{$\ilphypergraphfunc$}
		\end{subfigure}\hfill
		\begin{subfigure}[t]{0.30\linewidth}
			\centering
			\safeincludegraphics[height=2.5cm]{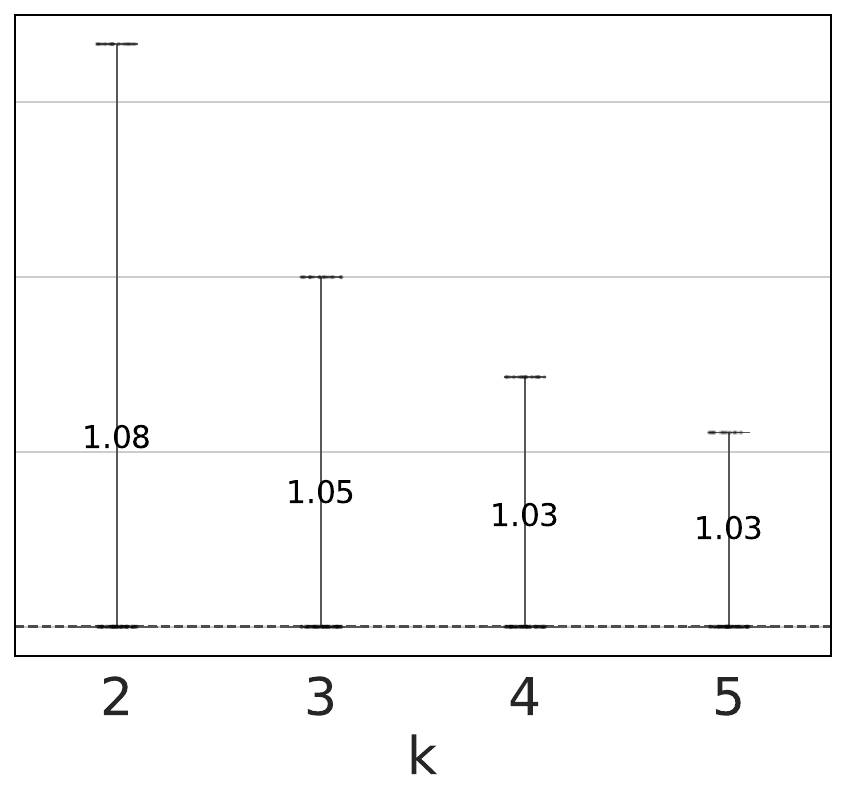}
			\caption{$\ilpunhypergraphfunc$}
		\end{subfigure}\hfill
		\begin{subfigure}[t]{0.30\linewidth}
			\centering
			\safeincludegraphics[height=2.5cm]{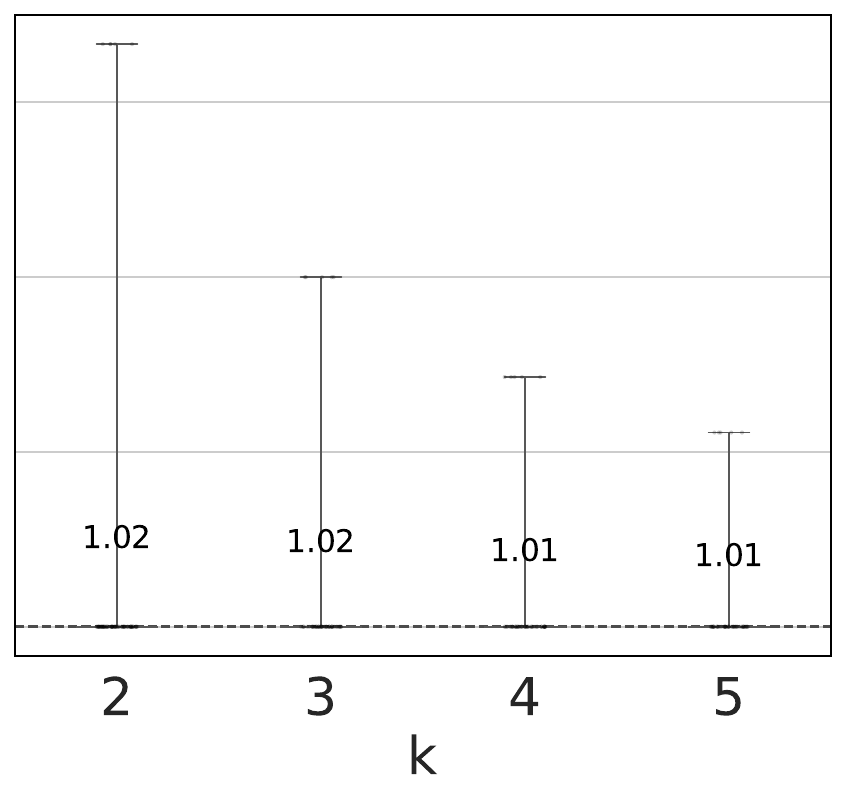}
			\caption{$\ilpunproofdagfunc$}
		\end{subfigure}
		
		\par\vspace{1em}
		\textbf{\doctors\ (100k)}
	\end{minipage}
	
	\caption{Ratio (optimal / greedy) distribution across $k$ for the \doctors\ scenario.}
	\label{fig:doctors-rate}
\end{figure}

\begin{figure}[!h]
	\centering
	\small
	
	%==================== Facebook ====================%
	\begin{minipage}[t]{\linewidth}
		\centering
		
		\begin{subfigure}[t]{0.36\linewidth}
			\centering
			\safeincludegraphics[height=2.57cm]{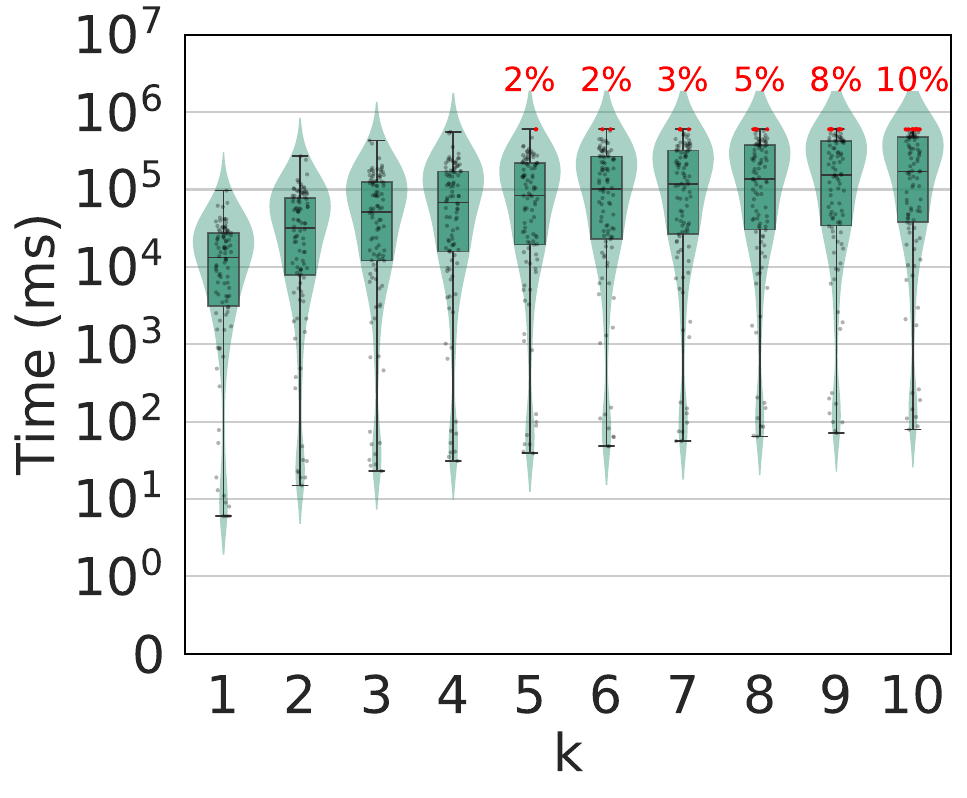}
			\caption{$\ilphypergraphfunc$}
		\end{subfigure}\hfill
		\begin{subfigure}[t]{0.30\linewidth}
			\centering
			\safeincludegraphics[height=2.5cm]{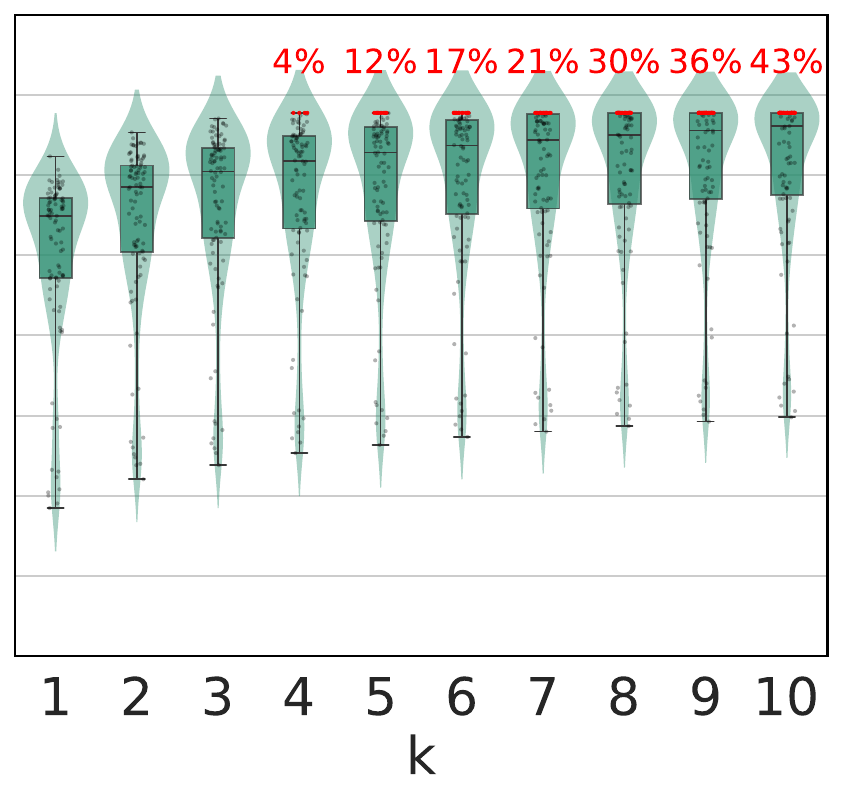}
			\caption{$\ilpunhypergraphfunc$}
		\end{subfigure}\hfill
		\begin{subfigure}[t]{0.30\linewidth}
			\centering
			\safeincludegraphics[height=2.5cm]{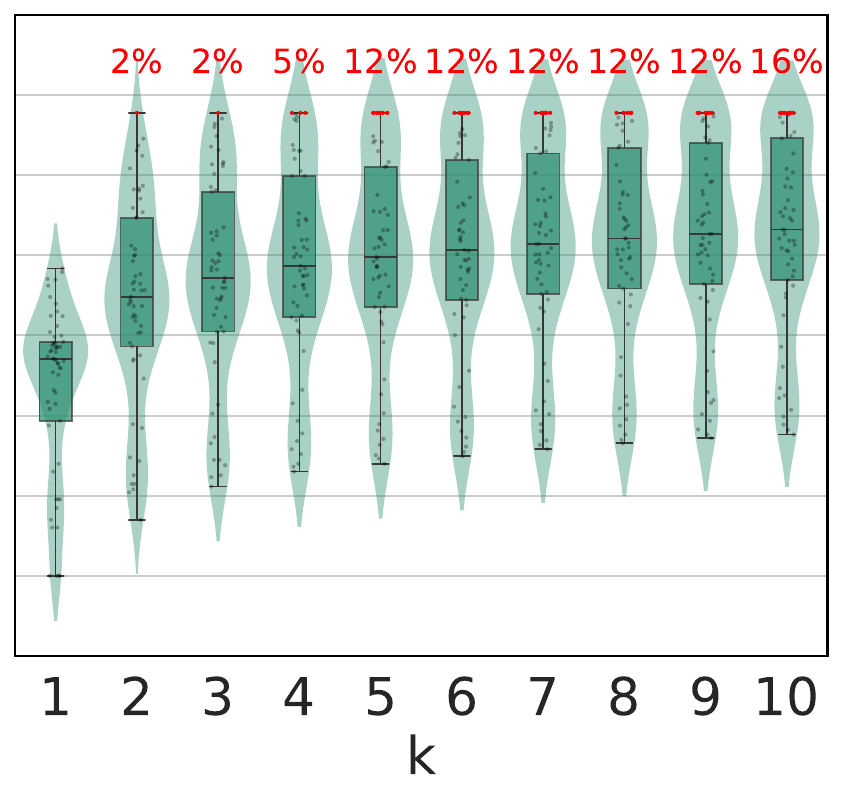}
			\caption{$\ilpunproofdagfunc$}
		\end{subfigure}
		
		\par\vspace{1em}
		\textbf{\graphanalytics (\facebook)}
	\end{minipage}
	
	\par\vspace{1em}
	
	%==================== Bitcoin ====================%
	\begin{minipage}[t]{\linewidth}
		\centering
		
		\begin{subfigure}[t]{0.36\linewidth}
			\centering
			\safeincludegraphics[height=2.57cm]{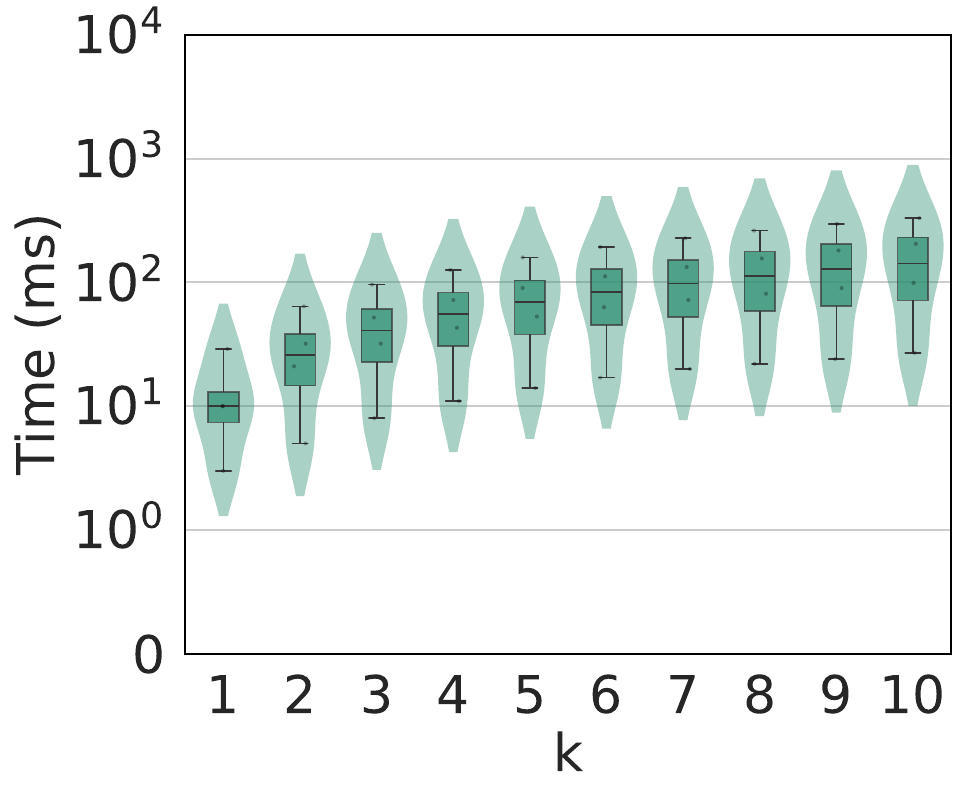}
			\caption{$\ilphypergraphfunc$}
		\end{subfigure}\hfill
		\begin{subfigure}[t]{0.30\linewidth}
			\centering
			\safeincludegraphics[height=2.5cm]{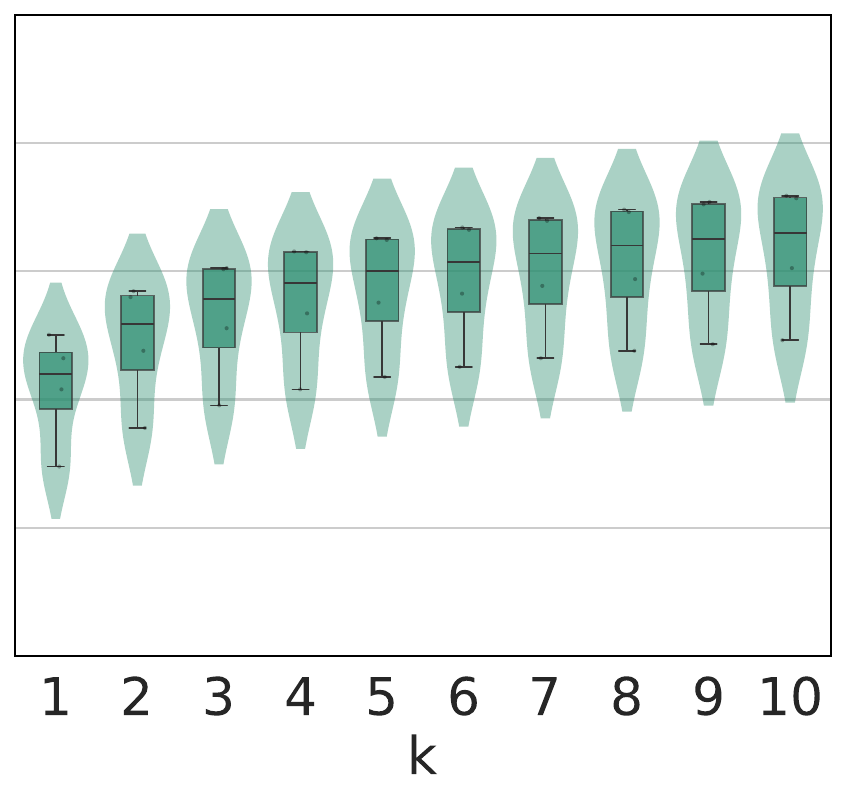}
			\caption{$\ilpunhypergraphfunc$}
		\end{subfigure}\hfill
		\begin{subfigure}[t]{0.30\linewidth}
			\centering
			\safeincludegraphics[height=2.5cm]{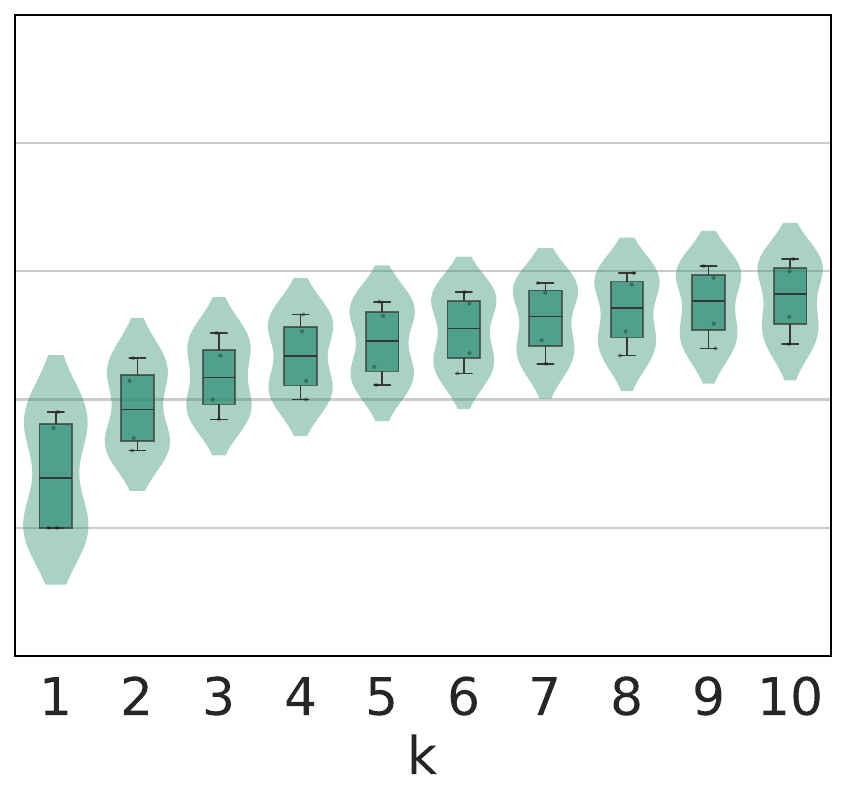}
			\caption{$\ilpunproofdagfunc$}
		\end{subfigure}
		
		\par\vspace{1em}
		\textbf{\graphanalytics (\bitcoin)}
	\end{minipage}
	
	\caption{Runtimes for the \plotlabel{\graphanalytics} scenario using the approximation algorithms ($k$ up to 10).}
	\label{fig:k10-graph-analytics-approximate}
\end{figure}

% Part 1: doctors 1-4
\begin{figure}[!h]
	\centering
	\small
	
	%==================== doctors 1 ====================%
	\begin{minipage}[t]{\linewidth}
		\centering
		
		\begin{subfigure}[t]{0.36\linewidth}
			\centering
			\safeincludegraphics[height=2.57cm]{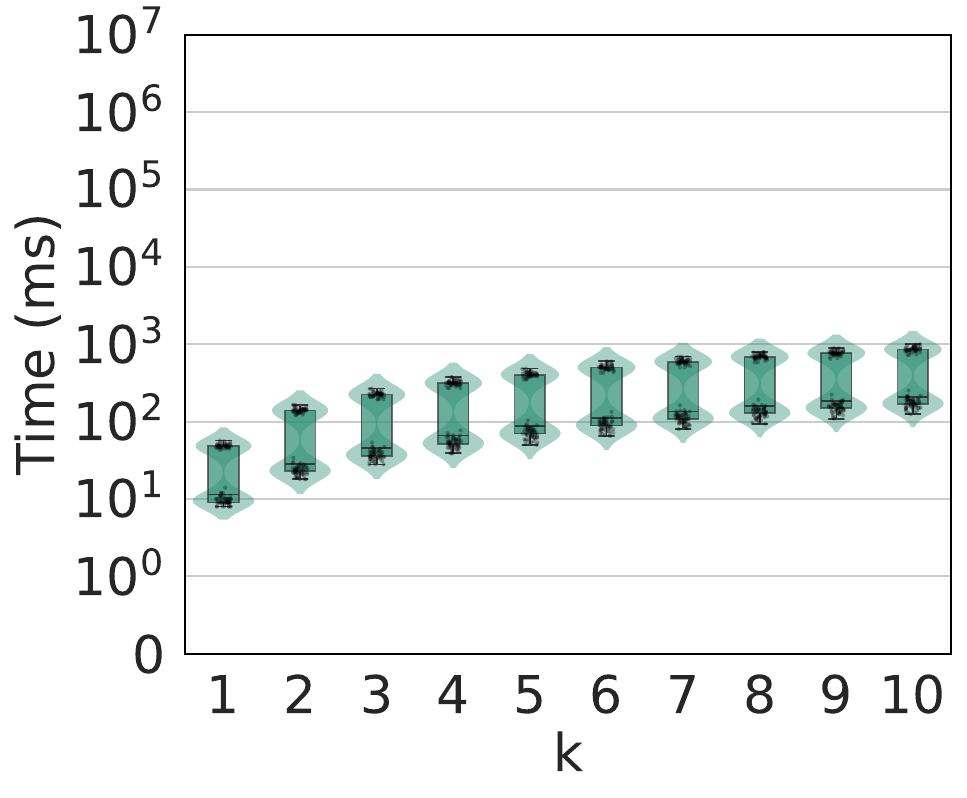}
			\caption{$\ilphypergraphfunc$}
		\end{subfigure}\hfill
		\begin{subfigure}[t]{0.30\linewidth}
			\centering
			\safeincludegraphics[height=2.5cm]{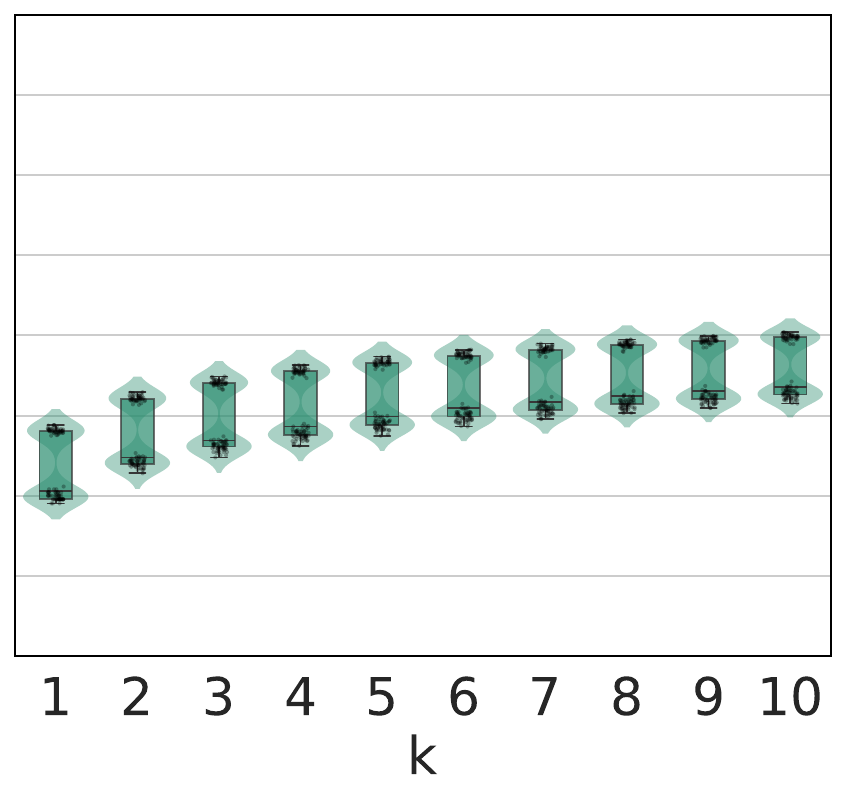}
			\caption{$\ilpunhypergraphfunc$}
		\end{subfigure}\hfill
		\begin{subfigure}[t]{0.30\linewidth}
			\centering
			\safeincludegraphics[height=2.5cm]{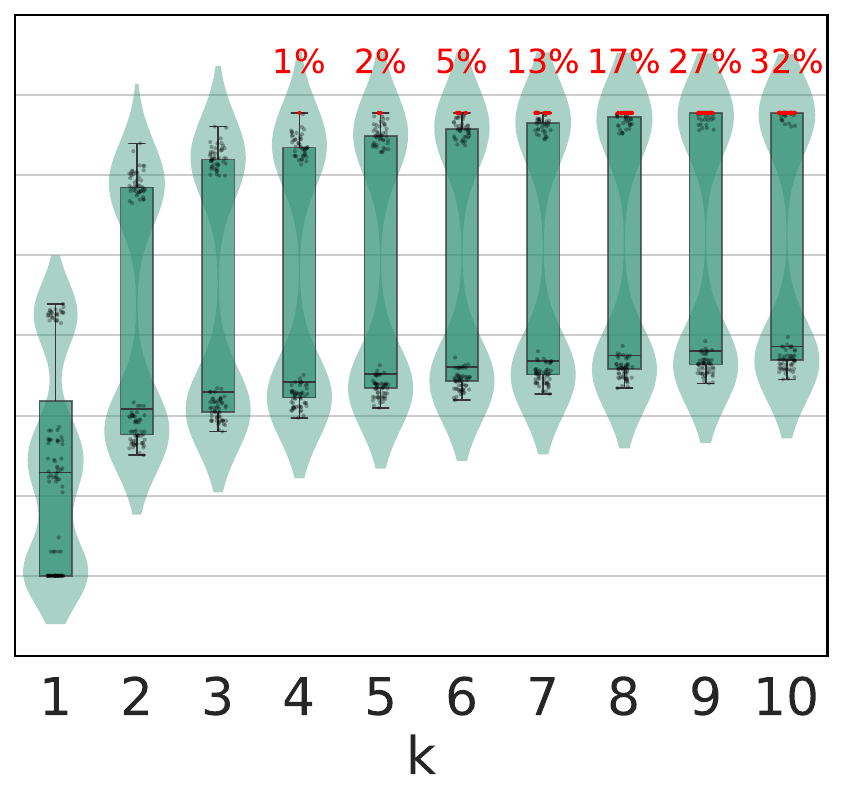}
			\caption{$\ilpunproofdagfunc$}
		\end{subfigure}
		
		\par\vspace{1em}
		\textbf{\doctors\ 1 (100k)}
	\end{minipage}
	
	\par\vspace{1em}
	
	%==================== doctors 2 ====================%
	\begin{minipage}[t]{\linewidth}
		\centering
		
		\begin{subfigure}[t]{0.36\linewidth}
			\centering
			\safeincludegraphics[height=2.57cm]{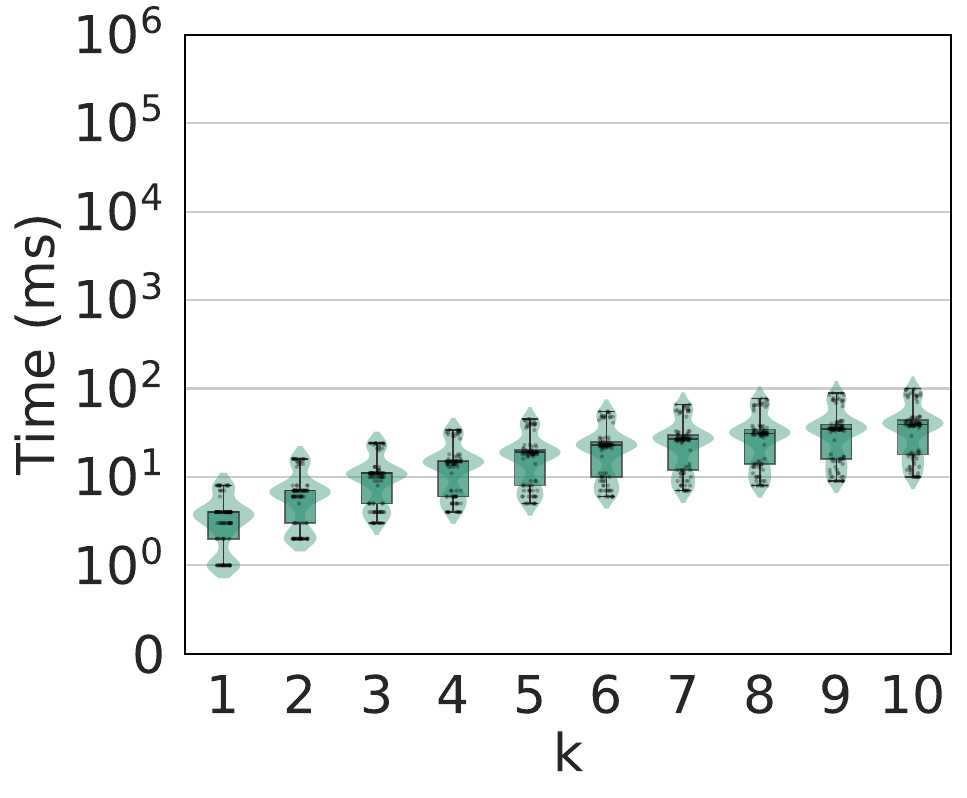}
			\caption{$\ilphypergraphfunc$}
		\end{subfigure}\hfill
		\begin{subfigure}[t]{0.30\linewidth}
			\centering
			\safeincludegraphics[height=2.5cm]{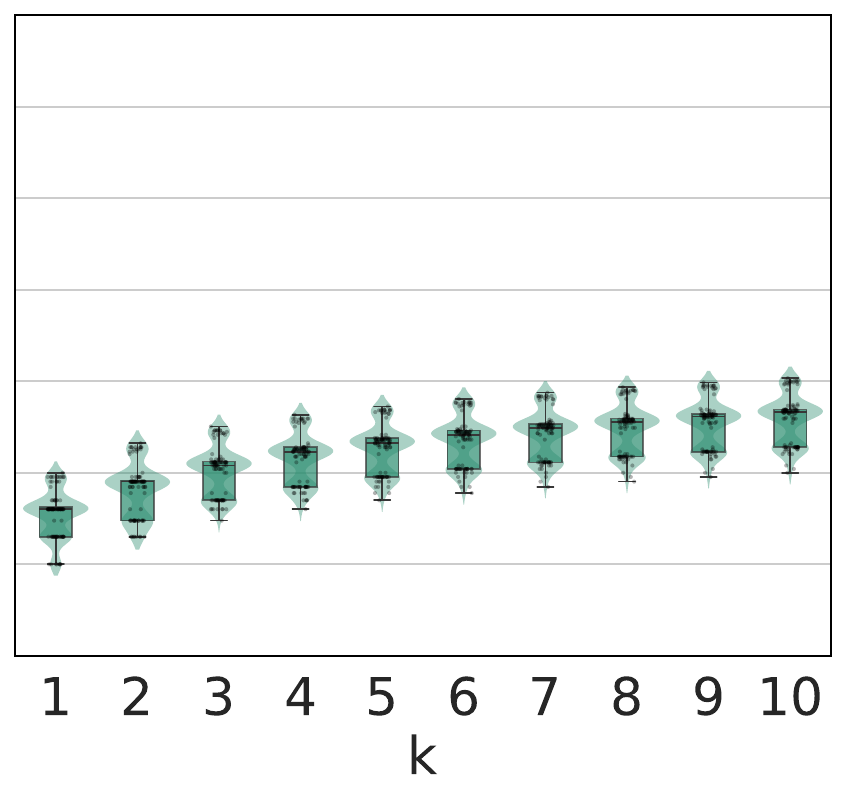}
			\caption{$\ilpunhypergraphfunc$}
		\end{subfigure}\hfill
		\begin{subfigure}[t]{0.30\linewidth}
			\centering
			\safeincludegraphics[height=2.5cm]{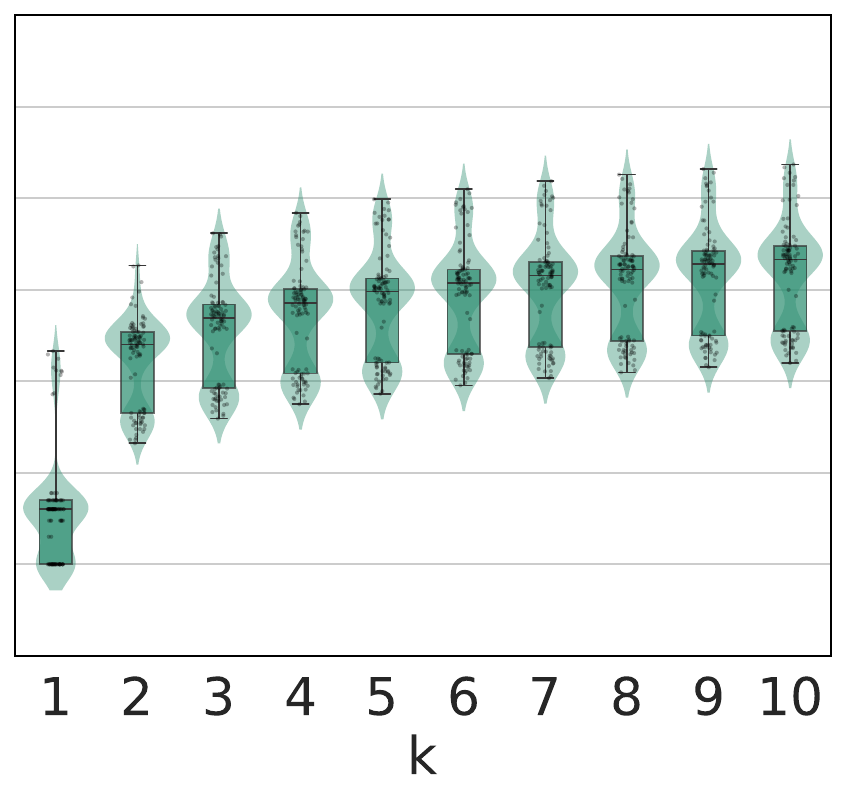}
			\caption{$\ilpunproofdagfunc$}
		\end{subfigure}
		
		\par\vspace{1em}
		\textbf{\doctors\ 2 (100k)}
	\end{minipage}
	
	\par\vspace{1em}
	
	%==================== doctors 3 ====================%
	\begin{minipage}[t]{\linewidth}
		\centering
		
		\begin{subfigure}[t]{0.36\linewidth}
			\centering
			\safeincludegraphics[height=2.57cm]{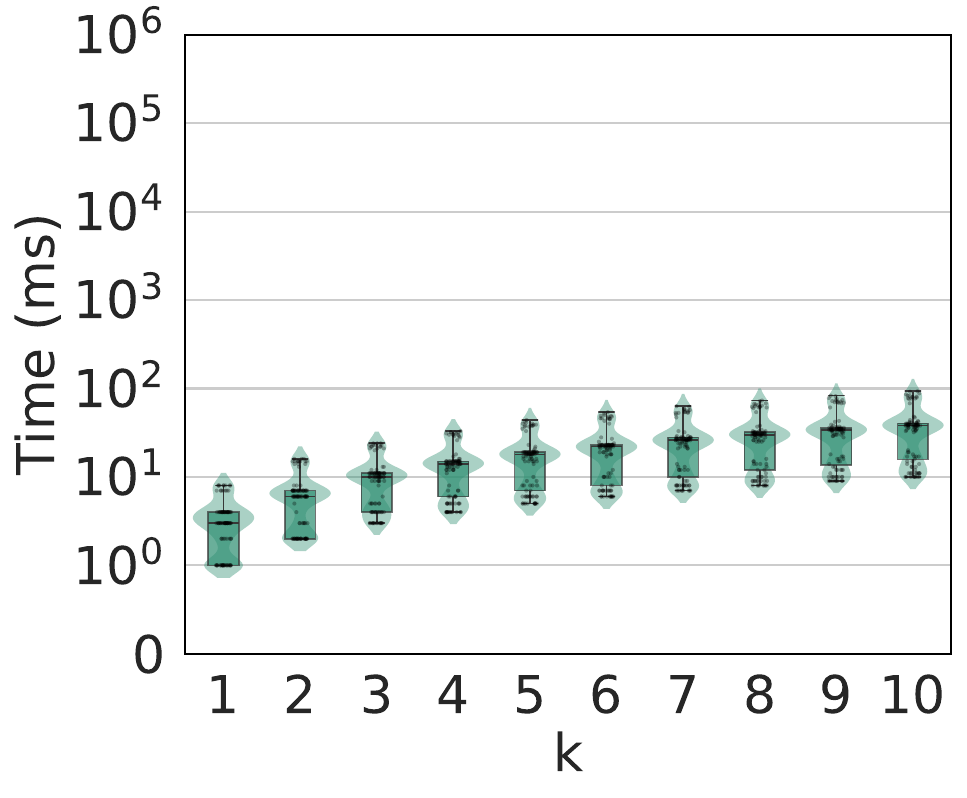}
			\caption{$\ilphypergraphfunc$}
		\end{subfigure}\hfill
		\begin{subfigure}[t]{0.30\linewidth}
			\centering
			\safeincludegraphics[height=2.5cm]{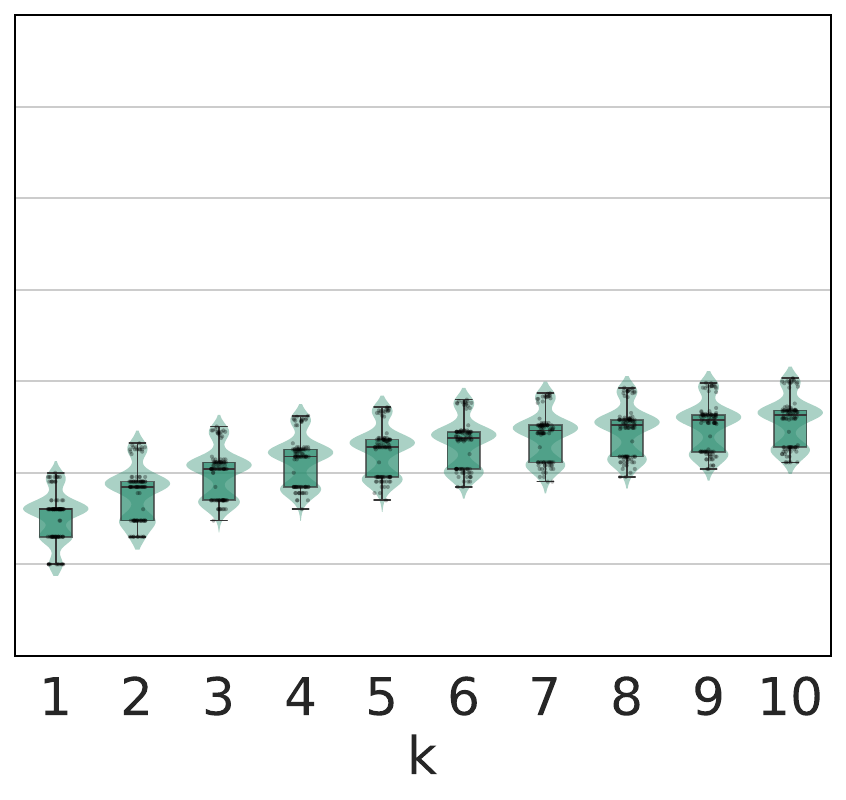}
			\caption{$\ilpunhypergraphfunc$}
		\end{subfigure}\hfill
		\begin{subfigure}[t]{0.30\linewidth}
			\centering
			\safeincludegraphics[height=2.5cm]{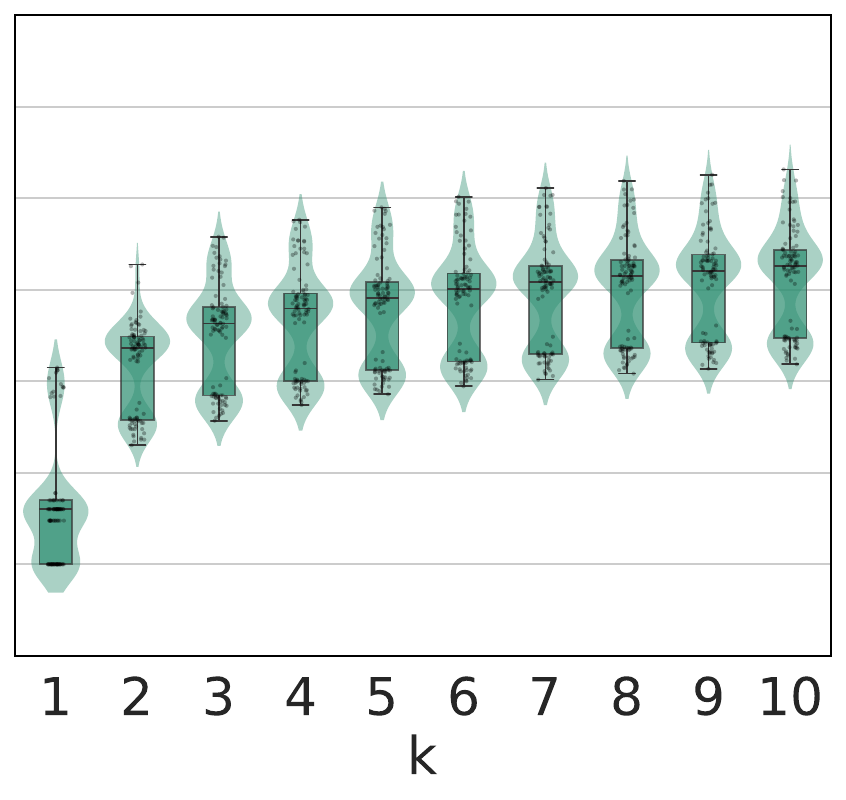}
			\caption{$\ilpunproofdagfunc$}
		\end{subfigure}
		
		\par\vspace{1em}
		\textbf{\doctors\ 3 (100k)}
	\end{minipage}
	
	\par\vspace{1em}
	
	%==================== doctors 4 ====================%
	\begin{minipage}[t]{\linewidth}
		\centering
		
		\begin{subfigure}[t]{0.36\linewidth}
			\centering
			\safeincludegraphics[height=2.57cm]{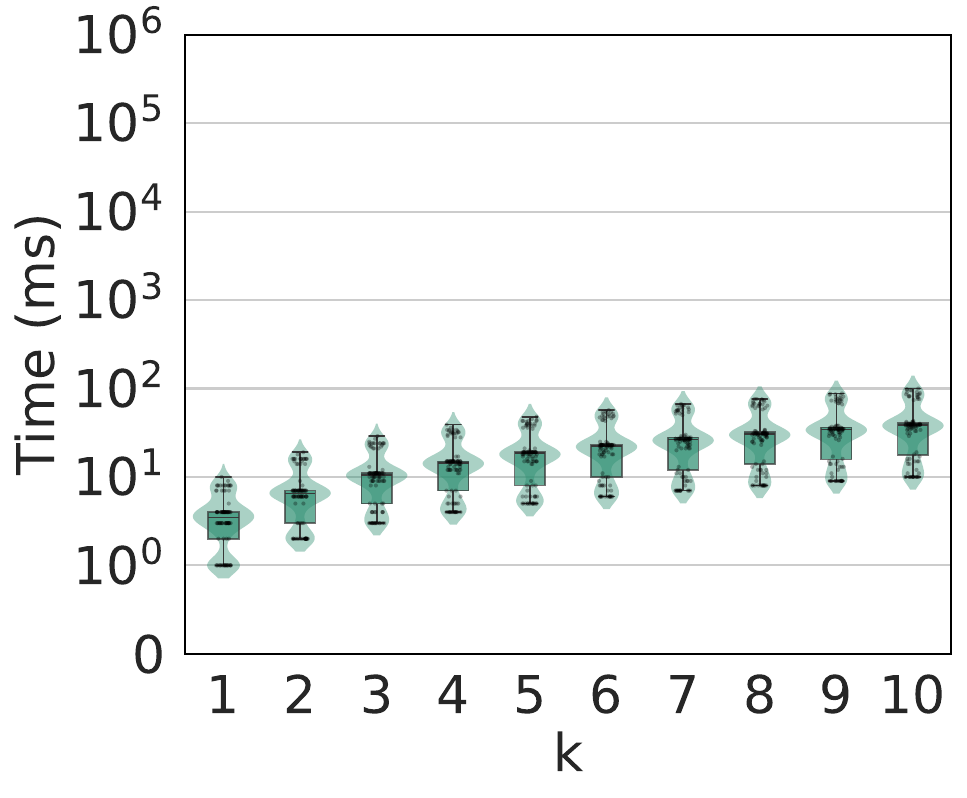}
			\caption{$\ilphypergraphfunc$}
		\end{subfigure}\hfill
		\begin{subfigure}[t]{0.30\linewidth}
			\centering
			\safeincludegraphics[height=2.5cm]{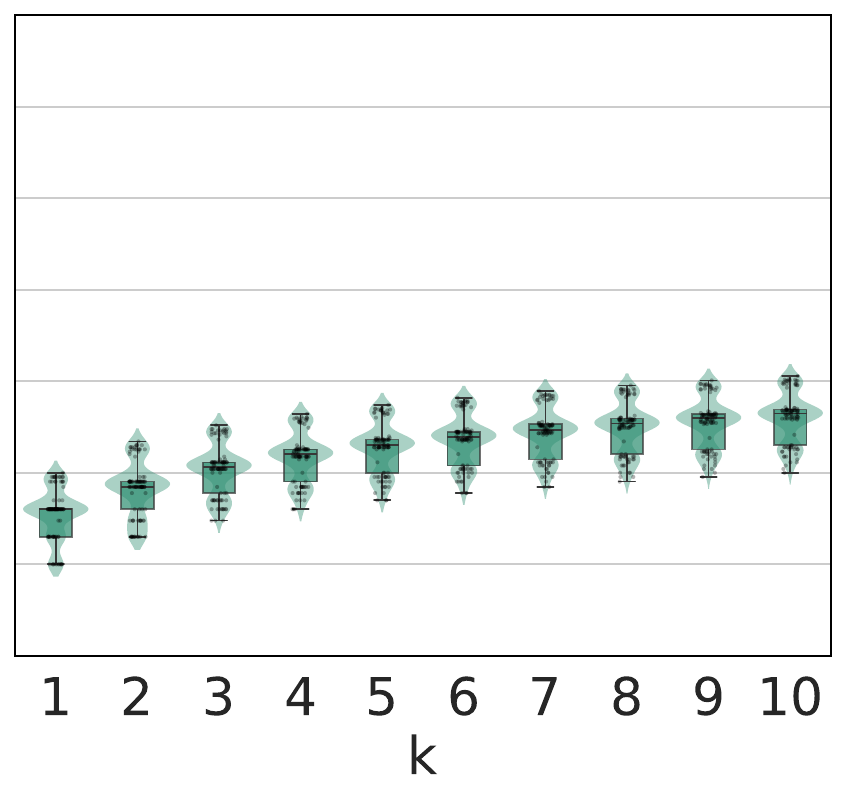}
			\caption{$\ilpunhypergraphfunc$}
		\end{subfigure}\hfill
		\begin{subfigure}[t]{0.30\linewidth}
			\centering
			\safeincludegraphics[height=2.5cm]{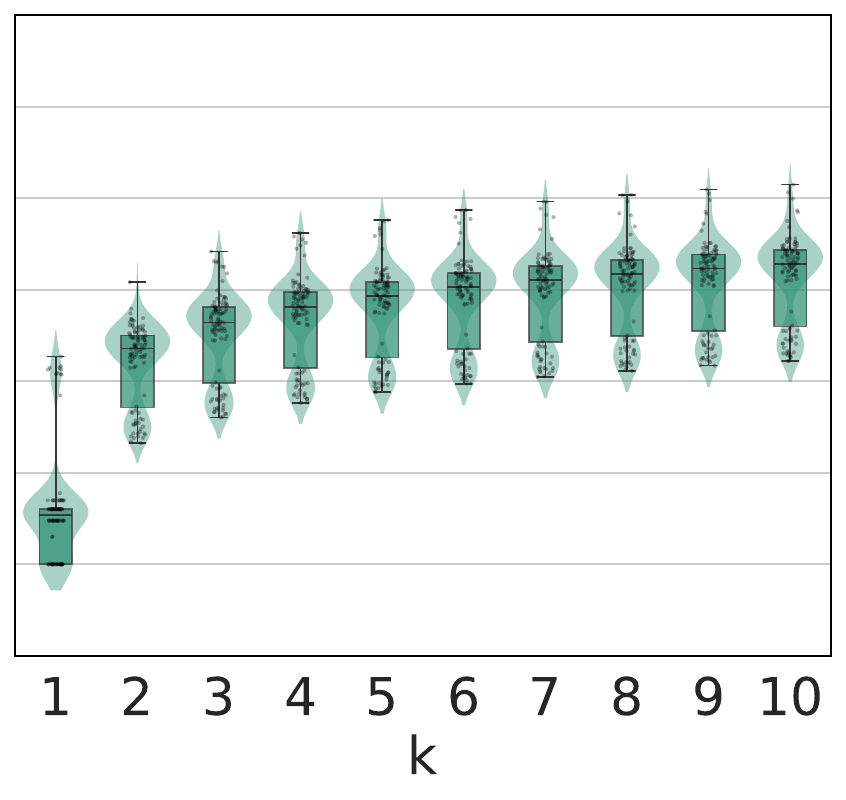}
			\caption{$\ilpunproofdagfunc$}
		\end{subfigure}
		
		\par\vspace{1em}
		\textbf{\doctors\ 4 (100k)}
	\end{minipage}
	
	\caption{Runtimes for the \doctors\ (1--7) scenarios using the approximation algorithms ($k$ up to 10).}
	\label{fig:k10-doctors-approximate}
\end{figure}

% Part 2: doctors 5-7 (same figure number; subfigure letters continue)
\begin{figure}[!h]
	\ContinuedFloat
	\centering
	\small
	
	%==================== doctors 5 ====================%
	\begin{minipage}[t]{\linewidth}
		\centering
		
		\begin{subfigure}[t]{0.36\linewidth}
			\centering
			\safeincludegraphics[height=2.57cm]{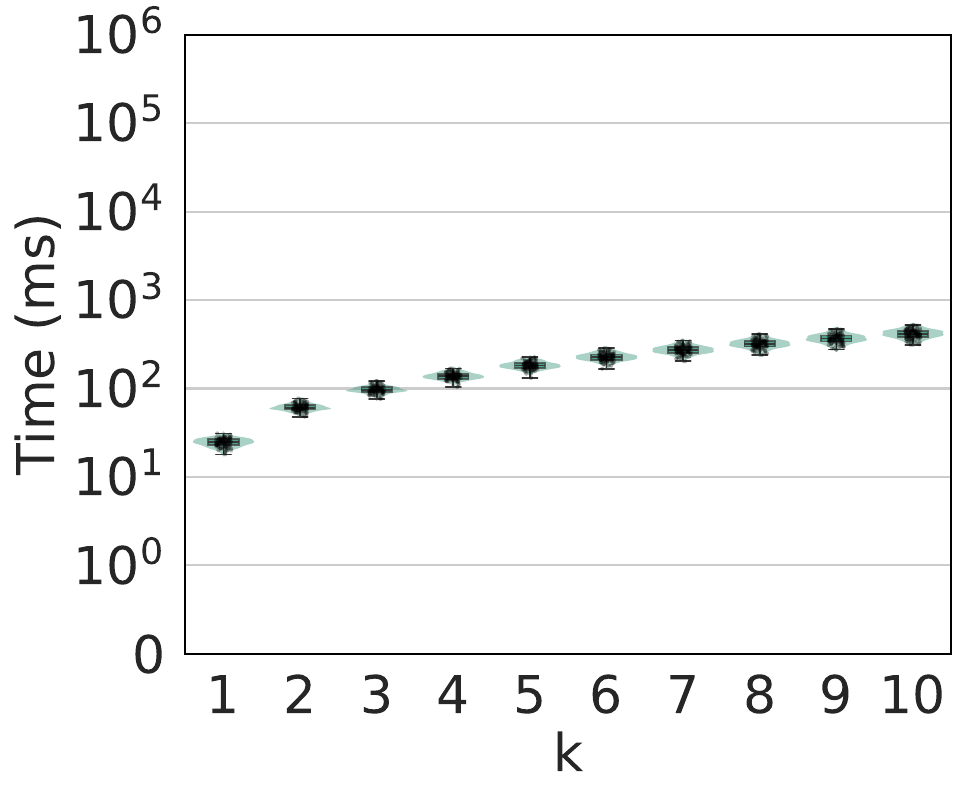}
			\caption{$\ilphypergraphfunc$}
		\end{subfigure}\hfill
		\begin{subfigure}[t]{0.30\linewidth}
			\centering
			\safeincludegraphics[height=2.5cm]{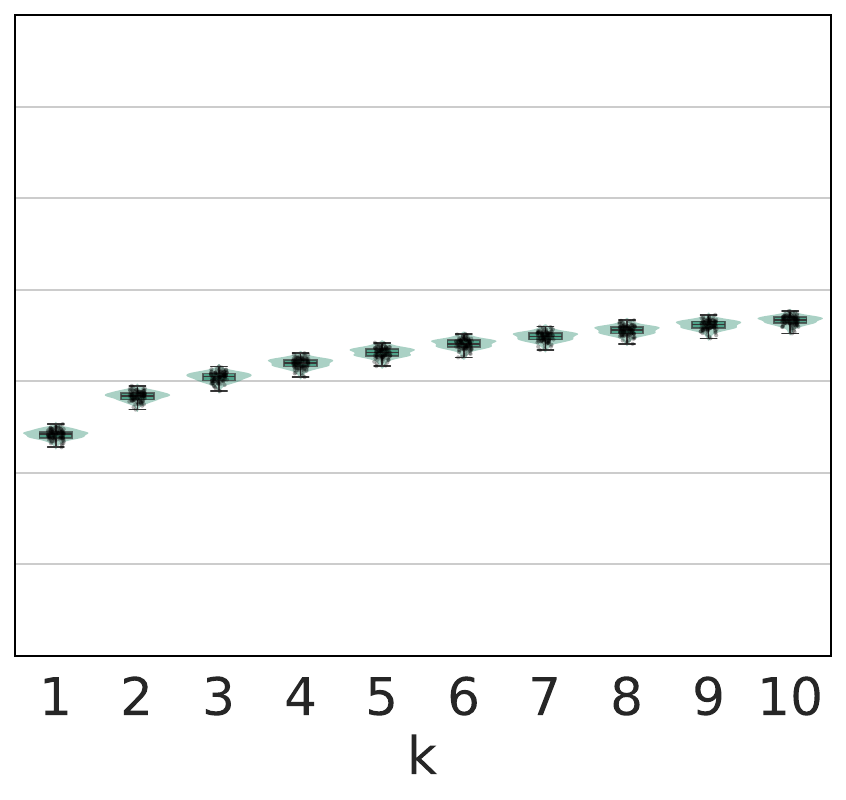}
			\caption{$\ilpunhypergraphfunc$}
		\end{subfigure}\hfill
		\begin{subfigure}[t]{0.30\linewidth}
			\centering
			\safeincludegraphics[height=2.5cm]{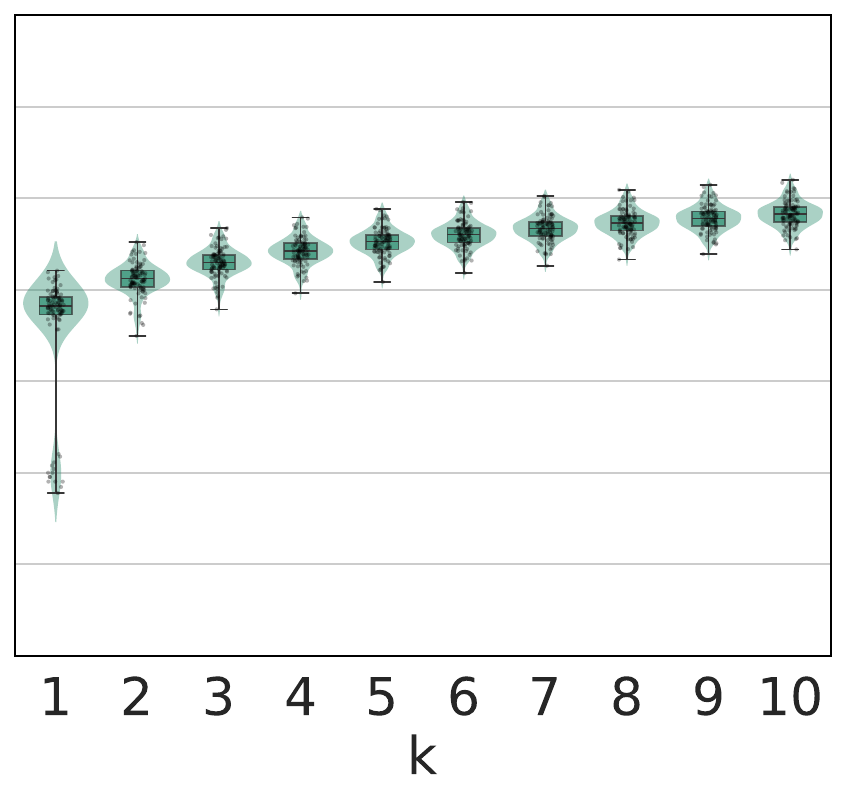}
			\caption{$\ilpunproofdagfunc$}
		\end{subfigure}
		
		\par\vspace{1em}
		\textbf{\doctors\ 5 (100k)}
	\end{minipage}
	
	\par\vspace{1em}
	
	%==================== doctors 6 ====================%
	\begin{minipage}[t]{\linewidth}
		\centering
		
		\begin{subfigure}[t]{0.36\linewidth}
			\centering
			\safeincludegraphics[height=2.57cm]{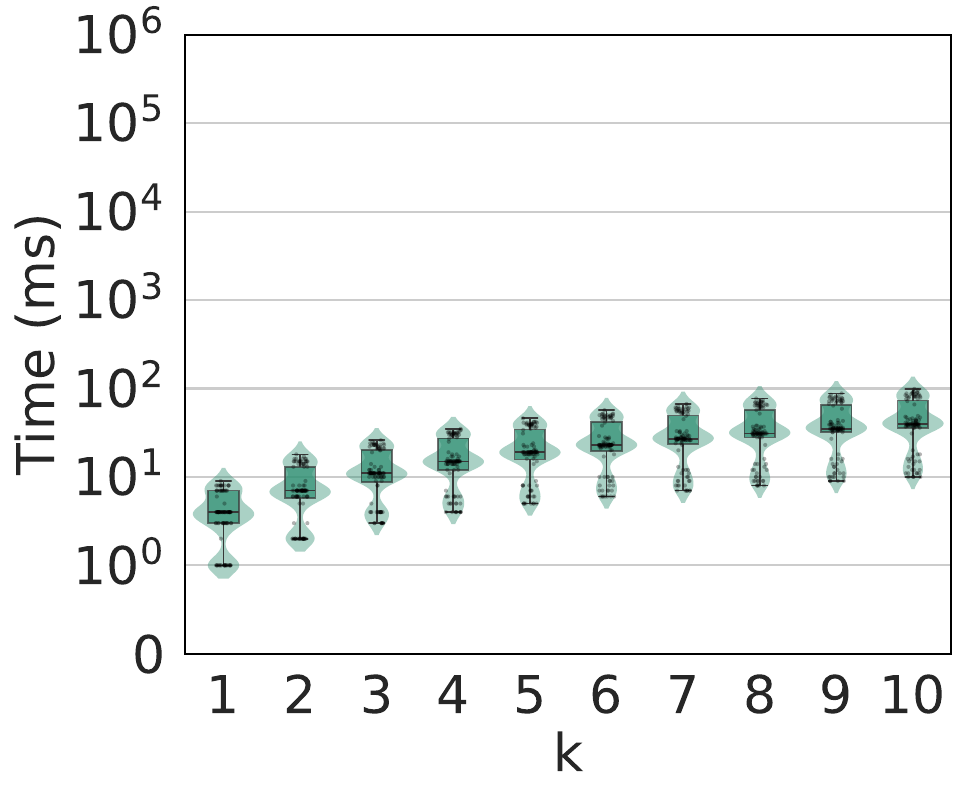}
			\caption{$\ilphypergraphfunc$}
		\end{subfigure}\hfill
		\begin{subfigure}[t]{0.30\linewidth}
			\centering
			\safeincludegraphics[height=2.5cm]{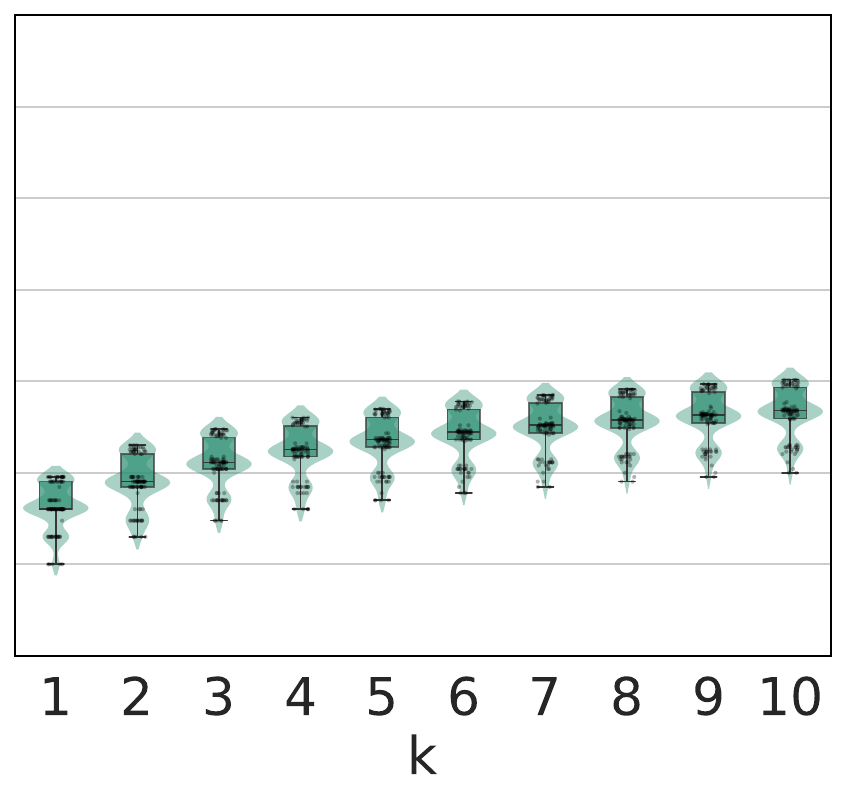}
			\caption{$\ilpunhypergraphfunc$}
		\end{subfigure}\hfill
		\begin{subfigure}[t]{0.30\linewidth}
			\centering
			\safeincludegraphics[height=2.5cm]{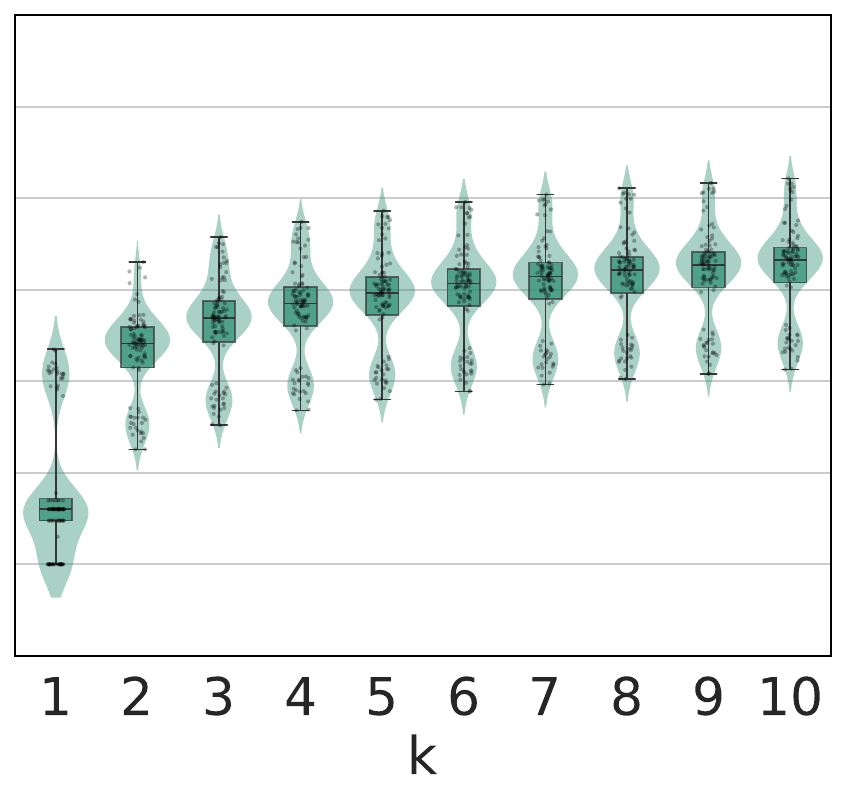}
			\caption{$\ilpunproofdagfunc$}
		\end{subfigure}
		
		\par\vspace{1em}
		\textbf{\doctors\ 6 (100k)}
	\end{minipage}
	
	\par\vspace{1em}
	
	%==================== doctors 7 ====================%
	\begin{minipage}[t]{\linewidth}
		\centering
		
		\begin{subfigure}[t]{0.36\linewidth}
			\centering
			\safeincludegraphics[height=2.57cm]{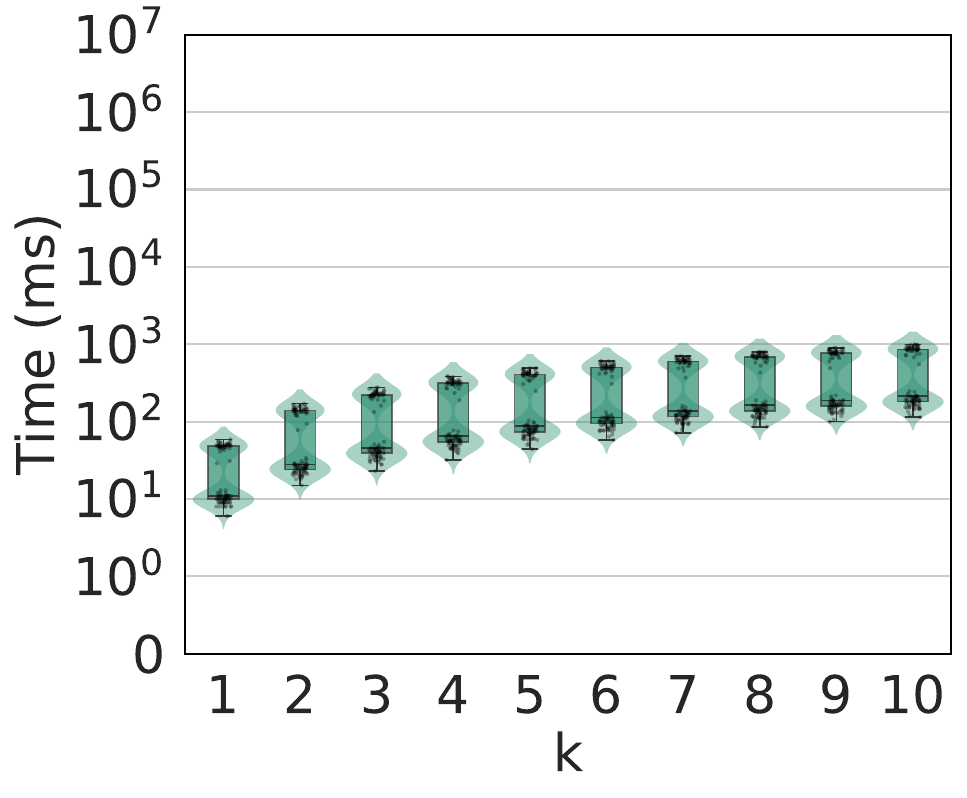}
			\caption{$\ilphypergraphfunc$}
		\end{subfigure}\hfill
		\begin{subfigure}[t]{0.30\linewidth}
			\centering
			\safeincludegraphics[height=2.5cm]{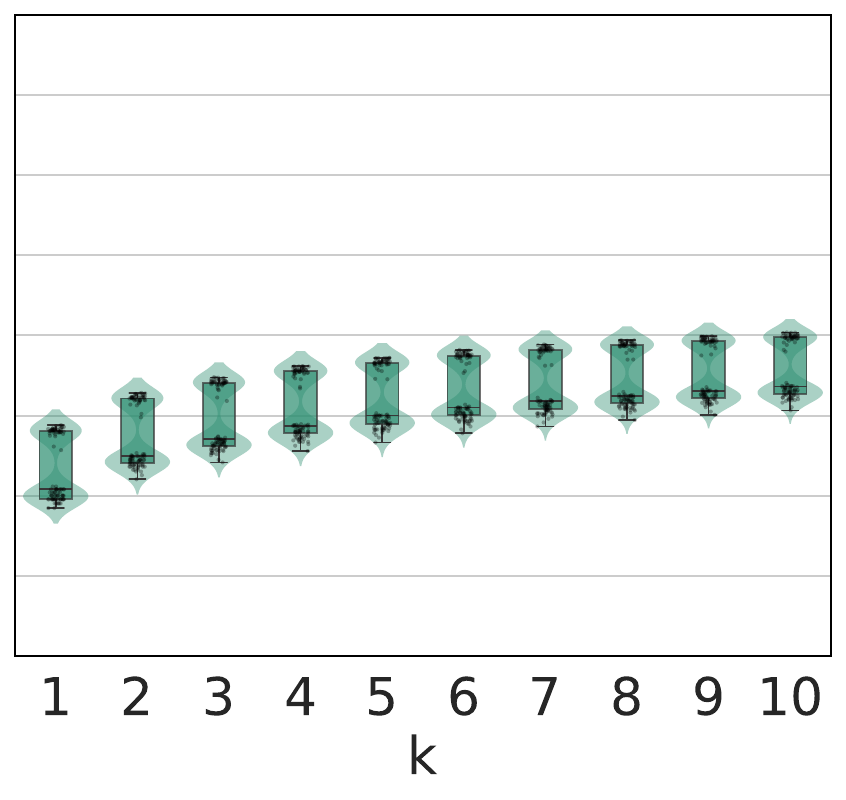}
			\caption{$\ilpunhypergraphfunc$}
		\end{subfigure}\hfill
		\begin{subfigure}[t]{0.30\linewidth}
			\centering
			\safeincludegraphics[height=2.5cm]{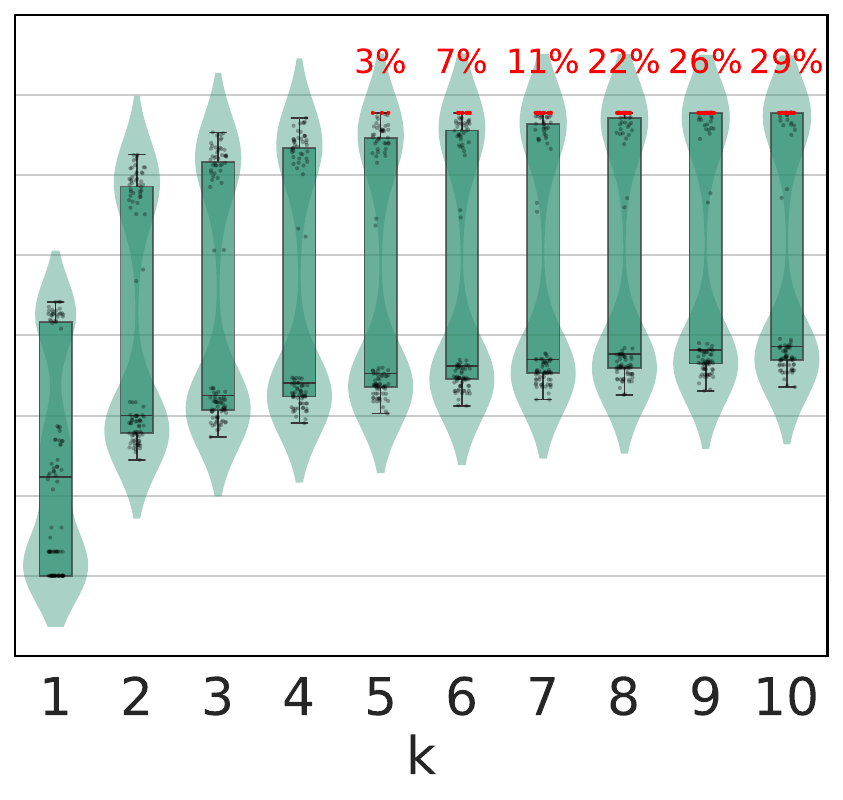}
			\caption{$\ilpunproofdagfunc$}
		\end{subfigure}
		
		\par\vspace{1em}
		\textbf{\doctors\ 7 (100k)}
	\end{minipage}
	
	\caption[]{Runtimes for the \doctors\ (1--7) scenarios using the approximation algorithms ($k$ up to 10) (continued).}
\end{figure}

\begin{figure}[!h]
	\centering
	\small
	
	%==================== HTTPD ====================%
	\begin{minipage}[t]{\linewidth}
		\centering
		
		\begin{subfigure}[t]{0.36\linewidth}
			\centering
			\safeincludegraphics[height=2.57cm]{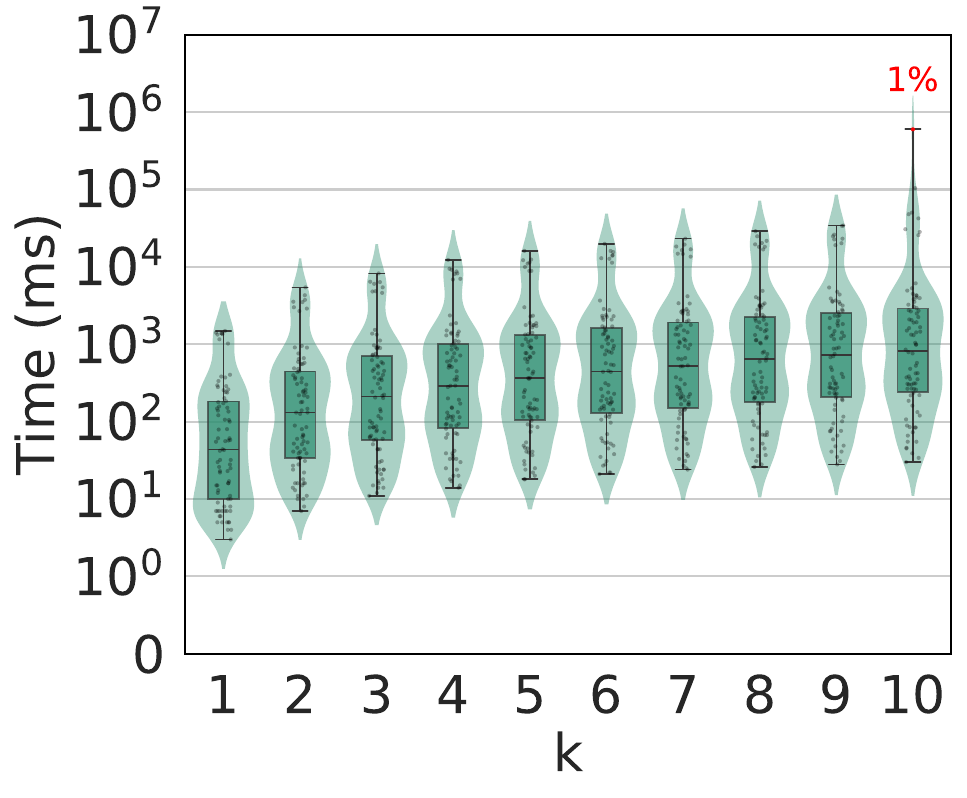}
			\caption{$\ilphypergraphfunc$}
		\end{subfigure}\hfill
		\begin{subfigure}[t]{0.30\linewidth}
			\centering
			\safeincludegraphics[height=2.5cm]{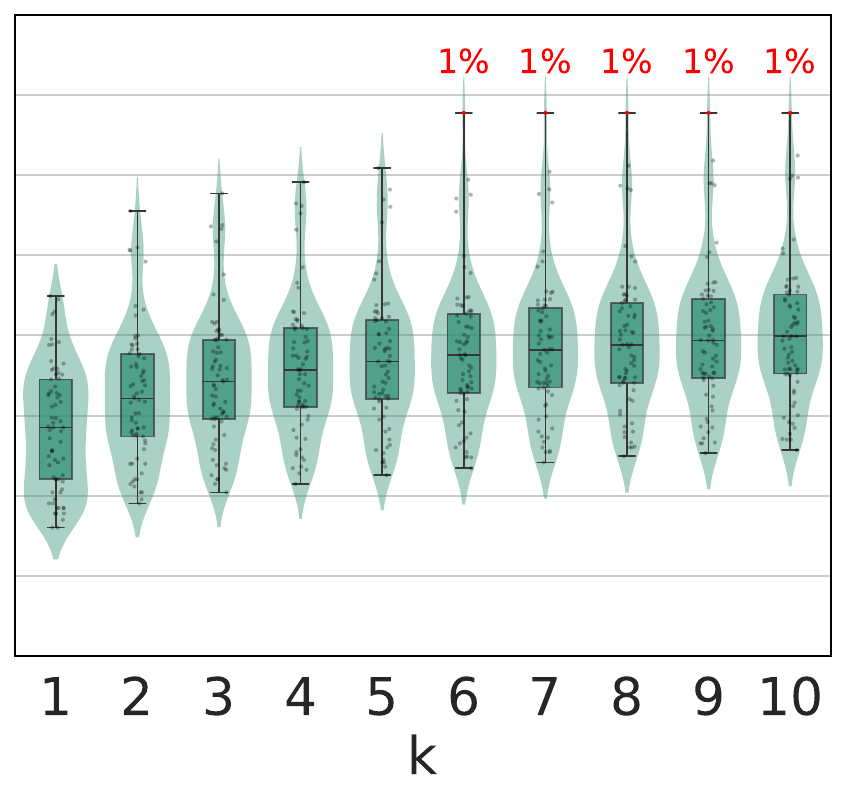}
			\caption{$\ilpunhypergraphfunc$}
		\end{subfigure}\hfill
		\begin{subfigure}[t]{0.30\linewidth}
			\centering
			\safeincludegraphics[height=2.5cm]{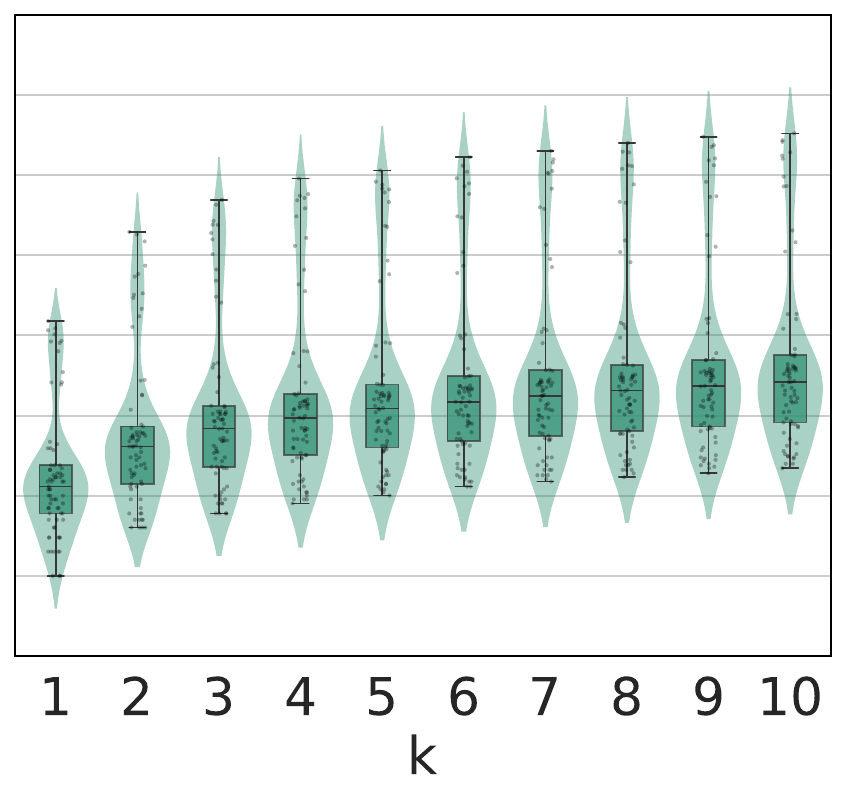}
			\caption{$\ilpunproofdagfunc$}
		\end{subfigure}
		
		\par\vspace{1em}
		\textbf{\csda (\httpd)}
	\end{minipage}
	
	\par\vspace{1em}
	
	%==================== PostgreSQL ====================%
	\begin{minipage}[t]{\linewidth}
		\centering
		
		\begin{subfigure}[t]{0.36\linewidth}
			\centering
			\safeincludegraphics[height=2.57cm]{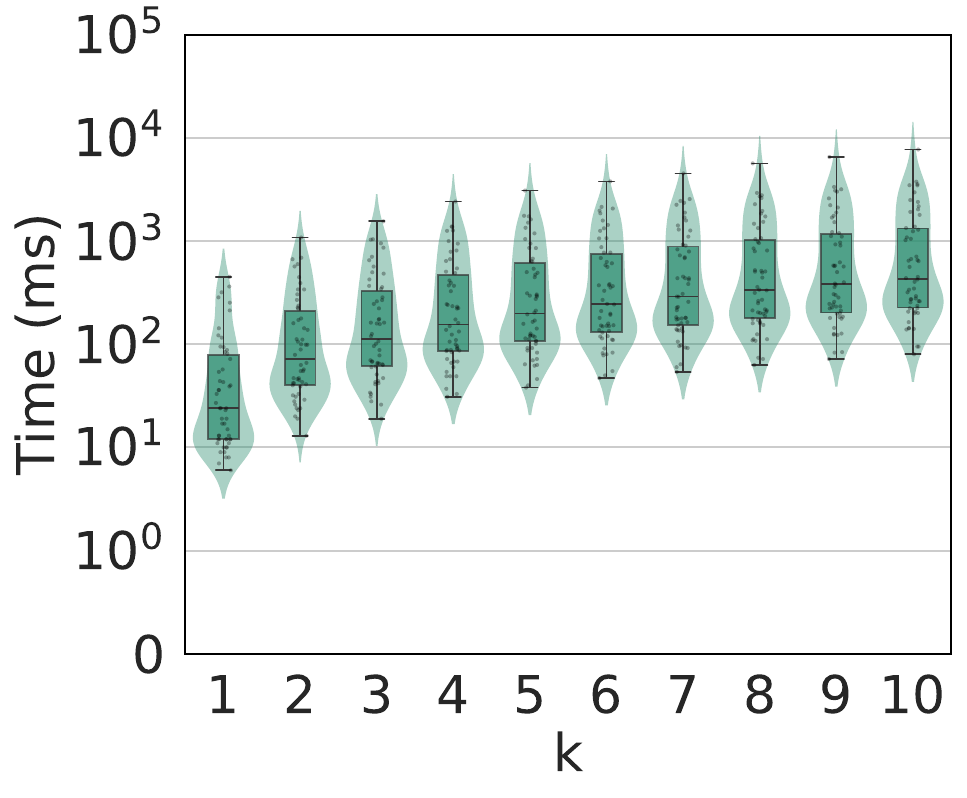}
			\caption{$\ilphypergraphfunc$}
		\end{subfigure}\hfill
		\begin{subfigure}[t]{0.30\linewidth}
			\centering
			\safeincludegraphics[height=2.5cm]{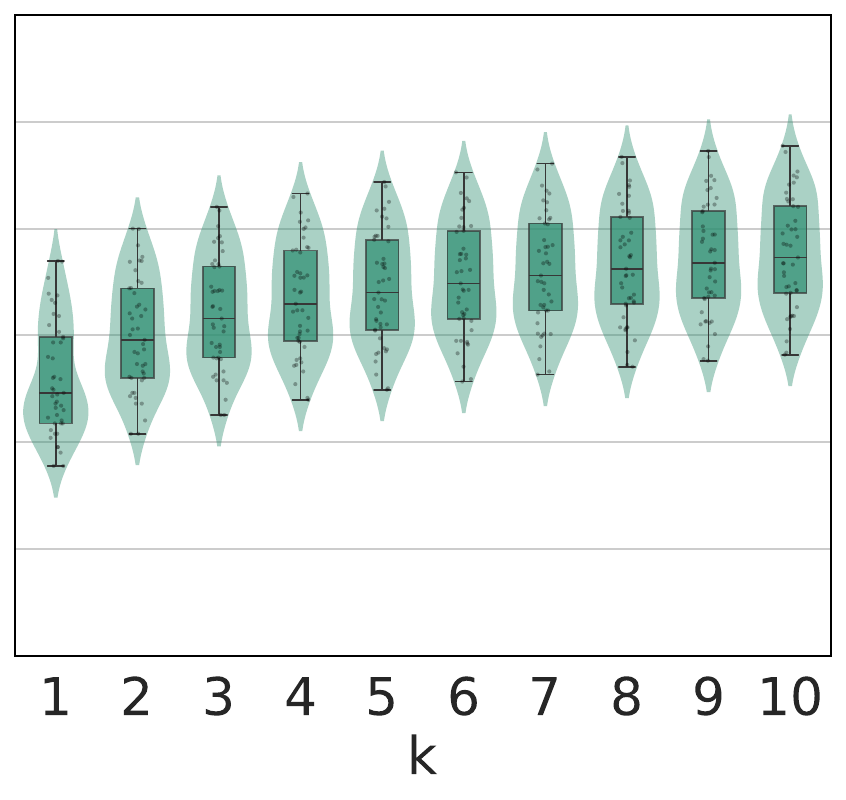}
			\caption{$\ilpunhypergraphfunc$}
		\end{subfigure}\hfill
		\begin{subfigure}[t]{0.30\linewidth}
			\centering
			\safeincludegraphics[height=2.5cm]{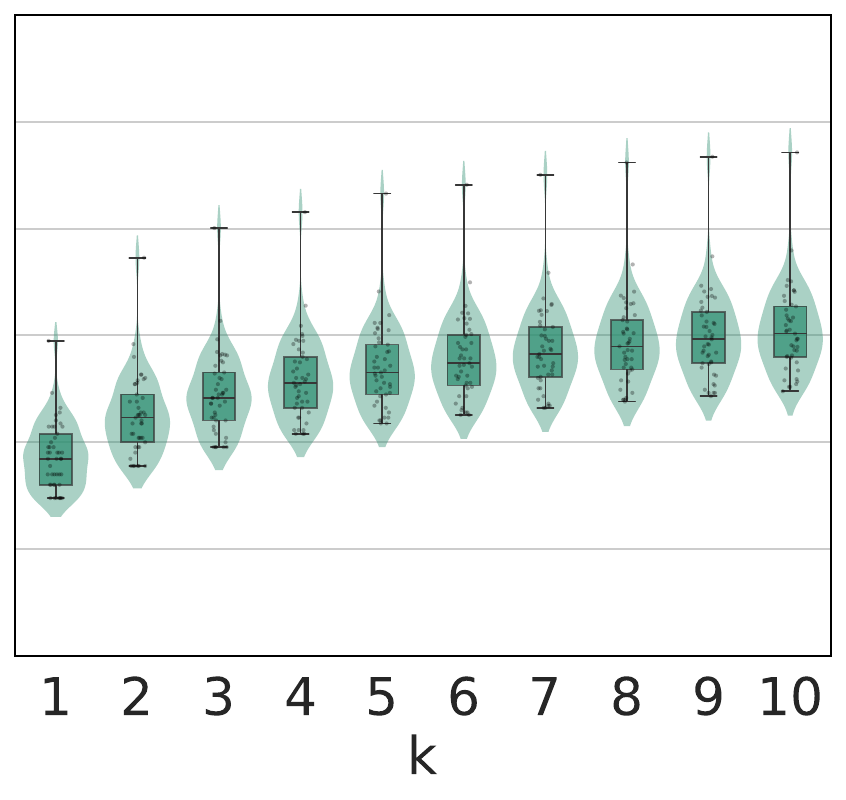}
			\caption{$\ilpunproofdagfunc$}
		\end{subfigure}
		
		\par\vspace{1em}
		\textbf{\csda (\postgres)}
	\end{minipage}
	
	\par\vspace{1em}
	
	%==================== Linux ====================%
	\begin{minipage}[t]{\linewidth}
		\centering
		
		\begin{subfigure}[t]{0.36\linewidth}
			\centering
			\safeincludegraphics[height=2.57cm]{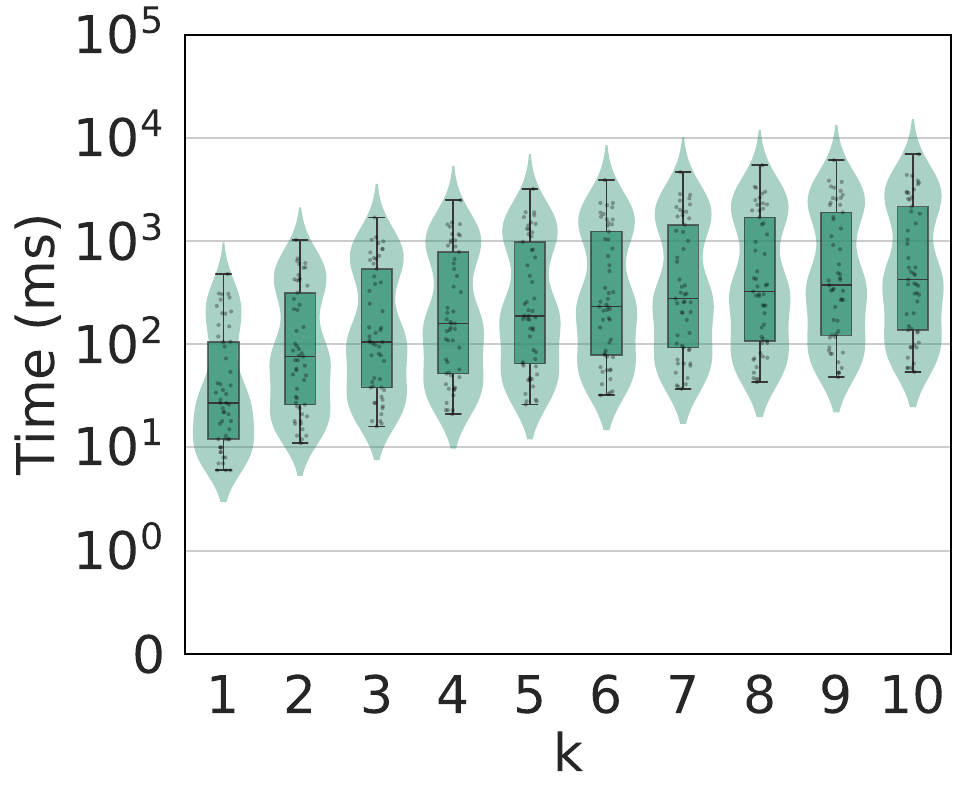}
			\caption{$\ilphypergraphfunc$}
		\end{subfigure}\hfill
		\begin{subfigure}[t]{0.30\linewidth}
			\centering
			\safeincludegraphics[height=2.5cm]{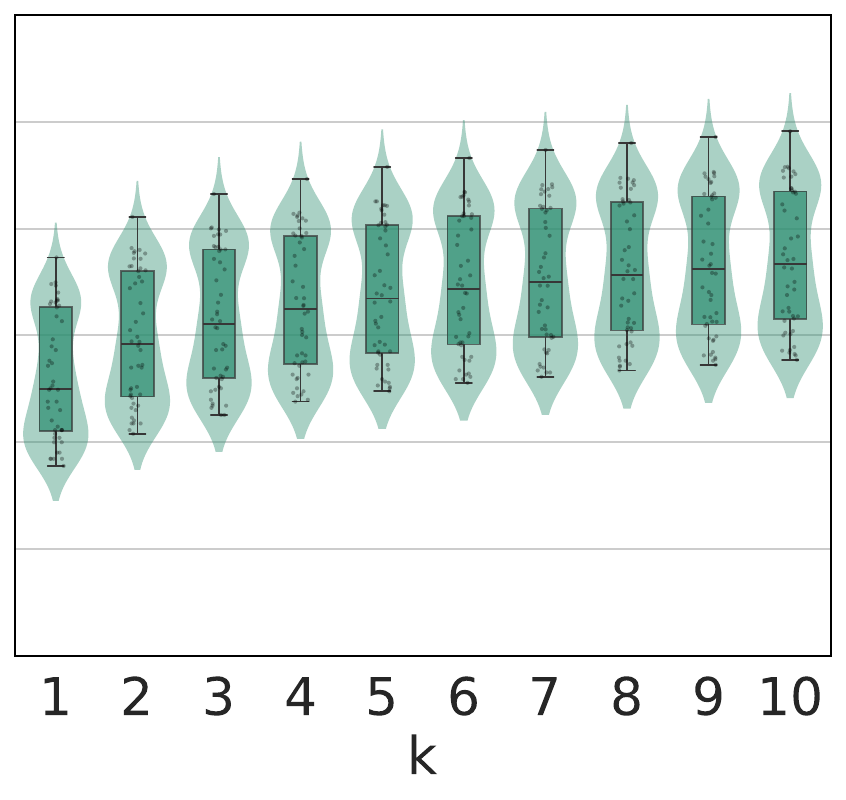}
			\caption{$\ilpunhypergraphfunc$}
		\end{subfigure}\hfill
		\begin{subfigure}[t]{0.30\linewidth}
			\centering
			\safeincludegraphics[height=2.5cm]{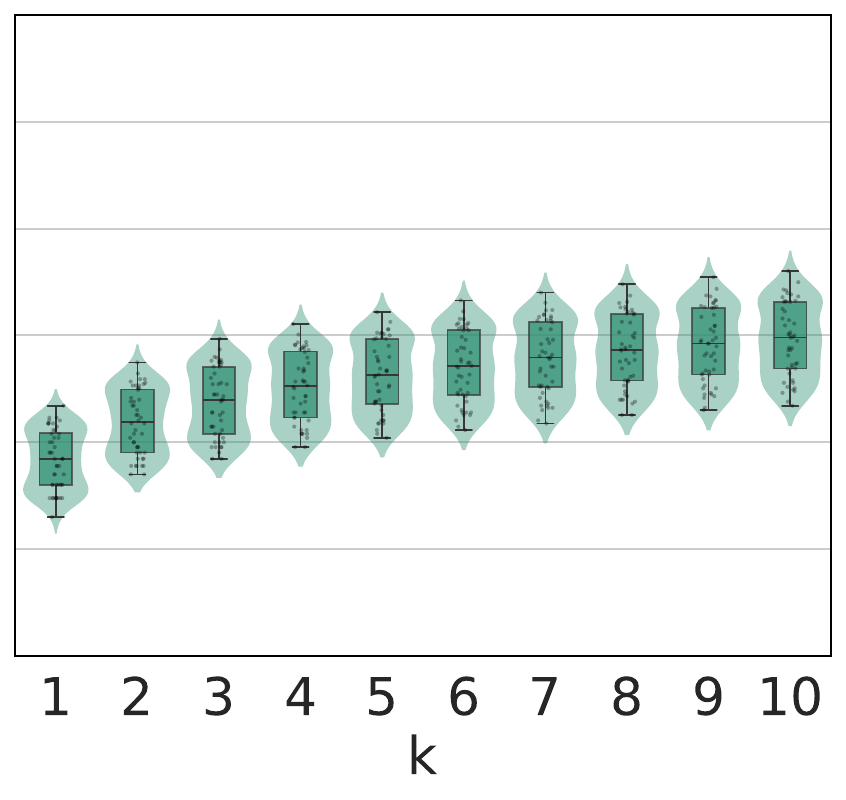}
			\caption{$\ilpunproofdagfunc$}
		\end{subfigure}
		
		\par\vspace{1em}
		\textbf{\csda (\linux)}
	\end{minipage}
	
	\caption{Runtimes for the \csda\ scenario using the approximation algorithms ($k$ up to 10).}
	\label{fig:k10-csda-approximate}
\end{figure}

\subsection{Runtime Analysis} 
\label{sec:approximate-runtime}

As already discussed in Section~\ref{sec:exact-experiments}, constructing the ILP encoding via $\ilpproofdagfunc$ requires more than the given 10-minutes timeout. Therefore, as in Section~\ref{sec:exact-experiments}, we only report plots for the remaining three approximation algorithms $\mathsf{Approx}[\ilpstargammafunc]$ with $(\gamma,\star) \in \{ (\mathsf{hyper},\mathsf{ALL}), (\mathsf{hyper}, \mathsf{UN}), (\mathsf{graph},\mathsf{UN})\}$.
Figures~\ref{fig:graph-analytics-approximate}-\ref{fig:csda-approximate} report the runtime of the above three approximation algorithms across all scenarios, following the same organization of Figures~\ref{fig:graph-analytics-exact}-\ref{fig:csda-exact}. Note that for the \galen\ scenario, most answer tuples have only $k=1$ explanations. Thus, we do not report the runtime of the approximation algorithms for the \galen\ scenario since, for $k=1$, the approximation algorithms solve the same maximization problem as their exact counterpart and the runtimes are identical.

As we can observe, in contrast to the exact computation (Section~\ref{sec:exact-experiments}), the approximate computation scales smoothly as $k$ increases and achieves substantial improvements in runtime.
The key observation is that the runtime is reduced by several orders of magnitude across all settings. Even as $k$ increases, runtimes remain low, with the vast majority of runs completing well below $10^3$ milliseconds (i.e., under one second). This represents a dramatic improvement over the exact computation and demonstrates that the approximate computation is suitable for practical, interactive settings that require near-real-time responses.

As expected, and similar to the exact algorithm, the runtime of the approximation algorithms increases with $k$, reflecting its incremental nature: it constructs the set of explanations by building one explanation at the time and by maximizing the marginal diversity at each step. However, this growth remains modest and controlled since the algorithms can avoid the combinatorial explosion of the search space of the exact algorithms that, at each of the $k$ iterations, it coincides (in the worst case) with the set of candidate explanations from $\starwhy{\bar t}{D}{Q}$.
A direct consequence of this improvement is that the vast majority of timeouts observed with exact algorithms disappear, and only a small number of more demanding instances still reach the timeout threshold.

Regarding differences across the family of proof trees and the ILP encodings considered, we observe similar trends to those discussed for the exact algorithms in Section~\ref{sec:exact-experiments}, although the differences are significantly less pronounced due to the reduced search space.

\subsection{Quality of the Approximation} \label{sec:quality}

We now assess the quality of the approximate solutions. Figures~\ref{fig:graph-analytics-rate}-\ref{fig:csda-rate} report the ratio between the optimal diversity and the diversity achieved by the approximation algorithms. The closer the ratio is to one the better the approximate solution, where a ratio equal to one means that the approximation algorithm managed to construct an optimal set of explanations; the number on the middle line in each box-plot refers to the average ratio.

As expected, the ratio is exactly one for $k=1$ since any singleton explanation provides optimal diversity, i.e., diversity equal to zero. For larger values of $k$, the ratio typically increases slightly above one, indicating a small divergence from the optimum. Moreover, all the three algorithms provide similar ratios across all scenarios, and thus, there is no clear preference w.r.t.~the quality of approximation.

A consistent pattern across all scenarios is that the ratio is highest for small values of $k>1$ (often at $k=2$), and then gradually decreases as $k$ increases. 
This behavior is due to the greedy nature of the approximation algorithm: early choices are made without knowledge of future choices, which may lead to suboptimal combinations when only a few explanations are selected. 
As $k$ grows, the solution space becomes more constrained, and the greedy choices align more closely with the optimum.

Summing up, the ratios remain quite close to one in all cases, with an average ratio of almost one for $k \ge 3$. The worst-case average ratio is $1.4$ for $k = 2$ in the \graphanalytics\ scenario with the \facebook\ database. This ratio, however, is still 30\% better than the theoretical approximation ratio established in Theorem~\ref{thm:approximation-quality}.

\subsection{Scalability to Larger Diversity Parameter}
\label{sec:large-k}

To further assess scalability, we evaluate the approximation algorithms for values of $k$ up to 10 on all scenarios except for \galen; \galen is excluded for the usual reason discussed above. In this evaluation, we keep the 10 minutes timeout and report the runtime required to complete the computation for the answer tuples that have at least 10 members of the why-provenance.

The results in Figures~\ref{fig:k10-graph-analytics-approximate}-\ref{fig:k10-csda-approximate} show that the approximation algorithms keep scaling gracefully even for larger $k$. Although the runtime increases as expected with $k$, the growth remains controlled and stable. This experiment demonstrates that all approximation algorithms provide a reliable and scalable solution to generate diverse why-provenance explanations. This, together with our analysis of the quality of the approximation, allows us to conclude that all three approximation algorithms provide practical solutions.

\section{Conclusion}\label{sec:conclusions}

In this work, we initiated the study of the problem of computing a small, sufficiently diverse subset of the set of why-provenance explanations of a certain answer of interest. After formalizing the problem in question as an optimization problem, and showing that it is intractable in data complexity, we focussed our efforts on the question whether ILP optimization solvers can be employed towards practical solutions.
Our thorough analysis and experimental evaluation showed that the latter is achievable as long as we keep the requested set of explanations small, or we employ approximation algorithms (with error guarantees).

An interesting direction for future research is to consider alternative diversity measures that are conceptually relevant and efficiently computable; the next natural choice is the min function.
Another interesting diection for future work is to study the problem of interest for more expressive query langauges, in particular, Datalog queries with negation interprested according to the standard well-founded or stable model semantics from Logic Programming. Note that the basic problem of computing why-provenance explanations for such expressive query languages has been recently studied in~\cite{BCPPE26}, and thus, we can build on the notions introduced in that paper.

%\begin{acks}
%\end{acks}

% \clearpage

\bibliographystyle{ACM-Reference-Format}
\bibliography{references}

% \clearpage

\appendix
\section{Proofs of Section~\ref{sec:diverse-why-provenance}}

\subsection{Proof of Theorem~\ref{the:diversity-complexity}}

We first introduce the decision problem of checking whether a fixed Datalog query has $k$ answers over a given database $D$ such that their diversity is above a certain threshold; we know from~\cite{MePS25} that this is an NP-hard problem. We then reduce this problem via a polynomial-time Turing reduction to our problems, which in turn establishes Theorem~\ref{the:diversity-complexity}.

\medskip
\noindent
\underline{\textbf{Diverse Answers}}

\smallskip
\noindent 
Consider an integer $n > 0$. For an $n$-tuple $\bar t = (t_1,\ldots,t_n)$, let $\bar t[i]$ be the term $t_i$, with $i \in [n]$. For two $n$-tuples $\bar t$ and $\bar u$, the \emph{Hamming distance between $\bar t$ and $\bar u$}, denoted $\hammd{\bar t}{\bar u}$, is defined as the integer $|\{i \in [n] \mid  \bar t[i] \neq \bar u[i]\}|$. Moreover, for a set $\mathcal{A} = \{\bar t_1,\ldots,\bar t_k\}$ of $n$-tuples, the \emph{diversity of $\mathcal{A}$}, denoted $\diversitytup{\mathcal{A}}$, is defined as%
\[
\sum\limits_{i,j \in [k] \text{ and }i<j} \hammd{\bar t_i}{\bar t_j}.
\]

For a Datalog query $Q = (\dep,R)$, a database $D$ over $\esch{\dep}$, and an integer $k > 0$, we define $\ansk{k}{D}{Q}$ as the collection of all sets $\mathcal{A} \subseteq Q(D)$ of tuples with $|\mathcal{A}| = k$. For a fixed Datalog query $Q = (\dep,R)$, we consider the following decision problem:

% \medskip

\begin{center}
	\fbox{\begin{tabular}{ll}
			{\small PROBLEM} : & $\mathsf{DiverseAnswersCheck}[Q]$
			\\[2mm]
			{\small INPUT} : & A database $D$ over $\esch{\dep}$,\\
			& an integer $k > 0$ encoded in unary, and \\
			& an integer $d \ge 0$\\[2mm]
			{\small QUESTION} : & Is there a set $\mathcal{A} \in \ansk{k}{D}{Q}$ with $\diversitytup{\mathcal{A}} \ge d?$
	\end{tabular}}
\end{center}

\medskip

We know that there is a Datalog query for which the above problem is NP-hard; this is a consequence of~\cite[Theorem 3.13]{MePS25}.

\begin{theorem}\label{the:diverse-answers}
	$\mathsf{DiverseAnswersCheck}[Q^*]$ is NP-complete, where $Q^* = (\dep^*,R)$ with $\dep^* = \{R(x_1,\ldots,x_5)\ \assign\ P(x_1,\ldots,x_5)\}$.
\end{theorem}

\medskip
\noindent
\underline{\textbf{The Reduction}}

\smallskip
\noindent 
Consider now the Datalog query $\hat{Q} = (\hat{\dep},\hat{R})$, where $\hat{R}$ is a $0$-ary predicate and $\hat{\dep}$ consists of the single rule
\[
\hat{R}\ \assign\ P(x_1,\ldots,x_5),P_1(x_1),\ldots,P_5(x_5).
\]
We proceed to show that, for every $\star \in \{\mathsf{ALL},\mathsf{UN}\}$, there exists a polynomial-time Turing reduction from $\mathsf{DiverseAnswersCheck}[Q^*]$ to $\star\text{-}\mathsf{DiverseWhyProvenance}[\hat{Q}]$, i.e., there is a procedure with access to an oracle for the problem $\star\text{-}\mathsf{DiverseWhyProvenance}[\hat{Q}]$ that decides $\mathsf{DiverseAnswersCheck}[Q^*]$ in polynomial time.

%	First, we note that for the Datalog query $\hat{Q}$, $\unwhy{\bar t}{D}{\hat{Q}}$ and $\allwhy{\bar t}{D}{\hat{Q}}$ coincide, for every database $D$ over $\esch{\hat{\dep}}$, and (thee only) tuple $\bar t \in \adom{D}^{\arity{\hat{R}}}$. Hence, it is enough to prove the claim for $\star = \mathsf{ALL}$. We now proceed by describing the desired polynomial-time procedure.

%For a database $D$ over $\esch{\dep^*}$, we define the database $\hat{D}$ over $\esch{\hat{\dep}}$ as the set:
%
%$$ \hat{D} = D \cup \{P_1(t_1),\ldots,P_5(t_5) \mid P(t_1,\ldots,t_5) \in D\}.$$
%
%Fix any $\star \in \{\mathsf{ALL},\mathsf{UN}\}$, and 
Consider an input to $\mathsf{DiverseAnswersCheck}[Q^*]$, i.e., a database $D$ over $\esch{\dep^*}$, an integer $k > 0$ encoded in unary, and an integer $d \ge 0$.
The procedure deciding $\mathsf{DiverseAnswersCheck}[Q^*]$, which we dub $\mathsf{DecideDiverseAnswers}$, does the following:
\begin{enumerate}
	\item Construct the database 
	$$ \hat{D}\ =\ D \cup \{P_1(t_1),\ldots,P_5(t_5) \mid P(t_1,\ldots,t_5) \in D\}.$$
	%$\hat{D}$ over $\esch{\hat{\dep}}$.
	\item Ask the oracle to compute some $\mathcal{\hat{S}} \in \argmax\limits_{\mathcal{S} \in \starwhyk{k}{()}{\hat{D}}{\hat{Q}}} \diversity{\mathcal{S}}$. 
	\item If the oracle returns \textbf{fail}, then reject. 
	\item Accept if $\diversity{\mathcal{\hat{S}}} \ge k \cdot (k-1) + 2\cdot d$; otherwise, reject.
\end{enumerate} 

\medskip
\noindent
\textbf{Running Time.}
The procedure $\mathsf{DecideDiverseAnswers}$ clearly runs in polynomial time since constructing $\hat{D}$ only requires iterating over the facts in $D$. Moreover, the set $\mathcal{\hat{S}}$ returned by the oracle, when it exists, contains linearly many elements w.r.t.\ the size of the encoding of $k$, and each element in $\mathcal{\hat{S}}$ is no larger than $|\hat{D}|$. Finally, computing $\diversity{\mathcal{\hat{S}}}$ requires adding $k(k-1)$ symmetric distances, and the number $k(k-1) + 2\cdot d$ can be easily computed.

\medskip
\noindent
\textbf{Correctness.} We proceed to show that the procedure indeed decides $\mathsf{DiverseAnswersCheck}[Q^*]$.
We first make some observations. 	For a $5$-tuple $\bar t$, let $D_{\bar t} = \{P(\bar t),P_1(\bar t[1]),\ldots,P_5(\bar t[5])\}$.
It is not difficult to verify that by construction of $\hat{D}$ and by definition of $\hat{Q} = (\hat{\dep},\hat{R})$, we have that:
\begin{observation}\label{obs:why-ans}
	For each $\star \in \{\mathsf{ALL},\mathsf{UN}\}$,
	$$ \starwhy{()}{\hat{D}}{\hat{Q}}\ =\ \{D_{\bar t} \mid \bar t \in Q^*(D)\}.$$
\end{observation}

Moreover, for two distinct databases $D_{\bar t},D_{\bar u} \in \starwhy{()}{\hat{D}}{\hat{Q}}$, the symmetric difference between $D_{\bar t}$ and $D_{\bar u}$ is of the form
\begin{multline*}
	\diff{D_{\bar t}}{D_{\bar u}}\ =\ \{P(\bar t),P(\bar u)\} \cup 
	\diff{D_{\bar t} \setminus \{P(\bar t)\}}{D_{\bar u} \setminus \{P(\bar u)\}},
\end{multline*}
and thus,
\[
\dist{D_{\bar t}}{D_{\bar u}}\ =\ |\diff{D_{\bar t}}{D_{\bar u}}|\ =\ 2 + 2 \cdot \hammd{\bar t}{\bar u}.
\]
The above allows us to conclude the following:

\begin{observation}\label{obs:diversity-ans}
	For each $\mathcal{S} = \{D_{\bar t_1},\ldots,D_{\bar t_k}\} \in \starwhyk{k}{()}{\hat{D}}{\hat{Q}}$, 
	%the diversity of $\mathcal{S}$ is:
	$$\diversity{\mathcal{{S}}} = \sum\limits_{i,j \in [k] \text{ and } i<j} (2 + 2 \cdot \hammd{\bar t_i}{\bar t_j}) = k \cdot (k-1) + 2 \cdot \diversitytup{\{\bar t_1,\ldots,\bar t_k\}}.$$
\end{observation}

We now proceed with the proof of correctness. Let $D$, $k$, and $d$ be an input to $\mathsf{DecideDiverseAnswers}$. 

\medskip

($\Rightarrow$) Assume first that there exists a set $\mathcal{A} \in \ansk{k}{D}{Q^*}$ with $\diversitytup{\mathcal{A}} \ge d$. We argue that $\mathsf{DecideDiverseAnswers}$ accepts with input $D$, $k$, and $d$. Note that by the existence of $\mathcal{A}$, there also exists $\mathcal{A}^* \in \argmax\limits_{\mathcal{A} \in \ansk{k}{D}{Q^*}} \diversitytup{\mathcal{A}}$, and all such $\mathcal{A}^*$ have the same diversity such that $\diversitytup{\mathcal{A}^*} \ge d$. Let $d^* = \diversitytup{\mathcal{A}^*}$.
By the existence of $\mathcal{A}$ and Observation~\ref{obs:why-ans}, the oracle call does not fail and returns a set $\mathcal{\hat{S}} \in \starwhyk{k}{()}{\hat{D}}{\hat{Q}}$. Moreover, by Observations~\ref{obs:why-ans} and 2, and the fact that $\mathcal{\hat{S}} \in \argmax\limits_{\mathcal{S} \in \starwhyk{k}{()}{\hat{D}}{\hat{Q}}} \diversity{\mathcal{S}}$, we conclude that
\[
\diversity{\mathcal{\hat{S}}}\ =\ k \cdot (k-1) + 2\cdot d^*.
\]
Since $d^* \ge d$, we have that $\diversity{\mathcal{\hat{S}}} \ge k \cdot (k-1) + 2\cdot d$, and thus, $\mathsf{DecideDiverseAnswers}$ accepts, as needed.

\medskip

($\Leftarrow$) Assume now that there is no set $\mathcal{A} \in \ansk{k}{D}{Q^*}$ with $\diversitytup{\mathcal{A}} \ge d$, i.e., either $\ansk{k}{D}{Q^*}$ is empty, or $\ansk{k}{D}{Q^*}$ is not empty but every set $\mathcal{A}^* \in \argmax\limits_{\mathcal{A} \in \ansk{k}{D}{Q^*}} \diversitytup{\mathcal{A}}$ is such that $\diversitytup{\mathcal{A}^*} < d$. We argue that $\mathsf{DecideDiverseAnswers}$ rejects with input $D$, $k$, and $d$. 
If $\ansk{k}{D}{Q^*}$ is empty then, by Observation~\ref{obs:why-ans}, $\starwhyk{k}{()}{\hat{D}}{\hat{Q}}$ is also empty and the oracle call returns \textbf{fail}, which in turn implies that $\mathsf{DecideDiverseAnswers}$ rejects, as needed. 
Assume now that $\ansk{k}{D}{Q^*}$ is not empty, but every $\mathcal{A}^* \in \argmax\limits_{\mathcal{A} \in \ansk{k}{D}{Q^*}} \diversitytup{\mathcal{A}}$ is such that $\diversitytup{\mathcal{A}^*} < d$. Let $d^* = \diversitytup{\mathcal{A}^*}$, for an arbitrarily chosen $\mathcal{A}^* \in \argmax\limits_{\mathcal{A} \in \ansk{k}{D}{Q^*}} \diversitytup{\mathcal{A}}$.
By Observation~\ref{obs:why-ans}, and the fact that $\ansk{k}{D}{Q^*}$ is not empty, the oracle call will succeed and the set $\mathcal{\hat{S}} \in \argmax\limits_{\mathcal{S} \in \starwhyk{k}{()}{\hat{D}}{\hat{Q}}} \diversity{\mathcal{S}}$ is constructed.
By Observations~\ref{obs:why-ans} and \ref{obs:diversity-ans}, and the fact that $\mathcal{\hat{S}}$, has maximal diversity, we get that
\[
\diversity{\mathcal{\hat{S}}}\ =\ k \cdot (k-1) + 2\cdot d^*.
\]
Since $d^* < d$, we have that $\diversity{\mathcal{\hat{S}}} < k(k-1) + 2\cdot d$, and thus, $\mathsf{DecideDiverseAnswers}$ rejects, as needed.\qed

\section{Proofs of Section~\ref{sec:exact-encoding}}

\subsection{Proof of Theorem~\ref{thm:generic-encoding}}
	For notational convenience, in the rest of the proof we will write $\am$ for the set
	\[
	\argmax\limits_{\mathcal{S} \in \starwhyk{k}{\bar t}{D}{Q}} \diversity{\mathcal{S}}.
	\]
	We proceed to prove the two directions separately.

	\medskip
	\noindent
	$(1) \Rightarrow (2)$.
	Consider a set $\{D_1,\ldots,D_k\}$ of databases over $\esch{\dep}$ such that $\{D_1,\ldots,D_k\} \in \am$.
	%$\{D_1,\ldots,D_k\} \in \argmax\limits_{\mathcal{S} \in \starwhyk{k}{\bar t}{D}{Q}} \diversity{\mathcal{S}}$. 
	We construct an optimal solution $\mu$ of the $(T,k)$-problem $(f,S)$ such that, for each $i \in [k]$, $D_i = \{\alpha_\ell \mid \ell \in [n] \text{ and } \mu(\copyof{x_{\alpha_\ell}}{i}) = 1\}$.

	Since $\{D_1,\ldots,D_k\} \in \am$,
	%\argmax\limits_{\mathcal{S} \in \starwhyk{k}{\bar t}{D}{Q}} \diversity{\mathcal{S}}$, 
	$D_1,\ldots,D_k \in \starwhy{\bar t}{D}{Q}$, which, by definition of the ILP encoding of $\starwhy{\bar t}{D}{Q}$ (Definition~\ref{def:ilp-encoding-why}), implies that there exist $\mu_1,\ldots,\mu_k \in \sol{T}$ such that $D_i = \{\alpha_\ell \mid \ell \in [n] \text{ and } \mu_i(x_{\alpha_\ell}) = 1\}$, for each $i \in [k]$.
	By exploiting $\mu_1,\ldots,\mu_k$, we construct a solution $\mu : \mathcal{X} \cup \mathcal{B} \rightarrow \mathbb{Z}$ of the $(T,k)$-problem $(f,S)$ as follows; recall that $(f,S)$ is over the integer variables $\mathcal{X} = \bigcup_{i \in [k]} X^{(i)}$ and the Boolean variables $\mathcal{B} = \bigcup_{i \in [k]} B^{(i)} \cup \{\auxvar^{i,j}_\ell \mid 1 \leq \ell \leq n \text{ and } 1 \leq i < j \leq k\}$: 
	%We define $\mu : \mathcal{X} \cup \mathcal{B} \rightarrow \mathbb{Z}$ such that, 
	for each variable $x \in \mathcal{X} \cup \mathcal{B}$,
	\begin{eqnarray*}
		&& \mu(x)\\
		&=&
		\hspace{-1mm} 
		\begin{cases}
			\mu_i(z) & \text{if } x = z^{(i)}, i \in [k] \text{ and } z \in X \cup B \\
			\mu_i(x_{\alpha_\ell}) \oplus \mu_j(x_{\alpha_\ell}) & \text{if } x = \auxvar^{i,j}_\ell, \ell \in [n] \text{ and } 1 \le i < j \le k.
		\end{cases}
	\end{eqnarray*}
	In other words, $\mu$ assigns to each copy variable $z^{(i)}$ of the $(T,k)$-problem $(f,S)$ the same value assigned by $\mu_i$ to $z$, while it assigns to each auxiliary variable $\auxvar^{i,j}_\ell$ of $(f,S)$ the xor of the values of the (Boolean) database variable $x_{\alpha_\ell}$ given by $\mu_i$ and $\mu_j$, respectively. The function $\mu$ is indeed a solution of $(f,S)$ by construction, and by the fact that $\mu_1,\ldots,\mu_k$ are all solutions of $T$; note, in particular, that the constraints of the form $\sum_{1 \leq \ell \leq n} \auxvar^{i,j}_\ell \ge 1$ are satisfied since $D_1,\ldots,D_k$ are pairwise distinct. 
	
	We now proceed to argue that $\mu$ is optimal. Recall that the objective function $f$ of the $(T,k)$-problem is:
	\[
	\sum\limits_{1 \leq i < j \leq k} \left(\sum\limits_{1 \leq \ell \leq n} \auxvar^{i,j}_\ell\right).
	\]
	By definition of $\mu$, and the fact that $D_i = \{\alpha_\ell \mid \ell \in [n] \text{ and } \mu_i(x_{\alpha_\ell}) = 1\}$, for each $i \in [k]$, we conclude that $\mu(f) = \delta(\{D_1,\ldots,D_k\})$. Assume, towards a contradiction, that $\mu$ is not optimal. Thus, there exists a solution $\rho : \mathcal{X} \cup \mathcal{B} \rightarrow \mathbb{Z}$ of $(f,S)$ with $\rho(f) > \mu(f)$. Since $\rho$ is a solution of $(f,S)$, it satisfies the set of constraints of $S$ of the form $T^{(1)},\ldots,T^{(k)}$.
	%	 i.e., the projection $\proj{\mu}{X^{(i)} \cup B^{(i)}}$ of $\mu$ satisfies $T^{(i)}$, for each $i \in [k]$. 
	%
	For each $i \in [k]$, let $\rho_i : X \cup B \rightarrow \mathbb{Z}$ be the assignment such that $\rho_i(x) = \rho(x^{(i)})$, for each $x \in X \cup B$. It follows that $\rho_i$ is a solution of the ILP encoding $T$, for each $i \in [k]$. Hence, by definition of ILP encoding, the database $D'_i = \{\alpha_\ell \mid \ell \in [n] \text{ and } \rho_i(x_{\alpha_\ell}) = 1\}$ belongs to $\starwhy{\bar t}{D}{Q}$, for each $i \in [k]$. From the latter, and the fact that $\rho$ is a solution of $(f,S)$, it follows that $\rho(f) = \delta(\{D'_1,\ldots,D'_k\})$, with $\{D'_1,\ldots,D'_k\} \in \starwhyk{k}{\bar t}{D}{Q}$. 
	Hence, from our assumption $\rho(f) > \mu(f)$, we conclude that $\delta(\{D'_1,\ldots,D'_k\}) > \delta(\{D_1,\ldots,D_k\})$, which contradicts the fact that $\{D_1,\ldots,D_k\} \in \am$, and the claim follows.
	%\argmax\limits_{\mathcal{S} \in \starwhyk{k}{\bar t}{D}{Q}} \diversity{\mathcal{S}}$.
	
	\medskip
	\noindent
	$(2) \Rightarrow (1)$. Assume now that there is an optimal solution $\mu$ of the $(T,k)$-problem $(f,S)$ such that, for each $i \in [k]$, $D_i = \{\alpha_\ell \mid \ell \in [n] \text{ and } \mu(\copyof{x_{\alpha_\ell}}{i}) = 1\}$. We need to show that $\{D_1,\ldots,D_k\} \in \am$.
	%\argmax\limits_{\mathcal{S} \in \starwhyk{k}{\bar t}{D}{Q}} \diversity{\mathcal{S}}$.
	
	For each $i \in [k]$, let $\mu_i : X \cup B \rightarrow \mathbb{Z}$ be the assignment such that $\mu_i(x) = \mu(x^{(i)})$, for each $x \in X \cup B$. Since $\mu$ is a solution of $(f,S)$, and $S$ contains all the constraints of $T^{(i)}$, we get that $\mu_i$ is a solution of the ILP encoding $T$, for each $i \in [k]$. Therefore, by definition of ILP encoding of $\starwhy{\bar t}{D}{Q}$ (see Definition~\ref{def:ilp-encoding-why}), $D_i = \db{\mu_i} \in \starwhy{\bar t}{D}{Q}$, for each $i \in [k]$.
	
	We now observe that, for every solution $\nu$ of $(f,S)$, and for all $1 \le i < j \le k$, the xor constraints of $S$ imply that $\nu(\auxvar^{i,j}_\ell) = 1$ iff $\nu(\copyof{x_{\alpha_\ell}}{i}) \neq \nu(\copyof{x_{\alpha_\ell}}{j})$, i.e., iff $\alpha_\ell$ belongs to exactly one of the two databases induced by the $i$-th and the $j$-th copy of the database variables. Consequently, with $D^\nu_i = \{\alpha_\ell \mid \ell \in [n] \text{ and } \nu(\copyof{x_{\alpha_\ell}}{i}) = 1\}$,
	\begin{equation}\label{eq:aux-distance}
		\sum\limits_{1 \leq \ell \leq n} \nu(\auxvar^{i,j}_\ell)\ =\ |\diff{D^\nu_i}{D^\nu_j}|\ =\ \dist{D^\nu_i}{D^\nu_j}.
	\end{equation}
	Since $S$ also contains the constraints $\sum_{1 \leq \ell \leq n} \auxvar^{i,j}_\ell \ge 1$, for all $1 \le i < j \le k$, we get that $\dist{D_i}{D_j} \ge 1$, i.e., $D_i \neq D_j$, for all $1 \le i < j \le k$. Hence, $\{D_1,\ldots,D_k\}$ consists of exactly $k$ distinct members of $\starwhy{\bar t}{D}{Q}$, that is, $\{D_1,\ldots,D_k\} \in \starwhyk{k}{\bar t}{D}{Q}$. Moreover, by~\eqref{eq:aux-distance} and the definition of the objective function $f$,
	\[
	\mu(f)\ =\ \sum\limits_{1 \leq i < j \leq k} \dist{D_i}{D_j}\ =\ \diversity{\{D_1,\ldots,D_k\}}.
	\]
	
	It remains to show that $\{D_1,\ldots,D_k\}$ maximizes the diversity among the members of $\starwhyk{k}{\bar t}{D}{Q}$. Assume, towards a contradiction, that this is not the case. Since $\starwhyk{k}{\bar t}{D}{Q}$ is non-empty, as it contains $\{D_1,\ldots,D_k\}$, and finite, there exists a set $\{D'_1,\ldots,D'_k\} \in \starwhyk{k}{\bar t}{D}{Q}$ such that $\delta(\{D'_1,\ldots,D'_k\}) > \delta(\{D_1,\ldots,D_k\})$. By definition of ILP encoding, there are $\rho_1,\ldots,\rho_k \in \sol{T}$ such that $D'_i = \{\alpha_\ell \mid \ell \in [n] \text{ and } \rho_i(x_{\alpha_\ell}) = 1\}$, for each $i \in [k]$. Let $\rho$ be the assignment defined from $\rho_1,\ldots,\rho_k$ exactly as $\mu$ has been defined from $\mu_1,\ldots,\mu_k$ in the direction $(1) \Rightarrow (2)$ above. As in that direction, $\rho$ is a solution of $(f,S)$; in particular, the constraints $\sum_{1 \leq \ell \leq n} \auxvar^{i,j}_\ell \ge 1$ are satisfied since the databases $D'_1,\ldots,D'_k$ are pairwise distinct. But then, by~\eqref{eq:aux-distance}, $\rho(f) = \delta(\{D'_1,\ldots,D'_k\}) > \delta(\{D_1,\ldots,D_k\}) = \mu(f)$, which contradicts the fact that $\mu$ is an optimal solution of $(f,S)$. Therefore, $\{D_1,\ldots,D_k\} \in \am$, as needed.\qed %\argmax\limits_{\mathcal{S} \in \starwhyk{k}{\bar t}{D}{Q}} \diversity{\mathcal{S}}$, as needed.
	
\section{Proofs of Section~\ref{sec:characterizations}}

\subsection{Proof of Theorem~\ref{thm:hyper-correct}}
	Let $\downc{R(\bar t)}{D}{\dep} = (V,E)$, and let $T$ be the set of constraints $\ilphypergraph{\bar t}{D}{Q}$. Consider a database $D'$ over $\esch{\dep}$. Since every member of $\allwhy{\bar t}{D}{Q}$ is a subset of $D$, and $\{\alpha_i \mid i \in [n] \text{ and } \mu(z_{\alpha_i}) = 1\} \subseteq D$ for every assignment $\mu$, we can safely assume that $D' \subseteq D$; otherwise, both items of Definition~\ref{def:ilp-encoding-why} are false, and the statement holds trivially.
	Moreover, by Theorem~\ref{thm:arbitrary-hypergraph}, $D' \in \allwhy{\bar t}{D}{Q}$ iff there is a description of a proof tree of $R(\bar t)$ w.r.t.~$D$ and $\dep$ with support $D'$. Hence, it suffices to show that the following statements are equivalent:
	\begin{enumerate}
		\item There is a description $E^*$ of a proof tree of $R(\bar t)$ w.r.t.~$D$ and $\dep$ with support $D'$.
		\item There is $\mu \in \sol{T}$ with $\db{\mu} = D'$.
	\end{enumerate}
	
	Before we proceed further, let us introduce some auxiliary notation and present a couple of simple observations.
	For an assignment $\mu$ and a linear expression $\mathcal{E}$ over the variables of $T$, we write $\mu(\mathcal{E})$ for the value of $\mathcal{E}$ under $\mu$.
	For a multiset $E^*$ of hyperedges from $\downc{R(\bar t)}{D}{\dep}$, we write $\hat{E}^*$ for the set of (standard) edges induced by the hyperedges occurring in $E^*$, i.e., $\hat{E}^* = \{(\alpha,\beta) \mid (\alpha,T) \in E^* \text{ and } \beta \in T\}$; clearly, $\hat{E}^* \subseteq \hat{E}$. By the definition of reachability in a hypergraph, for all nodes $\alpha,\beta \in \nodes{E^*}$, $\alpha$ reaches $\beta$ in $\hypergraphof{E^*}$ iff there is a directed path from $\alpha$ to $\beta$ in the directed graph $(\nodes{E^*},\hat{E}^*)$.
	Furthermore, since $\esch{\dep} \cap \isch{\dep} = \emptyset$, no fact of $D$ occurs as the source of a hyperedge of $\gri{D}{\dep}$, and hence, for each $\alpha \in V \cap D$, the summation $\mathcal{O}_\alpha$ has no terms, and $\alpha$ has no outgoing edge in $\hat{E}$.
	We now show the equivalence claimed above.
	
	\medskip
	\noindent
	$(1) \Rightarrow (2)$. Let $E^*$ be a description of a proof tree of $R(\bar t)$ w.r.t.~$D$ and $\dep$ with support $D'$. By property~(4) of Definition~\ref{def:hypergraph-properties}, every node of $\nodes{E^*}$ is reachable from $R(\bar t)$ in the directed graph $(\nodes{E^*},\hat{E}^*)$. Hence, we can fix a spanning tree $\mathcal{T}$ of $(\nodes{E^*},\hat{E}^*)$ rooted at $R(\bar t)$, e.g., a spanning tree constructed via breadth-first search. Thus, the set of nodes of $\mathcal{T}$ is $\nodes{E^*}$, and its set of edges is a subset of $\hat{E}^*$. For a node $\beta$ of $\mathcal{T}$, let $n_\beta$ be the number of nodes of the subtree of $\mathcal{T}$ rooted at $\beta$. We define the assignment $\mu$ for $X$ and $B$ as follows:
	\begin{itemize}
		\item[-] $\mu(x_{(\alpha,T)}) = \card{(\alpha,T)}{E^*}$, for each $(\alpha,T) \in E$,
		\item[-] $\mu(z_\alpha) = 1$ if $\alpha \in \nodes{E^*}$, and $\mu(z_\alpha) = 0$ otherwise, for each $\alpha \in V \cup D$, and
		\item[-] $\mu(y_{(\alpha,\beta)}) = n_\beta$ if $(\alpha,\beta)$ is an edge of $\mathcal{T}$, and $\mu(y_{(\alpha,\beta)}) = 0$ otherwise, for each $(\alpha,\beta) \in \hat{E}$.
	\end{itemize}
	We proceed to show that $\mu$ is a solution of $T$ with $\db{\mu} = D'$.
	
	By construction, $\mu(\mathcal{I}_\alpha) = \indegree{\alpha}{E^*}$ and $\mu(\mathcal{O}_\alpha) = \outdegree{\alpha}{E^*}$, for each $\alpha \in V$, and all the values assigned to the variables of $X$ are non-negative, i.e., the constraints of $T_\mathsf{pos}$ are satisfied.
	Concerning $T_\mathsf{degree}$, the constraint $\mathcal{O}_{R(\bar t)} - \mathcal{I}_{R(\bar t)} = 1$ is satisfied due to property~(1) of Definition~\ref{def:hypergraph-properties}. For a node $\alpha \in V \setminus (D \cup \{R(\bar t)\})$, if $\alpha \in \nodes{E^*}$, then $\mu(\mathcal{O}_\alpha) - \mu(\mathcal{I}_\alpha) = 0$ due to property~(2), whereas, if $\alpha \not\in \nodes{E^*}$, then $\mu(\mathcal{O}_\alpha) = \mu(\mathcal{I}_\alpha) = 0$ since $\alpha$ occurs in no hyperedge of $E^*$; in both cases, the corresponding constraint is satisfied. Finally, since $\indegree{\alpha}{E^*} = 0$ for each $\alpha \in V \setminus \nodes{E^*}$, property~(3) implies that $\mu$ satisfies the constraint $\sum_{\alpha \in V \cap D} \mathcal{I}_\alpha \ge 1$.
	
	Concerning $T_\mathsf{bind}$, for a node $\alpha \in V$, $\alpha \in \nodes{E^*}$ iff $\alpha$ occurs in some hyperedge of $E^*$ iff $\indegree{\alpha}{E^*} + \outdegree{\alpha}{E^*} \ge 1$. Hence, $\mu(z_\alpha) = 1$ iff $\mu(\mathcal{I}_\alpha + \mathcal{O}_\alpha) \ge 1$, i.e., $\mu$ satisfies the indicator constraints of $T_\mathsf{bind}$. Moreover, $\mu(z_\alpha) = 0$ for each $\alpha \in D \setminus V$ since $\nodes{E^*} \subseteq V$.
	
	It remains to argue for $T_\mathsf{reach}$. Consider an edge $(\alpha,\beta) \in \hat{E}$ with $\mu(y_{(\alpha,\beta)}) \neq 0$, and thus, $(\alpha,\beta)$ is an edge of $\mathcal{T}$, which implies that $(\alpha,\beta) \in \hat{E}^*$, i.e., there is a hyperedge $(\alpha,T) \in E^*$ with $\beta \in T$, and therefore $\mu(\mathcal{C}_{(\alpha,\beta)}) \ge 1$. Since the subtree of $\mathcal{T}$ rooted at $\beta$ does not contain $R(\bar t)$, we get that $\mu(y_{(\alpha,\beta)}) = n_\beta \le |\nodes{E^*}| - 1 \le |V| - 1$ and the constraint $y_{(\alpha,\beta)} \le (|V|-1) \cdot \mathcal{C}_{(\alpha,\beta)}$ is satisfied; the same constraint is satisfied when $\mu(y_{(\alpha,\beta)}) = 0$ since $\mu(\mathcal{C}_{(\alpha,\beta)}) \ge 0$. Concerning the remaining constraints of $T_\mathsf{reach}$, observe that, for a node $\alpha$ of $\mathcal{T}$, the total incoming flow of $\alpha$ is $n_\alpha$ if $\alpha \neq R(\bar t)$, since the only edge of $\mathcal{T}$ entering $\alpha$ is the one coming from its parent, and $0$ if $\alpha = R(\bar t)$, while its total outgoing flow is $\sum_{\beta \text{ child of } \alpha \text{ in } \mathcal{T}} n_\beta = n_\alpha - 1$. Hence,
	\[
	\mu(\outflow{R(\bar t)}) - \mu(\inflow{R(\bar t)})\ =\ n_{R(\bar t)} - 1\ =\ |\nodes{E^*}| - 1\ =\ \mu\left(\sum\limits_{\alpha \in V} z_\alpha\right) - 1,
	\]
	as needed, and, for each $\alpha \in \nodes{E^*} \setminus \{R(\bar t)\}$, 
	\[
	\mu(\outflow{\alpha})\ =\ n_\alpha - 1\ =\ \mu(\inflow{\alpha}) - \mu(z_\alpha),
	\]
	as needed. 
	For $\alpha \in V \setminus \nodes{E^*}$, no edge of $\mathcal{T}$ is incident to $\alpha$, and thus,
	\[
	\mu(\outflow{\alpha})\ =\ \mu(\inflow{\alpha})\ =\ \mu(z_\alpha)\ =\ 0, 
	\]
	and the corresponding constraint is satisfied as well.
	
	From the above discussion, we get that $\mu \in \sol{T}$. Finally, by property~(5) of Definition~\ref{def:hypergraph-properties}, and since $\mu(z_\alpha) = 0$ for each $\alpha \in D \setminus V$, $\db{\mu} = \{\alpha \in D \mid \mu(z_\alpha) = 1\} = \nodes{E^*} \cap D = D'$, and item (2) follows.
	
	\medskip
	\noindent
	$(2) \Rightarrow (1)$. Let $\mu \in \sol{T}$ with $\db{\mu} = D'$, and let $E^*$ be the multiset of hyperedges from $\downc{R(\bar t)}{D}{\dep}$ such that $\card{(\alpha,T)}{E^*} = \mu(x_{(\alpha,T)})$, for each $(\alpha,T) \in E$; this is well-defined since, due to $T_\mathsf{pos}$, $\mu(x_{(\alpha,T)}) \ge 0$. We show that $E^*$ is a description of a proof tree of $R(\bar t)$ w.r.t.~$D$ and $\dep$ with support $D'$, which implies item (1).
	
	As above, $\mu(\mathcal{I}_\alpha) = \indegree{\alpha}{E^*}$ and $\mu(\mathcal{O}_\alpha) = \outdegree{\alpha}{E^*}$, for each $\alpha \in V$. Moreover, since all the multiplicities are non-negative,
	\[
	\nodes{E^*}\ =\ \left\{\alpha \in V \mid \indegree{\alpha}{E^*} + \outdegree{\alpha}{E^*} \ge 1 \right\} = \left\{\alpha \in V \mid \mu(z_\alpha) = 1\right\};
	\]
	the second equality holds due to the indicator constraints of $T_\mathsf{bind}$.
	
	Property~(1) of Definition~\ref{def:hypergraph-properties} holds since $\mu$ satisfies the constraint $\mathcal{O}_{R(\bar t)} - \mathcal{I}_{R(\bar t)} = 1$, which also implies that $\outdegree{R(\bar t)}{E^*} \ge 1$, i.e., $R(\bar t) \in \nodes{E^*}$. Property~(2) holds since $\nodes{E^*} \setminus (D \cup \{R(\bar t)\}) \subseteq V \setminus (D \cup \{R(\bar t)\})$, and $\mu$ satisfies the second family of constraints of $T_\mathsf{degree}$. Property~(3) holds since $\indegree{\alpha}{E^*} = 0$ for each node $\alpha \in (V \cap D) \setminus \nodes{E^*}$, and $\mu$ satisfies the constraint $\sum_{\alpha \in V \cap D} \mathcal{I}_\alpha \ge 1$. Property~(5) holds since, by the above characterization of $\nodes{E^*}$, and the constraints of the form $z_\alpha = 0$, with $\alpha \in D \setminus V$, of $T_\mathsf{bind}$, we get that $\nodes{E^*} \cap D = \{\alpha \in D \mid \mu(z_\alpha) = 1\} = \db{\mu} = D'$.

	The rest of the proof is devoted to establishing property~(4). Let $F = \{(\alpha,\beta) \in \hat{E} \mid \mu(y_{(\alpha,\beta)}) > 0\}$, and let $U$ be the set of nodes of $V$ that are \emph{not} reachable from $R(\bar t)$ in the directed graph $(V,F)$; clearly, $R(\bar t) \not\in U$. By summing up the constraints of the third family of $T_\mathsf{reach}$ over all the nodes of $U$, we get that
	\[
	\sum\limits_{\alpha \in U} \mu(\outflow{\alpha})\ =\ \sum\limits_{\alpha \in U} \mu(\inflow{\alpha}) - \sum\limits_{\alpha \in U} \mu(z_\alpha).
	\]
	By definition of $U$, there is no edge $(\beta,\alpha) \in F$ with $\beta \in V \setminus U$ and $\alpha \in U$, since otherwise $\alpha$ would be reachable from $R(\bar t)$ in $(V,F)$. Hence, letting $W$ (resp., $W'$) be the sum of $\mu(y_{(\alpha,\beta)})$ over all the edges $(\alpha,\beta) \in \hat{E}$ with $\alpha,\beta \in U$ (resp., with $\alpha \in U$ and $\beta \not\in U$), the above equality can be rewritten as $W + W' = W - \sum_{\alpha \in U} \mu(z_\alpha)$, i.e., $\sum_{\alpha \in U} \mu(z_\alpha) = -W'$. Since $W' \ge 0$ due to $T_\mathsf{pos}$, and the variables of the form $z_\alpha$ are binary, we conclude that $\mu(z_\alpha) = 0$ for each $\alpha \in U$. Therefore, every node of $\nodes{E^*}$ is reachable from $R(\bar t)$ in $(V,F)$.
	
	We now observe that $F \subseteq \hat{E}^*$. Indeed, for an edge $(\alpha,\beta) \in F$, the first family of constraints of $T_\mathsf{reach}$ gives $0 < \mu(y_{(\alpha,\beta)}) \le (|V|-1) \cdot \mu(\mathcal{C}_{(\alpha,\beta)})$, and thus $\mu(\mathcal{C}_{(\alpha,\beta)}) \ge 1$, i.e., $\card{(\alpha,T)}{E^*} \ge 1$ for some hyperedge $(\alpha,T) \in E$ with $\beta \in T$, which in turn means that $(\alpha,\beta) \in \hat{E}^*$. Consequently, every node of $\nodes{E^*}$ is reachable from $R(\bar t)$ via a directed path consisting of edges of $\hat{E}^*$. Since both endpoints of an edge of $\hat{E}^*$ belong to $\nodes{E^*}$, such a path is a path of the directed graph $(\nodes{E^*},\hat{E}^*)$, and property~(4) follows.\qed
	
\subsection{Proof of Theorem~\ref{thm:un-hyper-correct}}
	Let $\downc{R(\bar t)}{D}{\dep} = (V,E)$, and let $T$ be the set of constraints $\ilpunhypergraph{\bar t}{D}{Q}$. As in the proof of Theorem~\ref{thm:hyper-correct}, we can safely assume that the database $D'$ over $\esch{\dep}$ under consideration is a subset of $D$, and, by Theorem~\ref{thm:arbitrary-hypergraph}, it suffices to show that the following statements are equivalent:
	\begin{enumerate}
		\item There is a description $E^*$ of an unambiguous proof tree of $R(\bar t)$ w.r.t.~$D$ and $\dep$ with support $D'$.
		\item There is $\mu \in \sol{T}$ with $\db{\mu} = D'$.
	\end{enumerate}
	
	Recall that $T = \ilphypergraph{\bar t}{D}{Q} \cup T_\mathsf{UN}$, and that the variables of $T$ are those of $\ilphypergraph{\bar t}{D}{Q}$ together with the binary variables of the form $\hat{x}_{(\alpha,T)}$, with $(\alpha,T) \in E$, which occur only in $T_\mathsf{UN}$. Consider an assignment $\mu$ for the variables of $T$ that satisfies $T_\mathsf{pos}$, and let $E^*$ be the multiset of hyperedges with $\card{(\alpha,T)}{E^*} = \mu(x_{(\alpha,T)})$, for each $(\alpha,T) \in E$. We claim that $\mu$ satisfies $T_\mathsf{UN}$ iff the following two conditions hold:
	\begin{enumerate}
		\item[(i)] for each $(\alpha,T) \in E$, $\mu(\hat{x}_{(\alpha,T)}) = 1$ iff $(\alpha,T) \in E^*$, and
		\item[(ii)] for every pair $(\alpha,T),(\alpha,T') \in E^*$, it holds that $T = T'$.
	\end{enumerate}
	Indeed, since $\card{(\alpha,T)}{E^*} = \mu(x_{(\alpha,T)}) \ge 1$ iff $(\alpha,T) \in E^*$, the first family of constraints of $T_\mathsf{UN}$ states precisely condition~(i). Further, assuming condition~(i), the second family of constraints of $T_\mathsf{UN}$ states that, for each $\alpha \in V$, at most one hyperedge with source $\alpha$ occurs in $E^*$, which is precisely condition~(ii). The claim implies that an assignment for the variables of $\ilphypergraph{\bar t}{D}{Q}$ can be extended to a solution of $T$ iff it is a solution of $\ilphypergraph{\bar t}{D}{Q}$ and the induced multiset $E^*$ enjoys condition~(ii). Note also that such an extension does not affect the values of the variables $z_{\alpha_1},\ldots,z_{\alpha_n}$.
	
	The desired equivalence now follows from the proof of Theorem~\ref{thm:hyper-correct}. In particular, for $(1) \Rightarrow (2)$, let $E^*$ be a description of an unambiguous proof tree of $R(\bar t)$ w.r.t.~$D$ and $\dep$ with support $D'$. Since $E^*$ is, in particular, a description of a proof tree of $R(\bar t)$ w.r.t.~$D$ and $\dep$ with support $D'$, the assignment $\mu$ constructed in the direction $(1) \Rightarrow (2)$ of the proof of Theorem~\ref{thm:hyper-correct} is a solution of $\ilphypergraph{\bar t}{D}{Q}$ with $\db{\mu} = D'$ that induces the multiset $E^*$. By Definition~\ref{def:hypergraph-properties}, $E^*$ enjoys condition~(ii), and thus, by the above claim, $\mu$ can be extended to a solution of $T$ with $\db{\mu} = D'$.
	For $(2) \Rightarrow (1)$, let $\mu \in \sol{T}$ with $\db{\mu} = D'$. Since $T \supseteq \ilphypergraph{\bar t}{D}{Q}$, $\mu$ is also a solution of $\ilphypergraph{\bar t}{D}{Q}$, and thus, by the direction $(2) \Rightarrow (1)$ of the proof of Theorem~\ref{thm:hyper-correct}, the multiset $E^*$ induced by $\mu$ is a description of a proof tree of $R(\bar t)$ w.r.t.~$D$ and $\dep$ with support $D'$. Moreover, by the above claim, $E^*$ enjoys condition~(ii); thus, $E^*$ is a description of an unambiguous proof tree of $R(\bar t)$ w.r.t.~$D$ and $\dep$ with support $D'$, as needed.\qed

\subsection{Proof of Theorem~\ref{thm:dag-correct}}
	Let $\downc{R(\bar t)}{D}{\dep} = (V,E)$, let $T$ be the set $\ilpproofdag{\bar t}{D}{Q}$. Recall that $\rho$ denotes the amplified copy $R(\bar t)^1$.
	Recall further that $T_\mathsf{acyclic}$ is a system of constraints over the binary variables of $B_E$, together with auxiliary binary variables that do not occur in the other constraints of $T$, such that an assignment $\nu$ for $B_E$ can be extended to a solution of $T_\mathsf{acyclic}$ iff the directed graph whose edges are $\{(\alpha^i,\beta^j) \mid \nu(z_{(\alpha^i,\beta^j)}) = 1\}$ is acyclic; this is precisely the guarantee provided by the vertex elimination encoding of~\cite{RankoohR22}.
	As in the proof of Theorem~\ref{thm:hyper-correct}, we can safely assume that the database $D'$ over $\esch{\dep}$ under consideration is a subset of $D$.
	Finally, let us discuss two simple observations that will be used throughout the proof: (i) every fact of $V \setminus D$ is intensional since, by definition of $\gri{D}{\dep}$, each of its nodes is either a fact of $D$ or a fact derived via a rule of $\dep$, and (ii) no fact of $D$ occurs as the source of a hyperedge of $\downc{R(\bar t)}{D}{\dep}$ since $\esch{\dep} \cap \isch{\dep} = \emptyset$, which implies that $\hat{E}$ contains no edge of the form $(\alpha,\beta)$ with $\alpha \in D$.
	
	We start by showing that Theorem~\ref{thm:arbitrary-proof-dag} can be strengthened by additionally requiring that the proof DAG in question has at most one node labeled with each database fact.
	
	\begin{lemma}\label{lem:auxiliary-1}
		The following are equivalent:
		\begin{enumerate}
			\item $D' \in \allwhy{\bar t}{D}{Q}$. 
			\item There is a proof DAG $G = (V_G,E_G,\lambda)$ of $R(\bar t)$ w.r.t.~$D$ and $\dep$ with support $D'$ that enjoys conditions (a), (b, (c) of item (2) of Theorem~\ref{thm:arbitrary-proof-dag}, as well as the following condition: for each fact $\alpha \in D$, $|\{v \in V_G \mid \lambda(v) = \alpha\}| \le 1$.
		\end{enumerate}
	\end{lemma}
	
	%\medskip
	%\noindent
	%\textbf{Claim 1.} $D' \in \allwhy{\bar t}{D}{Q}$ if and only if there is a proof DAG $G = (V_G,E_G,\lambda)$ of $R(\bar t)$ w.r.t.~$D$ and $\dep$ with support $D'$ that enjoys conditions (a), (b), and (c) of item (2) of Theorem~\ref{thm:arbitrary-proof-dag}, as well as the condition
	%\begin{itemize}
	%	\item[(d)] $|\{v \in V_G \mid \lambda(v) = \alpha\}| \le 1$, for each fact $\alpha \in D$.
	%\end{itemize}
	
	%\smallskip
	%\noindent
	%\emph{Proof of Claim 1.} 
	\begin{proof}[Proof of Lemma~\ref{lem:auxiliary-1}]
		The direction $(2) \Rightarrow (1)$ is an immediate consequence of Theorem~\ref{thm:arbitrary-proof-dag}. 
		We now focus on direction $(1) \Rightarrow (2)$. For brevity, we will refer to the additional condition given in item (2) of Lemma~\ref{lem:auxiliary-1} as condition (d). 
		Clearly, Theorem~\ref{thm:arbitrary-proof-dag} provides a proof DAG $G = (V_G,E_G,\lambda)$ of $R(\bar t)$ w.r.t.~$D$ and $\dep$ with support $D'$ that enjoys conditions (a), (b), and (c). We proceed to convert $G$ into a proof DAG $G' = (V_{G'},E_{G'},\lambda')$ of $R(\bar t)$ w.r.t.~$D$ and $\dep$ with support $D'$ that enjoys all the conditions (a), (b), (c), and (d).

		We first observe that every node $v \in V_G$ with $\lambda(v) \in D$ is a leaf of $G$; if $v$ had outgoing edges, then, by item (3) of Definition~\ref{def:proof-dag}, the predicate of $\lambda(v)$ would 
		%occur in the head of a rule of $\dep$, i.e., it would 
		be intensional, which contradicts the fact that $\lambda(v) \in D$.
		Let $\sim$ be the equivalence relation over $V_G$ such that $u \sim v$ iff either $u = v$ or $\lambda(u) = \lambda(v) \in D$; let $[v]$ be the equivalence class of a node $v \in V_G$. We define the labeled directed graph $G' = (V_{G'},E_{G'},\lambda')$, where
		\[
		V_{G'} = \{[v] \mid v \in V_G\} \qquad E_{G'} = \{([u],[v]) \mid (u,v) \in E_G\},
		\]
		and $\lambda'([v]) = \lambda(v)$, which is well-defined since equivalent nodes carry the same label. In other words, $G'$ is obtained from $G$ by merging, for each fact $\alpha \in D$, all the nodes labeled with $\alpha$ into a single node. We proceed to first show that $G'$ is indeed a proof DAG of $R(\bar t)$ w.r.t.~$D$ and $\dep$ with support $D'$, and then argue that it enjoys the desired conditions (a), (b), (c), and (d).
			
			\medskip
			\noindent \underline{$G'$ is a proof DAG of $R(\bar t)$ w.r.t.~$D$ and $\dep$ with support $D'$}
			\smallskip
			
			\noindent Since every node $v \in V_G$ with $|[v]| > 1$ is labeled with a fact of $D$, and thus it is a leaf of $G$, we get that
			\begin{itemize}
				\item[$(\dagger)$] for every node $v \in V_G$ with at least one outgoing edge, $[v] = \{v\}$, and the outgoing edges of $[v]$ in $G'$ are precisely the edges $([v],[u])$ with $(v,u) \in E_G$.
			\end{itemize}
			Hence, a node $[v]$ of $G'$ is a leaf of $G'$ iff all the nodes of $[v]$ are leaves of $G$, which implies that $\lambda'$ maps every leaf of $G'$ to a fact of $D$, i.e., item (2) of Definition~\ref{def:proof-dag} holds, and that $\support{G'} = \support{G} = D'$.
			Moreover, $G'$ is acyclic: by $(\dagger)$, a cycle of $G'$ can mention only singleton classes, and thus, it would give rise to a cycle of $G$, which is not possible; in particular, $G'$ has no self-loops.

			Let now $r$ be the root of $G$. We have that $[r] = \{r\}$: if $\lambda(r) = R(\bar t)$ and $R \in \isch{\dep}$, then $[r] = \{r\}$ by definition of the equivalence relation $\sim$, while if $\lambda(r) = R(\bar t)$ and $R \in \esch{\dep}$, then $G$ mentions only the node $r$, by definition of proof DAG, and thus $[r] = \{r\}$ as well. We show that $[r]$ is the only node of $G'$ without incoming edges, and thus, $G'$ is a rooted DAG with root $[r]$ such that $\lambda'([r]) = R(\bar t)$, i.e., item (1) of Definition~\ref{def:proof-dag} holds. Assume first that $([u],[r]) \in E_{G'}$; hence, there is $(u',v') \in E_G$ with $[u'] = [u]$ and $[v'] = [r] = \{r\}$, i.e., $v' = r$, which contradicts the fact that $r$ has no incoming edges in $G$. Consider now a node $[v] \neq [r]$ of $G'$, and let $u \in [v]$. Since $[v] \neq [r]$, $u \neq r$, and thus, as $G$ is rooted, there is an edge $(w,u) \in E_G$, which implies that $([w],[v]) \in E_{G'}$ is an incoming edge of $[v]$; note that $[w] \neq [v]$ since $G'$ has no self-loops.
			It remains to argue for item (3) of Definition~\ref{def:proof-dag}.
			
			Consider a node $[v]$ of $G'$ with at least one outgoing edge. By $(\dagger)$, $[v] = \{v\}$ and $v$ has outgoing edges $(v,u_1),\ldots,(v,u_n)$ in $G$, for some $n \ge 1$. Since $G$ is a proof DAG, there are a rule $R_0(\bar x_0)\ \assign\ R_1(\bar x_1),\ldots,R_n(\bar x_n) \in \dep$ and a function $h : \bigcup_{\ell \in [n]} \bar x_\ell \ra \ins{C}$ such that $\lambda(v) = R_0(h(\bar x_0))$, and $\lambda(u_\ell) = R_\ell(h(\bar x_\ell))$ for each $\ell \in [n]$. Since rule-bodies are self-join-free, the predicates $R_1,\ldots,R_n$ are pairwise distinct, and thus the facts $\lambda(u_1),\ldots,\lambda(u_n)$ are pairwise distinct, which in turn implies that the nodes $[u_1],\ldots,[u_n]$ of $G'$ are pairwise distinct. Therefore, by $(\dagger)$, the outgoing edges of $[v]$ in $G'$ are precisely $([v],[u_1]),\ldots,([v],[u_n])$, and the above rule together with $h$ witness item (3) of Definition~\ref{def:proof-dag} for $[v]$.
			
			\medskip
			\noindent \underline{$G'$ enjoys the desired conditions}
			\smallskip
			
			\noindent We now argue that the proof DAG $G'$ of $R(\bar t)$ w.r.t.~$D$ and $\dep$ devised above enjoys the conditions (a), (b), (c), and (d):
			\begin{itemize}
				\item[-] It enjoys (a) since $\lambda'(V_{G'}) = \lambda(V_G)$. 
				
				\item[-] It enjoys (b) since 
				$(\lambda'([v]),\{\lambda'([u_i])\}_{i \in [n]})$ coincides with $(\lambda(v),\{\lambda(u_i)\}_{}i \in [n])
				$ and $G$ also enjoys (b).
				
				\item[-] It enjoys (c) since $|\{[v] \in V_{G'} \mid \lambda'([v]) = \alpha\}| \le |\{v \in V_G \mid \lambda(v) = \alpha\}|$, for every fact $\alpha$.
				
				\item[-] Finally, condition (d) holds by construction.
			\end{itemize}
			This completes the proof of Lemma~\ref{lem:auxiliary-1}.  
	\end{proof}

	%\medskip
	
	Having Lemma~\ref{lem:auxiliary-1} in place, to complete the proof of Theorem~\ref{thm:dag-correct}, it suffices to show the following technical lemma:
	
	\begin{lemma}\label{lem:auxiliary-2}
		For a database $D' \subseteq D$, the following are equivalent:
		\begin{enumerate}
			\item There is a proof DAG $G$ of $R(\bar t)$ w.r.t.~$D$ and $\dep$ with support $D'$ that enjoys conditions (a), (b), (c), and (d).
			\item There is $\mu \in \sol{T}$ with $\db{\mu} = D'$.
		\end{enumerate}
	\end{lemma}    
	
	\begin{proof}[Proof of Lemma~\ref{lem:auxiliary-2}]
		We start with the direction $(1)\Rightarrow (2)$.
		By hypothesis, there is a proof DAG $G =(V_G,E_G,\lambda)$ of $R(\bar t)$ w.r.t.~$D$ and $\dep$ with support $D'$ that enjoys conditions (a), (b), (c), and (d). Let $r$ be the root of $G$. By condition (a), $\lambda(v) \in V$ for each $v \in V_G$. Hence, we can fix an injective function $\iota$ that maps each node $v \in V_G$ to an amplified copy $\lambda(v)^i$ of $\lambda(v)$ in such a way that $\iota(r) = \rho$. Such a function exists: for a node $v$ with $\lambda(v) \in D$ we are forced to let $\iota(v) = \lambda(v)^1$, since $\lambda(v)$ admits a single amplified copy, and this is compatible with injectivity because of condition (d); for the remaining nodes, injectivity is guaranteed by condition (c), and we can always exchange the indices assigned to $r$ and to the node of $G$ mapped to $R(\bar t)^1$, if such a node exists, since $\lambda(r) = R(\bar t) \not\in D$.
		We define the assignment $\mu$ for the variables of $T$ as follows:
		\begin{itemize}
			\item[-] $\mu(x_{\alpha^i}) = 1$ iff $\alpha^i \in \iota(V_G)$,
			\item[-] $\mu(z_{(\alpha^i,\beta^j)}) = 1$ iff $\alpha^i,\beta^j \in \iota(V_G)$ and $(\iota^{-1}(\alpha^i),\iota^{-1}(\beta^j)) \in E_G$, and
			\item[-] $\mu(y_{(\alpha^i,T')}) = 1$ iff there is $v \in V_G$ with $\iota(v) = \alpha^i$ and $T' = \{\iota(u) \mid (v,u) \in E_G\}$.
		\end{itemize}
		Note that $\mu$ is well-defined since, for every $v \in V_G$ with outgoing edges $(v,u_1),\ldots,(v,u_n)$, condition (b) implies that the hyperedge $(\lambda(v),\{\lambda(u_1),\ldots,\lambda(u_n)\}) \in E$, and thus $(\iota(v),\{\iota(u_1),\ldots,\iota(u_n)\})$ is an amplified copy of a hyperedge of $\downc{R(\bar t)}{D}{\dep}$, i.e., it holds that $y_{(\iota(v),\{\iota(u_1),\ldots,\iota(u_n)\})} \in B_H$. Moreover, since the directed graph whose edges are the pairs $(\alpha^i,\beta^j)$ with $\mu(z_{(\alpha^i,\beta^j)}) = 1$ is isomorphic, via $\iota$, to the DAG $G$, it is acyclic, and thus $\mu$ can be extended to the auxiliary variables of $T_\mathsf{acyclic}$ in such a way that all the constraints of $T_\mathsf{acyclic}$ are satisfied.
		
		We now show that $\mu$ satisfies the remaining constraints of $T$. The constraints of $T_\mathsf{graph}$ are satisfied since, by construction, $\mu(z_{(\alpha^i,\beta^j)}) = 1$ implies $\alpha^i,\beta^j \in \iota(V_G)$, i.e., $\mu(x_{\alpha^i}) = \mu(x_{\beta^j}) = 1$.
		Concerning $T_\mathsf{root}$, $\mu(x_\rho) = 1$ since $\iota(r) = \rho$. Moreover, for every $z_{(\alpha^i,\rho)} \in B_E$, $\mu(z_{(\alpha^i,\rho)}) = 0$; otherwise, $r = \iota^{-1}(\rho)$ would have an incoming edge in $G$, which contradicts the fact that $r$ is the root of $G$. Finally, for every $x_{\alpha^i} \in B_N$ with $\alpha \in V$ and $\alpha^i \neq \rho$ such that $\mu(x_{\alpha^i}) = 1$, the node $v = \iota^{-1}(\alpha^i)$ is different from $r$, and thus, since $G$ is rooted, there is an edge $(u,v) \in E_G$, which implies that $\mu(z_{(\iota(u),\alpha^i)}) = 1$; hence, the third family of constraints of $T_\mathsf{root}$ is satisfied.
		Concerning $T_\mathsf{proof}$, its first family of constraints is satisfied since, by condition (a), none of the nodes of $G$ is labeled with a fact of $D \setminus V$. For the second family, consider $x_{\alpha^i} \in B_N$ with $\alpha \in V$ being intensional and $\mu(x_{\alpha^i}) = 1$, and let $v = \iota^{-1}(\alpha^i)$. Since $\lambda(v) = \alpha \not\in D$, $v$ is not a leaf of $G$, and thus, by construction, $\mu(y_{(\alpha^i,T')}) = 1$ with $T' = \{\iota(u) \mid (v,u) \in E_G\}$. For the third family, consider $y_e \in B_H$, with $e = (\alpha^i,T)$, such that $\mu(y_e) = 1$; hence, $T = \{\iota(u) \mid (v,u) \in E_G\}$, where $v = \iota^{-1}(\alpha^i)$. By construction, for every $z_{(\alpha^i,\beta^j)} \in B_E$, $\mu(z_{(\alpha^i,\beta^j)}) = 1$ iff $\beta^j \in T$, and thus, $\mu(\ell_{e,\beta^j}) = 1$ for every $z_{(\alpha^i,\beta^j)} \in B_E$, as needed. If instead $\mu(y_e) = 0$, then the corresponding constraints are trivially satisfied since $\mu(\ell_{e,\beta^j}) \ge 0$.
		From the above discussion, $\mu \in \sol{T}$. 
		
		To conclude the proof of $(1)\Rightarrow (2)$, it remains to observe that $\db{\mu} = \support{G}$. As shown in the proof of Lemma~\ref{lem:auxiliary-1}, a node $v$ of $G$ is labeled with a fact of $D$ iff $v$ is a leaf of $G$. Hence,
		\begin{eqnarray*}
			\db{\mu} &=& \{\alpha \in D \mid \mu(x_{\alpha^1}) = 1\}\\
			&=& \{\lambda(v) \mid v \in V_G \text{ and } \lambda(v) \in D\}\\ 
			&=& \support{G},
		\end{eqnarray*}
		and direction $(1)\Rightarrow (2)$ follows.

		\medskip
		
		We now establish the direction $(2)\Rightarrow (1)$. By hypothesis, there is $\mu \in \sol{T}$ with $\db{\mu} = D'$. We define the labeled directed graph $G_\mu = (V_\mu,E_\mu,\lambda_\mu)$, where
			\begin{eqnarray*}
				V_\mu &=& \{\alpha^i \mid x_{\alpha^i} \in B_N \text{ and } \mu(x_{\alpha^i}) = 1\},\\
				E_\mu &=& \{(\alpha^i,\beta^j) \mid z_{(\alpha^i,\beta^j)} \in B_E \text{ and } \mu(z_{(\alpha^i,\beta^j)}) = 1\},
			\end{eqnarray*}
			and $\lambda_\mu(\alpha^i) = \alpha$ for each $\alpha^i \in V_\mu$.
			
			We proceed to first show that $G_\mu$ is indeed a proof DAG of $R(\bar t)$ w.r.t.~$D$ and $\dep$ with support $\db{\mu}$, and then argue that it enjoys the desired conditions (a), (b), (c), and (d).
			
			\medskip
			\noindent \underline{$G_\mu$ is a proof DAG of $R(\bar t)$ w.r.t.~$D$ and $\dep$ with support $\db{\mu}$}
			\smallskip
			
			\noindent By the constraints of $T_\mathsf{graph}$, $E_\mu \subseteq V_\mu \times V_\mu$, i.e., $G_\mu$ is indeed a (finite) labeled directed graph, and, by the constraints of $T_\mathsf{acyclic}$, it is acyclic.
			By the constraints of $T_\mathsf{root}$, $\rho \in V_\mu$ and $\rho$ has no incoming edges in $G_\mu$, while every node $\alpha^i \in V_\mu$ with $\alpha^i \neq \rho$ has at least one incoming edge in $G_\mu$; note that the third family of constraints of $T_\mathsf{root}$ indeed applies to all the nodes of $V_\mu$ different from $\rho$ since they are all of the form $\alpha^i$ with $\alpha \in V$. Therefore, $G_\mu$ is a finite rooted DAG with root $\rho$, and $\lambda_\mu(\rho) = R(\bar t)$, i.e., item (1) of Definition~\ref{def:proof-dag} holds.
			
			We proceed with the following auxiliary statement:
			\begin{itemize}
				\item[$(\ddagger)$] for every node $\alpha^i \in V_\mu$ with $\alpha$ intensional, there are a hyperedge $(\alpha,T) \in E$ and an amplified copy $(\alpha^i,T')$ of $(\alpha,T)$ such that the outgoing edges of $\alpha^i$ in $G_\mu$ are precisely the edges $(\alpha^i,\beta^j)$ with $\beta^j \in T'$.
			\end{itemize}
			Since $\mu(x_{\alpha^i}) = 1$, with $\alpha \in V$ intensional, the second family of constraints of $T_\mathsf{proof}$ implies that $\mu(y_{(\alpha^i,T')}) = 1$ for some amplified copy $(\alpha^i,T')$ of a hyperedge $(\alpha,T) \in E$. Then, the third family of constraints of $T_\mathsf{proof}$ implies that, for every $z_{(\alpha^i,\beta^j)} \in B_E$, $\mu(z_{(\alpha^i,\beta^j)}) = 1$ if $\beta^j \in T'$, and $\mu(z_{(\alpha^i,\beta^j)}) = 0$ otherwise; since $\beta^j \in T'$ implies $(\alpha,\beta) \in \hat{E}$, and thus $z_{(\alpha^i,\beta^j)} \in B_E$, statement $(\ddagger)$ follows.
			
			Consider now a leaf $\alpha^i$ of $G_\mu$. If $\alpha$ were intensional, then, by $(\ddagger)$, $\alpha^i$ would have at least one outgoing edge, since the target of a hyperedge is non-empty. Hence, $\alpha \in D$, as every fact of $V \setminus D$ is intensional, and item (2) of Definition~\ref{def:proof-dag} holds. Conversely, every node $\alpha^i \in V_\mu$ with $\alpha \in D$ is a leaf of $G_\mu$ since $\hat{E}$, and thus $B_E$, contains no edge with $\alpha^i$ as the source. Summing up, the leaves of $G_\mu$ are precisely the nodes of the form $\alpha^1$ with $\alpha \in D$ and $\mu(x_{\alpha^1}) = 1$, and therefore $\support{G_\mu} = \{\alpha \in D \mid \mu(x_{\alpha^1}) = 1\} = \db{\mu}$.
			
			It remains to establish item (3) of Definition~\ref{def:proof-dag}. Consider a node $\alpha^i \in V_\mu$ with at least one outgoing edge; as discussed above, $\alpha$ is intensional. Let $(\alpha,T) \in E$ and $(\alpha^i,T')$ be as in $(\ddagger)$. By definition of $\gri{D}{\dep}$, there are a rule $R_0(\bar x_0)\ \assign\ R_1(\bar x_1),\ldots,R_n(\bar x_n) \in \dep$ and a function $h : \bigcup_{\ell \in [n]} \bar x_\ell \ra \ins{C}$ such that $\alpha = R_0(h(\bar x_0))$ and $T = \{R_1(h(\bar x_1)),\ldots,R_n(h(\bar x_n))\}$. Since rule-bodies are self-join-free, the facts $\beta_\ell = R_\ell(h(\bar x_\ell))$, for $\ell \in [n]$, are pairwise distinct, and thus $T = \{\beta_1,\ldots,\beta_n\}$ with $|T| = n$, and $T' = \{\beta^{k_1}_1,\ldots,\beta^{k_n}_n\}$ with $|T'| = n$. Therefore, by $(\ddagger)$, the node $\alpha^i$ has exactly $n$ outgoing edges, namely $(\alpha^i,\beta^{k_1}_1),\ldots,(\alpha^i,\beta^{k_n}_n)$ and, since $\lambda_\mu(\alpha^i) = \alpha = R_0(h(\bar x_0))$ and $\lambda_\mu(\beta^{k_\ell}_\ell) = \beta_\ell = R_\ell(h(\bar x_\ell))$ for each $\ell \in [n]$, the above rule together with $h$ witness item (3) of Definition~\ref{def:proof-dag}.
			
			\medskip
			\noindent \underline{$G_\mu$ enjoys the desired conditions}
			\smallskip
			
			\noindent By the first family of constraints of $T_\mathsf{proof}$, every node of $V_\mu$ is of the form $\alpha^i$ with $\alpha \in V$; hence, condition (a) holds. Condition (c) holds since a fact has at most $\boundall{R(\bar t)}{D}{\dep}$ amplified copies, whereas condition (d) holds since a fact of $D$ has a single amplified copy, and thus at most one node of $G_\mu$ can be labeled with it. Condition (b) follows from the fact that $G_\mu$ is a proof DAG of $R(\bar t)$ w.r.t.\ $D$ and $\dep$, as established above, and the fact that all its labels are facts of $\downc{\bar t}{D}{Q}$ by construction. This completes the proof of $(2)\Rightarrow (1)$.
	\end{proof}
	Theorem~\ref{thm:dag-correct} follows from Lemmas~\ref{lem:auxiliary-1} and~\ref{lem:auxiliary-2}.\qed
	
\subsection{Proof of Theorem~\ref{thm:unambiguous-dag-correct}}
		Let $\downc{R(\bar t)}{D}{\dep} = (V,E)$, and let $T$ be the set of constraints $\ilpunproofdag{\bar t}{D}{Q}$. Recall that $T$ is obtained via the same construction that gives $\ilpproofdag{\bar t}{D}{Q}$, with the only difference that, since $\boundun{R(\bar t)}{D}{\dep} = 1$, \emph{every} fact $\alpha$, and not only a database fact, has a single amplified copy, namely $\alpha^1$; in particular, $\rho = R(\bar t)^1$ is the only amplified copy of $R(\bar t)$. As in the proof of Theorem~\ref{thm:dag-correct}, we can safely assume that the database $D'$ over $\esch{\dep}$ under consideration is a subset of $D$.
		
		By Theorem~\ref{thm:arbitrary-proof-dag}, $D' \in \unwhy{\bar t}{D}{Q}$ iff there is an unambiguous proof DAG $G$ of $R(\bar t)$ w.r.t.~$D$ and $\dep$ with support $D'$ that enjoys conditions (a), (b), and (c) of item (2) of Theorem~\ref{thm:arbitrary-proof-dag}, where the bound of condition (c) is $\boundun{R(\bar t)}{D}{\dep} = 1$. Note that, in this case, condition (c) states that no two distinct nodes of $G$ have the same label, which is precisely the additional condition in the definition of unambiguous proof DAG (see Definition~\ref{def:proof-dag}); in particular, condition (c) subsumes condition (d) of Lemma~\ref{lem:auxiliary-1} in the proof of Theorem~\ref{thm:dag-correct}, and thus, no analogue of that lemma is needed here.
		Hence, it suffices to show that the following holds:
		\begin{lemma}\label{lem:equiv-unambiguous-graph}
		The following are equivalent:
		\begin{enumerate}
			\item There is an unambiguous proof DAG $G$ of $R(\bar t)$ w.r.t.~$D$ and $\dep$ with support $D'$ that enjoys conditions (a), (b), and (c) with bound $1$.
			
			\item There is $\mu \in \sol{T}$ with $\db{\mu} = D'$.
		\end{enumerate}
		\end{lemma}
		\begin{proof}[Proof of Lemma~\ref{lem:equiv-unambiguous-graph}]
		This is shown as Lemma~\ref{lem:auxiliary-2} used in the proof of Theorem~\ref{thm:dag-correct}; we only need to discuss how the proof of Lemma~\ref{lem:auxiliary-2} specializes to the case where each fact has a single amplified copy.
		
		\medskip
		\noindent
		$(1) \Rightarrow (2)$. By hypothesis, there is an unambiguous proof DAG $G = (V_G,E_G,\lambda)$ of $R(\bar t)$ w.r.t.~$D$ and $\dep$ with support $D'$ that enjoys conditions (a), (b), and (c) with bound $1$. Let $r$ be the root of $G$, and let $\iota$ be the function that maps each node $v \in V_G$ to the amplified copy $\lambda(v)^1$. Note that $\iota$ is the function used in the proof of Lemma~\ref{lem:auxiliary-2}; it is injective since $G$ is unambiguous, it trivially maps every node labeled with a fact of $D$ to the amplified copy with index $1$, and $\iota(r) = R(\bar t)^1 = \rho$ since $\lambda(r) = R(\bar t)$. The rest of the proof of $(1) \Rightarrow (2)$ of Lemma~\ref{lem:auxiliary-2} goes through verbatim, and provides an assignment $\mu \in \sol{T}$ with $\db{\mu} = \support{G} = D'$, as needed.
		
		\medskip
		\noindent
		$(2) \Rightarrow (1)$.
		By hypothesis, there is $\mu \in \sol{T}$ with $\db{\mu} = D'$. Let $G_\mu$ be the labeled directed graph defined in the proof of $(2) \Rightarrow (1)$ of Lemma~\ref{lem:auxiliary-2}. 
		The proof of $(2) \Rightarrow (1)$ of Lemma~\ref{lem:auxiliary-2} goes through verbatim, and thus, $G_\mu$ is a proof DAG of $R(\bar t)$ w.r.t.~$D$ and $\dep$ with support $\db{\mu} = D'$ that enjoys conditions (a), (b), and (c). Moreover, since every node of $G_\mu$ is of the form $\alpha^1$ and $\lambda_\mu(\alpha^1) = \alpha$, no two distinct nodes of $G_\mu$ have the same label, i.e., $G_\mu$ is an unambiguous proof DAG, and the claim follows.
		\end{proof}
		
		Theorem~\ref{thm:unambiguous-dag-correct} follows from Lemma~\ref{lem:equiv-unambiguous-graph}.\qed

\section{Proofs of Section~\ref{sec:approximations}}

\subsection{Proof of Theorem~\ref{thm:approximation-quality}}
	Consider a Datalog query $Q = (\dep,R)$, a database $D = \{\alpha_1,\ldots,\alpha_n\}$ over $\esch{\dep}$, a tuple $\bar t \in \adom{D}^{\arity{R}}$, and an integer $k > 0$, and consider the execution of $\mathsf{Approx}[\ConstructILPEncoding]$ on input $Q$, $D$, $\bar t$, and $k$. For $i \in [k]$, we write $T_i$ and $f_i$ for the system of constraints and the objective function, respectively, at the beginning of the $i$-th iteration of the for loop, assuming that the algorithm reaches that iteration, and $\mu_i$ for the assignment computed during the $i$-th iteration, assuming it exists; recall that $D_i = \db{\mu_i}$. Note that $T_1 = \ConstructILPEncoding(\bar t,D,Q)$ is an ILP encoding of $\starwhy{\bar t}{D}{Q}$ with database variables $x_{\alpha_1},\ldots,x_{\alpha_n}$, and $f_1 = 0$.
	
	We start with a simple observation. Recall that, for an assignment $\mu$ for the variables of $T_1$ and $\ell \in [n]$, $g_{\mu,\ell}$ denotes the expression $1 - x_{\alpha_\ell}$ if $\mu(x_{\alpha_\ell}) = 1$, and the expression $x_{\alpha_\ell}$ if $\mu(x_{\alpha_\ell}) = 0$. Hence, for every assignment $\nu$, $\nu(g_{\mu,\ell}) = 1$ if $\nu(x_{\alpha_\ell}) \neq \mu(x_{\alpha_\ell})$, and $\nu(g_{\mu,\ell}) = 0$ otherwise, i.e., $\nu(g_{\mu,\ell}) = 1$ precisely when $\alpha_\ell \in \diff{\db{\nu}}{\db{\mu}}$. Therefore, we get the following equation, which we dub ($\lozenge$):	%\begin{equation}\label{eq:marginal}
	\[
	\nu\left(\sum\limits^n_{\ell = 1} g_{\mu,\ell}\right)\ =\ |\diff{\db{\nu}}{\db{\mu}}|\ =\ \dist{\db{\nu}}{\db{\mu}}.
	\]
	%\end{equation}
	%
	%	Let us also observe that $\sum^n_{\ell=1} g_{\mu,\ell}$ is a linear expression of the form $c_1 z_1 + \cdots + c_m z_m + c$, for some constant $c$; since an additive constant affects neither the set of solutions of a constraint, as it can be moved to the right-hand side, nor the set of maximizers of an objective function, such expressions can be safely used in the algorithm.
	
	We now prove our theorem by first proving two auxiliary results. The first result establishes a loop invariant of the algorithm.
	
	%\medskip
	%\noindent
	%\textbf{Claim 1.} 
	
	\begin{lemma}\label{lem:apx-1}
		For each $i \in [k]$ such that the algorithm reaches the $i$-th iteration, the following hold:
		\begin{enumerate}
			\item $\left\{\db{\mu} \mid \mu \in \sol{T_i}\right\} = \starwhy{\bar t}{D}{Q} \setminus \{D_1,\ldots,D_{i-1}\}$,
			\item $\mu(f_i) = \sum_{j < i} \dist{\db{\mu}}{D_j}$, for every $\mu \in \sol{T_i}$.
		\end{enumerate}
	\end{lemma}
	
	%\smallskip
	%\noindent
	%\emph{Proof of Claim 1.} 
	\begin{proof}[Proof of Lemma~\ref{lem:apx-1}]
		We proceed by induction on $i$. For $i=1$, since $T_1$ is an ILP encoding of $\starwhy{\bar t}{D}{Q}$, Definition~\ref{def:ilp-encoding-why} implies that $\{\db{\mu} \mid \mu \in \sol{T_1}\} = \starwhy{\bar t}{D}{Q}$ (that is, item (1)), whereas item (2) holds trivially since $f_1 = 0$.

		Assume now that the claim holds for $i \in [k-1]$, and the algorithm reaches the $(i+1)$-th iteration. By construction, $T_{i+1} = T_i \cup \{\sum^n_{\ell=1} g_{\mu_i,\ell} \ge 1\}$ and $f_{i+1} = f_i + \sum^n_{\ell=1} g_{\mu_i,\ell}$. By ($\lozenge$), for an assignment $\mu$, $\sum^n_{\ell=1} g_{\mu_i,\ell} \ge 1$ iff $\dist{\db{\mu}}{D_i} \ge 1$ iff $\db{\mu} \neq D_i$. Hence, $\sol{T_{i+1}} = \{\mu \in \sol{T_i} \mid \db{\mu} \neq D_i\}$, which, together with the induction hypothesis, implies that $\{\db{\mu} \mid \mu \in \sol{T_{i+1}}\} = \starwhy{\bar t}{D}{Q} \setminus \{D_1,\ldots,D_i\}$, and item (1) follows. 
		Now, for every $\mu \in \sol{T_{i+1}}$, by the induction hypothesis and ($\lozenge$), $\mu(f_{i+1}) = \mu(f_i) + \dist{\db{\mu}}{D_i} = \sum_{j < i+1} \dist{\db{\mu}}{D_j}$, and item (2) follows.
	\end{proof}
	
	Observe that Lemma~\ref{lem:apx-1} implies that $D_1,\ldots,D_{i}$ are pairwise distinct members of $\starwhy{\bar t}{D}{Q}$, for every $i \in [k]$ such that the algorithm reaches the $i$-th iteration and does not fail during its execution.
	Let us also observe that, whenever $T_i$ is feasible, an optimal solution of $(f_i,T_i)$ exists. Indeed, by Lemma~\ref{lem:apx-1}, the value $\mu(f_i)$ depends only on $\db{\mu}$; hence, the set $\{\mu(f_i) \mid \mu \in \sol{T_i}\}$ is a finite non-empty set of integers, and every $\mu \in \sol{T_i}$ that attains its maximum is an optimal solution of $(f_i,T_i)$. We now move to the second auxiliary result that characterizes when the algorithm fails.

	%\medskip
	%\noindent
	%\textbf{Claim 2.} 
	\begin{lemma}\label{lem:apx-2}
		The algorithm fails iff $\starwhyk{k}{\bar t}{D}{Q} = \emptyset$.
	\end{lemma}
	
	%\smallskip
	%\noindent
	%\emph{Proof of Claim 2.} 
	\begin{proof}[Proof of Lemma~\ref{lem:apx-2}]
		We first assume that the algorithm fails. By the observation above, the algorithm fails during the $i$-th iteration, for $i \in [k]$, precisely when $T_i$ is unfeasible. By Lemma~\ref{lem:apx-1}, $T_i$ is unfeasible implies $\starwhy{\bar t}{D}{Q} \setminus \{D_1,\ldots,D_{i-1}\} = \emptyset$. Thus, since $D_1,\ldots,D_{i-1}$ are pairwise distinct members of $\starwhy{\bar t}{D}{Q}$, we get that $|\starwhy{\bar t}{D}{Q}| = i - 1 < k$, which in turn implies that $\starwhyk{k}{\bar t}{D}{Q} = \emptyset$.
		Conversely, assume that $\starwhyk{k}{\bar t}{D}{Q} = \emptyset$, and thus, $|\starwhy{\bar t}{D}{Q}| \ge k$. Then, for each $i \in [k]$, $|\starwhy{\bar t}{D}{Q} \setminus \{D_1,\ldots,D_{i-1}\}| \ge k - (i-1) \ge 1$. Hence, by Lemma~\ref{lem:apx-1}, $T_i$ is feasible and the algorithm does not fail, as needed. 
		%We first assume that the algorithm fails. By the observation above, the algorithm fails during the $i$-th iteration, for some $i \in [k]$, precisely when $T_i$ is unfeasible. By Lemma~\ref{lem:apx-1}, $T_i$ is unfeasible iff $\starwhy{\bar t}{D}{Q} \setminus \{D_1,\ldots,D_{i-1}\} = \emptyset$. The latter, since $D_1,\ldots,D_{i-1}$ are pairwise distinct members of $\starwhy{\bar t}{D}{Q}$, holds iff $|\starwhy{\bar t}{D}{Q}| = i - 1 < k$. Conversely, if $|\starwhy{\bar t}{D}{Q}| \ge k$, then, for each $i \in [k]$, $|\starwhy{\bar t}{D}{Q} \setminus \{D_1,\ldots,D_{i-1}\}| \ge k - (i-1) \ge 1$, and thus, by Claim 1, $T_i$ is feasible and the algorithm does not fail. Summing up, the algorithm fails iff $|\starwhy{\bar t}{D}{Q}| < k$, i.e., iff $\starwhyk{k}{\bar t}{D}{Q} = \emptyset$.
	\end{proof}
	
	%\medskip
	
	We now exploit Lemmas~\ref{lem:apx-1} and~\ref{lem:apx-2} to complete the proof of Theorem~\ref{thm:approximation-quality}.
	Since $\argmax_{\mathcal{S} \in \starwhyk{k}{\bar t}{D}{Q}} \diversity{\mathcal{S}}$ is non-empty precisely when $\starwhyk{k}{\bar t}{D}{Q}$ is non-empty, Lemma~\ref{lem:apx-2} tells us that the algorithm fails exactly when it is allowed to fail. Hence, for the rest of the proof, we assume that $\starwhyk{k}{\bar t}{D}{Q} \neq \emptyset$, and the algorithm outputs the set $\hat{\mathcal{S}} = \{D_1,\ldots,D_k\}$ that, by Lemma~\ref{lem:apx-1}, consists of $k$ pairwise distinct members of $\starwhy{\bar t}{D}{Q}$; in other words, $\hat{\mathcal{S}} \in \starwhyk{k}{\bar t}{D}{Q}$. It remains to show that $\diversity{\hat{\mathcal{S}}} \ge \frac{1}{2} \cdot \diversity{\mathcal{S}_\mathsf{opt}}$, where $\mathcal{S}_\mathsf{opt} \in \argmax_{\mathcal{S} \in \starwhyk{k}{\bar t}{D}{Q}} \diversity{\mathcal{S}}$.
	
	Let $G = (V,E,w)$ be the edge-weighted, complete, undirected graph with $V = \starwhy{\bar t}{D}{Q}$ and $w(D_1,D_2) = |\diff{D_1}{D_2}|$, for $D_1,D_2 \in V$, i.e., the graph obtained from $Q$, $D$, and $\bar t$ via the reduction to $\mathsf{Remote\text{-}Clique}$ discussed above; recall that $w$ satisfies the triangle inequality. By construction, $\cliques{k}{G} = \starwhyk{k}{\bar t}{D}{Q}$, and, for every $V' \in \cliques{k}{G}$, the diversity of $V'$ as a set of nodes of $G$ coincides with the diversity of $V'$ as a set of databases.
	We claim that $D_1,\ldots,D_k$ is a possible sequence of nodes selected by the $2$-approximation algorithm for $\mathsf{Remote\text{-}Clique}$ on input $G, k$.

	Consider an integer $i \in [k]$. By Lemma~\ref{lem:apx-1}, and the fact that $\mu_i$ is an optimal solution of $(f_i,T_i)$, we get that $D_i = \db{\mu_i}$ belongs to $\starwhy{\bar t}{D}{Q} \setminus \{D_1,\ldots,D_{i-1}\}$, and
	\begin{eqnarray*}
		&& \sum\limits_{j<i} \dist{D_i}{D_j}\\
		&=&\max\left\{\sum\limits_{j<i} \dist{D'}{D_j}\ \Big|\ D' \in \starwhy{\bar t}{D}{Q} \setminus \{D_1,\ldots,D_{i-1}\}\right\}.
	\end{eqnarray*}
	Here, we exploit the fact that, by Lemma~\ref{lem:apx-1}, the value $\mu(f_i)$ of the objective function depends only on the database $\db{\mu}$ induced by $\mu$. Since, for every set $\mathcal{S}$ of databases and every database $D'$,
	\[
	\sum\limits_{D'' \in \mathcal{S}} \dist{D'}{D''}\ =\ \diversity{\mathcal{S} \cup \{D'\}} - \diversity{\mathcal{S}},
	\]
	we conclude that $D_i$ is a node of $G$ that is different from the previously selected nodes $D_1,\ldots,D_{i-1}$ and maximizes the marginal diversity $\diversity{\{D_1,\ldots,D_{i}\}} - \diversity{\{D_1,\ldots,D_{i-1}\}}$, which is precisely the choice performed by the greedy algorithm at its $i$-th iteration. Note that this is trivially true for $i=1$, where $f_1 = 0$ and every node of $G$ has marginal diversity $0$.
	
	From the above discussion, $\hat{\mathcal{S}}$ is a possible output of the greedy $2$-approximation algorithm for $\mathsf{Remote\text{-}Clique}$ on input $G, k$, and thus, $\diversity{\hat{\mathcal{S}}} \ge \frac{1}{2} \cdot \diversity{V_\mathsf{opt}}$, where $V_\mathsf{opt} \in \argmax_{V' \in \cliques{k}{G}} \diversity{V'}$. Since $\cliques{k}{G} = \starwhyk{k}{\bar t}{D}{Q}$, and the two notions of diversity coincide, $\diversity{\hat{\mathcal{S}}} \ge \frac{1}{2} \cdot \diversity{\mathcal{S}_\mathsf{opt}}$, and the claim follows.\qed

\end{document}